\documentclass[aps,twocolumn,prd,superscriptaddress,floatfix,showpacs,showkeys]{revtex4-2}
\usepackage{newtxtext}
\usepackage[latin1]{inputenc}
\usepackage[T1]{fontenc}
\usepackage{lmodern}

\usepackage{empheq}
\usepackage{mathtools}
\usepackage{amsmath}
\usepackage{etoolbox}
\usepackage[mathlines]{lineno}
\numberwithin{equation}{section}

\usepackage{amscd}
\usepackage{amsmath}
\usepackage{amssymb}
\usepackage{graphicx}%
\graphicspath{{figures/}}
\usepackage{scalerel}

\usepackage{dsfont}
\usepackage{mathrsfs}
\newcommand{\G}{\mathcal{G}}
\usepackage{upgreek}
\newcommand{\sympomega}{\boldsymbol{\upomega}}

\usepackage[ruled]{algorithm2e}
\usepackage{needspace}
\usepackage{newtxtext}
\usepackage{booktabs}
\usepackage{multirow}
\usepackage{bm,comment}

\usepackage[colorlinks=true, pdfstartview=FitV, linkcolor=blue,
            citecolor=blue, urlcolor=blue]{hyperref} 
\usepackage{subcaption}
\usepackage{sidecap}
\usepackage{grffile}
\usepackage{xcolor}
\usepackage{tikz}
\usetikzlibrary{arrows.meta, positioning, shapes.geometric}
\usetikzlibrary{decorations.pathmorphing} % <--- The magic line!

\makeatletter
\renewcommand{\subsubsection}{\@startsection{subsubsection}{3}{\z@}%
  {2ex \@plus 1ex \@minus .1ex}%  <-- Space before
  {1ex \@plus .1ex}%              <-- Space after
  {\normalfont\small\bfseries}}
\makeatother
\usepackage{dsfont}
\usepackage{amsfonts, amssymb, amsmath, amsthm}
\newtheorem{thm}{Theorem}[section]
\newtheorem{prop}[thm]{Proposition}
\newtheorem{corr}[thm]{Corollary}

\def\bprop{\begin{prop}}
\def\eprop{\end{prop}}

\def\bc{\begin{corr}}
\def\ec{\end{corr}}

\def\bt{\begin{thm}}
\def\et{\end{thm}}
\usepackage[ruled]{algorithm2e}
\newtheorem{rem}{Remark}[section]
\def\br{\begin{rem}}
\def\er{\end{rem}}
\def\be{\begin{equation}}
\def\ee{\end{equation}}
\def\bes{\begin{equation*}}
\def\ees{\end{equation*}}
\def\bea{\begin{equation} \begin{aligned}}
\def\eea{\end{aligned} \end{equation}}
\def\beas{\begin{equation*} \begin{aligned}}
\def\eeas{\end{aligned} \end{equation*}}
\def\bi{\begin{itemize}}
\def\ei{\end{itemize}}

\def\d{\, \mathrm{d}}
\newcommand{\x}{\mathbf{x}}

\definecolor{rred}{rgb}{0.7,0,0.1}
\definecolor{ccyan}{rgb}{0,.5,1}
\definecolor{greenrb}{rgb}{0.2,0.6,0.2}
\definecolor{ppink}{rgb}{1,0.1,0.6}

\begin{document}

\title{A Symplectic Theory of Turbulence Closure: Hidden Reservoir Dynamics, Endogenous Stochastic Transport, and Kraichnan Dual Cascades}

\author{Micka{\"e}l D. Chekroun}
\email{mchekroun@atmos.ucla.edu}
\affiliation{Department of Atmospheric and Oceanic Sciences, University of California, Los Angeles, CA 90095-1565, USA}
\affiliation{Department of Earth and Planetary Sciences, Weizmann Institute of Science, Rehovot 76100, Israel} 

%---------------------------------------------%
\author{James C. McWilliams}
\affiliation{Department of Atmospheric and Oceanic Sciences, University of California, Los Angeles, CA 90095-1565, USA}
\affiliation{Institute of Geophysics and Planetary Physics, 
University of California, Los Angeles, CA 90095-1565, USA}
\date{\today}

\begin{abstract}
The equations of fluid motion are known; retaining their physics after removing most degrees of freedom remains a fundamental problem. We attack this difficulty through the interaction itself: closure must preserve the geometry of the dynamics it replaces. The Symplectic Geometric Closure  (SGC) carries Euler's Hamiltonian, area-preserving transport into a prognostic stochastic field theory for two-dimensional and $\beta$-plane turbulence, while advancing a broader closure principle: unresolved dynamics should enter through geometrically admissible transport.

At its core, SGC couples resolved vorticity to a hidden stochastic reservoir, conserving augmented enstrophy while permitting bidirectional energy transfer. This constraint yields stability manifested by a compact random attractor in the hyperviscous Navier--Stokes--$\beta$ core. Numerically, SGC sustains jets, vortices, and filaments in fully developed turbulence at high Reynolds number, with high fidelity to filtered DNS when the cutoff lies inside the inertial range, without scale separation, while satisfying geometric conservation laws to machine precision.

The structural reason behind success is the SGC induced Hamiltonian transport that folds, stretches, and rearranges vorticity while preserving area. Its cross-gradient coupling selects interactions through gradient misalignment: aligned gradients switch off the inner interaction, while transverse gradients activate it. Transport retains geometric selectivity and memory.

This same mechanism reaches into the central difficulty of renormalized turbulence theory. A longstanding obstacle to Eulerian closure is spurious sweeping decorrelation. The SGC's interaction vertex excludes uniform translation geometrically, permitting random Galilean compatibility. Eliminating the reservoir generates finite memory, stochastic backscatter, and a Dyson--Volterra equation for the dressed macroscopic propagator whose one-loop, line-renormalized self-energy reproduces Kraichnan-type Direct-Interaction Approximation response architecture, with fourth-order infrared suppression of sweeping modes. SGC's 
deformation-controlled memory, together with conservative triads, then yields the classical $k^{-5/3}$ inverse-energy and $k^{-3}$ forward-enstrophy cascades under standard assumptions. SGC thus provides an Eulerian realization of Kraichnan's program, bridging field-theoretic renormalization to geometrically constrained data-driven closures.

\end{abstract}

\maketitle

{\small 
\tableofcontents
}

\section{Introduction}
\label{sec:intro}

The pursuit of accurate subgrid-scale (SGS) parameterizations for two-dimensional and quasi-geostrophic turbulence remains a central challenge in geophysical fluid dynamics.  Coarse-graining the Navier--Stokes equations leaves the nonlinear advective terms unclosed, requiring a macroscopic model for the effect of truncated scales.  In two-dimensional turbulence, this closure problem is unusually constrained.  A successful SGS model must remove enstrophy cascading to small scales while remaining compatible with the inverse transfer of kinetic energy to large scales, including the emergence of coherent vortices and zonal jets on a $\beta$-plane \cite{Kraichnan1967,Batchelor1969,Rhines1975}.  It must therefore do more than dissipate unresolved variance.  It must reproduce, at the coarse-grained level, the conservation and memory structure that makes the dual cascade possible.

Classical SGS models, such as the Smagorinsky \cite{smagorinsky1963general} and Leith closures \cite{leith1968diffusion}, rely on Boussinesq-type diagnostic assumptions, prescribing unresolved stress or vorticity flux in terms of resolved gradients.  These closures are robust and useful, but their logic is essentially Markovian and sign-definite: unresolved triadic dynamics are replaced by an instantaneous drain.  This makes it difficult for them to represent the finite-memory, sign-indefinite exchanges needed for kinetic-energy backscatter and sustained inverse transfer.  Nonlinear Gradient Models reconstruct Leonard-type stresses and retain more information about local anisotropy \cite{Clark1979}, but they lack strict conservative bounds and can inject spurious energy, leading to online instabilities or numerical blow-ups \cite[e.g.][]{VremanGeurtsKuerten1997,Jakhar2024}.  Modern data-driven closures face a complementary difficulty.  Standard neural networks can fit diagnostic tendencies with high offline accuracy, but without built-in geometric constraints they often suffer climate drift and numerical instability when deployed outside their training manifolds \cite[e.g.][]{BrenowitzBretherton2019,BrenowitzEtAl2020JAS,BrenowitzEtAl2020OfflineOnline}.

These failures point to the same obstruction.  SGS closure is not only an approximation problem; it is a structure-preservation problem.  The unresolved dynamics must generate memory, stochastic backscatter and response damping, but they must do so without violating the energy--enstrophy geometry of the inviscid flow or introducing the Eulerian sweeping contamination that obstructs bare Eulerian renormalized theories.  This is precisely the problem Kraichnan's program exposed \cite{Zhou2021}.  Eulerian closures can produce memory kernels and self-energies, but the bare Eulerian response is contaminated by random sweeping.  Lagrangian-history closures were introduced to repair this defect by following fluid histories rather than fixed Eulerian observations \cite{kraichnan1965lagrangian,kraichnan1971almost}.  The present paper proposes a different resolution: retain an Eulerian field-level response, but change the interaction geometry before renormalization.  The central claim is that the Symplectic Geometric Closure proposed in this work, provides an Eulerian realization of Kraichnan's program in which finite memory, realizability, sweeping suppression and dual-cascade consistency arise from one symplectic transport architecture rather than from an appended Lagrangian-history correction.

Thus, we introduce the Symplectic Geometric Closure (SGC), a geometric framework in which unresolved degrees of freedom modify the transport geometry of the resolved flow rather than appearing as an externally prescribed SGS forcing.  The resolved vorticity $\bar\zeta$ is coupled to a prognostic hidden reservoir $r(x,y,t)$ through a symplectic interaction.  The reservoir is stochastic, but it is not an imposed random velocity field.  It is an evolving hidden sector that stores unresolved fluctuations, carries memory, and feeds delayed geometric response back into the resolved flow.  In this sense, the closure is neither a diagnostic stress model nor a black-box learned tendency.  It is a Markovian extension of the resolved dynamics whose hidden sector generates non-Markovian transport after elimination.

The first foundational requirement is stability.  Once the unresolved sector is made prognostic, it can no longer be treated as a passive correction.  It must remain active enough to store memory and generate delayed response, but it must not become an artificial source of energy or enstrophy.  The SGC enforces this by constraining the resolved--reservoir exchange through a generating functional $\G$ that is orthogonal, in the Poisson sense, to an augmented enstrophy functional $V[\bar{\zeta},r] = \frac12\int \left(\bar{\zeta}^2+r^2\right)d\mathbf x$:
\be
\{\G,V\}=0 .
\ee
This identity is the geometric cancellation at the heart of the construction.  It says that the cross-interaction may exchange information between $\bar\zeta$ and $r$, but it cannot create augmented enstrophy.  In the abstract setting, this gives the conservative symplectic cross-interaction principle of Corollary~\ref{Main_corr}.  Combined with the dissipative skeleton and standard Ornstein--Uhlenbeck (OU) regularity assumptions from stochastic analysis \cite{da2006introduction}, it yields pullback absorption through Theorem~\ref{Main_thm}.  In the hyperviscous realization, fourth-order smoothing supplies the compactness required to prove the existence of a compact global random pullback attractor (Theorem~\ref{thm:SGC_attractor}).  In the physically classical forced Navier--Stokes--$\beta$ realization, without hyperviscosity, the same augmented-enstrophy structure still yields pullback boundedness in $L^2(\mathbb T^2)\times L^2(\mathbb T^2)$.  Thus the stability mechanism is guaranteed by symplectic properties, while hyperviscosity is used only to obtain the strongest compact-attractor theorem.

Specializing this construction to two-dimensional vorticity dynamics yields the coupled resolved--reservoir system
\begin{align}
&\partial_t \bar{\zeta} + J(\bar{\psi}, \bar{\zeta}) + \beta \bar{v} = -\mu \bar{\zeta} + \nu \nabla^2 \bar{\zeta} \nonumber\\
& \hspace{4cm}+ F_\zeta + \frac{1}{2} J\big( \gamma J(\bar{\zeta}, r), \bar{\zeta} \big) \label{eq:zeta_closed_intro} \\
&\partial_t r = -D r + \frac{1}{2} J\big( \gamma J(\bar{\zeta}, r), r \big) + \Sigma \dot{W}_t.\label{eq:r_closed_intro}
\end{align}
Here, $\bar{\zeta}$ denotes the resolved macroscopic relative vorticity, advected by the resolved streamfunction $\bar{\psi}$ (where $\bar{\zeta} = \nabla^2 \bar{\psi}$) via the standard Jacobian determinant $J(A, B) = \partial_x A \partial_y B - \partial_y A \partial_x B$. The macroscopic flow is subject to Ekman bottom drag $\mu$, kinematic viscosity $\nu$, and a large-scale deterministic forcing $F_\zeta$. The term $\beta \bar{v}$ captures the differential planetary rotation; when non-zero, this effect drives the spontaneous reorganization of turbulence into persistent zonal jets. The reservoir operator $D$ supplies relaxation and smoothing, while the spatially correlated cylindrical white noise $\Sigma\dot W_t$  \cite{da2006introduction} maintains unresolved fluctuations.  The coupling coefficient $\gamma$ determines the strength and spatial structure of the resolved--reservoir exchange.

The dynamic transfer of enstrophy and kinetic energy between the macroscopic flow and the reservoir is governed exclusively by the conjugate nonlinear cross-interactions.
The physical meaning of this coupling is exposed by the emergent subgrid drift potential
\be
\psi_{sgs}
=
\frac12\gamma J(\bar\zeta,r).
\label{eq:psi_sgs_intro}
\ee
With this definition, the SGS forcing is not a generic stress divergence.  It is a transport operator,
\be\label{Eq_sympl_transport}
\Pi=J(\psi_{sgs},\bar\zeta).
\ee
%--------------------------------------------------------------------------%
Theorem~\ref{thm:symplectic_transport} proves that this term is the Lie derivative of the resolved vorticity along the Hamiltonian vector field
\be\label{Eq_def_u_sgs}
\mathbf u_{sgs}
=
\nabla^\perp \psi_{sgs},
\ee
so that $\mathbf u_{sgs}$ belongs to the Lie algebra
$\mathfrak{symp}(\mathcal D)$ of the symplectomorphism group over the two-dimensional periodic domain $\mathcal D=\mathbb T^2$.  Defining the dressed streamfunction
\be
\psi_d = \bar\psi-\psi_{sgs},
\ee
the resolved equation can be rewritten as transport by the effective streamfunction $\psi_d$; see Section~\ref{Sec_Closed_Symplectic}.  Thus the SGS term does not simply push, damp or stir the resolved vorticity.  It renormalizes the advecting geometry itself.

 Figure~\ref{fig:reservoir_loop} summarizes the resulting architecture: the resolved flow and the hidden reservoir exchange enstrophy conservatively through a single emergent drift potential, and the net effect on the resolved dynamics is a renormalization of its transport geometry rather than an added subgrid force. The fields shown there are not schematic. They are taken from a $\beta$-plane integration of Eqns.~\eqref{eq:zeta_closed_intro}--\eqref{eq:r_closed_intro} initialized from a coarse-grained snapshot of the forced $\mathrm{Re}=2.5\times10^4$ direct numerical simulations of \cite{srinivasan2024turbulence}, in which the closure sustains zonal jets together with the vortices and filaments embedded in them, and the run reproduces the structural identities established below---vanishing spatial mean of the subgrid drift, area preservation, and detailed enstrophy neutrality---to machine precision.  The integration shown is carried over ten
eddy-turnover times with a calibrated model against filtered Direct Numerical Simulation (DNS).  Over that window the closure
holds the enstrophy spectrum of the filtered reference simulation across the
resolved range, keeps the grid-scale enstrophy fraction at its reference
value, and sustains the jets at their climatological strength. The emergent drift is moreover found to condense into chains of counter-rotating dipoles straddling the resolved vorticity filaments, that is, into precisely the stirring cells that fold and stretch those filaments while preserving area---an organization that the homogeneous, isotropic reservoir forcing cannot supply and that is produced by the symplectic transport operator (Eq.~\eqref{Eq_sympl_transport}). The numerical construction is described in Section~\ref{Sec_numerics} below.
%---------------------------------------------------------------%
\begin{figure*}[t]
\centering
\includegraphics[width=\textwidth]{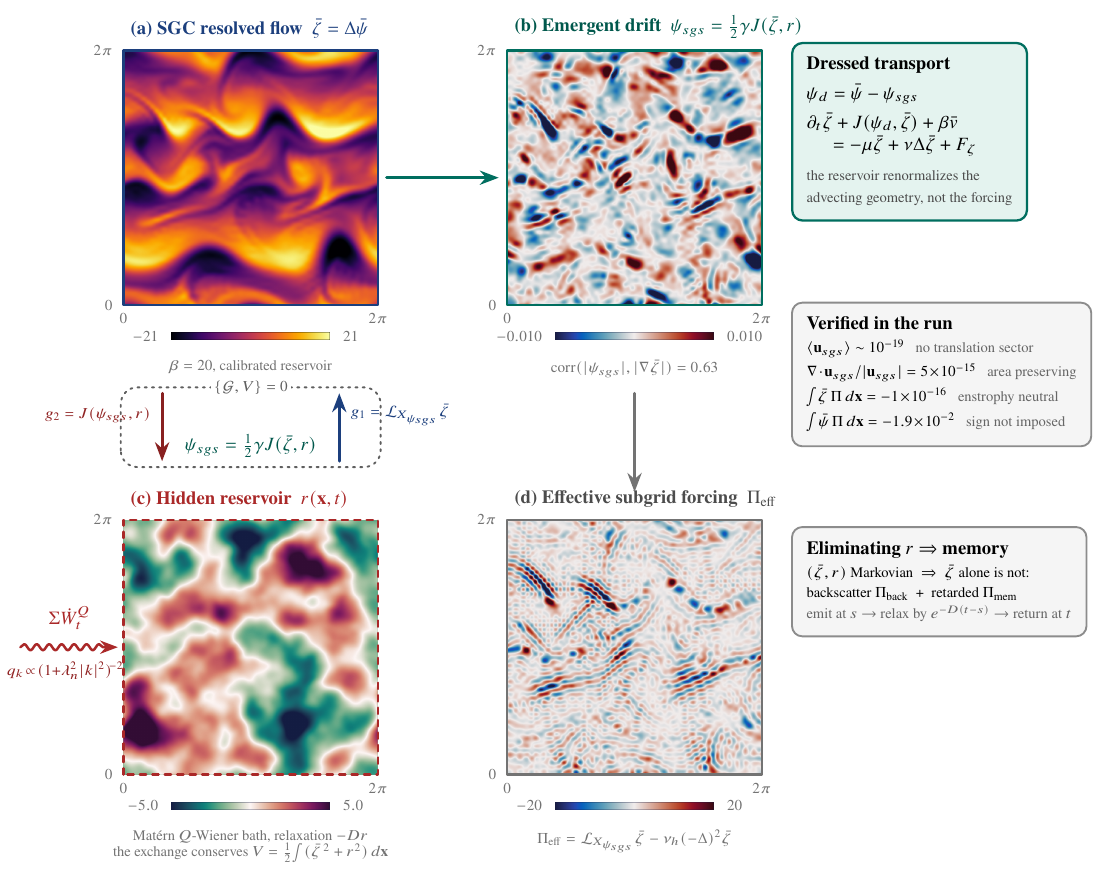}
\caption{\label{fig:reservoir_loop}%
\textbf{The Symplectic Geometric Closure in action: a hidden reservoir
and endogenous dressed transport.}
All four fields are from an actual integration of
Eqns.~\eqref{eq:zeta_closed_intro}--\eqref{eq:r_closed_intro} on a $64^2$
macroscopic grid, displayed over the same domain $(0,2\pi)^2$ at the same
time instant. The run is on the $\beta$-plane, $\beta=20$, initialized from a
jet-bearing coarse-grained snapshot of the forced
$\mathrm{Re}=2.5\times10^4$ DNS of
\cite{srinivasan2024turbulence} at $k_f=4$ (forcing wavenumber). The snapshot is taken after
ten eddy-turnover times of closure integration, with the calibrated model
parameters of Table~\ref{tab:dns_params}. The setup, the calibration of
$\gamma$, and the diagnostics quoted below are detailed in
Section~\ref{Sec_numerics}.
\textbf{(a)} The resolved vorticity \emph{carried by the closure}: the
closure sustains the zonal jets characteristic of $\beta$-plane
turbulence, together with the vortices, filaments and shear layers
embedded in and between them, all of them transported by the dressed
velocity $\psi_d$.
\textbf{(b)} The realized drift potential, at the same instant. It is
markedly finer-grained than either field that generates it
(Table~\ref{tab:scales}), and condenses into chains of alternating $\pm$
dipoles strung along the resolved vorticity filaments. Since
$\mathbf u_{sgs}=\nabla^\perp\psi_{sgs}$ (Eq.~\eqref{Eq_def_u_sgs}), each dipole is a
counter-rotating pair of stirring cells that folds and stretches the
filament it sits on while preserving area
[Eq.~\eqref{eq:dipole_mechanism}].
\textbf{(c)} The prognostic hidden reservoir, sustained by a Mat\'ern
$Q$-Wiener process \cite{matern1986spatial,lord2014introduction} and relaxed by the hyperviscous reservoir operator $D$.
The two sectors exchange enstrophy through a single generating functional
$\G$: the resolved flow stirs the bath through the emission vertex $g_2$,
and the bath returns a delayed back-reaction through the re-absorption
vertex $g_1$. Both are governed by the same emergent drift potential
$\psi_{sgs}$, and stochastic forcing enters \emph{only} the reservoir, so
the randomness felt by the resolved flow is endogenous. The Poisson
orthogonality $\{\G,V\}=0$ (dashed box) makes this exchange conserve
the augmented enstrophy $V=\frac12\int(\bar\zeta^2+r^2)\,d\mathbf x$: the
reservoir may store, delay and return enstrophy, but never create it.
 \textbf{(d)} The effective subgrid forcing
$\Pi_{\rm eff}=\mathcal{L}_{X_{\psi_{sgs}}}\bar\zeta-\nu_h(-\Delta)^2\bar\zeta$,
i.e.\ the closure's total action on the resolved flow. We display
$\Pi_{\rm eff}$ rather than the Lie-transport term alone because it is the
object comparable to the subgrid forcing diagnosed from filtered DNS, which
carries no hyperviscosity; the transport part by itself is exactly
enstrophy-neutral (Table~\ref{tab:diagnostics}) and organizes into the same type of
dipolar sheets along the filaments.
\emph{Right:} the net effect is a renormalization of the advecting
geometry through $\psi_d=\bar\psi-\psi_{sgs}$; the structural identities
of Theorem~\ref{thm:symplectic_transport} verified to machine precision
(Table~\ref{tab:diagnostics}); and the finite memory left behind when $r$
is eliminated.}
\end{figure*}
%---------------------------------------------------------------%

In the language of field theory and solid-state physics, the resolved flow is continuously ``dressed'' by fluctuations propagating through the hidden reservoir, analogously to a bare excitation dressed by its interaction with virtual or medium-induced fluctuations \cite{gurari1953,feynman1955slow,Devreese2009,PeskinSchroeder1995,ZinnJustin2002,Tauber2014}.
The content of Theorem~\ref{thm:symplectic_transport} is that this dressing operation preserves the Hamiltonian architecture of the fluid.  The dressed state therefore remains in the same geometric category as the bare state: the closure modifies the transport field, but not the fundamental symplectic symmetry that governs incompressible two-dimensional motion.  In particular, the subgrid drift is confined to $\mathfrak{symp}(\mathcal D)$, so unresolved stochastic activity can act on the macroscopic vorticity only through continuous area-preserving coordinate transformations.  This structurally forbids artificial geometric smearing: the reservoir may fold, shear and rearrange the resolved field, but it cannot introduce an arbitrary symmetry-breaking algebraic diffusion in place of transport.  Figure \ref{fig:symplectic_backbone} below displays the mechanism generating $\psi_{sgs}$, the Lie transport it induces, and the geometric sector in which the emergent drift is confined.

%In the language of field theory, the resolved flow is continuously ``dressed'' by fluctuations propagating through the hidden reservoir \cite{gurari1953,feynman1955slow,PeskinSchroeder1995,ZinnJustin2002,Devreese2009,huang2013critical,Tauber2014}.  

This observation explains the nomenclature \emph{Symplectic Geometric Closure}.  Arnold's seminal geometric formulation identifies incompressible fluid motion with geodesic evolution on the group of volume-preserving diffeomorphisms; in two dimensions, the corresponding area-preserving transformations coincide with symplectomorphisms \cite{arnold1966geometrie,ebin1970groups,arnold1998topological}.  By requiring the unresolved sector to couple through the Poisson orthogonality condition $\{\G,V\}=0$, the present closure inherits this symplectic structure at the coarse-grained level.  The hidden reservoir therefore does not merely parameterize unresolved transport.  It participates in a geometrically constrained extension of the fluid phase space whose stability and transport form both follow from the underlying symplectic architecture.
%--------------------------------------------------------------------------%

This is the first major result of the paper.  The unresolved field does not enter as an additive forcing that must be tuned to mimic an energy budget.  It enters as a Hamiltonian transport correction generated by the phase misalignment between the resolved vorticity and the hidden reservoir.  Because the same symplectic constraint that generates this drift also enforces the augmented-enstrophy cancellation, the transport correction is stable by construction.  The closure therefore links stability and stochastic transport at the level of the same geometric object, rather than treating them as separate modeling requirements.

This transport viewpoint connects SGC to several established geometric theories, but the distinction is important.  Generalized Lagrangian Mean (GLM) theory interprets unresolved wave fluctuations as producing an induced mean drift \cite{andrews1978generalized}, and the Gent--McWilliams (GM) parameterization represents mesoscale eddy effects through an eddy-induced transport streamfunction \cite{gent1990isopycnal}.  Stochastic Advection by Lie Transport (SALT) \cite{holm2015variational} and the Location Uncertainty (LU) framework \cite{resseguier2017geophysical_I,resseguier2017geophysical_II} show that stochasticity can be incorporated directly into transport geometry.  The SGC shares with these theories the principle that unresolved variability should enter through transport rather than through an additive force.  Its difference is the origin of the transport correction.  The stochastic drift is not prescribed kinematically and is not derived from an asymptotic homogenization limit.  It emerges dynamically from the finite-memory evolution of the hidden reservoir.  Consequently, the transport is both stochastic and non-Markovian, with statistics inherited from the coupled resolved--reservoir dynamics.

The second organizing idea is renormalization.  The analogy with quantum field theory is conceptual but useful: a bare object becomes a dressed object through repeated interactions with hidden fluctuations \cite{feynman1948space,feynman1949space,feynman1955slow,PeskinSchroeder1995,ZinnJustin2002}.  In turbulence, this idea appears in a more specific form.  Renormalization-group approaches eliminate bands of unresolved hydrodynamic modes and generate effective transport coefficients and induced noise \cite{ForsterNelsonStephen1977,yakhot1986renormalization}.  Kraichnan's Direct-Interaction Approximation (DIA) and the related Wyld/Martin--Siggia--Rose-type field-theoretic formalisms \cite{wyld1961formulation,MartinSiggiaRose1973} show, in complementary language, that unresolved nonlinear interactions produce retarded memory kernels, self-energy corrections, eddy damping and fluctuating backscatter \cite{kraichnan1959structure,kraichnan1961dynamics,wyld1961formulation,DeDominicisMartin1979,kraichnan1980review,BereraSalewskiMcComb2013}.  In most practical closures, however, the hidden sector responsible for this dressing is eliminated, approximated perturbatively, or collapsed into an effective coefficient.  Here it remains a prognostic stochastic field.

This also clarifies the relation with Mori--Zwanzig theory.  We do not claim to solve the full Mori--Zwanzig projection problem for Navier--Stokes turbulence \cite{mori_transport_1965,zwanzig_memory_1961,Chorin_al02,Chorin_Hald-book,LucariniChekroun2023}.  Rather, SGC constructs an explicit Markovian resolved--reservoir system whose elimination produces memory and fluctuation terms analytically.  The enlarged system is Markovian; the resolved dynamics alone are not.  The reservoir therefore realizes, at the level of a stochastic field model, the hidden degrees of freedom that projection and response theories identify abstractly.  The memory kernel is generated by an evolving geometric reservoir rather than reconstructed from an unknown exact projection.

At leading order by means of Ornstein--Uhlenbeck (OU) approximation, eliminating the reservoir yields a non-Markovian SGS response with two components: a fluctuating backscatter term and a deterministic viscoelastic memory (Section \ref{Sec_MZ_OUapprox}).  This first-order expansion already shows why the closure cannot be reduced to an eddy viscosity.  The resolved field emits fluctuations into the reservoir; the reservoir relaxes through its own stochastic dynamics; and the delayed response returns as a retarded transport correction.  The resulting drift depends on the history of the resolved state, not merely on its instantaneous gradient.  Thus stochastic backscatter and deterministic memory appear as two projections of the same hidden-sector dynamics.

%---------------------------------------------------------------------------------------------------------------------------------------%
The leading OU reduction already generates finite memory, but it does not yet produce a deterministic renormalization of the reservoir propagator. That additional field-theoretic structure appears at the next non-vanishing level of the Volterra--Duhamel hierarchy of the reservoir dynamics. In Section~\ref{Sec_higher_order_memory}, the response expansion is organized around the centered Gaussian reservoir bath. Gaussian parity then provides a simple selection rule. The second-order contribution to the deterministic memory is cubic in the Gaussian reservoir and therefore vanishes under Wick--Isserlis averaging \cite{isserlis1918formula,janson1997gaussian,peccati2011wiener} (Section~\ref{sec:gaussian_parity_third_order}). The third-order correction is the first level containing quartic reservoir moments. Their Wick contractions pair two successive resolved--reservoir interactions and generate a systematic feedback on the bare propagator. Physically, this is the first order at which a disturbance emitted into the hidden reservoir can propagate through the unresolved bath and return coherently to modify the subsequent evolution of the resolved flow. This returned feedback is the onset of the geometric self-energy.

A key technical point is that these contractions are not treated here merely as formal diagrams, as is common in standard perturbative approaches \cite{kraichnan1961dynamics,wyld1961formulation,nakano1972theory}. The stochastic reservoir is an infinite-dimensional Gaussian field, and the nested Jacobian vertices differentiate its covariance. The resulting point-split gradient covariances must therefore exist as genuine stochastic objects rather than only as symbolic diagrammatic factors. Within the framework of infinite-dimensional stochastic analysis \cite{da2006introduction,DZ96,DPZ08}, this requirement is enforced by the strengthened gradient trace condition \eqref{eq:OU_gradient_trace_condition}, or equivalently by \eqref{eq:OU_gradient_trace_fourier} in the commuting Fourier-diagonal setting, as detailed in Appendix~\ref{App_Macro_Diffusion_Ope}. This condition supplies precisely the ultraviolet covariance regularity needed to define the differentiated Wick contractions generated by the symplectic vertices.

The OU reservoir therefore gives rigorous stochastic content to the
algebraic organization of the Dyson--Volterra series.
Theorem~\ref{thm:one_loop_emergence} shows that the connected
third-order reservoir contractions generate the two-time
transport-covariance insertion $\mathbf\Sigma_{\rm sweep}$, while
Corollary~\ref{cor:reservoir_tangent_self_energy} shows that the same
pair of OU-induced Hamiltonian advection operators, with the
intermediate response inserted between them, forms the leading
tangent-space self-energy $\boldsymbol\eta^{(0)}$. Thus, at the
line-renormalized level, the reservoir dynamics provides an explicit
stochastic realization of a Kraichnan-type DIA response architecture.
%---------------------------------------------------------------------------------------------------------------------------------------%

%---------------------------------------------------------------------------------------------------------------------------------------%
The corresponding statement at the level of the full tangent dynamics is
given by Theorem~\ref{thm:Dyson_propagator}. Ensemble averaging the
pathwise response of the coupled resolved--reservoir system yields a
Dyson--Volterra equation for the dressed macroscopic propagator. Its
memory is encoded by the geometric self-energy
\be
\boldsymbol{\eta}(t,\tau)\Phi
=
\mathbb{E}
\Big[
J\Big(
\psi_{sgs}(t),
G(t,\tau)
J\big(\psi_{sgs}(\tau),\Phi\big)
\Big)
\Big],
\label{eq:eta_intro}
\ee
or, equivalently in Lie-transport notation,
\beas
\boldsymbol{\eta}(t,\tau)\Phi
=
\mathbb{E}
\Big[
\mathcal{L}_{X_{\psi_{sgs}(t)}}
\Big(
G(t,\tau)
\mathcal{L}_{X_{\psi_{sgs}(\tau)}}\Phi
\Big)
\Big].
\eeas
Theorem~\ref{thm:one_loop_emergence} and
Corollary~\ref{cor:reservoir_tangent_self_energy} then identify the
leading OU contribution to this operator: the connected reservoir
contractions generate its two-time covariance core, while insertion of
the intermediate response produces the corresponding tangent-space
self-energy $\boldsymbol{\eta}^{(0)}$.

%----Here we describe Ares​(t)G(t,τ)Ares​(τ)
Despite its name, the self-energy is not an energy. It is the
operator-valued memory that records how a disturbance of the resolved
flow is first acted upon by the reservoir-induced Hamiltonian advection
operator at time $\tau$, propagated to time $t$ by the dressed tangent
response, and then acted upon again by the reservoir-induced advection
operator. Through this two-time sequence, the reservoir modifies the
propagation, decorrelation, and lifetime of resolved structures.
%----------------------------------------------------------------%
This is the second major result of the paper: a DIA-type response
architecture is generated natively by the prognostic symplectic
reservoir. Moreover, the Gaussian contractions entering the leading OU
realization of this self-energy are mathematically justified by the
ultraviolet covariance regularity of the underlying spatio-temporal OU
field, as established in
Appendix~\ref{App_Macro_Diffusion_Ope}, rather than introduced as purely
formal field-theoretic contractions. The hidden reservoir thereby provides an explicit prognostic stochastic
realization of degrees of freedom that renormalized turbulence theories
encode implicitly through response functions, covariances, and
self-energy kernels.

% This is the second
%major result of the paper: a DIA-type response architecture is generated
%natively by the prognostic symplectic reservoir. Moreover, the Gaussian
%contractions entering its leading OU realization are supported by the
%ultraviolet covariance regularity of the underlying spatio-temporal
%process, as established in
%Appendix~\ref{App_Macro_Diffusion_Ope}, rather than being introduced as
%purely formal field-theoretic contractions. The hidden reservoir thus
%provides an explicit stochastic realization of degrees of freedom that
%renormalized turbulence theories ordinarily eliminate at the statistical
%level.
%---------------------------------------------------------------------------------------------------------------------------------------%

Deriving a self-energy, however, is only part of the closure problem. One must still determine which motions contribute to it. A conventional Eulerian response does not distinguish automatically between a large-scale velocity that merely carries an eddy past a fixed observer and a velocity gradient that actually stretches, twists, or deforms the eddy. The next question is therefore whether the geometric self-energy (Eq.~\eqref{eq:eta_intro}) remains contaminated by this rigid sweeping, as in the original Eulerian DIA, or whether its symplectic vertex selects only the relative deformation responsible for cascade transfer.

%=============================================================%
\paragraph*{Random Galilean Invariance, Sweeping, and Geometric Self-Energy.}

%---------------------------------------------------------------------------------------------------------------------------------------%
The third major result concerns what this geometric self-energy actually measures. A classical obstacle in statistical turbulence theory is the failure of Eulerian closures to distinguish between advection and deformation. Large-scale motions can strongly decorrelate a turbulent field at a fixed observation point simply by carrying smaller structures through space, even when they produce little of the stretching, twisting, or relative displacement responsible for cascade transfer. The central question is therefore whether the effective feedback encoded by the self-energy reflects genuine deformation of the resolved eddies or is instead dominated by their rigid sweeping past an Eulerian observer.
%---------------------------------------------------------------------------------------------------------------------------------------%

A longstanding difficulty is that the largest turbulent motions possess the largest velocities and therefore dominate Eulerian decorrelation measured at a fixed spatial point. Yet these energetic structures often act primarily by transporting smaller eddies almost rigidly through space. Such \emph{sweeping} motions translate eddies without significantly stretching, twisting, or deforming them and therefore contribute little to the nonlinear transfer of energy across scales. Classical Eulerian closures nevertheless incorporate this rigid-body advection into the self-energy, producing an infrared contamination of the effective decorrelation rate \cite{kraichnan1964kolmogorov,kraichnan1971almost,chenKraichnan1989,eyink2011robert}. Kraichnan's famous Lagrangian-history closures were introduced precisely to remove this defect by reorganizing the response around fluid trajectories rather than fixed observation points \cite{kraichnan1965lagrangian,kraichnan1971almost}. By following the moving fluid, the common translational motion is factored out, allowing the response to depend primarily on the relative deformation of neighboring fluid elements that actually drives the turbulent cascade.

%---------------------------------------------------------------------------------------------------------------------------------------%
The Symplectic Geometric Closure reaches the same physical objective by a fundamentally different route. Rather than changing the representation of the response, it changes the geometry of the interaction being propagated. The response remains Eulerian and Volterra-like, but the quantity being renormalized is no longer the bare Eulerian velocity vertex. Instead, it is the nested Hamiltonian transport vertex
\be\label{Eq_nested_vertex}
\Gamma_{\rm symp}
\sim
J\circ\gamma J.
\ee
When emphasizing its geometric action, we write the previously defined subgrid velocity $\mathbf{u}_{sgs}$ (Eq.~\eqref{Eq_def_u_sgs}) as the Hamiltonian vector field $X_{\psi_{sgs}}=\mathbf{u}_{sgs}$; see Section  \ref{Sec_Symplectic_Transport}.  Theorem \ref{thm:symplectic_transport} shows then that its action on the resolved vorticity is then 
\be\label{Eq_def_LX_psi_intro}
\mathcal{L}_{X_{\psi_{sgs}}}\bar{\zeta}
=
\mathbf{u}_{sgs}\cdot\nabla\bar{\zeta}
=
J(\psi_{sgs},\bar{\zeta}).
\ee
Thus, $X_{\psi_{sgs}}$ is used only when highlighting the Lie-transport interpretation of the same dynamical velocity field. Because the Lie derivative is the infinitesimal Eulerian generator of transport by the associated flow map, the interaction is expressed through relative transport geometry rather than absolute Eulerian velocity. In this sense, the Hamiltonian transport operator internalizes, at the Eulerian field level, the same separation between rigid translation and relative deformation that Kraichnan achieved by reorganizing the response around Lagrangian fluid histories.
Then uniform Galilean translations leave the vorticity unchanged, generate no relative deformation, and therefore belong to the null sector of the symplectic interaction. Proposition~\ref{prop:RGI_compatibility} formalizes this property as Random Galilean compatibility of the bare symplectic vertex. Section~\ref{sec:diagrammatic_interpretation} translates this construction into Wyld-type diagrams \cite{wyld1961formulation}, while Section~\ref{sec:causal_chain} shows that the line-renormalized self-energy inherits the same geometric selection rule.
%---------------------------------------------------------------------------------------------------------------------------------------%

This point marks the sharpest conceptual departure from conventional Eulerian closures. Rather than removing sweeping by reformulating the response in Lagrangian coordinates or by introducing phenomenological eddy-damping times, the present theory incorporates the essential Lagrangian rectification directly into the Eulerian transport generator. The admissible interaction geometry is modified before renormalization takes place. Consequently, the self-energy is built from Hamiltonian transport operators that are already insensitive to rigid-body translation. Theorem~\ref{thm:inertial_ranges} shows that this geometric constraint generates an intrinsic infrared suppression proportional to $p^4E(p)$ in the self-energy. Uniform sweeping modes are therefore filtered because they merely advect coherent structures, whereas strain-producing motions survive because they generate relative deformation, phase scrambling, and finite-memory triad interactions. The effective decorrelation time is therefore controlled by deformation-induced loss of coherence rather than by absolute Eulerian sweeping.   Figure~\ref{fig:sweeping_cascade} below contrasts this route with Kraichnan's Lagrangian-history construction and traces the chain from the transport vertex to the dual-cascade spectra.

From this perspective, the Symplectic Geometric Closure provides an Eulerian field-theoretic realization of the central physical objective that motivated Kraichnan's Lagrangian program. The cascade time is governed not by how rapidly small eddies are swept past a fixed observer, but by how efficiently neighboring fluid elements separate, stretch, twist, and lose phase coherence through nonlinear interactions. The essential distinction lies in where this separation is introduced. In Lagrangian-history DIA, it is achieved by reformulating the response around fluid histories. In the present theory, it is encoded from the outset in the Hamiltonian Lie-transport vertex itself, so that the subsequent renormalization automatically inherits the same immunity to rigid sweeping.
%---------------------------------------------------------------------------------------------------------------------------------------%

%---------------------------------------------------------------------------------------------------------------------------------------%
The fourth major result is the recovery of Kraichnan dual-cascade phenomenology.  Sweeping suppression provides the temporal ingredient: it yields a finite, dynamically generated memory time controlled by strain rather than by random translation.  But the inverse cascade also requires an energetic ingredient.  The transfer must preserve the conservation architecture that makes upscale kinetic-energy migration admissible.  Section~\ref{Sec_phase_lag_upscale} supplies this complementary result.  Projecting the geometric transfer onto triad mechanics shows that the Jacobian architecture preserves the energy--enstrophy bookkeeping underlying Fj\o{}rtoft's theorem \cite{fjortoft1953}.  Thus the inverse-cascade mechanism factorizes into two requirements: Fj\o{}rtoft admissibility, which constrains the allowed direction of transfer, and finite-memory coherence, which determines whether the allowed transfers persist long enough to generate a macroscopic flux.

The SGC supplies both requirements.  The symplectic Jacobian structure preserves the conservative triad constraints, while the geometric self-energy supplies the strain-controlled memory.  Under the usual locality and stationary-flux assumptions, Theorem~\ref{thm:inertial_ranges} proves that the dressed inertial-range dynamics admit the classical inverse-energy cascade
\be
E(k)\sim \epsilon^{2/3}k^{-5/3}
\ee
and the forward-enstrophy cascade
\be
E(k)\sim \eta_Z^{2/3}k^{-3}
\ee
as self-consistent stationary solutions.  The point is not that these spectra are imposed.  They emerge as consistent fixed points of a closure that simultaneously provides finite memory, sweeping suppression, and conservative triad admissibility.

This synthesis reframes the renormalization problem for turbulence closure.  Standard Eulerian closures ask how to renormalize a response whose bare vertex still allows sweeping contamination.  The SGC asks what happens if the field closure is built from a vertex that already distinguishes translation from deformation.  The answer is that the line-renormalized self-energy inherits this distinction.  Uniform convection is projected out of the infrared memory kernel, while strain-producing interactions remain available to decorrelate triads and sustain inertial-range flux.  In this sense, the geometry natively implements the central fixes that Kraichnan's program required: memory is retained, realizability is protected, sweeping is suppressed, and the dual cascade is recovered without an externally imposed Lagrangian-history construction.

The implications extend to data-driven closure.  If the physically relevant object is not an unconstrained SGS tendency but a symplectic generator $G[\bar\zeta,r]$, then learning the closure should mean learning an admissible geometric action rather than regressing a force.  Section~\ref{Sec_prospective_outlook} develops this idea in the language of Hamiltonian neural closures and symplectic networks \cite{greydanus2019hamiltonian,jin2020sympnets,desai2021port}.  The goal is to build ML models whose online stability is inherited from the Poisson orthogonality and Lyapunov structure, rather than imposed afterward through clipping, tuning, or empirical damping.  The same geometric principle that stabilizes the reservoir and selects the sweeping-blind vertex can therefore serve as a blueprint for structure-preserving data-driven parameterization.

The paper is organized as follows.  Section~\ref{Sec_Field_Theory} constructs the infinite-dimensional resolved--reservoir system, proves the abstract symplectic stability principle (Theorem~\ref{Main_thm} and Corollary~\ref{Main_corr}), and establishes the compact random attractor for the hyperviscous realization (Theorem~\ref{thm:SGC_attractor}) together with pullback boundedness for the Navier--Stokes--$\beta$ core.  Section~\ref{Sec_Symplectic_Transport} proves that the SGS forcing is a Hamiltonian Lie-transport term and introduces the dressed streamfunction (Theorem~\ref{thm:symplectic_transport}).  Section~\ref{Sec_numerics} then verifies numerically that this geometry is realized: it details the hyperviscous reservoir, the Mat\'ern $Q$-Wiener forcing \cite{matern1986spatial,lord2014introduction} and the Lilly-type calibration \cite{lilly1992proposed} of the coupling behind Fig.~\ref{fig:reservoir_loop}, and reports the measured co-location of the emergent drift with the resolved vorticity gradients.  Section~\ref{Sec_dressed_streamfunction_GLM} compares the resulting endogenous transport with GLM, GM, SALT, and LU.  Section~\ref{Sec_MZ_OUapprox} derives the leading-order non-Markovian memory and stochastic backscatter.  Section~\ref{Sec_higher_order_memory} develops the higher-order Volterra--Duhamel expansion and proves the emergence of the geometric self-energy (Theorem~\ref{thm:one_loop_emergence}).  Section~\ref{Sec_inertial_ranges} derives the infrared suppression and proves the dual-cascade scaling theorem (Theorem~\ref{thm:inertial_ranges}).  Section~\ref{Sec_phase_lag_upscale} establishes the energetic admissibility of upscale transfer by combining the geometric memory with Fj\o{}rtoft--Kraichnan triad mechanics.  Section~\ref{Sec_Galilean_Invariance} translates the symplectic vertex into Wyld-type diagrammatics \cite{wyld1961formulation}, proves its Random-Galilean-compatible selection rule, and explains why sweeping suppression survives line resummation.  Finally, Section~\ref{Sec_prospective_outlook} outlines how the same framework can guide geometrically constrained Hamiltonian neural closures that learn the symplectic generator rather than an unconstrained SGS forcing.

%======================================================%
\section{Symplectic Geometric Closure for  2D Turbulent Flows}\label{Sec_Field_Theory}

Transitioning to the continuous, infinite-dimensional framework of the 2D Navier-Stokes equations on a $\beta$-plane \cite{Tem97,FMRT01} elevates this closure theory to a macroscopic physical tier. By lifting our approach from vector calculus to variational calculus, we can bypass the diagnostic bottleneck of the deviatoric stress tensor $\bm\tau^d$ entirely, directly deriving a dynamically consistent, prognostic Subgrid-Scale (SGS) vorticity forcing $\Pi$ that is globally stable by construction.

\subsection{Field-Theoretic Foundation of the Subgrid Kinematic Reservoir}
\label{Sec_QFT_reservoir}

To theoretically ground our proposed closure framework, we build upon the Multilayer Stochastic Model (MSM) paradigm introduced by \cite{Majda_Harlim2012,MSM2015} and connected to response-theoretic model reduction by \cite{wouters2012,wouters2013} as shown in \cite{santos2021reduced}. In that line of work, hidden variables arise as finite-dimensional Markovian realizations (in an extended phase space) of the memory and fluctuation terms generated by the elimination of unresolved degrees of freedom. 
The resulting closures provide practical approximations to the fluctuation--dissipation structure predicted by Mori--Zwanzig projection theory and nonequilibrium response theory \cite{LucariniChekroun2023}.

While the geometric structure underlying the present construction can already be traced to \cite[Corollary 3.2]{MSM2015}, where hidden variables emerge naturally from nonlinear closure architectures, the interpretation remained primarily that of a reduced-order modeling framework designed to address the empirical stability issues identified in \cite{Majda_Harlim2012}. The response-theoretic foundations of these hidden variables were subsequently clarified by Santos et al. \cite{santos2021reduced}, who established explicit connections between MSM closures, Mori--Zwanzig projection methods, and nonequilibrium response theory. However, the geometric implications of the symplectic architecture highlighted in \cite[Corollary 3.2]{MSM2015} were not further explored.

The present work elevates this framework to the infinite-dimensional setting of stochastic fluid dynamics  \cite{DPZ08} and investigates the profound consequences of these symplectic geometric constraints for turbulence closure; see e.g.~Theorems \ref{thm:symplectic_transport} and
\ref{thm:inertial_ranges} below. In doing so, we uncover a complementary interpretation of the hidden MSM reservoir as a prognostic renormalization sector. Rather than serving solely as an auxiliary variable that approximates unresolved dynamics, the reservoir becomes the dynamical carrier of the renormalization process itself. As shown below, eliminating this hidden sector generates higher-order memory kernels (Section \ref{Sec_higher_order_memory}), stochastic backscatter, and self-energy corrections (Theorem \ref{thm:Dyson_propagator}) that recover the same structural ingredients appearing in Kraichnan's renormalized turbulence theories \cite{kraichnan1959structure,kraichnan1961dynamics,eyink2011robert}. From this perspective, the Symplectic Geometric Closure provides a bridge between response-theoretic model reduction, topological fluid dynamics  \cite{arnold1966geometrie,ebin1970groups,arnold1998topological}, and the field-theoretic description of turbulence \cite{Zhou2021}, while furnishing an explicit dynamical realization of the hidden degrees of freedom that conventional renormalized closures integrate out. 

%..and complements recent field-theoretic renormalization formulations based on effective actions and functional renormalization groups
%\cite{Canet2022,Verma2025}.

%---------------------------------------------------------------------------------------------------------------------------------------%
In that regard, we depart from traditional static closures and draw a conceptual parallel to quantum field theory (QFT) \cite{feynman1948space,feynman1955slow}. In QFT, a theoretical `bare'' particle is physically unobservable; it exists within a boiling bath of vacuum fluctuations. As it propagates, it continuously emits virtual particles into the bath, which eventually loop back to interact with the particle at a later time---forming closed topological loops in a Feynman diagram \cite{feynman1948space,feynman1949space}. This continuous temporal feedback permanently alters the particle as seen by an outside observer, resulting in a ``dressed'' physical state with finite, renormalized properties such as mass and charge \cite{gurari1953,feynman1955slow,Devreese2009,PeskinSchroeder1995,ZinnJustin2002,Tauber2014}.

Turbulence closure demands a closely analogous renormalization. A truncated, ``bare'' macroscopic flow cannot exist independently of the cascades and fluctuations it generates. In the dynamic Renormalization Group (RNG) approach to turbulence, this dressing is formalized by decomposing the velocity field into large-scale, slowly evolving resolved modes and a bath of small-scale, rapidly evolving subgrid modes \cite{yakhot1986renormalization}. Systematically eliminating the fast modes generates both an induced random force and a renormalized effective viscosity for the slow modes. More recent developments have recast this program within nonperturbative functional renormalization, where fluctuations are progressively integrated out scale by scale into an effective action constrained by the symmetries of the Navier--Stokes field theory and their associated Ward identities \cite{Canet2022}. In this formulation, renormalization is not limited to a finite perturbative correction: it becomes a continuous flow through a hierarchy of scale-dependent effective theories.

Concurrently, as developed in the Direct-Interaction Approximation (DIA) and its modern retrospectives \cite{kraichnan1959structure,kraichnan1961dynamics,kraichnan1980review,kraichnan1987eddy,mccomb1990physics,eyink2011robert}, Kraichnan mapped the Navier--Stokes nonlinearity onto a Dyson-type response architecture analogous to that of QFT. In this framework, the effective dynamics of the macroscopic flow are governed by non-Markovian memory: a disturbance emitted into the surrounding turbulent field relaxes through energy exchange with a vast bath of interacting modes, whose delayed feedback produces both retarded eddy damping and fluctuating backscatter on the original large-scale structure. Together, these theories demonstrate that the influence of eliminated degrees of freedom is fundamentally nonlocal in time and cannot, in general, be reduced to an instantaneous transport coefficient.
%---------------------------------------------------------------------------------------------------------------------------------------%

However, in most practical implementations of renormalized turbulence theory, the unresolved sector is ultimately eliminated. The resulting memory effects are either approximated perturbatively, absorbed into effective transport coefficients, or replaced by Markovian closures. Consequently, the hidden degrees of freedom responsible for the renormalization no longer possess an independent dynamical evolution.

The guiding principle of the present work is fundamentally different. Rather than integrating out the unresolved sector, we retain a minimal prognostic representation of it through a continuously evolving subgrid kinematic reservoir $r(x,y,t).$ This reservoir functions as a dynamical realization of the hidden sector that generates renormalization. As shown in Section \ref{Sec_Enstrophy_Fluctuations} below, it absorbs enstrophy fluctuations emitted by the resolved flow, propagates them through its own stochastic dynamics, and subsequently re-injects their influence back into the macroscopic evolution. 
In this sense, the feedback loops that appear only implicitly through self-energy corrections in DIA or through mode elimination in renormalization-group theories \cite{kraichnan1987eddy,yakhot1986renormalization} become explicit dynamical objects whose evolution is tracked alongside the resolved flow as will be articulated in Sections \ref{Sec_higher_order_memory} and \ref{Sec_inertial_ranges}.

Viewed through this lens, the closure introduced in Eqns.~\eqref{eq:zeta_closed_intro}-\eqref{eq:r_closed_intro} admits a natural interpretation as an infinite-dimensional MSM in which the hidden variables play the role of a prognostic renormalization sector. To do so, let $u = (\bar{\zeta}, r) \in \mathcal{H} \times \mathcal{H}$ represent the joint state space of this coupled system, where $\mathcal{H}$ is an appropriate Hilbert space, e.g., $L^2(\mathcal{D})$ equipped with the inner product $\langle f, g\cdot \rangle = \int_{\mathcal{D}} f(x) g(x) d\mathbf{x}$. Throughout this work, we assume the spatial domain $\mathcal{D}$ is a doubly periodic two-dimensional torus, $\mathcal{D}= \mathbb{T}^2$.

Casting the coupled system into the following compact form:
\begin{align}
\d \bar{\zeta} &= ( f_1(\bar{\zeta}) + g_1(\bar{\zeta}, r))\d t \label{eq:MSM_zeta_gen} \\
\d  r &= ( f_2(r) + g_2(\bar{\zeta}, r))\d t+ \Sigma \d W_t,  \label{eq:MSM_r_gen}
\end{align}
where $f_1$ governs the dynamics of the bare resolved flow in the absence of coupling to the hidden sector, while $f_2$ describes the internal relaxation of the reservoir.  
 The nonlinear couplings $g_1$ and $g_2$ 
respectively transfer information from the resolved flow into the reservoir and return its delayed influence back to the resolved dynamics. 
In field-theoretic language, $g_2$, acts as an emission vertex through which resolved fluctuations excite the hidden sector, whereas, $g_1$, acts as a re-absorption vertex through which the hidden sector feeds delayed response back onto the resolved dynamics.
 The stochastic forcing $\Sigma dW_t$ sustains unresolved fluctuations within the hidden layer.

Viewed through this lens, the coupled MSM  (Eqns.~\eqref{eq:MSM_zeta_gen}-\eqref{eq:MSM_r_gen}) performs a  continuous dynamical renormalization. Classical renormalized closures integrate out the unresolved degrees of freedom analytically, under assumptions of scale separation and statistical equilibrium, yielding an effective transport law fixed before the simulation begins. In contrast, the reservoir equation (Eq.~\eqref{eq:MSM_r_gen}) evolves concurrently with the resolved flow (Eq.~\eqref{eq:MSM_zeta_gen}) and continuously reconstructs the influence of unresolved dynamics.

The renormalization is therefore never frozen.  The effective transport experienced by the resolved state is recomputed at every instant through the evolving interaction between the macroscopic flow and the hidden reservoir. The resulting ``dressed'' macroscopic dynamics emerge from the accumulated phase lag between the resolved-scale modulation encoded by $g_2$ and the delayed geometric back-reaction encoded by $g_1$.
The coupled system therefore provides a prognostic realization of non-Markovian turbulent transport rather than an \emph{a priori} parameterization of it. The reservoir is therefore no longer merely a closure variable. It becomes a concrete dynamical realization of the hidden sector whose elimination generates the memory kernels, stochastic backscatter, and self-energy corrections characteristic of renormalized turbulence theory.

%===========================================================%
\subsection{The Symplectic Stability Condition and Random Attractors}
\label{Sec_Symplectic_Stab_Condition}

The introduction of a prognostic renormalization sector raises an immediate mathematical challenge. Once the hidden reservoir is granted its own prognostic dynamics, the closure problem is no longer merely one of representing unresolved transport. One must also guarantee that the continuous exchange of information between the resolved flow and the hidden sector does not destabilize the coupled closure system.

This issue is generic to closure architectures that evolve additional degrees of freedom. The hidden variables must remain sufficiently active to store memory and generate delayed feedback, yet their interaction with the resolved flow cannot be allowed to inject uncontrolled variance into the dynamics \cite{LucariniChekroun2023}. In conventional reduced-order models this balance is often enforced through empirical damping, ad hoc regularization, energy-preserving constraints  \cite{Majda_Harlim2012,MSM2015} or tuned transport coefficients.

The guiding principle adopted here is geometric. Rather than stabilizing the closure by adding dissipative corrections to the cross-scale interaction, we constrain the interaction itself. The unresolved sector is permitted to exchange information with the resolved flow, but only through directions that preserve the underlying conservative geometry of the augmented system extending ideas of \cite{MSM2015} to the infinite-dimensional setting of stochastic fluid dynamics \cite{DPZ08}. In this way, the hidden reservoir acts as a carrier of memory and renormalization while remaining incapable of acting as an artificial source of the Lyapunov functional controlling the long-time dynamics.

To formalize this idea, consider the abstract coupled MSM system (Eqs.~\eqref{eq:MSM_zeta_gen}--\eqref{eq:MSM_r_gen}) defined over the doubly periodic torus $\mathbb T^2$. We specialize the reservoir relaxation to a linear dissipative operator, $f_2(r)=-Dr$, where
\[
D:\mathrm{dom}(D)\subset\mathcal H\to\mathcal H,
\]
is a strictly positive self-adjoint spatial operator on a Hilbert space $\mathcal{H}$ \cite{brezis_book}. Let the joint state vector be
\[
u=(\bar\zeta,r),
\]
and introduce the base separable Hilbert space
\[
\mathcal H=L^2(\mathbb T^2)\times L^2(\mathbb T^2).
\]

Following the Lyapunov-function framework for random attractors \cite{Chueshov02}, we assume the existence of a continuously Fr\'echet-differentiable functional
\[
V:\mathcal H\to\mathbb R^+
\]
whose sublevel sets are bounded in $\mathcal H$. To handle the shifted variables in the pathwise estimates following the standard methods of proof of random attractors \cite{crauel1994attractors,crauel1997random} and random invariant manifolds \cite{CLW15_vol1}, we also assume that $V$ satisfies a polynomial translation bound of the form
\be
V(\tilde u+z)
\le
(1+\epsilon)V(\tilde u)
+
C_\epsilon\big(1+\|z\|^m\big),
\qquad
\epsilon>0,
\label{eq:V_translation_bound}
\ee
for some exponent $m\ge2$, where $z$ belongs to the regularity class supplied by the stationary Ornstein--Uhlenbeck (OU) process, solving:
\be
dz + Dz\,dt = \Sigma\,dW_t,
\label{eq:OU_process}
\ee
where $W_t$ is a cylindrical Wiener process  and  $\Sigma$ denotes  the noise covariance operator \cite{da2006introduction}.

We require the uncoupled deterministic dynamics to satisfy the dissipation condition
\be
\tag{H$_1$}
\exists \, \alpha,\Gamma>0:
\left\langle f_1(\bar\zeta),\frac{\delta V}{\delta\bar\zeta}\right\rangle + \left\langle -Dr,\frac{\delta V}{\delta r}\right\rangle
+
\alpha V(u)
\le
\Gamma ,
\label{cond1}
\ee
for all $u\in\mathcal H$, where $\delta V/\delta\phi$ denotes the variational derivative. The cross-interactions are assumed to satisfy the structural growth bound
\be
\tag{H$_2$}
\exists \, a,b\ge0:
\left\langle g_1(u),\frac{\delta V}{\delta\bar\zeta}\right\rangle
+
\left\langle g_2(u),\frac{\delta V}{\delta r}\right\rangle
\le
aV(u)+b .
\label{cond2}
\ee

Finally, the stochastic forcing must be sufficiently regular for the pathwise Ornstein--Uhlenbeck (OU) shift to be well defined in the function spaces used below. We impose:
\be
\tag{H$_3$}
\begin{aligned}
&\text{The stationary OU process generated by }(D,\Sigma)\\
& \text{admits }\mathbb{P}\text{-a.s. continuous trajectories in }H^2(\mathbb T^2).
\end{aligned}
\label{cond3}
\ee

We now state the abstract confinement result. %Its proof is given in Appendix \ref{app:theorem_proof}.
\begin{thm}[Abstract Random Attractor Theorem]
\label{Main_thm}
Assume that conditions \eqref{cond1}, \eqref{cond2}, and \eqref{cond3} hold, and that the structural constants satisfy
\[
a<\alpha.
\]
Assume moreover that the shifted system generates a well-defined, continuous, shift-compatible solution operator
\[
S(t,s;\omega):\mathcal H\to\mathcal H,
\]
where $\omega$ denotes the noise realization.
Then $S(t,s;\omega)$ possesses a pullback absorbing set in $\mathcal H$. If, in addition, the solution operator is pullback asymptotically compact in $\mathcal H$, then it possesses a unique measurable compact global random attractor $\mathcal A(\omega)$.
\end{thm}

This theorem separates the two ingredients required for the existence of a random attractor. Conditions \eqref{cond1}--\eqref{cond3} yield a pullback absorbing set: the uncoupled dynamics dissipate the Lyapunov functional, the reservoir noise is regular enough to perform the pathwise OU shift, and the cross-interaction is not allowed to overcome the dissipative margin. The proof is provided in Appendix \ref{app:theorem_proof}. The remaining compactness requirement is analytic rather than geometric; for the Navier--Stokes realization of the Symplectic Geometric Closure, it is verified by the hyperviscous smoothing argument developed in {\it Supplementary Note 1}, combined with a splitting argument used to prove asymptotic compactness for semigroups generated by viscoelastic fluid models  with memory \cite{CGH12}.

The following corollary identifies the geometric condition under which the cross-interaction contributes no Lyapunov growth at all, thereby guaranteeing condition \eqref{cond2}. This is the abstract form of the symplectic stability principle used throughout the paper.

%===========================================================%
\begin{corr}[Conservative Symplectic Cross-Interactions]
\label{Main_corr}
Assume that the system generates a valid solution operator $S(t,s;\omega)$ and satisfies the uncoupled dissipation condition \eqref{cond1} and the noise regularity condition \eqref{cond3}. Let there exist a continuously Fr\'echet-differentiable scalar functional $\G[\bar\zeta,r]$ such that the cross-interactions are generated by conjugate variational derivatives:
\be
g_1(u)=\frac{\delta \G}{\delta r},
\qquad
g_2(u)=-\frac{\delta \G}{\delta\bar\zeta}.
\label{eq:antideriv_functional}
\ee
If the coupling functional $\G$ and the Lyapunov functional $V$ are in involution under the continuous Poisson structure,
\be
\{\G,V\}
=
\int_{\mathbb T^2}
\left(
\frac{\delta \G}{\delta r}
\frac{\delta V}{\delta\bar\zeta}
-
\frac{\delta \G}{\delta\bar\zeta}
\frac{\delta V}{\delta r}
\right)
d\mathbf x
=
0,
\label{eq:Poisson_zero}
\ee
then the cross-interaction condition \eqref{cond2} is exactly satisfied with $a=b=0$. In particular, the geometric coupling contributes neither growth nor dissipation to the Lyapunov functional $V$; its effect is purely conservative and cannot obstruct the dissipative mechanism responsible for pullback absorption.
\end{corr}

\begin{proof}
We compute the contribution of the cross-interactions to the Lyapunov balance. By definition,
%---------------------------------------------------%
\begin{widetext}
\bea
\left\langle g_1(u),\frac{\delta V}{\delta\bar\zeta}\right\rangle + \left\langle g_2(u),\frac{\delta V}{\delta r}\right\rangle
&= \int_{\mathbb T^2} \left( g_1(u)\frac{\delta V}{\delta\bar\zeta} + g_2(u)\frac{\delta V}{\delta r} \right)d\mathbf x .
\eea
\end{widetext}
%---------------------------------------------------%
Using the conjugate variational representation \eqref{eq:antideriv_functional}, this becomes
\bea
\int_{\mathbb T^2}
\left(
\frac{\delta \G}{\delta r}
\frac{\delta V}{\delta\bar\zeta}
-
\frac{\delta \G}{\delta\bar\zeta}
\frac{\delta V}{\delta r}
\right)d\mathbf x
=
\{\G,V\}.
\eea
By the involution assumption \eqref{eq:Poisson_zero}, $\{\G,V\}=0$. Hence
\be
\left\langle g_1(u),\frac{\delta V}{\delta\bar\zeta}\right\rangle
+
\left\langle g_2(u),\frac{\delta V}{\delta r}\right\rangle
=
0.
\ee
Therefore condition \eqref{cond2} holds with $a=b=0$. Since $\alpha>0$ in \eqref{cond1}, the strict stability requirement $a<\alpha$ is automatically satisfied. Thus the symplectic cross-interaction is conservative at the level of the Lyapunov functional and cannot interfere with the pullback absorption mechanism established in Theorem \ref{Main_thm}.
\end{proof}

%%======================================================%
\subsection{Conservative Enstrophy Exchange and the Symplectic Closure Stability}
\label{Sec_Enstrophy_Fluctuations}
To specialize the abstract symplectic stability framework to two-dimensional turbulence, we must identify the Lyapunov functional that organizes the exchange between the resolved flow and the hidden renormalization sector. The reservoir field $r(x,y,t)$ was introduced as a dynamical carrier of unresolved enstrophy fluctuations. It therefore plays a role analogous to the eliminated ultraviolet sector in renormalized turbulence theories: it stores unresolved variance, transports it through an internal relaxation dynamics, and returns its delayed influence to the resolved flow.

From this perspective, the quantity that naturally constrains the coupled dynamics is not the kinetic energy but the enstrophy. In two-dimensional turbulence, enstrophy governs the forward cascade and controls the accumulation of variance at small scales. We therefore define the canonical Lyapunov functional to be the total augmented enstrophy of the extended phase space:
%--------------------------------%
\be
V[\bar{\zeta},r] = \frac12\int_{\mathbb T^2} \left(\bar{\zeta}^2+r^2\right)d\mathbf x .
\label{Eq_V_enstrophy}
\ee
%--------------------------------%
This functional measures the enstrophy stored jointly within the resolved vorticity and the hidden reservoir. In the language of the previous section, it provides the geometric quantity whose level sets constrain the admissible renormalization dynamics. With this choice of $V$, the variational derivatives are simply
%--------------------------------%
\be
\frac{\delta V}{\delta\bar{\zeta}}=\bar{\zeta},
\qquad
\frac{\delta V}{\delta r}=r.
\ee
%--------------------------------%

The central physical requirement is that renormalization must remain conservative at the level of this augmented enstrophy. The hidden sector may store enstrophy, delay its return, and generate memory effects, but it must not act as an artificial source of total enstrophy. Otherwise, the reservoir would cease to represent unresolved transport and would instead become an unphysical forcing mechanism. The symplectic stability condition provides exactly this guarantee. Rather than constraining the magnitude of the closure terms, it constrains their geometry: the unresolved sector is allowed to exchange enstrophy with the resolved flow only through directions that preserve the augmented Lyapunov structure.

This geometric constraint alone guarantees stability of the cross-interaction in the Lyapunov budget, but it does not by itself provide the compactness required for a random attractor. Compactness is an analytic property of the dissipative skeleton. Accordingly, in the Navier--Stokes realization below, we supplement the conservative symplectic coupling with standard fourth-order hyperviscous dissipation. The resulting division of labor is essential: the symplectic interaction preserves the augmented enstrophy geometry, while hyperviscosity provides the smoothing needed to verify pullback asymptotic compactness in the sense required by Theorem \ref{Main_thm}.

A natural geometric coupling between two scalar fields on a two-dimensional domain is provided by the Jacobian determinant:
%--------------------------------%
\be
J(\bar{\zeta},r)=\partial_x\bar{\zeta}\,\partial_y r-\partial_y\bar{\zeta}\,\partial_x r=\nabla^\perp\bar{\zeta}\cdot\nabla r .
\label{Eq_J_fluid_style}
\ee
%--------------------------------%
This quantity is Galilean invariant, rotationally covariant, and vanishes whenever the gradients of the two fields are locally aligned. It therefore measures the cross-gradient twisting between the resolved vorticity field and the hidden reservoir. Motivated by these geometric considerations, we introduce the quartic generating functional
%--------------------------------%
\be
\G[\bar{\zeta},r]=\frac14\int_{\mathbb T^2}\gamma(\x)\left[J(\bar{\zeta},r)\right]^2d\mathbf{x},
\label{eq:G_quartic}
\ee
%--------------------------------%
where $\gamma(\x)$ is a scale-dependent coupling coefficient that may be specified dynamically or learned from data. The functional $\G$ generates the cross-interactions through conjugate variational derivatives. Consequently, the coupling is not imposed as a dissipative stress; it is generated as a symplectic exchange mechanism between the resolved vorticity and the hidden enstrophy reservoir.

We then have the following fundamental result. It states that, once the symplectic coupling is combined with the hyperviscous dissipative skeleton, the coupled resolved--reservoir dynamics possess a compact global random attractor. The proof is given in Appendix \ref{app:SGC_attractor_proof}.

%-------------------------------------------------------------%
\begin{thm}[Global Attractor of the Symplectic Geometric Closure]
\label{thm:SGC_attractor}
Let the macroscopic vector field $f_{1,h}$ be governed by the two-dimensional
Navier--Stokes dynamics augmented with fourth-order hyperviscous dissipation,
\be
f_{1,h}(\bar{\zeta})
=-J(\bar{\psi}, \bar{\zeta}) -\beta \partial_x \bar{\psi} -\mu \bar{\zeta} +\nu \Delta \bar{\zeta}
-\nu_h \Delta^2\bar{\zeta}
+F_\zeta,
\ee
with $\nu_h>0,$ and where $\bar{\zeta}=\Delta\bar{\psi}$ on the periodic torus $\mathbb T^2$.
Let the reservoir relaxation be governed by the strongly elliptic operator
\be
D= \nu_r(-\Delta)^2+\kappa(-\Delta)+\mu_r I,
\qquad
\nu_r,\kappa,\mu_r>0.
\ee
Let $A=-\Delta$ and assume that the noise covariance $\Sigma$ satisfies the
Hilbert--Schmidt regularity condition
\be
\mathrm{Tr}\left(A^2D^{-1}\Sigma\Sigma^*\right)<\infty.
\label{eq:OU_trace_condition}
\ee
Let the canonical Lyapunov functional be the total augmented enstrophy defined in Eq.~\eqref{Eq_V_enstrophy}.

Assume that the spatially dependent coupling coefficient  defined in Eq.~\eqref{eq:G_quartic} satisfies
$\gamma\in W^{2,\infty}(\mathbb T^2)$, and let the cross-interactions be
generated by the conjugate variational derivatives of the symplectic functional
$\G$ in Eq.~\eqref{eq:G_quartic}:
\be
g_1 = \frac{\delta \G}{\delta r}
= \frac12 J\Big(\gamma J(\bar{\zeta},r),\bar{\zeta}\Big),
\label{eq:g1_nested}
\ee
and
\be
g_2
=
-\frac{\delta \G}{\delta\bar{\zeta}}
=
\frac12
J\Big(\gamma J(\bar{\zeta},r),r\Big).
\label{eq:g2_nested}
\ee
Then the resulting hyperviscous Symplectic Geometric Closure,
%---------------------------------------%
\begin{align}
\partial_t\bar{\zeta}
&=f_{1,h}(\bar{\zeta})+\frac12J\big(\gamma J(\bar{\zeta},r),\bar{\zeta}\big),\label{eq:zeta_closed}\\
\partial_t r&=-Dr+\frac12J\big(\gamma J(\bar{\zeta},r),r\big)+\Sigma\dot W_t,
\label{eq:r_closed}
\end{align}
%---------------------------------------%
possesses a unique, measurable, compact global random attractor in
$\mathcal H=L^2(\mathbb T^2)\times L^2(\mathbb T^2)$, provided the associated
shifted solution operator is well defined and continuous on $\mathcal H$.
\end{thm}

%------------------------------------------------------------------------------------------------------------------------------------------------%
It is important to distinguish the compact-attractor statement above from the weaker stability property needed for the physical forced Navier--Stokes realization of the closure. The fourth-order hyperviscous term is used here as a sufficient smoothing mechanism to verify pullback asymptotic compactness and hence obtain a compact global random attractor. If this hyperviscous assumption is relaxed, the rigorous proof of compact random-attractor existence is left for future work. However, the enstrophy stability mechanism itself does not rely on hyperviscosity. Indeed, for the forced beta-plane Navier--Stokes skeleton
\be
f_1^{NS,\beta}(\bar{\zeta})
=
-J(\bar{\psi},\bar{\zeta})
-\beta\partial_x\bar{\psi}
-\mu\bar{\zeta}
+\nu\Delta\bar{\zeta}
+F_\zeta,
\ee
the augmented enstrophy functional $V$ still yields a pullback absorbing estimate in $L^2(\mathbb T^2)\times L^2(\mathbb T^2)$. The nonlinear advection satisfies the standard cancellation $\langle J(\bar{\psi},\bar{\zeta}),\bar{\zeta}\rangle=0$, while the beta-plane term is also skew with respect to the enstrophy pairing:
\beas
\left\langle -\beta\partial_x\bar{\psi},\bar{\zeta}\right\rangle &=
-\beta\int_{\mathbb T^2}\partial_x\bar{\psi}\,\Delta\bar{\psi}\,d\mathbf x \\
&= \frac{\beta}{2}\int_{\mathbb T^2}\partial_x|\nabla\bar{\psi}|^2\,d\mathbf x = 0.
\eeas
Thus the beta term contributes wave propagation and geophysical anisotropy, but no growth in the augmented enstrophy budget. The viscous, Ekman, reservoir-damping, and forcing terms then give condition \eqref{cond1}, while the symplectic identity $\{\G,V\}=0$ gives condition \eqref{cond2} with conservative cross-interactions.

For this weaker pullback-boundedness statement, one does not need the full $H^2$-valued Ornstein--Uhlenbeck regularity assumption \eqref{cond3}. If the reservoir operator is reduced to the second-order dissipative form
\be\label{Eq_D_no_hperviscous}
D=\kappa(-\Delta)+\mu_r I,
\qquad
\kappa,\mu_r>0,
\ee
then the natural replacement of Eq.~\eqref{eq:OU_trace_condition} for an $L^2$ absorbing theory is the finite-energy OU condition
\be
\mathrm{Tr}\left(D^{-1}\Sigma\Sigma^*\right)<\infty.
\label{eq:OU_L2_trace_condition}
\ee
This ensures that the stationary OU shift generated by $(D,\Sigma)$ has $\mathbb P$-a.s. continuous $L^2$ trajectories, which is sufficient for the random Gronwall argument producing a pullback absorbing ball in $\mathcal H=L^2(\mathbb T^2)\times L^2(\mathbb T^2)$. Hence the hyperviscous theorem above should be read as a compactness realization of the closure, whereas the physically classical forced Navier--Stokes--beta realization already possesses the pullback stability required for the conservative symplectic closure mechanism.
%------------------------------------------------------------------------------------------------------------------------------------------------%

%-------------------------------------------------%
Comparing Eq.~\eqref{eq:zeta_closed} with the classical filtered vorticity equation, we identify the macroscopic SGS vorticity forcing with the geometric cross-term:
\be
\Pi
=
\frac12 J\big(\gamma J(\bar\zeta,r),\bar\zeta\big)
=
J(\psi_{sgs},\bar\zeta).
\label{eq:Pi_geometric}
\ee
%-------------------------------------------------%
Because $g_1$ and $g_2$ are derived as conjugate variational derivatives of the generator  $\G$, and because $\int_\Omega (\bar{\zeta} g_1 + r g_2) d\mathbf{x} = 0$, the constraint $\{\G, V\} = 0$ is perfectly satisfied. According to Corollary \ref{Main_corr}, this structural cancellation proves that the symplectic cross-interaction conserves the augmented enstrophy while mediating exchange between the resolved and hidden sectors.

%-------------------------------------------------%
By expanding the outer Jacobian, the geometric SGS forcing can be written in conservative form:
\be\label{Eq_outer_Jacob}
\Pi
=
-\frac12
\nabla\cdot
\Big(
\gamma J(\bar\zeta,r)\nabla^\perp\bar\zeta
\Big),
\ee
since $J(A,\bar\zeta)=-\nabla \cdot (A \nabla^\perp\bar\zeta).$
The identity given by Eq.~\eqref{Eq_outer_Jacob} shows that the closure acts as an advective vorticity flux generated by the emergent subgrid drift potential $\psi_{sgs}=\frac12\gamma J(\bar\zeta,r)$.  The flux direction is fixed by the Hamiltonian geometry, through $\nabla^\perp\bar\zeta$, while its amplitude and sign are dynamically modulated by the cross-phase misalignment $J(\bar\zeta,r)$ between the resolved vorticity and the hidden reservoir.  Thus the SGS term is not introduced as an algebraic diffusion or as a prescribed deviatoric stress.  It is obtained directly as a rotationally covariant and Galilean-invariant transport flux in vorticity space.

%-------------------------------------------------%

%======================================%
\subsection{The Symplectic Geometric Closure and Emergent Subgrid Transport}
\label{Sec_Closed_Symplectic}

 While the continuous coupled system \eqref{eq:zeta_closed}--\eqref{eq:r_closed} guarantees a detailed enstrophy balance, the explicit physical mechanism of this cross-scale exchange is momentarily obscured by the apparent complexity of the nested geometric operators. Specifically, the hidden reservoir equation contains the quadratic cross-term $g_2 = \frac{1}{2} J\big( \gamma J(\bar{\zeta}, r), r \big)$. However, by unearthing the exact geometric identity of these operators, we reveal a profound structural connection to the cutting edge of stochastic fluid dynamics.

Let us define the scalar quantity generated continuously by the spatial phase-misalignment between the macroscopic flow and the hidden reservoir:
\be\label{Eq_def_psi_sgs}
\psi_{sgs}(x,y,t) = \frac{1}{2} \gamma J(\bar{\zeta}, r).
\ee
Structurally, this scalar occupies the mathematical position of an advecting potential. To understand this, we recall the canonical geometry of two-dimensional incompressible transport. The advection of any scalar field $q$ by a divergence-free velocity field $\mathbf{u} = \nabla^\perp \psi$ is expressed identically by the Jacobian determinant: $\mathbf{u} \cdot \nabla q = \nabla^\perp \psi \cdot \nabla q = J(\psi, q)$. Because the conjugate variational derivatives dictating the cross-scale exchange ($g_1$ and $g_2$) are formulated strictly as outer Jacobians, substituting $\psi_{sgs}$ places it exactly in the mathematical position of $\psi$ in the advective operator. Consequently, it acts as a rigorous \textit{drift potential}, generating a continuous, divergence-free velocity field $\mathbf{u}_{sgs} = \nabla^\perp \psi_{sgs}$. By definition, this orthogonal gradient ensures that the resulting transport actively stirs, folds, and advects the fluid fields without inducing artificial compressibility or divergence.

In classical Large Eddy Simulation (LES), the subgrid velocity $\mathbf{u}^\prime$ represents the literal, highly oscillatory physical velocity of turbulent eddies smaller than the grid filter width $\Delta$. Crucially, $\psi_{sgs}$ is not the literal streamfunction of these individual microscopic eddies. Instead, we classify $\psi_{sgs}$ as an \textit{emergent subgrid drift potential}. Its origin is fundamentally subgrid, meaning its existence is strictly predicated on the active unresolved micro-state $r(x,y,t)$; without this hidden kinematic reservoir, the phase-misalignment vanishes and the drift ceases to exist. Concurrently, it provides an effective back-reaction, representing the aggregated geometric influence of the subgrid bath on the resolved flow. In essence, it functions as a ``phantom'' large-scale drift that systematically reorganizes the macroscopic streamlines to account for the energy dynamically exchanged across the filter cutoff.

By substituting this emergent subgrid drift back into the conjugate variational derivatives $g_1$ and $g_2$, the nested Jacobians instantly collapse into pure advective operators. For the macroscopic forcing ($g_1$), the reduction is immediate:
\be
g_1 = \frac{1}{2} J\big( \gamma J(\bar{\zeta}, r), \bar{\zeta} \big) = J(\psi_{sgs}, \bar{\zeta}).
\ee
For the hidden reservoir forcing ($g_2$), we exploit the strict anti-symmetry of the inner Jacobian, $J(r, \bar{\zeta}) = -J(\bar{\zeta}, r)$, to yield an identical geometric collapse:
\bea
g_2 &= -\frac{1}{2} J\big( \gamma J(r, \bar{\zeta}), r \big) = \frac{1}{2} J\big( \gamma J(\bar{\zeta}, r), r \big) \\
&= J(\psi_{sgs}, r).
\eea

Substituting these purely advective forms back into the autonomous dynamics, the entire symplectic closure reorganizes into an elegant dual-transport system:
\begin{align}
\partial_t \bar{\zeta} + J(\bar{\psi} - \psi_{sgs}, \bar{\zeta}) + \beta \partial_x \bar{\psi} &= -\mu \bar{\zeta} + \nu \nabla^2 \bar{\zeta}+ F_\zeta \label{eq:zeta_transport} \\
%& \hspace{1cm} -\nu_h \nabla^4\bar{\zeta}+ F_\zeta \label{eq:zeta_transport} \\
\partial_t r + J(-\psi_{sgs}, r) &= -Dr + \Sigma \dot{W}_t. \label{eq:r_transport}
\end{align}

This geometric reduction reveals a powerful physical implication of the symplectic constraint $\{\G,V\}=0$: it does not merely impose an abstract energetic bound. Because the hidden reservoir $r$ is inherently stochastic, Eq.~\eqref{Eq_def_psi_sgs} structurally mandates that the cross-scale interaction manifests exclusively as a \textit{unified stochastic advection}. The macroscopic flow and the subgrid reservoir do not simply exchange variance; they generate a single, shared effective drift ($\psi_{sgs}$) that simultaneously advects both manifolds. The reservoir $r$ is not a passive Gaussian bath, but is actively stirred and folded by the very subgrid drift field it helps create.

Furthermore, the operator $J(\bar{\psi} - \psi_{sgs}, \bar{\zeta})$ formally establishes the fluid-dynamical equivalent of a ``one-loop'' topological dressing process. Drawing a conceptual parallel to macroscopic condensed matter physics, this mechanism explicitly mirrors Feynman's path-integral formulation of the polaron \cite{feynman1955slow}. Just as an electron moving through an ionic crystal continuously polarizes its environment---creating a ``phonon cloud'' that it must drag along, permanently altering its effective mass---the macroscopic primary flow continuously ``polarizes" the subgrid enstrophy bath. The macroscopic streamfunction $\bar{\psi}$ represents the ``bare'' theoretical flow, while the hidden reservoir $r$ acts as the boiling bath of virtual fluctuations. As the bare flow evolves, it continuously emits and reabsorbs fluctuations from this bath, forming a closed topological feedback loop. This renormalization interpretation is not only metaphorical, and will be actually formalized into the language of field-theoretic Feynman diagrams in Section \ref{sec:diagrammatic_interpretation} below.

 Within this interpretation, the resulting drift potential $\psi_d= \bar{\psi} - \psi_{sgs}$ therefore represents the \textit{dressed streamfunction}.
The resolved flow is no longer transported by the bare streamfunction
$\bar{\psi}$ alone, but by the renormalized transport geometry encoded in
$\psi_d$.
 Just as a bare particle is unobservable without its surrounding virtual cloud, the true topological evolution of the macroscopic cascade can only be evaluated along the trajectories of this dressed transport.

\paragraph*{\bf \small Relation to SALT and Location Uncertainty.}
The dual-transport formulation \eqref{eq:zeta_transport}--\eqref{eq:r_transport}
reveals that the symplectic closure generates stochastic advection through
an emergent drift velocity
$\mathbf{u}_{sgs}=\nabla^\perp\psi_{sgs}$.
This observation places the present framework in close dialogue with
geometric stochastic transport theories such as Stochastic Advection by
Lie Transport (SALT) \cite{holm2015variational,cotter2019numerically}
and the Location Uncertainty paradigm
\cite{resseguier2017geophysical_I,resseguier2017geophysical_II}.
Before making this connection precise, however, it is necessary to
determine the exact geometric nature of the transport generated by
$\psi_{sgs}$.
The next section establishes that the induced drift is not merely
divergence-free but is in fact generated by a Hamiltonian vector field
belonging to the Lie algebra of the symplectomorphism group.
%

%====================================================%
\section{Emergent Subgrid Transport as a Symplectomorphism}
\label{Sec_Symplectic_Transport}

%---------------------------------------------------------------%
\begin{figure*}[t]
\centering
\includegraphics[width=0.85\textwidth]{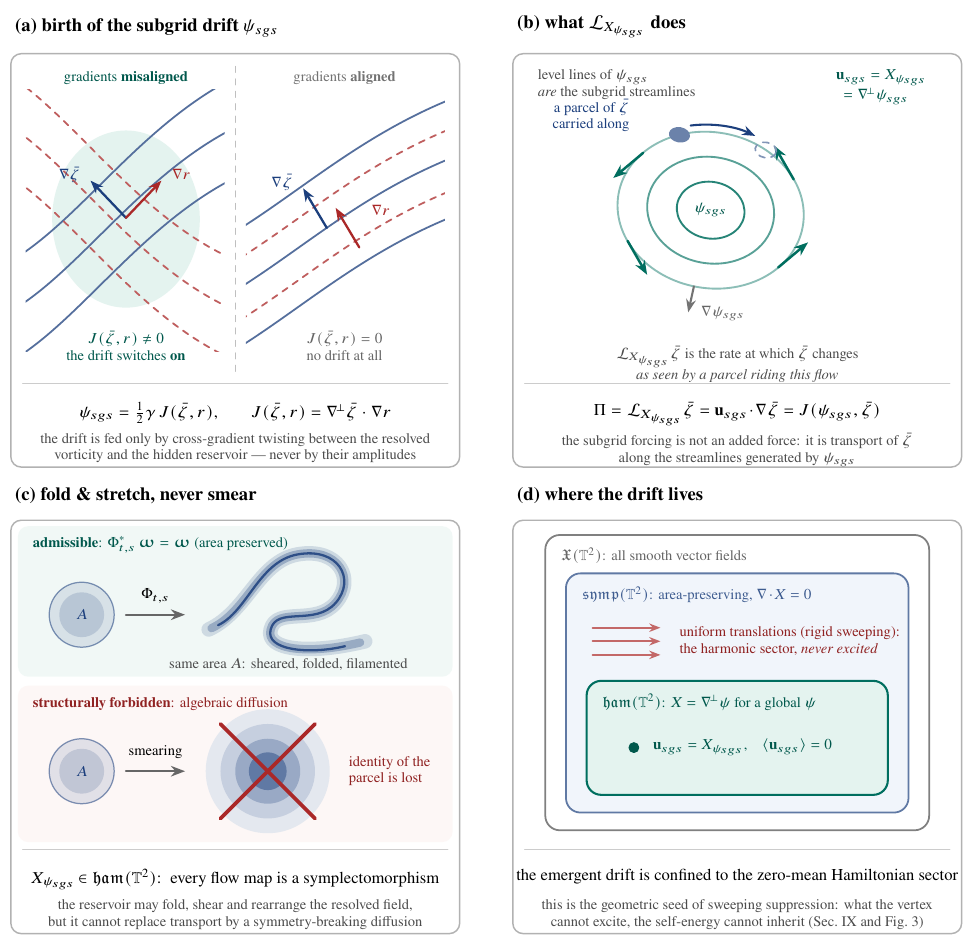}
\caption{\label{fig:symplectic_backbone}%
\textbf{The geometric backbone of the closure.}
(a) {\it Birth of the drift}. The subgrid potential
$\psi_{sgs}=\frac12\gamma J(\bar\zeta,r)$ is fed by the Jacobian
$J(\bar\zeta,r)=\nabla^{\perp}\bar\zeta\cdot\nabla r$, which measures
cross-gradient twisting between the resolved vorticity and the hidden
reservoir. Where their level sets cross transversally the drift switches
on; where the gradients are locally aligned it vanishes identically. The
drift therefore responds to \emph{relative} field geometry, never to
absolute amplitudes.
(b) {\it What the Lie derivative does}. The level lines of $\psi_{sgs}$ are the
streamlines of the Hamiltonian vector field
$X_{\psi_{sgs}}=\nabla^{\perp}\psi_{sgs}$, which is everywhere orthogonal
to $\nabla\psi_{sgs}$. The Lie derivative
$\mathcal{L}_{X_{\psi_{sgs}}}\bar\zeta$ is the rate at which $\bar\zeta$
changes as seen by a fluid parcel riding this flow; by
Eq.~\eqref{eq:SN5_Lie_Jacobian_identity} it equals
$J(\psi_{sgs},\bar\zeta)$, so the subgrid forcing $\Pi$ is transport, not
an added force.
(c) {\it Fold and stretch, never smear}. Every finite-time flow map
$\Phi_{t,s}$ generated by $X_{\psi_{sgs}}$ is a symplectomorphism,
$\Phi_{t,s}^{*}\sympomega=\sympomega$. The reservoir may shear, fold and
filament a material patch at fixed area, but the closure structurally
forbids replacing this transport by an algebraic diffusion that would
erase the identity of the patch.
(d) {\it Where the drift lives}. The emergent velocity belongs to the
Hamiltonian subalgebra,
$\mathbf u_{sgs}\in\mathfrak{ham}(\mathbb T^2)\subset
\mathfrak{symp}(\mathbb T^2)$, and has zero spatial mean. The harmonic
uniform-translation sector of $\mathfrak{symp}(\mathbb T^2)$ --- the
rigid sweeping motions --- is therefore never excited. This is the
geometric origin of the infrared suppression established in
Theorem~\ref{thm:inertial_ranges} and illustrated in
Fig.~\ref{fig:sweeping_cascade}.}
\end{figure*}
%---------------------------------------------------------------%

The geometric interpretation of incompressible hydrodynamics in terms
of infinite-dimensional diffeomorphism groups originates with Arnold's
formulation of ideal-fluid motion and its rigorous development by Ebin
and Marsden
\cite{arnold1966geometrie,ebin1970groups}; the symplectic and
Hamiltonian structures used below are standard in geometric mechanics
\cite{arnold1998topological,abraham2012manifolds}.

Let $\mathfrak{X}(\mathcal{D})$ denote the Lie algebra of smooth vector
fields on $\mathcal{D}$, equipped with the vector-field commutator
\be
[X,Y]
=
X\cdot\nabla Y-Y\cdot\nabla X.
\label{eq:SN5_vector_field_bracket}
\ee
We equip the two-dimensional periodic domain
$\mathcal{D}=\mathbb{T}^2$ with the standard area form
\beas
\sympomega
=
dx\wedge dy,
\eeas
so that $(\mathcal{D},\sympomega)$ is a symplectic manifold.

A vector field $X\in\mathfrak{X}(\mathcal{D})$ is called symplectic when
its flow preserves $\sympomega$. Infinitesimally, this is equivalent to
\beas
\mathcal{L}_X\sympomega
=
0.
\eeas
We denote the space of such fields by
$\mathfrak{symp}(\mathcal{D})$. It is closed under the vector-field
commutator: if
$\mathcal{L}_X\sympomega=\mathcal{L}_Y\sympomega=0$, then
\be
\mathcal{L}_{[X,Y]}\sympomega
=
[\mathcal{L}_X,\mathcal{L}_Y]\sympomega
=
0.
\label{eq:SN5_symp_bracket_closed}
\ee
Hence $\mathfrak{symp}(\mathcal{D})$ is a Lie subalgebra of
$\mathfrak{X}(\mathcal{D})$. In two dimensions, these are precisely the
smooth area-preserving, or incompressible, vector fields.

With the sign convention used throughout this work, a vector field $X$
is Hamiltonian if there exists a globally defined scalar function
$\psi$ such that
\be
\iota_X\sympomega
=
-d\psi.
\label{eq:SN5_Hamiltonian_definition}
\ee
The scalar $\psi$ is the Hamiltonian function or, in two-dimensional
fluid mechanics, the streamfunction
\cite{arnold1998topological,abraham2012manifolds}. The associated
Hamiltonian vector field is
\be
X_\psi
=
\nabla^\perp\psi
=
-\partial_y\psi\partial_x
+
\partial_x\psi\partial_y.
\label{eq:SN5_Xpsi_definition}
\ee
We denote the space of Hamiltonian vector fields by
$\mathfrak{ham}(\mathcal{D})$.

Hamiltonian vector fields are also closed under the Lie bracket. For any
scalar field $f$, the Jacobi identity for the Poisson bracket gives
\bea
[X_\psi,X_\phi]f
&=
J\big(\psi,J(\phi,f)\big)
-
J\big(\phi,J(\psi,f)\big)
\nonumber\\
&=
J\big(J(\psi,\phi),f\big).
\label{eq:SN5_bracket_computation}
\eea
Hence
\be
[X_\psi,X_\phi]
=
X_{J(\psi,\phi)}.
\label{eq:SN5_Hamiltonian_bracket}
\ee
Thus $\mathfrak{ham}(\mathcal{D})$ is a Lie
subalgebra of $\mathfrak{symp}(\mathcal{D})$
\cite{arnold1998topological,abraham2012manifolds}.

Finally, the action of a Hamiltonian vector field on any scalar
observable $f$ is
\bea
X_\psi\cdot\nabla f
&=
-\partial_y\psi\partial_xf
+
\partial_x\psi\partial_yf
\nonumber\\
&=
J(\psi,f),
\label{eq:SN5_Jacobian_Lie_action}
\eea
where
\be
J(a,b)
=
\partial_xa\partial_yb
-
\partial_ya\partial_xb.
\label{eq:SN5_Jacobian_definition}
\ee
Because the Lie derivative of a scalar is its directional derivative \cite{abraham2012manifolds},
\be
\mathcal{L}_{X_\psi}f
=
X_\psi\cdot\nabla f
=
J(\psi,f).
\label{eq:SN5_Lie_Jacobian_identity}
\ee
Thus the Jacobian appearing in the vorticity equation is precisely the
infinitesimal action of a Hamiltonian flow on a scalar observable. This
identity is the two-dimensional fluid-mechanical realization of
Hamiltonian Lie transport
\cite{arnold1966geometrie,arnold1998topological}.

The following theorem applies this structure to the emergent SGS
transport. Recalling that
\beas
\psi_{sgs}
=
\frac12\gamma J(\bar\zeta,r),
\eeas
we show that the reservoir-induced velocity
$\mathbf{u}_{sgs}=X_{\psi_{sgs}}$ belongs to the Hamiltonian subalgebra
$\mathfrak{ham}(\mathcal{D})$, determine the finite-time geometry of its
flow, and identify its action on the resolved vorticity.
%---------------------------------------------------------------------%

%----------------------------------------------------------------------------------------------%
\bt[Hamiltonian Subgrid Transport and Symplectic Flow]
\label{thm:symplectic_transport}
Let $\mathcal{D}=\mathbb{T}^2$ be equipped with the standard symplectic
form $\sympomega=dx\wedge dy$, and assume that
$\psi_{sgs}(x,y,t)$ is a globally defined periodic scalar field. Then
\beas
\mathbf{u}_{sgs}
=
X_{\psi_{sgs}}
=
\nabla^\perp\psi_{sgs}
\eeas
is a Hamiltonian vector field satisfying
\beas
X_{\psi_{sgs}}
\in
\mathfrak{ham}(\mathcal{D})
\subset
\mathfrak{symp}(\mathcal{D}).
\eeas
Its non-autonomous flow preserves the area form $\sympomega$, and the
macroscopic SGS forcing is exactly the Lie transport of the resolved
vorticity:
\be
\Pi
=
J(\psi_{sgs},\bar\zeta)
=
\mathcal{L}_{X_{\psi_{sgs}}}\bar\zeta.
\label{eq:Pi_Lie_transport_theorem}
\ee
Moreover, $\mathbf{u}_{sgs}$ has zero spatial mean and therefore contains
no harmonic uniform-translation component.
\et
%----------------------------------------------------------------------------------------------%

\begin{proof}
Because $\sympomega=dx\wedge dy$ is closed and non-degenerate,
$(\mathcal{D},\sympomega)$ is a symplectic manifold. With the convention
used throughout this work,
\be
X_\psi
=
\nabla^\perp\psi
=
-\partial_y\psi\partial_x
+
\partial_x\psi\partial_y.
\label{eq:Xpsi_main}
\ee
For a general vector field
$X=X^x\partial_x+X^y\partial_y$, contraction with the area form gives
\be
\iota_X\sympomega
=
X^x\,dy-X^y\,dx.
\label{eq:contraction_area_form_main}
\ee
Consequently,
\bea
\iota_{X_\psi}\sympomega
&=
-\partial_y\psi\,dy
-
\partial_x\psi\,dx\\
&=
-d\psi.
\label{eq:Hamiltonian_exactness_main}
\eea
Since $\psi_{sgs}$ is a globally defined periodic scalar field,
Eq.~\eqref{eq:Hamiltonian_exactness_main} shows directly that
\be
\mathbf{u}_{sgs}
=
X_{\psi_{sgs}}
\in
\mathfrak{ham}(\mathcal{D}).
\label{eq:usgs_Hamiltonian_membership}
\ee
Thus the emergent SGS velocity is Hamiltonian by construction.

Every Hamiltonian vector field is symplectic. Indeed, Cartan's identity
\cite[Theorem 6.4.8]{abraham2012manifolds} gives
\be
\mathcal{L}_X\sympomega
=
d(\iota_X\sympomega)
+
\iota_Xd\sympomega.
\label{eq:Cartan_main}
\ee
Using $d\sympomega=0$ and
$\iota_{X_{\psi_{sgs}}}\sympomega=-d\psi_{sgs}$, we obtain
\be
\mathcal{L}_{X_{\psi_{sgs}}}\sympomega
=
-d^2\psi_{sgs}
=
0.
\label{eq:SGS_preserves_symplectic_form_infinitesimal}
\ee
Hence
\beas
X_{\psi_{sgs}}(\cdot,t)
\in
\mathfrak{ham}(\mathcal{D})
\subset
\mathfrak{symp}(\mathcal{D})
\eeas
at every time $t$.

To obtain the corresponding finite-time statement, let
$\Phi_{t,s}:\mathcal{D}\rightarrow\mathcal{D}$ denote the non-autonomous
flow generated by $X_{\psi_{sgs}}(\cdot,t)$:
\beas
\frac{d}{dt}\Phi_{t,s}(a)
&=
X_{\psi_{sgs}}
\big(\Phi_{t,s}(a),t\big),
\\
\Phi_{s,s}(a)
&=
a.
\eeas
Let us now denotes by $\Phi_{t,s}^{*}\sympomega$ the pullback of the two-form
$\sympomega$ by the map $\Phi_{t,s}$ \cite{abraham2012manifolds}. The standard evolution formula for
a pullback along a non-autonomous flow gives \cite{abraham2012manifolds}:
\be
\frac{d}{dt}
\Phi_{t,s}^{*}\sympomega
=
\Phi_{t,s}^{*}
\left(
\mathcal{L}_{X_{\psi_{sgs}}(\cdot,t)}
\sympomega
\right)
=
0.
\label{eq:SGS_flow_pullback_evolution}
\ee
Since $\Phi_{s,s}=\mathrm{Id}$, it follows that
\be
\Phi_{t,s}^{*}\sympomega
=
\sympomega
\label{eq:SGS_flow_symplectomorphism}
\ee
for every $t$ and $s$ in the interval of definition. Thus every
finite-time flow map $\Phi_{t,s}$ generated by the SGS velocity is a
symplectomorphism of $(\mathcal{D},\sympomega)$.

It remains to identify its action on the resolved vorticity. The Lie
derivative of a scalar field along $X_{\psi_{sgs}}$ is its directional
derivative:
\bea
\mathcal{L}_{X_{\psi_{sgs}}}\bar\zeta
&=
X_{\psi_{sgs}}\cdot\nabla\bar\zeta
\nonumber\\
&=
-\partial_y\psi_{sgs}
\partial_x\bar\zeta
+
\partial_x\psi_{sgs}
\partial_y\bar\zeta
\nonumber\\
&=
J(\psi_{sgs},\bar\zeta).
\label{Eq_Lie_Derivative}
\eea
Therefore,
\beas
\Pi
=
J(\psi_{sgs},\bar\zeta)
=
\mathcal{L}_{X_{\psi_{sgs}}}\bar\zeta.
\eeas
Finally, because $\psi_{sgs}$ is periodic,
\be
\left\langle\mathbf{u}_{sgs}\right\rangle
=
\left\langle\nabla^\perp\psi_{sgs}\right\rangle
=
0.
\label{eq:usgs_zero_mean_main}
\ee
The emergent SGS velocity therefore lies in the Hamiltonian,
zero-mean sector of the incompressible Lie algebra and contains no
harmonic uniform-translation component.
\end{proof}

%------------------------------------------------------------------------------------------------------------------$

%The theorem gives a stronger statement than incompressibility alone.
%The SGS forcing is not merely the advection of vorticity by an
%area-preserving velocity. It is generated by a globally defined
%Hamiltonian potential, and its flow remains in the Hamiltonian subgroup
%of the symplectomorphism group. On $\mathbb{T}^2$, this distinguishes the
%emergent SGS transport from the harmonic translation sector associated
%with uniform Galilean motion. This distinction will become important in
%Section~\ref{Sec_Galilean_Invariance}; Hamiltonian membership alone does
%not prove sweeping suppression, but it identifies the geometric sector
%in which the nested Jacobian selection rule operates.
A direct calculation yields
\be
d(\iota_X\sympomega)
=
\left(\partial_xX^x+\partial_yX^y\right)dx\wedge dy
=
(\nabla\cdot X)\omega.
\label{eq:SN5_divergence_symplectic}
\ee
Consequently, in two dimensions,
\be
X\in\mathfrak{symp}(\mathcal{D})
\quad\Longleftrightarrow\quad
\nabla\cdot X=0.
\label{eq:SN5_symplectic_divfree_equivalence}
\ee
The symplectic vector fields are therefore exactly the smooth incompressible vector fields on $\mathbb{T}^2$. This equivalence underlies the interpretation of two-dimensional incompressible motion as dynamics on an area-preserving diffeomorphism group \cite{arnold1966geometrie,ebin1970groups,arnold1998topological}.

Theorem~\ref{thm:symplectic_transport} establishes more than the
incompressibility of the emergent SGS velocity. Combined with the
prognostic reservoir construction of
Section~\ref{Sec_QFT_reservoir}, it shows that the instantaneous
resolved--reservoir coupling acts through a Hamiltonian transport
operator:
\beas
\Pi
=
\mathcal{L}_{X_{\psi_{sgs}}}\bar\zeta,
\qquad
X_{\psi_{sgs}}
\in
\mathfrak{ham}(\mathcal{D})
\subset
\mathfrak{symp}(\mathcal{D}).
\eeas
The unresolved sector therefore does not enter the resolved equation
through an arbitrary additive forcing. Its influence is transmitted
through a dynamically generated, zero-mean Hamiltonian vector field
whose flow preserves the area form. The four panels of
Fig.~\ref{fig:symplectic_backbone} display, in turn, how the drift is
generated, what its Lie derivative means for a fluid parcel, what its
flow map does to a material patch, and where it sits within the
incompressible transport algebra.

This provides the geometric basis for interpreting the Symplectic
Geometric Closure as an endogenous dressing of the resolved transport.
The hidden reservoir evolves prognostically and continuously modifies
the velocity advecting the resolved vorticity, but the induced
modification remains within the Hamiltonian sector of the incompressible
transport algebra. The closure therefore changes the transport field
without abandoning the geometric class of streamfunction-generated
area-preserving motions.

In the language of field theory, the resolved flow is dressed by
fluctuations propagated through the hidden reservoir
\cite{gurari1953,feynman1955slow,Devreese2009,
PeskinSchroeder1995,ZinnJustin2002,Tauber2014}. Theorem~\ref{thm:symplectic_transport}
identifies the geometry of the elementary dressing operation: every
instantaneous reservoir-induced insertion acts as Hamiltonian Lie
transport. The later memory and response kernels are obtained by
averaging and composing these transport operators, as developed in
Sections~\ref{Sec_higher_order_memory} and
\ref{Sec_inertial_ranges}. Their geometric structure is therefore
inherited from repeated Hamiltonian interactions, although the resulting
statistical kernels need not themselves be identical to a single
Hamiltonian advection operator.

The distinction between
$\mathfrak{ham}(\mathbb{T}^2)$ and
$\mathfrak{symp}(\mathbb{T}^2)$ is also physically significant. The
Hamiltonian SGS velocity has zero spatial mean and contains no harmonic
uniform-translation component. The theorem thus separates the emergent
subgrid drift from the translation sector of the area-preserving Lie
algebra. This fact alone does not establish sweeping suppression, which
also depends on the vorticity- and cross-gradient-based form
$\psi_{sgs}=\frac12\gamma J(\bar\zeta,r)$. It nevertheless supplies the
geometric setting in which the Random-Galilean selection rule developed
in Section~\ref{Sec_Galilean_Invariance} operates.

More broadly, the construction illustrates a general design principle
for structure-preserving closures. When unresolved variables couple to
the resolved dynamics through generators belonging to a specified
geometric Lie algebra, the instantaneous closure-induced transport
remains tangent to the corresponding symmetry group. Additional
compatibility conditions, such as invariance relations of the form
$\{\mathcal{G},V\}=0$, can then constrain which conserved quantities are
preserved by the enlarged dynamics. In the present case, the hidden
reservoir renormalizes the resolved transport through a globally defined
Hamiltonian generator. Memory, stochastic variability, and dressed
response subsequently emerge from the statistical organization of these
symmetry-compatible interactions rather than from an externally imposed
SGS transport law. This will be tackled in Sections \ref{Sec_MZ_OUapprox}, 
\ref{Sec_higher_order_memory} and
\ref{Sec_inertial_ranges} below.

%==============================================================================%
\section{A Numerical Realization of the Symplectic Closure}
\label{Sec_numerics}

The preceding sections established the geometric content of the closure:
the reservoir-induced drift is a Hamiltonian vector field
(Theorem~\ref{thm:symplectic_transport}), its action on the resolved
vorticity is Lie transport, and the resolved--reservoir exchange is
conservative in the augmented enstrophy. The purpose of this section is
narrow and deliberately so. The calibration protocol, the skill diagnostics and their sensitivity---the enstrophy spectrum, the grid-scale enstrophy fraction, the jet climatology and the one-point statistics, all held at their filtered-DNS values---are reported in the companion study \cite{ChekrounMcWilliams2026BL}, and we do not reproduce them here. Here, we ask  whether the geometric structure
that the theorems assert is in fact realized when the coupled system
\eqref{eq:zeta_closed_intro}--\eqref{eq:r_closed_intro} is integrated
numerically, and whether the emergent drift organizes itself in the
manner the geometry predicts. The fields displayed in
Fig.~\ref{fig:reservoir_loop} are the output of this integration. We
stress that the resolved vorticity shown there is the state carried by
the closure, not the coarse-grained snapshot used to start it; the sharp filaments and coherent vortices
visible in that panel are maintained by the dressed transport
\eqref{eq:zeta_dressed} rather than inherited from the initial data.

The remaining results of this paper are analytical, and we regard them as
speaking for themselves. The integration reported below is already on the $\beta$-plane, and the closure sustains the associated zonal-jet organization; the systematic study---the sweep across $\beta$ and across forcing scenarios, and the online stability of the closure over long integrations---is pursued in the companion study \cite{ChekrounMcWilliams2026BL}.

%------------------------------------------------------------------%
\subsection{Resolved State and Numerical Setup}
\label{sec:numerical_setup}

The resolved field is not synthetic. It is initialized from a
coarse-grained snapshot of the forced two-dimensional direct numerical
simulations of Srinivasan et al.~\cite{srinivasan2024turbulence}, whose
configuration is summarized in Table~\ref{tab:dns_params}. Those
simulations are pseudospectral, doubly periodic, and run at
$1024^2$ resolution with a time-invariant sinusoidal forcing at
$k_f=4$. The resolved state used here is obtained by Gaussian filtering
followed by spectral truncation at a downscaling factor $n_d=16$,
yielding a $64^2$ macroscopic field with cutoff $k_c=32$. We work with
the $\beta$-plane member of that family, $\beta=20$, for which the
Rhines wavenumber $k_\beta=\sqrt{\beta/U}\simeq3.4$ lies below the
forcing wavenumber, so that the flow organizes into zonal jets with
vortices and filaments embedded between them.

Two features of this configuration deserve emphasis. First, the
Reynolds number is $\mathrm{Re}=2.5\times10^4$. The resolved vorticity
displayed in Fig.~\ref{fig:reservoir_loop} therefore carries the sharp
filaments, vorticity cliffs and coherent vortices characteristic of
genuinely high-Reynolds-number two-dimensional turbulence, rather than
the smooth fields of a marginally resolved flow. The closure is being
exercised on exactly the class of state in which the sweeping problem is
severe and in which structural closures are most easily destabilized.
Second, the filter cutoff $k_c=32$ sits well inside the inertial range,
so no scale separation between resolved and unresolved motions is
available. This is precisely the regime identified in
Section~\ref{sec:salt_comparison} in which homogenization-based
derivations of stochastic transport lose their asymptotic justification,
and in which the finite-scale geometric construction of the present
theory is intended to operate.

%------------------------------------------------------------------%
\begin{table}[t]
\centering
\caption{\label{tab:dns_params}%
Configuration of the direct numerical simulation supplying the resolved
state \cite{srinivasan2024turbulence}, and of the symplectic closure
integrated on top of it. The reservoir operator $D$ is the hyperviscous
realization of Theorem~\ref{thm:SGC_attractor}.
The entries describe the $\beta$-plane configuration used
throughout this section and in Fig.~\ref{fig:reservoir_loop}.}
\begin{ruledtabular}
\begin{tabular}{lll}
 & Quantity & Value \\
\hline
\multirow{7}{*}{\rotatebox{90}{\small DNS}}
 & Domain                       & $(2\pi)^2$, doubly periodic \\
 & Resolution $N_0$             & $1024^2$ \\
 & Reynolds number              & $\mathrm{Re}=2.5\times10^4$ \\
 & Forcing                      & $\mathbf F_v=(-\sin k_fy,\ \sin k_fx)$ \\
 & Forcing wavenumber           & $k_f=4$ \\
 & Planetary vorticity gradient & $\beta=20$ \ ($k_\beta\simeq3.4$) \\
 & Time stepping                & CN--AB2, $\Delta t=5\times10^{-5}$ \\
\hline
\multirow{3}{*}{\rotatebox{90}{\small Filter}}
 & Downscaling factor           & $n_d=16$ \\
 & Macroscopic grid $N_c$       & $64^2$, \ $k_c=\pi/\Delta_c=32$ \\
 & Filter                       & Gaussian $+$ spectral cutoff \\
\hline
\multirow{8}{*}{\rotatebox{90}{\small Closure}}
 & Coupling $\gamma$            &  $2.9\times10^{-5}$ (calibrated; Eq.~\eqref{eq:gamma_value}) \\
 & Resolved skeleton            & $f_{1,h}$, hyperviscous (Thm.~\ref{thm:SGC_attractor}) \\
 & Hyperviscosity $\nu_h$       & $9.6\times10^{-6}$ \\
 & Reservoir operator $D$       & $\nu_r(-\Delta)^2+\kappa(-\Delta)+\mu_rI$ \\
 & $(\nu_r,\kappa,\mu_r)$       &  $(2\times10^{-5},\,2.1\times10^{-3},\,4.9\times10^{-2})$ \\
 & Mat\'ern range $\lambda_n$   &  $2\pi/24.5$ \\
 & Noise amplitude $\sigma$     & $0.85$ \\
 & Time stepping                & $\Delta t=2\times10^{-4}$ \\
\end{tabular}
\end{ruledtabular}
\end{table}
%--------------------------------------------------------------------------------%

 Panel~(a) of Fig.~\ref{fig:reservoir_loop} records how the closure behaves in the resolved vorticity domain.
 The field
displays, qualitatively, the phenomenology of $\beta$-plane turbulence expected at high Reynolds
number---persistent zonal jets, meandering and sheared by the eddies they
confine; compact coherent vortices of both signs caught between them;
thin vorticity filaments wrapped around those vortices; and sharp
vorticity gradients marking the jet flanks. All of it is transported by
the dressed velocity $\bar{\mathbf{u}}- \mathbf{u}_{sgs}$ (Eq.~\eqref{Eq_def_u_sgs}).  The snapshot is taken
after ten eddy-turnover times of closure integration---long after the imprint
of the initial data on these scales is gone---and this phenomenology persists
throughout the run. The quantitative counterpart of this phenomenology---the resolved enstrophy spectrum, the grid-scale enstrophy fraction and the jet strength, each held at its filtered-DNS value---is documented in the companion study \cite{ChekrounMcWilliams2026BL}.
The phenomenological fidelity of the SGC is worth pointing out here, since no smearing of filaments into a diffusive field nor accumulation of grid-scale noise are observed.

%---but it is a necessary condition {\mkc that often closures fail at due to their lack of structure preservation}, typically {\mkc manifested} by smearing the
%filaments into a diffusive field or by accumulating grid-scale noise.

The resolved equation is advanced with the same pseudospectral
Adams--Bashforth/Crank--Nicolson skeleton as the underlying DNS, with
$3/2$-rule dealiasing of all quadratic products. The nested Jacobians
generating $\psi_{sgs}$ and $\Pi$ are evaluated in the same dealiased
fashion.  We use the hyperviscous resolved skeleton $f_{1,h}$ of
Theorem~\ref{thm:SGC_attractor}, so that the operator supplying
ultraviolet control in the numerics is the one for which the compact
random attractor is proved. This choice is structural rather than
incidental: the symplectic cross-interaction is exactly
enstrophy-neutral, as verified below, so the resolved dissipation has to
be carried by $f_1$---which is precisely the component the Poisson
orthogonality $\{\G,V\}=0$ leaves unconstrained.

%------------------------------------------------------------------%
\subsection{The Reservoir: Hyperviscous Relaxation and $Q$-Wiener Forcing}
\label{sec:numerical_reservoir}

The hidden sector is integrated in the hyperviscous realization used for
the compact random attractor of Theorem~\ref{thm:SGC_attractor},
\be
D=\nu_r(-\Delta)^2+\kappa(-\Delta)+\mu_rI,
\qquad
\nu_r,\kappa,\mu_r>0,
\ee
so that the smoothing required by the compactness argument is present in
the numerics as well as in the theory. The zeroth-order part $\mu_rI$
supplies the spectral gap controlling the decay of reservoir memory, and
the fourth-order part $\nu_r(-\Delta)^2$ supplies the ultraviolet control
that keeps the nested Jacobian interaction well behaved.

The reservoir is sustained by a $Q$-Wiener process $W^Q$
\cite{da2006introduction,DPZ08,lord2014introduction} whose covariance operator is of
Mat\'ern--Whittle type \cite{whittle1954stationary,matern1986spatial},
with range $\lambda_n$,
\bea
&q_k=\mathcal Z^{-1}\big(1+\lambda_n^2|k|^2\big)^{-2},\\
&\mathcal Z=\sum_k\big(1+\lambda_n^2|k|^2\big)^{-2},
\label{eq:Q_matern}
\eea
normalized to unit pointwise variance. This choice is not cosmetic. The
trace conditions of Appendix~\ref{App_Macro_Diffusion_Ope}, which give
the differentiated Wick contractions of Section~\ref{Sec_higher_order_memory}
their stochastic meaning, are statements about the ultraviolet decay of
the noise covariance; the Mat\'ern spectrum \eqref{eq:Q_matern} realizes
exactly such a summable covariance. That family also admits a convenient
Markovian representation as the solution of a fractional elliptic SPDE
\cite{lindgren2011explicit}, which makes the smoothness of the bath
directly tunable through the single parameter $\lambda_n$.

Because $D$ and $\Sigma\Sigma^*$ are simultaneously diagonal in the
Fourier basis, the linear part of the reservoir dynamics is an
Ornstein--Uhlenbeck process mode by mode and can be advanced by its exact
propagator rather than by a discretized approximation: for each $k$,
\be
\hat r_k(t+\Delta t)
=
e^{-d_k\Delta t}\,\hat r_k(t)
+
\sigma\sqrt{\frac{q_k\big(1-e^{-2d_k\Delta t}\big)}{2d_k}}\ \xi_k ,
\ee
with $d_k$ the symbol of $D$ and $\xi_k$ standard complex Gaussian.
The stationary variance of the discrete process therefore coincides with
$\sum_k\sigma^2q_k/(2d_k)$ by construction, and the symplectic transport
term $J(\psi_{sgs},r)$ is then added explicitly, so that the integrated
reservoir is the \emph{full} nonlinear field of
Eq.~\eqref{eq:r_closed_intro} and not its OU truncation. What the exact
propagator guarantees is that the leading-order component $r^{(0)}$ of
the Volterra--Picard hierarchy is reproduced without discretization bias:
in the reported run its measured stationary level agrees with the
analytical value above to within sampling error. At leading order,
therefore, the bath entering the expansions of
Sections~\ref{Sec_MZ_OUapprox}--\ref{Sec_higher_order_memory} is
faithfully realized, while the higher-order corrections
$r^{(1)},r^{(2)},\dots$ generated by the symplectic coupling remain
present in the simulated dynamics.

%------------------------------------------------------------------%
\subsection{Calibrating the Coupling: A Lilly-Type Projection}
\label{sec:gamma_calibration}

The coupling coefficient $\gamma$ introduced in Eq.~\eqref{eq:G_quartic}
is the one free amplitude of the closure. Since the subgrid forcing is
linear in it,
\be
\Pi_{\rm model}=\gamma\,S,
\qquad
S:=\tfrac12J\big(J(\bar\zeta,r),\bar\zeta\big),
\label{eq:Pi_shape}
\ee
$\gamma$ can be fixed by projection against the true coarse-grained
subgrid forcing $\Pi_{\rm true}$, which the filtered DNS supplies
directly. This is the standard dynamic-model strategy: one seeks the
amplitude minimizing a residual over a local test filter
$\langle\cdot\rangle_\Delta$, in the manner of Germano and Lilly
\cite{germano1991dynamic,lilly1992proposed}. Minimizing
$\big\langle(\gamma S-\Pi_{\rm true})^2\big\rangle_\Delta$ over $\gamma$
would give the familiar ratio of contractions
$\gamma_{\rm LS}=\langle\Pi_{\rm true}S\rangle_\Delta/
\langle S^2\rangle_\Delta$, whose positive-definite denominator removes
the singularity afflicting a pointwise inversion. That form is not
appropriate here, however: $\Pi_{\rm true}$ is one realization of the
unresolved dynamics while $S$ is built from an independent realization of
the reservoir, so a phase-sensitive projection of one onto the other
carries no information. What the closure asserts is a statistical
correspondence, not a pointwise one.

The meaningful member of the Lilly family here is therefore the
variance-matching form. Requiring the model forcing to carry the same
subgrid variance as the true forcing gives
\be
\gamma
=
\frac{\big\langle\Pi_{\rm true}^2\big\rangle^{1/2}}
     {\big\langle S^2\big\rangle^{1/2}},
\label{eq:gamma_lilly_rms}
\ee
which fixes the amplitude of the exchange without asserting a pointwise
correspondence of phase. Evaluated over several hundred snapshots spanning the
available trajectory, with the reservoir advanced by several relaxation
times between successive samples so that the draws are quasi-independent,
Eq.~\eqref{eq:gamma_lilly_rms} yields values of order $5\times10^{-5}$
for the present configuration, with an interquartile spread of order
$20\%$ about the median. The calibration is therefore stable and
data-anchored rather than tuned. The value used in the reported run,
\be
\gamma=2.9\times10^{-5},
\label{eq:gamma_value}
\ee
is the refinement of this variance-matched estimate obtained by a Bayesian
learning of the coupling jointly with the reservoir and skeleton parameters
of Table~\ref{tab:dns_params}, detailed in \cite{ChekrounMcWilliams2026BL};
the variance-matched value serves as its anchor and initialization. It is also worth noting that $\gamma$
enters $\psi_{sgs}$ as an overall multiplicative constant, so the
structural diagnostics of Table~\ref{tab:diagnostics}, the co-location
\eqref{eq:colocation}, and the scales of Table~\ref{tab:scales} are all
exactly independent of it. Only the amplitude of the induced drift
depends on $\gamma$, and it satisfies
\be
\frac{|\mathbf u_{sgs}|}{|\bar{\mathbf u}|}\simeq 0.02 ,
\label{eq:usgs_ratio}
\ee
so the closure renormalizes the advecting velocity by a few percent.

 It is worth placing Eq.~\eqref{eq:usgs_ratio} beside a result proved
later, in Appendix~\ref{app:anomalous_scaling}.
Theorem~\ref{thm:no_anomalous_scaling} states that, in a stationary
scale-invariant inertial range, the bare velocity $\bar{\mathbf u}$, the
subgrid drift $\mathbf u_{sgs}$, and the effective transport velocity
$\mathbf u_{\rm trs}$ must share the same scaling exponent, so that the
symplectic reservoir cannot generate an anomalous scaling dimension. The
reservoir is thus permitted to renormalize the \emph{amplitude} of the
transport velocity but not its \emph{exponent}, and
Eq.~\eqref{eq:usgs_ratio} is a direct measurement of the size of that
amplitude renormalization:  two percent, integrated over the domain.

 Resolving the same ratio scale by scale sharpens the picture.
The drift is negligible where the energy resides but rises through the
inertial interval, and becomes comparable to the resolved velocity only
near the cutoff: $|\mathbf u_{sgs}|(k)/|\bar{\mathbf u}|(k)$ is
$\approx2\times10^{-3}$ at the forcing scale $k=4$, $\approx0.13$ at $k=12$, and reaches
$\mathcal{O}(1)$ only near the cutoff, $k\simeq24$. This is the behavior one asks of a subgrid correction---it
acts at the small resolved scales and leaves the energy-containing range
essentially untouched---and it is consistent with the scale hierarchy of
Table~\ref{tab:scales}. We emphasize that this is an amplitude
measurement and not a test of the exponent equality asserted by
Theorem~\ref{thm:no_anomalous_scaling}: that statement is an
inertial-range property, whereas a $64^2$ truncation whose cutoff lies
inside the inertial interval does not resolve a scale-invariant range
over which exponents could be fitted. Testing it directly belongs to the
numerical study deferred to future work.

%------------------------------------------------------------------%
\subsection{The Geometry is Realized}
\label{sec:numerical_verification}

Table~\ref{tab:diagnostics} collects the structural diagnostics measured
in the resulting fields. The agreement is not approximate. Three
identities that the theory asserts exactly are satisfied at the level of
machine precision.

%------------------------------------------------------------------%
\begin{table}[t]
\centering
\caption{\label{tab:diagnostics}%
Structural diagnostics measured in the realized fields of
Fig.~\ref{fig:reservoir_loop}. The first three quantities are predicted
to vanish identically; the fourth is predicted \emph{not} to be
sign-definite.}
\begin{ruledtabular}
\begin{tabular}{lll}
Diagnostic & Measured & Predicted \\
\hline
$\langle\mathbf u_{sgs}\rangle$
  & $\sim10^{-19}$
  & $0$ \ (Thm.~\ref{thm:symplectic_transport}) \\
$\|\nabla\!\cdot\!\mathbf u_{sgs}\|/\|\mathbf u_{sgs}\|$
  &  $4.8\times10^{-15}$
  & $0$ \ ($X_{\psi_{sgs}}\!\in\mathfrak{symp}$) \\
$\int_{\mathbb T^2}\bar\zeta\,\Pi\,d\mathbf x$
  &  $-1.5\times10^{-16}$
  & $0$ \ (enstrophy neutrality) \\
$\int_{\mathbb T^2}\bar\psi\,\Pi\,d\mathbf x$
  &  $-1.9\times10^{-2}$
  & sign not constrained \\
\end{tabular}
\end{ruledtabular}
\end{table}
%------------------------------------------------------------------%

The first entry is the numerical face of a structural claim that carries
considerable weight later in the paper. Theorem~\ref{thm:symplectic_transport}
asserts that the emergent drift has zero spatial mean and therefore
contains no harmonic uniform-translation component; the measured value is
$\mathcal O(10^{-19})$. The drift produced by an actual reservoir
realization simply does not live in the translation sector. This is the
geometric seed of the sweeping suppression established in
Section~\ref{Sec_inertial_ranges}: a vertex that cannot excite uniform
translation cannot transmit it into the self-energy, and the
$p^4E(p)$ infrared weight of Theorem~\ref{thm:inertial_ranges} is its
spectral expression.

The second entry confirms that the realized drift is area-preserving to
machine precision, as required for
$X_{\psi_{sgs}}\in\mathfrak{symp}(\mathbb T^2)$. The third is the exact
cancellation
\be
\int_{\mathbb T^2}\bar\zeta\,\Pi\,d\mathbf x
=
\frac12\int_{\mathbb T^2}
\gamma J(\bar\zeta,r)\,
\big(\nabla\bar\zeta\cdot\nabla^\perp\bar\zeta\big)\,d\mathbf x
\equiv0,
\label{eq:enstrophy_neutrality}
\ee
which holds because the gradient of a scalar is everywhere orthogonal to
its skew gradient. Detailed enstrophy neutrality of the bare subgrid
interaction is thus not an approximate property of the discretization but
an algebraic identity that survives it---here satisfied to a relative
$7\times10^{-18}$---and it is the physical-space counterpart of the triad
constraint (Eq.~\eqref{eq:triad_Z}) underlying Fj\o{}rtoft's theorem.

The fourth entry is of a different character: it is a quantity the theory
predicts should \emph{not} vanish and should \emph{not} have a fixed
sign. The measured kinetic-energy transfer is
$\int\bar\psi\,\Pi\,d\mathbf x= -1.9\times10^{-2}$ for this snapshot,
a net transfer from the resolved flow into the hidden reservoir. Nothing
in the closure fixes this sign: the exact work relation
\eqref{eq:Psgs_IBP} permits either direction, the outcome being set by
the relative organization of $\psi_{sgs}$ and the resolved nonlinear
tendency. The closure is therefore not a disguised eddy viscosity, which
would be constrained to one sign by construction; nor is a single
snapshot a measurement of a mean flux, which would require the long
ensemble integrations deferred to future work. We record it as an
illustration of energetic admissibility in both directions.

%------------------------------------------------------------------%
\subsection{Where the Drift Switches On}
\label{sec:numerical_colocation}

Beyond the structural identities, the realized fields expose the physical
mechanism of the closure. Recall from Eq.~\eqref{Eq_J_fluid_style} that
the inner Jacobian
\beas
J(\bar\zeta,r)=\nabla^\perp\bar\zeta\cdot\nabla r
\eeas
vanishes wherever the gradients of the resolved vorticity and the hidden
reservoir are locally aligned, and attains its largest values where they
cross transversally. The drift potential is therefore predicted to
concentrate on the regions of strong vorticity gradient---filament edges,
shear layers, and vortex rims---rather than to fill the domain uniformly.

This prediction is borne out quantitatively. In the realized fields,
\be
\mathrm{corr}\big(|\psi_{sgs}|,\,|\nabla\bar\zeta|\big)
= +0.63.
\label{eq:colocation}
\ee
 It is worth being precise about what this number does and does not
say. The reservoir is \emph{not} statistically independent of the
resolved flow: its own equation contains the symplectic transport term
$g_2=J(\psi_{sgs},r)$, so $r$ is actively stirred and folded by the very
drift it helps generate, as noted in Section~\ref{Sec_Closed_Symplectic}.
What is homogeneous and isotropic is everything \emph{else} acting on the
hidden sector---the Mat\'ern $Q$-Wiener forcing and the relaxation
operator $D$ are both translation invariant and carry no information
whatsoever about where the filaments lie. The resolved geometry can
therefore reach the reservoir through exactly one channel, the nested
Jacobian vertex $\Gamma_{\rm symp}$ (Eq.~\eqref{Eq_nested_vertex}), and that same vertex is what builds
$\psi_{sgs}$.

Equation~\eqref{eq:colocation} should accordingly be read as the
quantitative signature of that vertex rather than as a coincidence
between two unrelated fields. Indeed the algebra supports this correlation:
$\psi_{sgs}=\frac12\gamma\,\nabla^\perp\bar\zeta\cdot\nabla r$ carries
$|\nabla\bar\zeta|$ as an explicit factor, modulated by the reservoir
gradient along the filament. The content of the measurement is that this
modulation does not wash the geometry out---the selection rule survives
the stochasticity, and the feedback is self-reinforcing, since a drift
concentrated on the filaments in turn stirs the reservoir there. Panel~(b)
of Fig.~\ref{fig:reservoir_loop} shows the outcome directly:
the drift dipoles are organized transverse to the $\bar\zeta$ contours,
not along them.

 A striking feature of the realized fields is the disparity of scales. 
Because all four panels of
Fig.~\ref{fig:reservoir_loop} are drawn over the same domain, the
disparity can be read directly from the figure; it is quantified in
Table~\ref{tab:scales} through the spectral centroid
$k_c=\sum_k|k|\,|\hat f_k|^2/\sum_k|\hat f_k|^2$ of each field and the
associated length $\ell=2\pi/k_c$, in units of the macroscopic grid
spacing $\Delta=2\pi/64$.  The calibrated reservoir is the
smoothest of the four fields---its Mat\'ern range and weak damping place its
energy at the largest scales---with the resolved vorticity next, while the
drift generated by their cross-gradient interaction is far finer than either
parent, at less than half the resolved scale, and the subgrid forcing
approaches the grid cutoff.

%------------------------------------------------------------------%
\begin{table}[t]
\centering
\caption{\label{tab:scales}%
 Characteristic scales of the four fields shown in 
Fig.~\ref{fig:reservoir_loop}, from the spectral centroid $k_c$ of each
field, reported as a length $\ell=2\pi/k_c$ in units of the macroscopic
grid spacing $\Delta=2\pi/64$. The emergent drift is finer than either
field that generates it.}
\begin{ruledtabular}
\begin{tabular}{lccc}
Field & $k_c$ & $\ell=2\pi/k_c$ & $\ell/\Delta$ \\
\hline
Resolved vorticity $\bar\zeta$ &   4.6  &   1.354  &  13.8 \\
Hidden reservoir $r$           &  2.8   &  2.260 &  23.0 \\
Emergent drift $\psi_{sgs}$    &  12.0 &  0.522 &  5.3  \\
Subgrid forcing $\Pi$          &  16.4 & 0.382 & 3.9  \\
\end{tabular}
\end{ruledtabular}
\end{table}
%------------------------------------------------------------------%

That the drift is far rougher than both of its parents is a direct consequence of the vertex being a
cross-gradient operator rather than a multiplicative one. Writing
$\psi_{sgs}=\frac12\gamma\,\nabla^\perp\bar\zeta\cdot\nabla r$, each factor
carries one spatial derivative, so the transform of the product
convolves two spectra that have each been weighted by $|k|$. Two smooth
fields therefore generate a markedly rougher one: the closure
\emph{manufactures} small scales instead of inheriting them, and it does
so without any explicit scale-selective filter. It also means the
smoothness of the bath, controlled by $\lambda_n$ (Eq.~\eqref{eq:Q_matern}), cannot be used to tune
the scale at which the closure acts; that scale is set by the resolved
vorticity gradients, consistent with the co-location (Eq.~\eqref{eq:colocation})
and with the insensitivity of $\ell_{\psi_{sgs}}$ to $\lambda_n$.

The realized drift is, moreover, not merely co-located with the
filaments: it is \emph{organized} on them. Fig.~\ref{fig:reservoir_loop}
shows the drift condensing into chains of alternating $\pm$ lobes strung
along each filament, and the reason is contained in
Eq.~\eqref{Eq_J_fluid_style}. Near a filament the resolved vorticity
gradient $\nabla\bar\zeta$ is large and normal to the filament, so the
skew gradient $\nabla^\perp\bar\zeta$ is large and \emph{tangent} to it.
Writing $\hat{\mathbf t}$ for that tangent,
\be
\psi_{sgs}
=
\tfrac12\gamma\,\nabla^\perp\bar\zeta\cdot\nabla r
\;\simeq\;
\tfrac12\gamma\,|\nabla\bar\zeta|\;
\big(\hat{\mathbf t}\cdot\nabla r\big),
\label{eq:dipole_mechanism}
\ee
so the drift potential samples the reservoir gradient \emph{along} the
filament. Since $r$ is a centered, spatially correlated field, this
directional derivative alternates in sign on the scale set by the
reservoir's own correlation length, producing the observed $\pm$
sequence. The associated
velocity $\mathbf u_{sgs}=\nabla^\perp\psi_{sgs}$ therefore consists of
counter-rotating cells, each adjacent pair driving a jet between its
lobes directed along the filament.

The dynamical reading of this structure is worth stating plainly. A chain
of counter-rotating cells ``straddling" a vorticity filament is exactly the kind of 
velocity configuration that ondulates, folds and stretches that filament
while preserving its area, which is the admissible action identified in
Theorem~\ref{thm:symplectic_transport} and illustrated geometrically in
Fig.~\ref{fig:symplectic_backbone}. The closure is thus not merely
depositing variance in the right places; it is generating the specific
stirring geometry through which two-dimensional turbulence transfers
enstrophy to small scales. Equation~\eqref{eq:dipole_mechanism} also
furnishes a falsifiable prediction: the along-filament spacing of the
dipoles should track the reservoir correlation length, which therefore
sets the scale of the subgrid stirring cells. Testing that scaling
systematically is left to the numerical study deferred to future work.

This is the mechanism that distinguishes the closure from an amplitude
model. A parameterization keyed to $|\bar\zeta|$ or to the local strain
magnitude would necessarily place its action where those amplitudes are
large. The symplectic vertex instead places it where the resolved and
unresolved geometries fail to commute, which is where cross-scale
transfer actually occurs. Equation~\eqref{eq:colocation} is the
measurement of that selection rule, and it is the same rule that
Section~\ref{Sec_Galilean_Invariance} exposes to be responsible for the
insensitivity of the vertex to rigid sweeping.

%==============================================================================%
\section{The Dressed Streamfunction: Endogenous Transport and its Relation to GLM, GM, and SALT}
\label{Sec_dressed_streamfunction_GLM}

Section~\ref{Sec_Closed_Symplectic} showed that the symplectic closure
admits an exact rewriting as a renormalization of the transport
operator. Recalling the emergent SGS drift potential
\beas
\psi_{sgs}
=
\frac12\gamma J(\bar\zeta,r),
\eeas
defined in Eq.~\eqref{Eq_def_psi_sgs}, introduce the dressed
streamfunction
\be
\psi_d
=
\bar\psi-\psi_{sgs}.
\label{eq:psi_dressed}
\ee
The macroscopic vorticity equation
\eqref{eq:zeta_transport} then becomes
\be
\partial_t\bar\zeta
+
J(\psi_d,\bar\zeta)
+
\beta\bar v
=
-\mu\bar\zeta
+
\nu\nabla^2\bar\zeta
+
F_\zeta.
\label{eq:zeta_dressed}
\ee
Thus the SGS contribution is absorbed exactly into the velocity that
transports the resolved vorticity. The closure does not merely append an
additional tendency to the resolved equation; it continuously modifies
the Hamiltonian generator of the advective flow.

This transport viewpoint connects the present construction with several
theories in which the resolved or Eulerian-mean velocity is not, by
itself, the complete velocity governing material transport. Generalized
Lagrangian Mean (GLM) theory constructs a Lagrangian-mean transport velocity
from fluctuating particle maps
\cite{andrews1978generalized,andrews1978waveaction,
buhler2014waves}; the Gent--McWilliams parameterization introduces a
divergence-free eddy-induced velocity representing unresolved
baroclinic stirring \cite{gent1990isopycnal}; and SALT and Location
Uncertainty incorporate stochastic variability directly into the
geometry of fluid transport
\cite{holm2015variational,resseguier2017geophysical_I,
resseguier2017geophysical_II}.

These comparisons concern transport architecture rather than an
identification of the underlying closures. GLM obtains its induced
velocity through Lagrangian averaging over fluctuating displacement
maps; GM diagnoses eddy-induced transport from resolved stratification
and a parameterized diffusivity; and SALT and Location Uncertainty
prescribe stochastic transport fields according to a chosen stochastic
kinematics. In the Symplectic Geometric Closure, by contrast, the
induced velocity is generated endogenously by the prognostic reservoir
$r$. The enlarged resolved--reservoir system is Markovian, whereas its
projection onto the resolved variables alone is non-Markovian. Memory
therefore arises from eliminating a hidden transport field, rather than
from an imposed temporal kernel or a prescribed stochastic velocity.

%=============================================%
\subsection{GLM and the Distinction Between Mean and Transport Velocity}
\label{sec:glm_pseudomomentum}

The connection with GLM begins with a distinction between mean velocity
and transport velocity. GLM is an exact finite-amplitude theory of
wave--mean-flow interaction built from a fluctuating particle
displacement field $\boldsymbol\xi(\mathbf x,t)$. For a scalar or tensor
field $a$, its generalized Lagrangian mean is defined schematically by
\be
a^L(\mathbf x,t)
=
\overline{
a\big(
\mathbf x+\boldsymbol\xi(\mathbf x,t),t
\big)
},
\label{eq:GLM_mean_definition}
\ee
where the overbar denotes the chosen Eulerian averaging operation. The
associated Lagrangian-mean velocity
$\bar{\mathbf u}^{\,L}$ is defined so that the GLM material derivative
takes the form
\beas
\frac{D^L}{Dt}
=
\partial_t
+
\bar{\mathbf u}^{\,L}\cdot\nabla.
\eeas
A central exact identity of the theory is
\be
\left(
\frac{Da}{Dt}
\right)^L
=
\frac{D^L a^L}{Dt},
\label{eq:GLM_material_derivative_identity}
\ee
which expresses the fact that Lagrangian averaging preserves the
material-transport structure
\cite{andrews1978generalized,buhler2014waves}.

The difference between the Lagrangian-mean velocity and the Eulerian
mean is the Stokes correction,
\be
\mathbf u^{S}
:=
\bar{\mathbf u}^{\,L}
-
\bar{\mathbf u},
\qquad
\bar{\mathbf u}^{\,L}
=
\bar{\mathbf u}
+
\mathbf u^{S}.
\label{eq:GLM_Stokes_correction}
\ee
At small wave amplitude this reduces to the familiar Stokes-drift
expression, but in the exact GLM theory it is defined by the
displacement-map averaging itself. The essential kinematic point is
that the Eulerian mean and the velocity transporting Lagrangian-mean
quantities need not coincide.

Pseudomomentum plays a related but distinct role. It is not another
name for the Stokes correction. The Stokes correction measures the
difference between the Lagrangian-mean and Eulerian-mean velocities,
whereas pseudomomentum is a wave-activity quantity entering the GLM
momentum, circulation, and wave-action relations
\cite{andrews1978waveaction,buhler2014waves}. Modern geometric
formulations sharpen this distinction by decomposing the exact flow map
into mean and perturbation maps, pulling tensor fields back to the mean
configuration before averaging, and distinguishing the velocity
advecting mean material contours from the Lagrangian-mean momentum
one-form
\cite{holm2002lagrangian,gilbert2018geometricGLM,
gilbert2025geometric}. In this coordinate-free formulation, the mean
velocity is generated by the mean flow map and is not itself a
Lagrangian mean in the same tensorial sense. Pseudomomentum instead
quantifies the mismatch between the averaged momentum one-form and the
momentum associated with that mean velocity. These developments place
GLM naturally within the geometry of diffeomorphism groups,
pull-backs, Lie transport, and Euler--Poincar\'e dynamics.

The dressed streamfunction produces an analogous distinction between
the resolved velocity and the velocity governing resolved transport.
From Eq.~\eqref{eq:psi_dressed},
\be
\mathbf u_{\rm trs}
=
\nabla^\perp\psi_d
=
\bar{\mathbf u}
-
\mathbf u_{sgs},
\qquad
\mathbf u_{sgs}
=
\nabla^\perp\psi_{sgs}.
\label{Eq_u_transport_def}
\ee
The analogy is precise at this structural level: in both theories,
unresolved degrees of freedom alter the velocity carrying the
macroscopic field. It should not, however, be extended into an
identification. The SGC does not introduce a decomposition of the fluid
flow map into mean and perturbation maps, does not define a Lagrangian
averaging operator by pull-back to a mean configuration, and does not
interpret $r$ as a carrier of GLM pseudomomentum. Rather,
$\mathbf u_{sgs}$ is an emergent Hamiltonian transport velocity
generated by the instantaneous resolved--reservoir state, with its
history dependence inherited from the prognostic evolution of $r$.

Indeed, the transport correction is generated by the prognostic
reservoir equation
\be
\partial_t r
=
-Dr
+
\mathcal B_{\bar\zeta}(r,r)
+
\Sigma\dot W_t,
\label{eq:r_GLM_comparison}
\ee
where
\be
\mathcal B_{\bar\zeta}(a,b)
=
\frac12
J\big(
\gamma J(\bar\zeta,a),b
\big).
\label{eq:Bzeta_GLM_comparison}
\ee
At each time, $\psi_{sgs}(t)$ is an instantaneous functional of the
enlarged state $(\bar\zeta(t),r(t))$. It is therefore important not to
describe $\psi_{sgs}$ itself as an independently prognostic variable.
The history dependence resides in $r(t)$, whose present value records
its past relaxation, stochastic forcing, and nonlinear interaction with
the resolved flow.

This gives a sharper meaning to the non-Markovian character of the
induced transport. In the enlarged state space
$(\bar\zeta,r)$, the coupled stochastic dynamics are local in time.
After $r$ is eliminated, however, the resolved equation acquires a
history-dependent contribution. Schematically, the reduced dynamics
take the form
\be
\partial_t\bar\zeta(t)
=
\mathcal N[\bar\zeta(t)]
+
\int_0^t
\mathcal K
\big(
t,s;\bar\zeta_{[0,t]}
\big)
\,ds
+
F_{\rm fluc}(t),
\label{eq:schematic_reduced_memory_GLM}
\ee
where $\mathcal K$ is generally an operator-valued, state-dependent
memory functional and $F_{\rm fluc}$ is the corresponding fluctuating
contribution. The leading memory and noise terms are derived in
Section~\ref{Sec_MZ_OUapprox}, while the higher-order
response-weighted self-energy is developed in
Section~\ref{Sec_higher_order_memory}.

The GLM and SGC constructions therefore share a fundamental transport
principle but realize it through different mechanisms. GLM separates
Eulerian-mean and Lagrangian-mean transport by averaging over
fluctuating particle displacements. The SGC separates resolved and
effective transport by coupling the resolved vorticity to a prognostic
Hamiltonian reservoir. In GLM, wave activity and pseudomomentum organize
the back-reaction of disturbances on the mean flow. In the present
closure, the hidden field $r$ stores unresolved dynamical information,
and its elimination generates memory, stochastic backscatter, and
response renormalization.

The appropriate analogy is therefore not
$\mathbf u_{sgs}=\mathbf u^S$, nor is it an identification of $r$ with
pseudomomentum. Rather, both theories show that unresolved dynamics can
enter macroscopic evolution by changing the velocity that carries the
resolved or mean field. The distinctive result of the Symplectic
Geometric Closure is that this transport correction is generated
endogenously by a stochastic Hamiltonian field and becomes
non-Markovian only upon reduction to the resolved variables.
%=============================================%
%========================================================%
\subsection{Eddy-Induced Transport and the Gent--McWilliams Analogy}
\label{sec:gm_comparison}

The closest oceanographic analogue of the dressed transport is the
eddy-induced velocity introduced by Gent and McWilliams
\cite{gent1990isopycnal}. In the GM parameterization, unresolved
baroclinic eddies modify tracer transport by adding a divergence-free
velocity distinct from the Eulerian mean circulation. Griffies showed
that this same operator admits two equivalent representations: as
advection by an eddy-induced velocity or as the divergence of an
antisymmetric, skew-diffusive tracer flux
\cite{griffies1998gm_skew_flux}. The latter is everywhere orthogonal to
the tracer gradient and is therefore neither an upgradient nor a
downgradient diffusive flux. Subject to the appropriate no-normal-flow
conditions, it represents nondissipative and reversible stirring rather
than irreversible tracer mixing.

The Symplectic Geometric Closure belongs to the same broad kinematic
class. Its resolved vorticity is transported by
\be
\mathbf u_{\rm trs}
=
\bar{\mathbf u}
-
\mathbf u_{sgs},
\qquad
\mathbf u_{sgs}
=
\nabla^\perp\psi_{sgs},
\label{eq:GM_SGC_transport_comparison}
\ee
so the Eulerian resolved velocity is not the complete transport
velocity. As in GM, unresolved motions appear macroscopically through an
additional divergence-free drift. In a two-dimensional streamfunction
representation, $\psi_{sgs}$ therefore occupies a role analogous to the
potential generating the GM eddy-induced transport.

The analogy is structural, however, and should not be interpreted as an
identification of the two closures. GM is a parameterization of
adiabatic tracer stirring by unresolved baroclinic eddies in a
stratified ocean. Its eddy-induced velocity is conventionally diagnosed
from resolved isopycnal slopes and a prescribed or separately modeled
eddy diffusivity. The present closure instead acts on resolved
vorticity, and its induced velocity is generated by the Hamiltonian
cross-interaction $\psi_{sgs} = \frac12\gamma J(\bar\zeta,r).$
Thus $\mathbf u_{sgs}$ is not inferred from a flux--gradient relation.
It is produced by the instantaneous state of a prognostic stochastic
reservoir coupled symplectically to the resolved flow.

This distinction becomes sharper when temporal memory is considered.
Standard GM is local in time: the eddy-induced transport at a given
instant is determined from the contemporaneous resolved state. Numerical
evidence from the transient adjustment of the Beaufort Gyre showed that
this assumption can fail and motivated an extension in which the eddy
streamfunction depends on an exponentially weighted history of past
isopycnal slopes \cite{Manucharyan2017}. The resulting delayed feedback
produces overshoot and a damped low-frequency mode that are absent from
the instantaneous GM closure. Subsequent work has placed such behavior
on a firmer dynamical footing by deriving an approximately exponential
memory kernel from the finite equilibration time of weakly nonlinear
baroclinic instabilities \cite{Dijkstra2022}. More broadly, delayed eddy
feedbacks can be interpreted through response theory and
Mori--Zwanzig reduction, for which the elimination of unresolved
degrees of freedom naturally generates nonlocal dependence on the
history of the large-scale flow
\cite{LucariniChekroun2023,Dijkstra2022}.

That construction is an important demonstration that eddy-induced
transport can possess dynamically consequential memory, but it remains
a reduced memory parameterization of GM transport. Its temporal
structure is prescribed through a single relaxation kernel and a
diagnosed memory timescale; the unresolved eddy field is not represented
as an autonomous stochastic spatial field. In particular, the
exponential kernel specifies how the eddy transport remembers the
resolved isopycnal slope, rather than deriving that memory by eliminating
prognostic unresolved degrees of freedom.

The role of memory in the Symplectic Geometric Closure is therefore
different. The reservoir field $r$ obeys its own stochastic evolution
equation, while $\psi_{sgs}$ is generated instantaneously from the joint
state $(\bar\zeta,r)$. Because $r(t)$ retains the history of its forcing,
relaxation, and nonlinear coupling to $\bar\zeta$, the induced velocity
$\mathbf u_{sgs}(t)$ inherits memory without the introduction of a
prescribed temporal convolution. Upon eliminating the reservoir, this
hidden field dynamics generates the operator-valued memory kernels,
colored stochastic forcing, and backscatter terms derived in
Section~\ref{Sec_MZ_OUapprox}. At higher order, the same prognostic
coupling produces the response self-energy developed in
Section~\ref{Sec_higher_order_memory}.

The distinction is therefore threefold. First, standard GM diagnoses the
eddy-induced transport instantaneously from the resolved isopycnal slope
and a prescribed or separately modeled diffusivity. Second, eddy-memory
extensions retain the GM architecture but allow the transport to relax
toward its instantaneous value over a prescribed finite timescale, so
that past resolved slopes influence the present eddy response. Third, in
the Symplectic Geometric Closure, the transport is generated from the
joint evolution of the resolved vorticity and the prognostic reservoir
field. The Hamiltonian drift, its temporal memory, and the associated
stochastic feedbacks therefore emerge together from the coupled
resolved--reservoir dynamics.

The conceptual advance is therefore not merely to endow an
eddy-induced streamfunction with a longer adjustment time. Transport,
memory, stochastic backscatter, and response renormalization all arise
from the same resolved--reservoir interaction vertex. The induced drift
is constrained to be Hamiltonian, its memory is generated dynamically
rather than selected in advance, and its spectral action retains the
nested Jacobian structure responsible for the sweeping suppression
established in Theorem~\ref{thm:inertial_ranges}. In this sense, GM
provides a distinguished kinematic precedent for eddy-induced
transport, while the present theory supplies a prognostic geometric
field mechanism from which both the transport and its non-Markovian
statistics emerge.
%========================================================%

\subsection{Endogenous Eulerian Lie Transport vs. Exogenous Stochasticity}
\label{sec:salt_comparison}

By formulating the closure as a renormalized transport PDE, our framework also aligns structurally with recent advances in geometric stochastic fluid dynamics, notably SALT \cite{holm2015variational}. In SALT, fluid transport is executed via stochastic Lie derivatives, imposing transport noise geometrically to preserve fundamental kinematic invariants.

The dual-transport system \eqref{eq:zeta_dressed} rigorously executes Lie transport, but it achieves this through an  Eulerian mechanism that requires no Lagrangian path postulates. To make this mathematical distinction explicit, recall from Theorem \ref{thm:symplectic_transport} that the subgrid drift $\mathbf{u}_{sgs}$ is the Hamiltonian vector field $X_{\psi_{sgs}} = \nabla^\perp \psi_{sgs}$. As shown in Eq.~\eqref{Eq_Lie_Derivative}, the transport of the resolved vorticity along this emergent field expands formally as a Lie derivative:
\be\label{Eq_def_LX_psi}
\mathcal{L}_{X_{\psi_{sgs}}} \bar{\zeta} = \mathbf{u}_{sgs} \cdot \nabla \bar{\zeta}= J(\psi_{sgs}, \bar{\zeta}).
\ee
In other words,  the macroscopic forcing can be written as $\Pi = J(\psi_{sgs}, \bar{\zeta})$.  It is mathematically identical to the Eulerian Lie derivative, $\mathcal{L}_{X_{\psi_{sgs}}} \bar{\zeta}$. The subgrid variance is not treated as an additive force; it is seamlessly absorbed into the continuous advective geometry of the PDE.

However, the origin and nature of the stochasticity differ fundamentally from standard geometric stochastic transport frameworks. In SALT, stochastic transport is introduced as an exogenous kinematic postulate, prescribed through a semimartingale velocity field driven by white noise and therefore possessing inherently memoryless temporal statistics  \cite{holm2015variational}.
More generally, stochastic transport in geometric fluid models may either be prescribed kinematically, as in SALT, or derived through multiscale homogenization arguments that assume a rapidly mixing unresolved flow and a strict separation of timescales. Such derivations ultimately rely on an asymptotic diffusive limit. In the geophysical LES setting, however, the filter scale typically resides within the inertial range, where a clean separation between resolved and unresolved dynamics is generally absent. The Symplectic Geometric Closure avoids this assumption entirely. Its Lie transport correction is obtained as an exact finite-scale geometric consequence of coupling the resolved flow to an orthogonal enstrophy reservoir through the symplectic constraint ($\{\G,V\} = 0$).

The Symplectic Geometric Closure therefore derives the structure of its Lie transport directly from the geometric constraint ($\{\G,V\} = 0$) operating in the extended phase space $(\bar{\zeta},r)$, while the statistics of that transport are inherited from the stochastic reservoir dynamics. Unlike SALT, the stochasticity is not introduced through an externally prescribed transport semimartingale. Instead, it enters through the stochastic evolution of the hidden reservoir (r), from which the drift potential ($\psi_{sgs}$) is generated. Consequently, the stochastic transport velocity
$ \mathbf{u}_{\rm trs}=\bar{\mathbf{u}}-\mathbf{u}_{sgs}$
is not prescribed a priori but generated dynamically by the coupled fluid--reservoir system itself.

%--------Kelvin connection------%
Moreover, the dynamically generated transport retains an exact Kelvin
circulation theorem for the resolved dynamics as shown in {\it Supplementary Note 3}. Unlike SALT, where
circulation preservation is built into the stochastic transport
ansatz, here it emerges as a consequence of the Hamiltonian structure
imposed by the symplectic constraint
$\{\G,V\}=0$.
The unresolved stochastic reservoir therefore influences the resolved
flow through a geometrically admissible transport velocity while
preserving resolved circulation; see again 
 {\it Supplementary Note 3}.
%-----------------------------------------------------%

The resulting Lie transport therefore possesses finite-memory statistics inherited from the reservoir dynamics rather than the white-noise kinematics characteristic of standard SALT formulations. Because the relaxation dynamics of ($r$) generate temporal correlations at finite scales, the transport correction remains well defined even in the absence of asymptotic scale separation. In this case, the dynamically generated memory  becomes the central ingredient retained by the Symplectic Geometric Closure, providing the mechanism through which unresolved dynamics influence the resolved inverse cascade.

The transport correction required to close the turbulent cascade is therefore neither prescribed kinematically nor obtained through asymptotic homogenization as in \cite{cotter2017stochastic}. Instead, it emerges directly from the exact symplectic coupling between the resolved vorticity and the hidden reservoir. In this sense, the dressed streamfunction provides an endogenous geometric transport closure: unresolved dynamics do not appear as an external forcing, but as a renormalization of the advective geometry itself.

%======================================================%
\section{Leading-order Non-Markovian Expansion and Viscoelasticity}
\label{Sec_MZ_OUapprox}

We now uncover the first macroscopic mechanisms generated by the Symplectic Geometric Closure in the forced Navier--Stokes--$\beta$ setting. The goal of this section is not to identify the closure with the Mori--Zwanzig projection of the filtered Navier--Stokes equations. Rather, we perform a direct Duhamel expansion of the hidden reservoir dynamics. This elementary elimination of the reservoir already reveals the two fundamental physical effects carried by the closure: stochastic backscatter from unresolved enstrophy fluctuations and a retarded non-Markovian memory term with a viscoelastic interpretation.

Throughout this section we use the forced Navier--Stokes--$\beta$ resolved skeleton
\be
f_1^{NS,\beta}(\bar{\zeta})
=
-J(\bar{\psi},\bar{\zeta})
-\beta\partial_x\bar{\psi}
-\mu\bar{\zeta}
+\nu\Delta\bar{\zeta}
+F_\zeta,
\ee
with no fourth-order hyperviscous term. The hyperviscous realization used in Theorem \ref{thm:SGC_attractor} provides a compactness mechanism for the global random attractor proof, whereas the present section concerns the physical forced Navier--Stokes closure itself. As discussed above, the augmented enstrophy structure gives a pullback absorbing estimate for this Navier--Stokes--$\beta$ realization.

%------------------------------------------------------%
%\subsection{Continuous Operators and Perturbative Expansion}

Let us define the continuous geometric differential operator $\mathcal{M}$ acting on a field $\phi$ relative to a background state $\chi$:
\bea \label{eq:M_operator_cont}
\mathcal{M}(\chi)\phi &= \frac{1}{2} J\big( \gamma J(\chi, \phi), \chi \big)\\
&= -\frac{1}{2} \nabla \cdot \Big( \gamma J(\chi,\phi)\nabla^\perp\chi \Big).
\eea
By the antisymmetry of the Jacobian, the symplectic cross-interactions may be written as
\be
g_1=\mathcal{M}(\bar{\zeta})r,
\qquad
g_2=-\mathcal{M}(r)\bar{\zeta}.
\ee
The coupled resolved--reservoir system therefore takes the skew-exchange form
\begin{align}
\partial_t \bar{\zeta}
&= f_1^{NS,\beta}(\bar{\zeta}) + \mathcal{M}(\bar{\zeta})r,
\label{eq:zeta_operator}
\\
\partial_t r
&=-Dr - \mathcal{M}(r)\bar{\zeta} + \Sigma\dot W_t.
\label{eq:r_operator}
\end{align}
This representation makes the division of labor transparent. The resolved Navier--Stokes--$\beta$ operator supplies the forced, dissipative macroscopic dynamics, while the reservoir exchanges augmented enstrophy with the resolved flow through the symplectic pair $\mathcal{M}(\bar{\zeta})r$ and $-\mathcal{M}(r)\bar{\zeta}$.

We now expand the reservoir dynamics perturbatively in the regime where the hidden field operates as a low-amplitude fluctuating sector. We write
\be
r=r^{(0)}+r^{(1)}+\cdots .
\ee
The leading component is the stationary unresolved bath, while the first correction is the retarded response induced by the resolved flow.

\textbf{\small Leading order: the stochastic reservoir bath.}
The unresolved background dynamics are governed by the Ornstein--Uhlenbeck (OU) equation
\be
\d r^{(0)}
=
-D r^{(0)}\,\d t
+
\Sigma\,\d W_t,
\label{Eq_OU_simple}
\ee
where, in the non-hyperviscous reservoir realization,
\be
D=\kappa(-\Delta)+\mu_r I,
\qquad
\kappa,\mu_r>0,
\label{Eq_D_no_hperviscous}
\ee
and the noise covariance satisfies the finite-energy condition (Appendix \ref{App_Macro_Diffusion_Ope}):
\be
\mathrm{Tr}\left(D^{-1}\Sigma\Sigma^*\right)<\infty.
\label{eq:OU_L2_trace_condition}
\ee
The corresponding stochastic convolution \cite{da2006introduction} is
\be
r^{(0)}(t)
=
\int_0^t e^{-D(t-s)}\Sigma\,\d W_s.
\label{eq:r0_stoch_conv}
\ee
This field represents the base fluctuating unresolved enstrophy reservoir. It is not a deterministic eddy viscosity; it is a temporally correlated bath capable of returning variance to the resolved scales.

\textbf{\small First order: the retarded reservoir response.}
The first correction is driven by the interaction of the leading reservoir bath with the resolved vorticity:
\be
\partial_t r^{(1)}
=
-D r^{(1)}
-
\mathcal{M}(r^{(0)})\bar{\zeta}.
\label{Eq_PDE_r1}
\ee
%--------------------------------------------------% 
The term $r^{(1)}$ is not an independent hidden state. It is the leading retarded correction obtained after replacing the full reservoir dynamics by the linear response problem \eqref{Eq_PDE_r1}. Thus
\be
r^{(1)}(t)
=
-\int_0^t
e^{-D(t-s)}
\mathcal{M}(r^{(0)}(s))\bar{\zeta}(s)
\,ds
\label{eq:r1_duhamel}
\ee
records the history of resolved--reservoir interactions only within this leading-order Duhamel approximation: the resolved vorticity deforms the stochastic reservoir bath at past times $s$, the deformation relaxes through the reservoir semigroup $e^{-D(t-s)}$, and the resulting delayed reservoir state feeds back into the resolved equation at time $t$.

Substituting the leading approximation
\be
r(t)\approx r^{(0)}(t)+r^{(1)}(t)
\ee
into the resolved equation \eqref{eq:zeta_operator} yields the leading-order non-Markovian closure:
\begin{widetext}
\be\label{eq:MZ_Continuous}
\partial_t \bar{\zeta}(t)
=
f_1^{NS,\beta}(\bar{\zeta}(t))
+
\underbrace{
\mathcal{M}(\bar{\zeta}(t))r^{(0)}(t)
}_{\Pi_{\rm back}}
-
\underbrace{
\mathcal{M}(\bar{\zeta}(t))
\int_0^t
e^{-D(t-s)}
\mathcal{M}(r^{(0)}(s))\bar{\zeta}(s)
\,ds
}_{\Pi_{\rm mem}}.
\ee
\end{widetext}
At this order, the SGS vorticity forcing decomposes into
\be
\Pi
=
\Pi_{\rm back}
-
\Pi_{\rm mem}.
\ee

 The first contribution,
\be
\Pi_{\rm back}
=
\mathcal{M}(\bar{\zeta}(t))r^{(0)}(t),
\ee
is naturally interpreted as a stochastic backscatter term. In classical SGS modeling, unresolved scales are often represented through a deterministic eddy viscosity, which removes resolved variance. Kraichnan's analysis of eddy viscosity in two-dimensional turbulence showed that this picture is incomplete: in two dimensions, the effective eddy viscosity associated with subgrid interactions can be negative, reflecting the ability of small-scale vorticity fluctuations interacting with large-scale strain to return activity to larger, resolved scales \cite{kraichnan1976eddy}. Leith later made this mechanism explicit in LES by supplementing deterministic SGS viscosity with stochastic backscatter, i.e., a random acceleration of the resolved scales generated by unresolved eddies \cite{leith1990stochastic}. Furthermore, Leith emphasized that the probabilistic properties of high-dimensional chaotic systems, such as turbulent flows, can be effectively approximated using linear stochastic models (Langevin equations) that balance linear eddy damping with Gaussian white noise \cite{leith1996stochastic}.

The present term $\Pi_{\rm back}$ realizes this physical and statistical philosophy geometrically. The random bath $r^{(0)}$ which is the OU process generated by Eq.~\eqref{Eq_OU_simple} represents the base stochastic unresolved enstrophy fluctuations. The operator $\mathcal{M}(\bar{\zeta})$ then converts these fluctuations into a resolved vorticity forcing through the symplectic Jacobian geometry. Unlike a scalar eddy viscosity, $\Pi_{\rm back}$ is not sign-definite and is not constrained to damp the resolved field. It injects structured stochastic activity into the resolved dynamics, mimicking the random forcing of chaotic subgrid scales advocated by Leith \cite{leith1996stochastic}, but it does so strictly through the conservative exchange geometry generated by $\G[\bar{\zeta},r]$. Thus the backscatter is stochastic and potentially upscale at the resolved level, while remaining embedded in the augmented-enstrophy-preserving resolved--reservoir exchange generated by $\{\G,V\}=0$.

The second contribution,
\be
\Pi_{\rm mem}
=
\mathcal{M}(\bar{\zeta}(t))
\int_0^t
e^{-D(t-s)}
\mathcal{M}(r^{(0)}(s))\bar{\zeta}(s)
\,ds,
\ee
is the leading non-Markovian memory term. It is generated by the deformation of the hidden reservoir by the past resolved flow, followed by relaxation through the semigroup $e^{-D(t-s)}$ and re-insertion into the present resolved dynamics through $\mathcal{M}(\bar{\zeta}(t))$. This is the precise sense in which the closure is viscoelastic \cite{ting1963certain,huilgol1968second,Coleman_Noll1974,joseph1985hyperbolicity}: the unresolved sector stores past cross-gradient twisting and returns a delayed stress-like response to the resolved vorticity equation.

It is useful to isolate the internal memory variable
\be
\Phi_{\rm mem}(t)
=
\int_0^t
e^{-D(t-s)}
\mathcal{M}(r^{(0)}(s))\bar{\zeta}(s)
\,ds,
\label{eq:Phi_mem_def}
\ee
so that
\be
\Pi_{\rm mem}
=
\mathcal{M}(\bar{\zeta}(t))\Phi_{\rm mem}(t).
\ee
The field $\Phi_{\rm mem}$ satisfies the relaxation equation
\be
\partial_t\Phi_{\rm mem}
=
-D\Phi_{\rm mem}
+
\mathcal{M}(r^{(0)}(t))\bar{\zeta}(t),
\qquad
\Phi_{\rm mem}(0)=0.
\ee
Thus the non-Markovian contribution may equivalently be viewed as a viscoelastic auxiliary dynamics: the hidden bath is strained by the resolved vorticity, relaxes through $D$, and feeds back through the same symplectic geometric operator.

Equation \eqref{eq:MZ_Continuous} is therefore not a phenomenological eddy-viscosity ansatz, nor is it a claim about the exact Mori--Zwanzig projection of Navier--Stokes. It is the leading Duhamel elimination of the reservoir dynamics. It shows that the Symplectic Geometric Closure naturally produces both stochastic backscatter and a causal memory force, while preserving the conservative augmented-enstrophy exchange structure inherited from the generating functional $\G$.

%=================================================================%
\section{Beyond the OU Limit: The Emergence of the Geometric Self-Energy}
\label{Sec_higher_order_memory}

Starting from the SGS closure written as Eqs.~\eqref{eq:zeta_operator}--\eqref{eq:r_operator}, one may formally eliminate the hidden reservoir by solving the $r$-equation along the past history of the resolved state. This yields a causal representation of the subgrid drift potential $\psi_{sgs}$ emerging in the dual-transport formulation of the SGS closure (Eqs.~\eqref{eq:zeta_transport}--\eqref{eq:r_transport}), taking the form
\be\label{Eq_mem_psi_sgs}
\psi_{sgs}(\mathbf x,t)
=
\frac12\gamma J\big(\bar{\zeta}(\mathbf x,t),\mathcal N(\bar{\zeta}_t)(\mathbf x)\big),
\ee
where $\bar{\zeta}_t=\{\bar{\zeta}(t+s):s\le0\}$ denotes the resolved history and $\mathcal N$ is the nonlinear reservoir-response operator generated by the full hidden dynamics in Eq.~\eqref{eq:r_operator}. Equation \eqref{Eq_mem_psi_sgs} functions as a formal causal closure for the subgrid drift potential, wherein the potential depends not merely on the instantaneous resolved vorticity, but on the entire past resolved trajectory filtered through the reservoir's relaxation and self-interaction.

The generating functional $\G[\bar{\zeta}, r] = \frac{1}{4}\int \gamma [J(\bar{\zeta}, r)]^2 d\mathbf{x}$ thus defines a fully non-Markovian stochastic transport system. While the leading-order Duhamel expansion provides a clear analytic baseline, treating the subgrid memory kernel as a passive, stationary exponential decay ($G_0(t,s) = e^{-D(t-s)}$, $D$ defined in Eq.~\eqref{Eq_D_no_hperviscous}) constitutes a severe asymptotic simplification. In true turbulence, the decorrelation of subgrid eddies is dynamically dictated by the non-local convective sweeping and severe nonlinear straining exerted by the macroscopic flow \cite{kraichnan1978sweeping}. 
 
Because the nonlinear operator $\mathcal N$ in Eq.~\eqref{Eq_mem_psi_sgs} is too implicit to expose these dynamical statistical mechanisms (damping, memory, and backscatter), we must expand the reservoir response in a Volterra--Duhamel series about the stochastic OU bath. Performing a direct perturbative expansion of the reservoir equation itself yields an explicit hierarchy of stochastic reservoir corrections whose low orders can be transparently interpreted physically.

To formalize this hierarchy, let us rewrite the self-advection of the subgrid reservoir using a continuous bilinear operator $\mathcal{B}_{\bar{\zeta}}(a, b)$, parameterized by the resolved macroscopic state $\bar{\zeta}$:
\be\label{Eq_def_B_zeta}
\mathcal{B}_{\bar{\zeta}}(a, b) = \frac{1}{2} J\big( \gamma J(\bar{\zeta}, a), b \big).
\ee
The reservoir component of the fully coupled SPDE (Eq.~\eqref{eq:r_closed}) then takes the compact form:
\be\label{Eq_full_r_recall}
\partial_t r = -D r + \mathcal{B}_{\bar{\zeta}}(r, r) + \Sigma \dot{W}_t,
\ee
with $D$ given in Eq.~\eqref{Eq_D_no_hperviscous} and $\Sigma$ satisfying the finite-energy condition in Eq.~\eqref{eq:OU_L2_trace_condition} in Appendix \ref{App_Macro_Diffusion_Ope}.

Formally, by Duhamel's principle, the mild solution $r(t)$ of this equation satisfies the stochastic integral equation:
\bea
r(t) = \int_0^t &G_0(t,s) \Sigma\,dW_s \nonumber\\
&+ \int_0^t G_0(t,s) \mathcal{B}_{\bar{\zeta}(s)}\big(r(s), r(s)\big)\,ds,
\eea
where $G_0(t,s) = e^{-D(t-s)}$ is the bare resolvent operator. 

Substituting the full perturbation hierarchy $r = r^{(0)} + r^{(1)} + r^{(2)} + r^{(3)} + \dots$ into the bilinear term and equating orders yields a cascade of governing equations. The hierarchical components $r^{(n)}(t)$ represent the elements of the Picard iteration (the Volterra series) solving this integral equation. The base state $r^{(0)}$ is the stationary Gaussian OU process driven by $\Sigma \dot{W}_t$ (Eq.~\eqref{Eq_OU_simple}). The first-order response $r^{(1)}$ is driven by the interaction of the base bath with itself via the resolved flow, $\mathcal{B}_{\bar{\zeta}}(r^{(0)}, r^{(0)})$ (Eq.~\eqref{Eq_PDE_r1}), yielding the baseline viscoelastic memory derived previously. 

By substituting the expansion $r(t) = \sum_{n=0}^\infty r^{(n)}(t)$ into the quadratic bilinear term and matching polynomial orders, the components are rigorously generated by the recursive sequence:
\be
r^{(0)}(t) = \int_0^t G_0(t,s) \Sigma\,dW_s,  
\ee
and for $n \ge 1$:
\be
r^{(n)}(t) = \sum_{i+j=n-1} \int_0^t G_0(t,s) \mathcal{B}_{\bar{\zeta}(s)}\Big( r^{(i)}(s), r^{(j)}(s) \Big)\,ds. 
\label{eq:r_n_recursive}
\ee

A fundamental vulnerability of perturbative expansions in fluid dynamics is that formal series may fail to converge, or may lose control of high-gradient events, unless the nonlinear hierarchy is embedded within a sufficiently regular dissipative structure. In the hyperviscous realization of the Symplectic Geometric Closure, the required analytic control is provided by the Sobolev  $H^2(\mathbb T^2)\times H^2(\mathbb T^2)$ estimate \cite{brezis_book,FMRT01} established in Lemma 1 of {\it Supplementary Note~1}. Then, as shown in Lemma 3 of  {\it Supplementary Note~2} since the resolved vorticity $\bar\zeta(t)$ is controlled in $H^2$, the Volterra--Picard series $\sum_{n=0}^\infty r^{(n)}(t)$  is locally convergent in 
\be
X_{T_*}=C([0,T_*];H^2(\mathbb T^2))
\ee
on a strictly positive, almost surely finite random time interval $T_*(\omega)>0$. The proof relies on the negative-Sobolev continuity \cite{FMRT01} of the quadratic reservoir interaction,
\be
\mathcal B_{\bar\zeta}:H^2\times H^2\longrightarrow H^{-\sigma},
\qquad
\sigma\in(3/2,2),
\ee
combined with the fractional smoothing of the fourth-order reservoir semigroup ({\it Supplementary Note~2}),
\be
e^{-tD}:H^{-\sigma}\longrightarrow H^2,
\ee
Thus, in the hyperviscous setting, the Volterra--Picard expansion is not merely formal: it converges locally to the mild solution of the reservoir equation driven by the prescribed resolved vorticity path $\bar\zeta|_{[0,T_*]}$.

In the non-hyperviscous Navier--Stokes realization considered for the physical derivation below, the augmented-enstrophy structure still yields pullback boundedness in 
\be
\mathcal H=L^2(\mathbb T^2)\times L^2(\mathbb T^2),
\ee
but without $H^2$ compactness and fractional smoothing, establishing the convergence of the Volterra--Picard series remains an open analytic question. However, this limitation does not affect  the one-loop calculation below, which uses only the finite set of terms up to $r^{(3)}$ to identify the first nontrivial geometric self-energy generated by the closure.

Evaluating the expected macroscopic forcing in Eq.~\eqref{eq:zeta_operator} up to third order, $\mathbb{E}\big[\mathcal{M}(\bar{\zeta}) \sum_{n=0}^3 r^{(n)}\big]$, therefore represents the explicit statistical evaluation of the first four reservoir dressings generated by the Duhamel hierarchy. The conservative symplectic identity $\{\G,V\}=0$ guarantees that these resolved--reservoir exchanges do not create augmented enstrophy at the structural level, while the dissipative reservoir operator $D$ supplies the relaxation scale entering each retarded integral. The purpose of the expansion is consequently not to assert global convergence of the infinite series, but to mathematically extract the leading nontrivial memory operator generated by the geometry of the coupled system.

%-----------------------------------------------------------------%
\subsection{Gaussian Parity and the Need for the Third-Order Dressing}
\label{sec:gaussian_parity_third_order}

The necessity of the $\mathcal O(r^{(3)})$ expansion is dictated by the algebra of the Gaussian expectation operator.  The Volterra--Picard hierarchy above already defines the stochastic reservoir corrections $r^{(n)}$.  We now ask which of these corrections contributes to the deterministic macroscopic forcing after averaging over the invariant Gaussian law of the background OU bath.

Write the $n$th-order macroscopic contribution schematically as
\be
\Pi^{(n)}(t)
=
\mathcal M_{\bar\zeta(t)} r^{(n)}(t),
\ee
where $\mathcal M_{\bar\zeta(t)}$ denotes the linear map from a reservoir perturbation to the resolved vorticity forcing at the current resolved state.  Since $r^{(0)}$ is centered Gaussian, the zeroth-order contribution has zero mean:
\be
\mathbb{E}[\Pi^{(0)}(t)]
=
0.
\ee

The first nonzero deterministic contribution comes from $r^{(1)}$, because $r^{(1)}$ is quadratic in the Gaussian field $r^{(0)}$.  Using the first Picard correction,
\be
r^{(1)}(t)
=
\int_0^t
G_0(t,s)\,
\mathcal B_{\bar\zeta(s)}
\big(r^{(0)}(s),r^{(0)}(s)\big)\,ds,
\ee
we obtain
\be
\mathbb{E}[\Pi^{(1)}(t)]
=
\mathcal M_{\bar\zeta(t)}
\int_0^t
G_0(t,s)\,
\mathbb{E}
\Big[
\mathcal B_{\bar\zeta(s)}
\big(r^{(0)}(s),r^{(0)}(s)\big)
\Big]\,ds .
\label{eq:Pi1_mean_revised}
\ee
This is the baseline deterministic memory.  It is nonzero because the expectation contains an even Gaussian contraction.  However, its propagation is still carried by the bare reservoir semigroup $G_0(t,s)=e^{-(t-s)D}$, with $D$ understood as the reservoir dissipation operator defined in Eq.~\eqref{Eq_D_no_hperviscous}.  Thus this term captures the first covariance-induced drag, but it does not yet produce a dressed, state-dependent response kernel.

The next correction is
\be
\partial_t r^{(2)}
=
-D r^{(2)}
+
\mathcal B_{\bar\zeta}(r^{(0)},r^{(1)})
+
\mathcal B_{\bar\zeta}(r^{(1)},r^{(0)}).
\label{eq:r2_governing}
\ee
Integrating by Duhamel's principle gives
%--------------------------------------%
\begin{widetext}
\be
r^{(2)}(t)
=
\int_0^t
G_0(t,s)
\Big[
\mathcal B_{\bar\zeta(s)}
\big(r^{(0)}(s),r^{(1)}(s)\big)
+
\mathcal B_{\bar\zeta(s)}
\big(r^{(1)}(s),r^{(0)}(s)\big)
\Big]\,ds .
\label{eq:r2_integrated}
\ee
\end{widetext}
%--------------------------------------%
Since $r^{(0)}$ is a centered Gaussian process and $r^{(1)}$ is quadratic in $r^{(0)}$, every term in the integrand of $r^{(2)}$ is cubic in the centered Gaussian field.  Therefore, by the Wick--Isserlis theorem for Gaussian moments \cite{isserlis1918formula,janson1997gaussian,peccati2011wiener}, its deterministic contribution vanishes:
\be
\mathbb{E}[\Pi^{(2)}(t)]
=
\mathbb{E}
\big[
\mathcal M_{\bar\zeta(t)} r^{(2)}(t)
\big]
=
0.
\label{eq:Pi2_mean_zero}
\ee

This cancellation is the key point.  The $\mathcal O(r^{(2)})$ term does not contribute to the deterministic subgrid memory.  It is nevertheless not dynamically empty: its variance is generically nonzero, so it contributes to the fluctuating part of the SGS forcing.  More precisely, it represents a state-dependent stochastic scattering channel generated by the interaction between the OU reservoir fluctuations and their first retarded response.  Such a channel may contribute to non-Gaussian backscatter statistics after nonlinear propagation, but the parity argument alone does not justify identifying it with a complete theory of intermittency.

The first genuinely new deterministic dressing must therefore occur at the next even Gaussian contraction level.  The third correction contains terms of the form
\beas
\partial_t r^{(3)}
=
-D r^{(3)}
+
\mathcal B_{\bar\zeta}(r^{(1)},r^{(1)})
&+
\mathcal B_{\bar\zeta}(r^{(0)},r^{(2)})\\
&+ \mathcal B_{\bar\zeta}(r^{(2)},r^{(0)}).
\eeas
Each contribution is quartic in the centered Gaussian background field.  Its expectation is not forced to vanish, because Wick contractions can pair the four Gaussian factors into covariance products.  This is the first order at which repeated resolved--unresolved interactions can feed back into the deterministic memory as a self-energy correction to the bare reservoir propagator $G_0(t,s)=e^{-(t-s)D}$.

Thus the hierarchy partitions the low-order closure into three roles.  The $\mathcal O(r^{(1)})$ term gives the first bare deterministic memory.  The $\mathcal O(r^{(2)})$ term is mean-zero by Gaussian parity and contributes only to fluctuations.  The $\mathcal O(r^{(3)})$ term is the first level at which a deterministic geometric self-energy can appear.  This prepares the next subsection, where the quartic Wick contractions are evaluated and reorganized as the leading line-dressing of the subgrid response.

%-----------------------------------------------------------------%
\subsection{Emergence of the Geometric Self-Energy}
\label{Sec_Dyson_Emergence}
In Kraichnan's amplitude representation of the Direct-Interaction Approximation
\cite{kraichnan1970convergents}, the macroscopic response is not governed by
passive viscous decay alone; see also Leith \cite{leith1971atmospheric}. Instead,
the infinitesimal response satisfies a Volterra equation in which the bare
viscous damping is supplemented by a history-dependent eddy-damping kernel. In
spectral notation, Kraichnan's response equation has the form ($t>s$)
\begin{equation}
\left(
\frac{\partial}{\partial t}
+
\nu k^2
\right)
G_i(t,s)
+
\int_s^t
\eta_i(t,\tau)G_i(\tau,s)\,d\tau
=
0,
\label{eq:Kraichnan_Volterra}
\end{equation}
where $\eta_i(t,\tau)$ is the self-consistent memory kernel \cite[Chap.~6]{orszag1977lectures}. The essential
structure of this kernel is the coupling of two Eulerian objects: a two-time
covariance measuring the memory of the advecting turbulent field, and an
infinitesimal response function measuring the propagation of perturbations.
Their product, schematically $G\times Y$, controls the rate at which nonlinear
interactions scramble phase coherence \cite{orszag1977lectures}.

The corresponding objects in the present theory arise from the stochastic
dressed-transport skeleton. The resolved vorticity is transported by the reservoir-induced
Hamiltonian velocity ($\mathbf{u}_{sgs} = \nabla^\perp \psi_{sgs}$ generated by $\psi_{sgs} =
\frac12\gamma J(\bar\zeta,r)$, 
so that the stochastic, reservoir-induced  transport operator is
\be
\mathcal{A}_{\rm res}(t)\Phi
=
J(\psi_{sgs}(t),\Phi)
=
\mathcal L_{X_{\psi_{sgs}(t)}}\Phi,
\ee
for any smooth test function $\Phi$. 

%----------------------------------------------------------------------------------------------------------------%
As reviewed in Appendix~\ref{app:variational_propagator}, the
macroscopic response $G(t,s)$ is the ensemble-averaged infinitesimal
response of a stochastic transport problem involving the
reservoir-induced advection operator $\mathcal{A}_{\rm res}(t)$ and the
resolved advection operator
$\mathcal{L}_0(t)\Phi=-J\big(\bar\psi(t),\Phi\big)$; see
Eq.~\eqref{eq:direct_advection_tangent_skeleton}.
Theorem~\ref{thm:Dyson_propagator} in
Appendix~\ref{app:Dyson_derivation} then shows that the associated
Eulerian transport self-energy is formed from two reservoir-induced
Lie-transport actions, with the intermediate response inserted between
them [Eq.~\eqref{eq:eta_reservoir_induced_operator}]:
\be
\boldsymbol{\eta}(t,\tau)\Phi
=
\mathbb{E}
\Big[
\mathcal{L}_{X_{\psi_{sgs}(t)}}
\Big(
G(t,\tau)
\mathcal{L}_{X_{\psi_{sgs}(\tau)}}\Phi
\Big)
\Big].
\ee
This formula has a direct dynamical interpretation. A disturbance of the
resolved field is first acted upon at time $\tau$ by the
reservoir-induced transport, propagated from $\tau$ to $t$ by the
intermediate response $G(t,\tau)$, and then acted upon again by the
reservoir-induced transport at time $t$. After ensemble averaging,
$\boldsymbol{\eta}(t,\tau)$ records the delayed feedback of this sequence
on the propagation, decorrelation, and lifetime of resolved structures.
It is in this sense the geometric analogue of Kraichnan's
response--covariance kernel.
%----------------------------------------------------------------------------------------------------------------%

%----------------------------------------------------------------------------------------------------------------%
Theorem~\ref{thm:one_loop_emergence} establishes how this response
architecture arises from the reservoir dynamics themselves. Expanding
the hidden field through the Volterra--Picard hierarchy, the first
non-vanishing deterministic feedback appears in the third-order mean
forcing $\mathbb E[\Pi^{(3)}]$. Its Gaussian contractions separate into
two contributions: a disconnected part, which yields the deterministic
macroscopic diffusion $\mathcal D_\gamma$ derived in
Appendix~\ref{App_Macro_Diffusion_Ope}, and a connected part, which
contracts two OU-induced Hamiltonian advection operators acting at
different times and produces the two-time operator
$\mathbf\Sigma_{\rm sweep}(s,\tau)$. Corollary~\ref{cor:reservoir_tangent_self_energy}
then shows how the same pair of time-separated Lie-transport actions
enters the tangent-space self-energy. In the Volterra--Picard
representation, their covariance appears as the separate insertion
$\mathbf\Sigma_{\rm sweep}(s,\tau)$, with the propagation displayed
outside it; in $\boldsymbol\eta(s,\tau)$, the intermediate response
$G(s,\tau)$ is inserted directly between the two operators. Thus the
theorem derives the transport covariance from the stochastic reservoir
SPDE, while the corollary identifies its response-dressed form with the
Eulerian geometric self-energy.
%----------------------------------------------------------------------------------------------------------------%

To state these results we introduce a few notations.
First, we  introduce the causalsemigroup 
\be\label{Eq_def_G_D}
G_D(t,s)=e^{-D(t-s)}\Theta(t-s),  
\ee
where $\Theta$ denotes the Heaviside function, and the $D$ denotes the reservoir dissipative operator given by Eq.~\eqref{Eq_D_no_hperviscous}.
For two time-ordered operator kernels $A$ and $B$, we define the Volterra
composition as
\be
(A\star B)(t,s)
:=
\int_s^t
A(t,\tau)B(\tau,s)\,d\tau.
\label{eq:Volterra_composition}
\ee
We have then the following Theorem.

%==================BEGINNING THEOREM AND COROLLARY=========================%
\bt[Third-Order Reservoir Memory and the Transport-Covariance Insertion]
\label{thm:one_loop_emergence}
Let the uncoupled reservoir $r^{(0)}$ be the stationary
Ornstein--Uhlenbeck (OU) process on $\mathcal{H}=L^2_0(\mathbb{T}^2)$ satisfying
the spatial homogeneity, local isotropy, and finite gradient-energy trace
conditions established in Appendix~\ref{App_Macro_Diffusion_Ope}. Let
$\mathbb{E}[\Pi^{(3)}(t)]$ denote the deterministic macroscopic memory
forcing generated by the third-order Volterra--Picard expansion of the
reservoir hierarchy.
Consider the contribution
\bes
\mathbb{E}[\Pi^{(3)}(t)]
\Big\vert_{r^{(1)}\times r^{(1)}}
\ees
arising from the quadratic interaction
$r^{(1)}\times r^{(1)}$ within the third-order state $r^{(3)}$ generated
by the recursive hierarchy \eqref{eq:r_n_recursive}. This contribution
admits a canonical decomposition into a disconnected mean--mean part and
a connected covariance part.
For
\be
Y_\tau
:=
\mathcal{B}_{\bar\zeta(\tau)}
\big(r^{(0)}(\tau),r^{(0)}(\tau)\big),
\ee
the disconnected part is determined by
\be
\mathbb{E}[Y_\tau]
=
\mathcal{D}_\gamma\bar\zeta(\tau),
\ee
where $\mathcal{D}_\gamma$ is the macroscopic diffusion operator given
explicitly by Eq.~\eqref{Eq_D_gamma_commuting_formula} in Proposition \ref{prop:macroscopic_diffusion_trace}. It therefore
produces local-in-time deterministic spatial-diffusion insertions.
Define the leading-order reservoir-induced drift potential
\be
\psi_s^{(0)}
:=
\frac12\gamma J\big(\bar\zeta(s),r^{(0)}(s)\big),
\label{eq:psi_sgs_zeroth}
\ee
and the associated OU-induced Hamiltonian advection operator defined for any smooth test function $\Phi$ by:
\bea
\mathcal{A}_s^{(0)}\Phi
&:=
J(\psi_s^{(0)},\Phi)
=
\mathcal{L}_{X_{\psi_s^{(0)}}}\Phi,
\\
X_{\psi_s^{(0)}}
&=
\nabla^\perp\psi_s^{(0)}.
\eea
Then the connected covariance part defines the deterministic two-time
spatial operator $\mathbf{\Sigma}_{\rm sweep}(s,\tau)$ by
\bea
\mathbf{\Sigma}_{\rm sweep}(s,\tau)\Phi
&:=
-2\,\mathbb{E}
\Big[
\mathcal{A}_s^{(0)}
\mathcal{A}_\tau^{(0)}\Phi
\Big]
\nonumber\\
&=
-2\,\mathbb{E}
\Big[
\mathcal{L}_{X_{\psi_s^{(0)}}}
\Big(
\mathcal{L}_{X_{\psi_\tau^{(0)}}}\Phi
\Big)
\Big]
\nonumber\\
&=
-2\,\mathbb{E}
\Big[
J\big(\psi_s^{(0)},J(\psi_\tau^{(0)},\Phi)\big)
\Big].
\label{eq:Sigma_sweep_Lie_covariance_definition}
\eea
Using the Volterra composition operation
\eqref{eq:Volterra_composition}, the connected covariance contribution
reorganizes over the causal simplex $t>s>\tau>\sigma$ as
\be
\mathbb{E}[\Pi^{(3)}(t)]
\Big\vert_{r^{(1)}\times r^{(1)}}^{\rm cov}
=
\Big(
G_D
\star
\mathcal{M}(\bar\zeta)
\star
G_D
\star
\mathbf{\Sigma}_{\rm sweep}
\star
G_D\bar\zeta
\Big)(t),
\label{eq:Volterra_Dyson_sequence_compact}
\ee
where $\mathcal{M}(\bar\zeta)$ is defined in
Eq.~\eqref{eq:M_operator_cont} and $G_D$ is given by
Eq.~\eqref{Eq_def_G_D}. Equivalently, its pointwise integrand has the
ordered structure
\bes
G_D(t,s) \,
\mathcal{M}(\bar\zeta(t)) \,
G_D(s,\tau) \,
\mathbf{\Sigma}_{\rm sweep}(s,\tau) \,
G_D(\tau,\sigma) \,
\bar\zeta(\sigma),
\ees
integrated over $0<\sigma<\tau<s<t$. The full integral representation is
given in Eq.~\eqref{eq:Volterra_Dyson_Simplex}.
\et
\paragraph*{Interpretation.}
The theorem separates two statistically distinct effects of the
third-order reservoir expansion. The disconnected contractions collapse
to the instantaneous operator $\mathcal{D}_\gamma$ and therefore
renormalize spatial diffusion without producing temporal memory. The
connected contractions instead retain correlations between two distinct
times and generate the transport-covariance insertion
$\mathbf{\Sigma}_{\rm sweep}(s,\tau)$. In the ordered Volterra sequence,
this insertion lies between successive reservoir propagators and acts on
a historical resolved input transported forward from $\sigma$ to
$\tau$. It is therefore the part of the reservoir calculation that
carries the genuinely non-Markovian feedback. Panel~(b) of
Fig.~\ref{fig:sweeping_cascade} renders this two-time operator
concretely: $\mathbf{\Sigma}_{\rm sweep}$ is assembled from two
successive Hamiltonian Lie transports acting on the same disturbance,
and it is this nested-Jacobian structure --- rather than a raw velocity
covariance --- that will decide which motions can contribute to the
memory.
%==================END THEOREM AND COROLLARY===========================%

%=========================================================%
\begin{proof}
The proof proceeds by separating four analytic operations: first, the
stationary OU reservoir and its covariance kernel are fixed; second, the
third-order Volterra--Picard term is written explicitly as a nested
reservoir integral; third, the quartic Gaussian expectation is decomposed
into a mean--mean part and a centered covariance part; finally, the
centered part is reorganized into the ordered self-energy insertion.

\textbf{Step 1: The Stationary Reservoir and its Two-Time Covariance.}
We distinguish throughout between the macroscopic tangent response and
the internal reservoir relaxation. Let
\be
G_D(t,s)
:=
e^{-D(t-s)}\Theta(t-s)
\ee
denote the causal semigroup generated by the reservoir dissipative operator $D$ given by Eq.~\eqref{Eq_D_no_hperviscous}.

The leading-order hidden reservoir is the stationary OU process
\be
r^{(0)}(t)
=
\int_{-\infty}^t
G_D(t,s)\Sigma\,dW_s.
\ee
Since $D$ has Fourier eigenvalues
\be
d_k=\kappa |k|^2+\mu_r,
\ee
different spatial modes decorrelate at different temporal rates. Hence
the stationary space-time covariance does not, in general, factor into a
scalar temporal covariance times a fixed spatial correlation. Instead,
under the stationary OU measure, we work with the two-time spatial
covariance kernel
\be
C_{t,\tau}(\mathbf x-\mathbf y)
:=
\mathbb{E}
\big[
r^{(0)}(\mathbf x,t)r^{(0)}(\mathbf y,\tau)
\big].
\label{eq:Ct_tau_def}
\ee
The dependence on $\mathbf x-\mathbf y$ follows from the translation
invariance of both $D$ and $\Sigma\Sigma^*$. The ultraviolet regularity
assumption of Appendix \ref{App_Macro_Diffusion_Ope} ensures that the
corresponding differentiated covariance kernels are meaningful in the
diagonal covariance sense.

\textbf{Step 2: The Third-Order Volterra--Picard Contribution.}
The third-order contribution to the mean resolved forcing is
\be
\mathbb{E}[\Pi^{(3)}(t)]
=
\mathbb{E}
\big[
\mathcal M(\bar\zeta(t))r^{(3)}(t)
\big].
\ee
In the Volterra--Picard hierarchy for the reservoir equation, the part of
$r^{(3)}$ generated by the quadratic self-interaction of the first-order
response is
\be
r^{(3)}_{1\times1}(t)
=
\int_0^t
G_D(t,s)
\mathcal B_{\bar\zeta(s)}
\big(r^{(1)}(s),r^{(1)}(s)\big)
\,ds,
\ee
where
\be
r^{(1)}(s)
=
\int_0^s
G_D(s,\tau)
\mathcal B_{\bar\zeta(\tau)}
\big(r^{(0)}(\tau),r^{(0)}(\tau)\big)
\,d\tau.
\ee
Substituting this expression twice gives the $r^{(1)}\times r^{(1)}$
contribution
\begin{widetext}
\bea\label{eq:Pi3_raw}
&\mathbb{E}[\Pi^{(3)}(t)]
\Big\vert_{r^{(1)}\times r^{(1)}} =
\int_0^t
G_D(t,s)\mathcal{M}(\bar{\zeta}(t))
\left(
\int_0^s\int_0^s
G_D(s,\tau_1)G_D(s,\tau_2)
\,
\mathcal N(s,\tau_1,\tau_2)
\,d\tau_1d\tau_2
\right)ds,
\\
&\textrm{with}\qquad
\mathcal N(s,\tau_1,\tau_2)
:=
\mathbb{E}
\Big[
\mathcal B_{\bar\zeta(s)}
\big(
Y_{\tau_1},
Y_{\tau_2}
\big)
\Big],
\eea
\end{widetext}
where we have introduced the quadratic reservoir forcing
\be
Y_\tau
:=
\mathcal B_{\bar\zeta(\tau)}
\big(r^{(0)}(\tau),r^{(0)}(\tau)\big).
\label{eq:Ytau_def}
\ee
This notation isolates the only random objects in the integrand:
$\mathcal N$ is the expectation of a deterministic bilinear operator
applied to two quadratic Gaussian functionals.

\textbf{Step 3: Mean--Covariance Decomposition of the Quadratic Vertices.}
By Proposition \ref{prop:macroscopic_diffusion_trace} in Appendix
\ref{App_Macro_Diffusion_Ope}, the equal-time mean of the quadratic
reservoir forcing is the positive macroscopic diffusion:
\be
\mathbb{E}[Y_\tau]
=
\mathcal D_\gamma\bar\zeta(\tau).
\label{eq:Ytau_mean}
\ee
We therefore split
\be
Y_\tau
=
\mathcal D_\gamma\bar\zeta(\tau)
+
Y_\tau^\circ,
\qquad
Y_\tau^\circ
:=
Y_\tau-\mathbb{E}[Y_\tau].
\label{eq:Ytau_centered}
\ee
Substituting \eqref{eq:Ytau_centered} into $\mathcal N$ gives
\bea
\mathcal N(s,\tau_1,\tau_2)
=\mathcal B_{\bar\zeta(s)}
&\big( \mathcal D_\gamma\bar\zeta(\tau_1),
\mathcal D_\gamma\bar\zeta(\tau_2) \big)\\
&+  \mathbb{E} \Big[ \mathcal B_{\bar\zeta(s)} \big( Y_{\tau_1}^\circ, Y_{\tau_2}^\circ
\big)
\Big],
\label{eq:N_decomp_mean_connected}
\eea
because the mixed terms vanish by centering. The first term is the
disconnected mean--mean contribution: it produces local-in-time
deterministic diffusion insertions and does not generate a two-time
memory kernel. The second term is the covariance of the two centered
quadratic reservoir vertices. This is the part responsible for the
operator-valued memory insertion.

Equivalently, writing the centered contribution as
\be
\mathcal N_{\rm cov}(s,\tau_1,\tau_2)
:=
\mathbb{E}
\Big[
\mathcal B_{\bar\zeta(s)}
\big(
Y_{\tau_1}^\circ,
Y_{\tau_2}^\circ
\big)
\Big],
\label{eq:Ncov_def}
\ee
the genuinely two-time part of the third-order forcing is
\begin{widetext}
\be\label{eq:Pi3_cov_part}
\mathbb{E}[\Pi^{(3)}(t)]
\Big\vert_{r^{(1)}\times r^{(1)}}^{\rm cov}\\
 =
\int_0^t
G_D(t,s)\mathcal{M}(\bar{\zeta}(t))
\left(
\int_0^s\int_0^s
G_D(s,\tau_1)G_D(s,\tau_2)
\,
\mathcal N_{\rm cov}(s,\tau_1,\tau_2)
\,d\tau_1d\tau_2
\right)ds.
\ee
\end{widetext}

At this point no diagrammatic terminology is needed. The only remaining
task is to evaluate the covariance of the centered quadratic vertices
$Y_{\tau_1}^\circ$ and $Y_{\tau_2}^\circ$ (Eq.~\eqref{eq:Ncov_def}). Since $Y_\tau$ is quadratic in
the centered Gaussian field $r^{(0)}(\tau)$, this covariance is determined
by the two cross-pairings in the Wick--Isserlis identity \cite{isserlis1918formula,wick1950evaluation,janson1997gaussian,peccati2011wiener}. One of these
pairings is antisymmetric under the exchange of the corresponding
Jacobian slots and vanishes after integration by parts on the periodic
domain. The remaining cross-pairing is the one that transfers spatial
gradient information between the two distinct historical times
$\tau_1$ and $\tau_2$. It is this covariance term that will be reorganized
below into the sweeping self-energy operator.

%------------------------------------------------------------------------------------------------------------------------------------------------------%
\textbf{Step 4: From the Square Integral to an Ordered Volterra Kernel.} We now reorganize the centered covariance contribution
Eq.~\eqref{eq:Pi3_cov_part} into an ordered Volterra kernel. The only
operation at this stage is deterministic: the square domain
$(\tau_1,\tau_2)\in[0,s]^2$ is decomposed into its two ordered triangles.
For any integrable operator-valued integrand $F$, we use the identity
\be\label{eq:square_to_ordered_triangles}
\int_0^s\int_0^s
F(\tau_1,\tau_2)\,d\tau_1d\tau_2
=
\int_0^s\int_0^\tau
\Big[
F(\tau,\sigma)+F(\sigma,\tau)
\Big]\,d\sigma d\tau.
\ee
Applying \eqref{eq:square_to_ordered_triangles} to
\be\label{eq:F_ordered_integrand}
F(\tau_1,\tau_2)
=
G_D(s,\tau_1)G_D(s,\tau_2)
\mathcal N_{\rm cov}(s,\tau_1,\tau_2),
\ee
we obtain an ordered representation over the simplex
\be
0<\sigma<\tau<s.
\ee
No probabilistic or diagrammatic assumption is involved in this step; it
is simply the exact decomposition of the square into two ordered sectors.

The second deterministic ingredient is the semigroup property of the OU
reservoir propagator. Since
\be
G_D(t,s)=e^{-D(t-s)}\Theta(t-s),
\ee
we have, for $s>\tau>\sigma$,
\be\label{eq:GD_semigroup_factorization}
G_D(s,\sigma)
=
G_D(s,\tau)G_D(\tau,\sigma).
\ee
Thus the older branch in the ordered integral, originally propagated
directly from $\sigma$ to $s$, can be factored through the intermediate
time $\tau$. This isolates the historical input
\be
G_D(\tau,\sigma)\bar\zeta(\sigma),
\ee
and leaves between the adjacent times $s$ and $\tau$ an operator-valued
insertion. We define this insertion to be the sweeping self-energy
$\mathbf\Sigma_{\rm sweep}(s,\tau)$.

More precisely, $\mathbf\Sigma_{\rm sweep}(s,\tau)$ is defined by the
identity
\begin{widetext}
\bea\label{eq:Sigma_sweep_definition_by_identity}
&G_D(s,\tau)\,
\mathbf\Sigma_{\rm sweep}(s,\tau)\,
G_D(\tau,\sigma)\bar\zeta(\sigma)\\
&\qquad :=
G_D(s,\tau)G_D(s,\sigma)
\mathcal N_{\rm cov}(s,\tau,\sigma)
+
G_D(s,\sigma)G_D(s,\tau)
\mathcal N_{\rm cov}(s,\sigma,\tau),
\eea
\end{widetext}
where the semigroup factorization \eqref{eq:GD_semigroup_factorization} is used in
the terms containing $G_D(s,\sigma)$. In the symmetric case, the two
ordered contributions are identical after exchanging the dummy variables,
and the right-hand side reduces to twice the first ordered term. In the
general notation above, \eqref{eq:Sigma_sweep_definition_by_identity}
keeps the full symmetrized ordered-sector contribution.

Substituting \eqref{eq:Sigma_sweep_definition_by_identity} into
\eqref{eq:Pi3_cov_part} gives the ordered Volterra--Dyson form
\begin{widetext}
\begin{equation}
\mathbb{E}[\Pi^{(3)}(t)]
\Big\vert_{r^{(1)}\times r^{(1)}}^{\rm cov}
=
\int_0^t
G_D(t,s)\mathcal{M}(\bar{\zeta}(t))
\left[
\int_0^s\int_0^\tau
G_D(s,\tau)
\mathbf{\Sigma}_{\rm sweep}(s,\tau)
G_D(\tau,\sigma)
\bar{\zeta}(\sigma)
\,d\sigma d\tau
\right]ds.
\label{eq:Volterra_Dyson_Simplex}
\end{equation}
\end{widetext}
This is the desired ordered memory sequence: a first reservoir branch
propagates from $\tau$ to $s$, the self-energy
$\mathbf\Sigma_{\rm sweep}(s,\tau)$ acts between adjacent interaction
times, and the historical state is transported from $\sigma$ to $\tau$ by
$G_D(\tau,\sigma)$.

It remains to identify the operator $\mathbf\Sigma_{\rm sweep}$ in
analytic terms. Using the bilinear vertex
\be\label{eq:B_advective_form}
\mathcal B_{\bar\zeta}(a,b)
=
J\big((\mathbf v_{\bar\zeta}\cdot\nabla)a,b\big),
\qquad
\mathbf v_{\bar\zeta}
=
\frac12\gamma\nabla^\perp\bar\zeta,
\ee
we see that the stochastic dependence of each quadratic vertex
$Y_\tau$ enters only through spatial derivatives of $r^{(0)}(\tau)$.
For two reservoir slots evaluated at times $\tau$ and $\sigma$, introduce
distinct spatial variables $\mathbf x$ and $\mathbf y$. By
Eq.~\eqref{eq:Ct_tau_def}, the stationary two-time covariance is
\be
C_{\tau,\sigma}(\mathbf x-\mathbf y)
=\mathbb{E}[
r^{(0)}(\mathbf x,\tau)r^{(0)}(\mathbf y,\sigma)
],
\ee
and therefore the differentiated point-split covariance is
\be\label{eq:point_split_covariance_tau_sigma}
\mathbb{E}[
\partial_{x_i}r^{(0)}(\mathbf x,\tau)
\partial_{y_j}r^{(0)}(\mathbf y,\sigma)]
=
\partial_{x_i}\partial_{y_j}
C_{\tau,\sigma}(\mathbf x-\mathbf y).
\ee
Equivalently, with $z=\mathbf x-\mathbf y$,
\be\label{eq:negative_hessian_covariance}
\partial_{x_i}\partial_{y_j}
C_{\tau,\sigma}(\mathbf x-\mathbf y)
=
-\partial_{z_i}\partial_{z_j}C_{\tau,\sigma}(z).
\ee
The ultraviolet trace condition in Appendix
\ref{App_Macro_Diffusion_Ope} ensures that
\eqref{eq:point_split_covariance_tau_sigma} has a well-defined diagonal
covariance limit.

Accordingly, $\mathbf\Sigma_{\rm sweep}$ is not introduced as a
heuristic diagrammatic object. It is the deterministic operator obtained
by evaluating the covariance of the centered quadratic reservoir
vertices in \eqref{eq:Ncov_def} through the point-split derivative
covariance \eqref{eq:point_split_covariance_tau_sigma}, and then
collecting the part that remains between the adjacent times $s$ and
$\tau$ after the historical line $G_D(\tau,\sigma)\bar\zeta(\sigma)$ has
been factored out.

%------------------------------------------------------------------------------------------------------------------------------------------------------%
To make this operator completely explicit, it is convenient to pass from the
reservoir covariance to the covariance of the induced subgrid drift potential.
For each resolved time $s$, define
\beas
\mathbf v_s(\mathbf x)
&:=
\frac12\gamma(\mathbf x)\nabla^\perp\bar\zeta(\mathbf x,s),\\
\psi_s^{(0)}(\mathbf x) &:=
\mathbf v_s(\mathbf x)\cdot\nabla r^{(0)}(\mathbf x,s).
\eeas
Thus
\be
\psi_s^{(0)}
=
\frac12\gamma J(\bar\zeta(s),r^{(0)}(s))
\ee
is the leading-order stochastic subgrid drift potential generated by the
OU reservoir. Its two-time spatial covariance kernel is
\be
K_{s,\tau}(\mathbf x,\mathbf y)
:=\mathbb{E}
\big[
\psi_s^{(0)}(\mathbf x)
\psi_\tau^{(0)}(\mathbf y)
\big].
\ee
Using the point-split covariance of the reservoir, this kernel is given by
\be
K_{s,\tau}(\mathbf x,\mathbf y)
=
v_s^i(\mathbf x)\,
v_\tau^j(\mathbf y)\,
\partial_{x_i}\partial_{y_j}
C_{s,\tau}(\mathbf x-\mathbf y),
\label{eq:K_subgrid_potential_covariance}
\ee
with summation over repeated spatial indices. The ultraviolet trace condition
in Appendix \ref{App_Macro_Diffusion_Ope} guarantees that this differentiated
covariance is well-defined in the diagonal covariance sense.

The sweeping self-energy is then the deterministic operator obtained by
contracting two stochastic Lie derivatives generated by these potentials.
With the ordered-sector factor included, we define, for a smooth test function
$\Phi$,
\be
\mathbf\Sigma_{\rm sweep}(s,\tau)\Phi
:=
-2\,\mathbb{E}
\Big[
J\big(\psi_s^{(0)},J(\psi_\tau^{(0)},\Phi)\big)
\Big].
\label{eq:Sigma_sweep_Lie_covariance_definition}
\ee
This formula is already deterministic, since the expectation acts only on the
two Gaussian potentials. Expanding the two Jacobians gives a fully explicit
point-split representation. Write
\be
J(f,g)=\epsilon_{ab}\,\partial_a f\,\partial_b g,
\qquad
\epsilon_{12}=1,\quad \epsilon_{21}=-1.
\ee
Then
%-------------------------------%
%\begin{widetext}
\bea
J\big(&\psi_s^{(0)},J(\psi_\tau^{(0)},\Phi)\big)
=
\epsilon_{ab}\epsilon_{cd}
\partial_a\psi_s^{(0)}
\partial_b
\Big(
\partial_c\psi_\tau^{(0)}
\partial_d\Phi
\Big)\\
&=
\epsilon_{ab}\epsilon_{cd}
\bigg[
\partial_a\psi_s^{(0)}
\partial_b\partial_c\psi_\tau^{(0)}
\partial_d\Phi+
\partial_a\psi_s^{(0)}
\partial_c\psi_\tau^{(0)}
\partial_b\partial_d\Phi
\bigg].
\eea
%\end{widetext}
%-------------------------------%
Taking expectation and using the covariance kernel $K_{s,\tau}$ yields
%-------------------------------%
\begin{widetext}
\bea
\big(\mathbf{\Sigma}_{\rm sweep}(s,\tau)\Phi\big)(\mathbf x)
&=-2\,\epsilon_{ab}\epsilon_{cd}\lim_{\mathbf y\to\mathbf x}
\Big[
\partial_{x_a}\partial_{y_b}\partial_{y_c}
K_{s,\tau}(\mathbf x,\mathbf y)\,
\partial_{y_d}\Phi(\mathbf y)\\
&\hspace{4.2cm}
+
\partial_{x_a}\partial_{y_c}
K_{s,\tau}(\mathbf x,\mathbf y)\,
\partial_{y_b}\partial_{y_d}\Phi(\mathbf y)
\Big].
\label{eq:Sigma_sweep_explicit_index}
\eea
\end{widetext}
%-------------------------------%
Equations \eqref{eq:K_subgrid_potential_covariance}--\eqref{eq:Sigma_sweep_explicit_index}
give an analytic definition of the sweeping self-energy as a deterministic
spatial differential operator. All coefficients are expressed in terms of the
two-time covariance kernel of the OU reservoir and deterministic derivatives of
the resolved field $\bar\zeta$. In particular, the symbol
$\mathbf\Sigma_{\rm sweep}$ denotes no additional modeling assumption: it is
the covariance of two point-split Hamiltonian transport operators generated by
the SPDE-defined reservoir.
%------------------------------------------------------------------------------------------------------------------------------------------------------%
\end{proof} 
%---------------------------------------------------------------------%
\bc[Relation to the Tangent-Space Geometric Self-Energy]
\label{cor:reservoir_tangent_self_energy}
Under the hypotheses of Theorem~\ref{thm:one_loop_emergence}, define the
reservoir-induced advection operator of the full stochastic
dressed-transport system by
\bea\label{Eq_A_res_Lie_operator}
\mathcal{A}_{\rm res}(s)\Phi
&:=
J\big(\psi_{sgs}(s),\Phi\big)
=
\mathcal{L}_{X_{\psi_{sgs}(s)}}\Phi,
\\
\psi_{sgs}(s)
&=
\frac12\gamma J\big(\bar\zeta(s),r(s)\big).
\eea
Assume, as in Theorem~\ref{thm:Dyson_propagator}, that
$\mathcal{A}_{\rm res}(s)$ is centered. The tangent-space geometric
self-energy derived in Appendix~\ref{app:Dyson_derivation} is then
\be
\boldsymbol{\eta}(s,\tau)\Phi
=
\mathbb{E}
\Big[
\mathcal{A}_{\rm res}(s)
\Big(
G(s,\tau)
\mathcal{A}_{\rm res}(\tau)\Phi
\Big)
\Big].
\label{eq_thm:eta_tangent}
\ee
At the leading-order OU baseline $r\to r^{(0)}$, the full
reservoir-induced operator reduces to the centered OU-induced advection
operator $\mathcal{A}_s^{(0)}$, and
\be
\boldsymbol{\eta}^{(0)}(s,\tau)\Phi
=
\mathbb{E}
\Big[
\mathcal{A}_s^{(0)}
\Big(
G(s,\tau)
\mathcal{A}_\tau^{(0)}\Phi
\Big)
\Big].
\label{eq:eta_leading_OU_baseline}
\ee
Define the unpropagated connected covariance core by
\be
\mathcal{C}^{(0)}(s,\tau)\Phi
:=
\mathbb{E}
\Big[
\mathcal{A}_s^{(0)}
\mathcal{A}_\tau^{(0)}\Phi
\Big].
\label{eq:OU_transport_covariance_core}
\ee
By Theorem~\ref{thm:one_loop_emergence},
\be
\mathbf{\Sigma}_{\rm sweep}(s,\tau)\Phi
=
-2\,\mathcal{C}^{(0)}(s,\tau)\Phi,
\label{eq:Sigma_covariance_core_relation}
\ee
whereas the tangent-space self-energy is the response-dressed
counterpart
\be
\boldsymbol{\eta}^{(0)}(s,\tau)\Phi
=
\mathbb{E}
\Big[
\mathcal{A}_s^{(0)}
\Big(
G(s,\tau)
\mathcal{A}_\tau^{(0)}\Phi
\Big)
\Big].
\label{eq:eta_response_dressed_core}
\ee
Thus $\mathbf{\Sigma}_{\rm sweep}(s,\tau)$ and
$\boldsymbol{\eta}^{(0)}(s,\tau)$ are generated by the same pair of
OU-induced Hamiltonian advection vertices, but they are not, in general,
identical operators. In the Volterra--Picard representation,
$\mathbf{\Sigma}_{\rm sweep}$ isolates the connected covariance insertion,
with the corresponding intermediate propagation displayed explicitly by
the adjacent factor $G_D(s,\tau)$ in
Eq.~\eqref{eq:Volterra_Dyson_Simplex}. In the tangent-space
representation, the intermediate response $G(s,\tau)$ is inserted
between the two vertices and incorporated into the geometric self-energy
kernel $\boldsymbol{\eta}^{(0)}(s,\tau)$.
\ec
%---------------------------------------------------------------------%
%------------------------------------------------------------------------%
\paragraph*{Interpretation.}
The corollary identifies the precise common structure of the two
calculations. The Volterra--Picard expansion and the tangent-space Dyson
theory do not introduce two unrelated memory mechanisms. Both are built
from the covariance of the same reservoir-induced Hamiltonian
transport vertex. They differ in how the intermediate response line is
organized: the reservoir calculation displays the propagation and the
covariance insertion as consecutive factors in an ordered causal
sequence, whereas the tangent-space formulation packages the
intermediate response together with the two vertices into a single
self-energy kernel.
Accordingly, $\mathbf{\Sigma}_{\rm sweep}$ should not be described as the
self-energy of the full Fr\'echet linearization of the coupled
$(\bar\zeta,r)$ system. It is the connected transport-covariance
insertion produced directly by the OU reservoir calculation. Once the
intermediate response is placed between its two Hamiltonian transport
vertices, the same mechanism takes the tangent-space form
$\boldsymbol{\eta}^{(0)}$. This establishes the one-loop correspondence
between the explicit reservoir expansion and the geometric
Dyson--Volterra response architecture. The DIA-type structure is
therefore generated by the stochastic symplectic reservoir dynamics
rather than introduced as an external statistical closure ansatz.

%------------------------------------------------------------------------------------------------------------------------------------------------------%
\begin{proof}
The proof relates the operator obtained from the reservoir
Volterra--Picard calculation to the tangent-space self-energy derived in
Appendix~\ref{app:Dyson_derivation}. The comparison is not between the
third-order reservoir expansion and the full Fr\'echet linearization of
the coupled $(\bar\zeta,r)$ system. Rather, it concerns two
representations of the same stochastic Eulerian transport mechanism.
For any realization of the reservoir $r$, the symplectic closure defines
the drift potential
\be
\psi_{sgs}(s)
=
\frac12\gamma J\big(\bar\zeta(s),r(s)\big).
\ee
The corresponding stochastic transport operator acting on a macroscopic
test function $\Phi$ is given by
Eq.~\eqref{Eq_A_res_Lie_operator}. Thus the random transport coefficient
in the dressed-transport skeleton is the reservoir-induced Hamiltonian
advection operator $\mathcal{A}_{\rm res}(s)$.
At the leading OU level used in the Volterra--Picard calculation, the
reservoir is replaced by $r^{(0)}$, so that
\be
\psi_{sgs}(s)
\quad\leadsto\quad
\psi_s^{(0)}
=
\frac12\gamma J\big(\bar\zeta(s),r^{(0)}(s)\big).
\ee
The corresponding OU-induced advection operator is
\be
\mathcal{A}_s^{(0)}\Phi
:=
J\big(\psi_s^{(0)},\Phi\big).
\label{eq:A0_transport_operator}
\ee
The transport-covariance insertion obtained in Step~4 of the proof of
Theorem~\ref{thm:one_loop_emergence} is therefore
\be
\mathbf{\Sigma}_{\rm sweep}(s,\tau)\Phi
=
-2\,\mathbb{E}
\Big[
\mathcal{A}_s^{(0)}
\mathcal{A}_\tau^{(0)}\Phi
\Big].
\label{eq:Sigma_sweep_A0_covariance}
\ee
The factor $2$ is the combinatorial contribution of the two ordered
sectors of the square time domain.
On the other hand, the tangent-space self-energy derived in
Appendix~\ref{app:Dyson_derivation} is
\be
\boldsymbol{\eta}(s,\tau)\Phi
=
\mathbb{E}
\Big[
\mathcal{A}_{\rm res}(s)
\Big(
G(s,\tau)
\mathcal{A}_{\rm res}(\tau)\Phi
\Big)
\Big].
\label{eq:eta_A_operator_form}
\ee
At the leading OU baseline $r\to r^{(0)}$, this becomes
\be
\boldsymbol{\eta}^{(0)}(s,\tau)\Phi
=
\mathbb{E}
\Big[
\mathcal{A}_s^{(0)}
\Big(
G(s,\tau)
\mathcal{A}_\tau^{(0)}\Phi
\Big)
\Big].
\label{eq:eta_OU_level}
\ee
Equations~\eqref{eq:Sigma_sweep_A0_covariance} and
\eqref{eq:eta_OU_level} contain the same pair of OU-induced Hamiltonian
advection vertices. Their difference is the organization of the
intermediate propagation.

In the tangent-space formulation,
$G(s,\tau)$ is inserted between the two vertices and forms part of the
self-energy kernel.

In the ordered Volterra--Picard formulation, the corresponding
propagation remains explicit in the full causal sequence
\bes
G_D(t,s)
\mathcal M(\bar\zeta(t))
G_D(s,\tau)
\mathbf\Sigma_{\rm sweep}(s,\tau)
G_D(\tau,\sigma)
\bar\zeta(\sigma),
\ees
appearing in Theorem~\ref{thm:one_loop_emergence}. In particular, the
intermediate reservoir propagation $G_D(s,\tau)$ is displayed as the
factor immediately preceding the covariance insertion
$\mathbf\Sigma_{\rm sweep}(s,\tau)$, rather than being incorporated
inside that operator as $G(s,\tau)$ is in the tangent-space
self-energy.

Define now the connected two-vertex covariance operator
\be
\mathcal{C}^{(0)}(s,\tau)\Phi
:=
\mathbb{E}
\Big[
\mathcal{A}_s^{(0)}
\mathcal{A}_\tau^{(0)}\Phi
\Big].
\label{eq:C0_covariance_core_proof}
\ee
Equation~\eqref{eq:Sigma_sweep_A0_covariance} then gives the exact
relation
\be
\mathcal{C}^{(0)}(s,\tau)\Phi
=
-\frac12
\mathbf{\Sigma}_{\rm sweep}(s,\tau)\Phi.
\label{eq:eta_sigma_core_relation}
\ee
Thus $\mathbf{\Sigma}_{\rm sweep}(s,\tau)$ records the connected
covariance of two OU-induced advection events acting successively on
$\Phi$, with no response evolution inserted between them. By contrast,
the tangent-space kernel
\be
\boldsymbol{\eta}^{(0)}(s,\tau)\Phi
=
\mathbb{E}
\Big[
\mathcal{A}_s^{(0)}
\Big(
G(s,\tau)
\mathcal{A}_\tau^{(0)}\Phi
\Big)
\Big]
\ee
describes the same pair of stochastic advection events while allowing the
perturbation generated at time $\tau$ to evolve through the intermediate
response $G(s,\tau)$ before it is acted upon again at time $s$. No
commutation of $G(s,\tau)$ with $\mathcal{A}_s^{(0)}$ is assumed or
required.
The correspondence is therefore an operator-level relation between two
different organizations of the same OU-induced transport covariance,
not an equality between
$\mathbf{\Sigma}_{\rm sweep}$ and $\boldsymbol{\eta}^{(0)}$. In the
Volterra--Picard calculation, the two-vertex covariance is evaluated
explicitly through the point-split SPDE kernel $K_{s,\tau}$ in
Eq.~\eqref{eq:K_subgrid_potential_covariance}, producing the differential
operator \eqref{eq:Sigma_sweep_explicit_index}. Its propagation remains
external to that operator in the ordered causal sequence
\bes
G_D(t,s) \,
\mathcal{M}(\bar\zeta(t)) \,
G_D(s,\tau) \,
\mathbf{\Sigma}_{\rm sweep}(s,\tau) \,
G_D(\tau,\sigma) \,
\bar\zeta(\sigma).
\ees
In the tangent-space formulation, the intermediate propagation is
instead placed directly between the same two Hamiltonian advection
vertices and thereby becomes part of the self-energy kernel
$\boldsymbol{\eta}^{(0)}(s,\tau)$. Accordingly,
$\mathbf{\Sigma}_{\rm sweep}$ supplies the explicit OU two-vertex
covariance operator, while $\boldsymbol{\eta}^{(0)}$ is its
response-weighted tangent-space form.
\end{proof}
%------------------------------------------------------------------------------------------------------------------------------------------------------%
%==================END OF PROOF=========================%
%------------------------------------------------------------------------------------------------------------------------------------------------------%

\br[Dependence on the resolved deterministic skeleton]
\label{rem:deterministic_skeleton_independence}
The statement of Theorem~\ref{thm:one_loop_emergence} is formulated
along the full forced Navier--Stokes--$\beta$ resolved skeleton. This
choice fixes the background trajectory and the deterministic part of the
tangent propagation. However, the covariance insertion identified in the
theorem is generated locally by the resolved--reservoir coupling,
namely by the OU process $r^{(0)}$ and the symplectic transport vertex
\be
J(\psi_{sgs}^{(0)},\cdot),
\qquad
\psi_{sgs}^{(0)}
=
\frac12\gamma J(\bar\zeta,r^{(0)}).
\ee
The deterministic resolved terms
$-\beta\partial_x\bar\psi$, $-\mu\bar\zeta$, $\nu\Delta\bar\zeta$, and
$F_\zeta$ may change the background path, the large-scale balance, and
the baseline propagator, but they do not introduce a second
resolved--reservoir vertex and therefore do not alter the
transport-covariance core $\mathbf\Sigma_{\rm sweep}$. Consequently, when
the theory is restricted in Section~\ref{sec:scaling_preliminaries} to
the local inertial skeleton
\be
\partial_t\bar\zeta+J(\bar\psi,\bar\zeta)=J(\psi_{sgs},\bar\zeta),
\ee
the same geometric self-energy core remains the object that dresses the
Eulerian response. The forcing used later to maintain stationarity fixes
the fluxes $\epsilon$ and $\eta_Z$; it does not affect the infrared
selection rule of the symplectic vertex.
\er

Theorem~\ref{thm:one_loop_emergence} and
Corollary~\ref{cor:reservoir_tangent_self_energy} complete the
structural identification of the reservoir-generated response kernel.
The theorem derives the connected two-time covariance
$\mathbf{\Sigma}_{\rm sweep}$ directly from the third-order
Volterra--Picard expansion, while the corollary shows how the same pair
of OU-induced Hamiltonian advection operators enters the
response-weighted tangent-space self-energy
$\boldsymbol{\eta}^{(0)}$. What remains unresolved at this stage is the
spectral content of that covariance: in particular, whether its
large-scale contribution is dominated by rigid Eulerian sweeping or by
the strain-producing motions relevant to cascade transfer. Section \ref{Sec_inertial_ranges} below
answers this question by projecting the nested Jacobian covariance onto
Fourier shells and determining its infrared weighting.

%===========================================================%

%===========================================================%
\section{Spectral Scaling: Geometric Suppression of Sweeping and Inertial Range Consistency}
\label{Sec_inertial_ranges}

%===========================================================%
\subsection{Preliminaries: Spectral Transfer and the Dressed Volterra Response}
\label{sec:scaling_preliminaries}
Section~\ref{Sec_higher_order_memory} established the operator structure
of the reservoir-generated memory. Theorem~\ref{thm:one_loop_emergence}
derived the connected two-time transport covariance
$\mathbf{\Sigma}_{\rm sweep}(s,\tau)$ from the third-order
Volterra--Picard expansion, while
Corollary~\ref{cor:reservoir_tangent_self_energy} showed how the same
pair of OU-induced Hamiltonian advection operators enters the
response-weighted tangent-space self-energy
$\boldsymbol{\eta}^{(0)}(s,\tau)$. What remains to be determined is the
spectral content of this memory: in particular, whether its
large-scale contribution is dominated by rigid Eulerian sweeping or by
the relative deformation that drives scale-to-scale transfer. The
purpose of this subsection is to place that question within the
inertial-range energy balance and to identify precisely where the
reservoir-generated response time enters.

As stated in
Remark~\ref{rem:deterministic_skeleton_independence}, the deterministic
Navier--Stokes--$\beta$ terms modify the background trajectory and the
baseline propagation, but they do not change the transport-covariance
core generated by the symplectic reservoir. We may therefore isolate the
local isotropic inertial mechanism by considering scales
$k\gg k_\beta=\sqrt{\beta/U}$, setting $\beta=0$, and omitting Ekman
friction, viscosity, and forcing from the local transfer equation. The
forcing needed to maintain statistical stationarity remains understood
to act outside the inertial interval and fixes the fluxes introduced
below; it does not enter the local spectral power counting.

The resolved vorticity equation then reduces to the dressed-transport
skeleton
\be
\partial_t\bar\zeta
+
J(\bar\psi,\bar\zeta)
=
J(\psi_{sgs},\bar\zeta).
\label{eq:zeta_inertial}
\ee
For later use, write
\beas
J_b
&:=
J(\bar\psi,\bar\zeta),
\quad
J_{sgs}
:=
J(\psi_{sgs},\bar\zeta),
\quad
\partial_t\bar\zeta=-J_b+J_{sgs}.
\eeas
The two Jacobians have complementary roles in the response formulation.
The resolved Jacobian $J_b$ is the nonlinear tendency whose work
determines the spectral transfer amplitude. The reservoir-induced
Jacobian $J_{sgs}$ enters the tangent dynamics and modifies the time over
which that nonlinear tendency remains correlated.

This temporal dressing is the result established in the preceding
section. At the leading OU level, define
\be
\psi_s^{(0)}
=
\frac12\gamma
J\big(\bar\zeta(s),r^{(0)}(s)\big),
\label{Eq_psi_s^0}
\ee
and the corresponding Hamiltonian advection operator
\beas
\mathcal{A}_s^{(0)}\Phi
=
J\big(\psi_s^{(0)},\Phi\big).
\eeas
Theorem~\ref{thm:one_loop_emergence} gives the unpropagated connected
transport covariance
\beas
\mathbf{\Sigma}_{\rm sweep}(s,\tau)\Phi
=
-2\,
\mathbb{E}
\Big[
\mathcal{A}_s^{(0)}
\mathcal{A}_\tau^{(0)}\Phi
\Big],
\eeas
as expressed in
Eq.~\eqref{eq:Sigma_sweep_A0_covariance}. Corollary~\ref{cor:reservoir_tangent_self_energy}
then gives its response-weighted counterpart,
\beas
\boldsymbol{\eta}^{(0)}(s,\tau)\Phi
=
\mathbb{E}
\Big[
\mathcal{A}_s^{(0)}
\Big(
G(s,\tau)
\mathcal{A}_\tau^{(0)}\Phi
\Big)
\Big],
\eeas
as expressed in Eq.~\eqref{eq:eta_response_dressed_core}. Thus
$\mathbf{\Sigma}_{\rm sweep}$ identifies the two-time covariance generated
by the reservoir, whereas the intermediate response $G(s,\tau)$ records
how a disturbance propagates between the two Hamiltonian advection
actions. The spectral theorem below examines the infrared content of
this same covariance structure.

To connect this response architecture with spectral transfer, let
$\mathcal{U}(t,s)$ denote the pathwise fundamental operator of the
direct-advection tangent equation
\eqref{eq:direct_advection_tangent_skeleton}. As reviewed in
Appendix~\ref{app:variational_propagator}, its ensemble average is the
macroscopic response
\be
G(t,s)
=
\mathbb{E}[\mathcal{U}(t,s)]
=
\mathbb{E}
\left[
\frac{\delta\bar\zeta(t)}
{\delta f(s)}
\right].
\label{eq:G_zeta_def}
\ee
The response $G$ is therefore not an independently prescribed
eddy-damping function. It is the averaged propagation of an
infinitesimal disturbance through the stochastic dressed-transport
background, with its memory generated by the reservoir-induced
self-energy derived above.

The spectral energy transfer is obtained from the work of the resolved
nonlinear tendency. Denoting
\beas
J_k(t)
:=
\mathcal{F}_k
\big\{
J(\bar\psi,\bar\zeta)
\big\},
\eeas
we write, up to the conventional sign and normalization,
\be
T(k)
\propto
\mathrm{Re}\,
\mathbb{E}
\Big[
\bar\psi_k^*(t)J_k(t)
\Big].
\label{Eq_unclosed_work}
\ee
In this instantaneous form, the transfer involves an unclosed triple
correlation. The response representation converts it into a two-time
covariance. Treating the resolved nonlinear tendency as the internal
source propagated by the dressed tangent dynamics gives schematically
\be
\bar\zeta(t)
\propto
\int_{-\infty}^{t}
\mathcal{U}(t,s)
J(\bar\psi(s),\bar\zeta(s))
\,ds.
\label{eq:zeta_response_raw}
\ee
%The minus sign follows from
%$\partial_t\bar\zeta=-J_b+J_{sgs}$ and is immaterial for the dimensional
%scaling.

Using the Fourier inversion
$\bar\zeta_k=-k^2\bar\psi_k$, passing to the ensemble-averaged response,
and imposing stationarity and isotropy yield
\be
\bar\psi_k(t)
\propto
k^{-2}
\int_0^\infty
G(k,\tau)
J_k(t-\tau)
\,d\tau.
\label{eq:psi_response_G}
\ee
Here $G(k,\tau)$ denotes the diagonal shell projection of the dressed
response. The factor $k^{-2}$ comes only from the elliptic inversion
between vorticity and streamfunction; it is not part of the dynamical
memory.

Substitution into Eq.~\eqref{Eq_unclosed_work} gives
\be
T(k)
=
\mathrm{Re}
\int_0^\infty
k^{-2}G(k,\tau)\,
\mathbb{E}
\Big[
\langle
J_k(t)J_k^*(t-\tau)
\rangle
\Big]
\,d\tau,
\label{eq:Tk_integral}
\ee
where $\langle\cdot\rangle$ denotes spatial averaging. This formula
separates the two ingredients required for the inertial-range argument.
The covariance of $J_k$ supplies the magnitude of the nonlinear
transfer, while $G(k,\tau)$ determines how long that transfer remains
temporally coherent. We therefore define the associated response time by
\be
\theta_k
:=
\int_0^\infty
G(k,\tau)
\,d\tau.
\label{eq:theta_k_def}
\ee
The central question for Theorem~\ref{thm:inertial_ranges} is whether
this time is shortened by rigid large-scale sweeping or is controlled by
the strain-producing interactions relevant to the cascade.

It remains to estimate the variance factor in
Eq.~\eqref{eq:Tk_integral}. The resolved vorticity is transported by the
effective transport velocity $\mathbf{u}_{\rm trs}=\bar{\mathbf{u}}- \mathbf{u}_{sgs}$ (Eq.~\eqref{Eq_u_transport_def}),
and we define $E(k)$ as the one-dimensional kinetic-energy spectrum of
this active transport field. The energy in a logarithmic band obeys
$\int_k^{2k}E(p)\,dp\sim kE(k)$, so the characteristic velocity amplitude
is
\be
u_k:=|\mathbf{u}_{\rm trs}|_k\sim\sqrt{kE(k)}.
\label{eq:uk_Ek_relation}
\ee
The resolved Jacobian $J_b$, however, is expressed in terms of the bare
macroscopic fields $\bar\psi$ and $\bar\zeta$. Appendix~\ref{app:anomalous_scaling},
and specifically Theorem~\ref{thm:no_anomalous_scaling}, shows that in a
stationary scale-invariant inertial range the bare velocity, the
reservoir-induced drift velocity, and the effective transport velocity
share the same scaling exponent. Consequently,
\be
|\bar{\mathbf{u}}|_k
\sim
u_k
\sim
\sqrt{kE(k)}
\label{eq:ubar_uk_same_scaling}
\ee
at the level of inertial-range power counting. The reservoir may change
amplitudes, phases, and response times, but it introduces no independent
velocity exponent.

It follows that the characteristic resolved amplitudes satisfy
\beas
\bar\psi_k
&\sim
k^{-1/2}E(k)^{1/2},
\\
\bar\zeta_k
&\sim
k^{3/2}E(k)^{1/2}.
\eeas
For local triads, the resolved Jacobian therefore scales as
\beas
J_k
\sim
k^2\bar\psi_k\bar\zeta_k
\sim
k^3E(k).
\eeas
Since $J_k^2$ represents the variance contained in a logarithmic shell,
the corresponding one-dimensional variance density scales as
\beas
\frac{\langle J_kJ_k^*\rangle}{k}
\sim
k^5E(k)^2.
\eeas
Combining this density with the elliptic factor $k^{-2}$ in
Eq.~\eqref{eq:Tk_integral} gives
\beas
k^{-2}
\frac{\langle J_kJ_k^*\rangle}{k}
\sim
k^3E(k)^2.
\eeas
Integration over the dressed response time then yields
\be
T(k)
\sim
\theta_k\,k^3E(k)^2.
\label{Eq_Tk_response}
\ee
Equation~\eqref{Eq_Tk_response} isolates the remaining issue. The
resolved nonlinear covariance fixes the transfer amplitude, but the
cascade scaling depends on the response time $\theta_k$. Section~\ref{Sec_higher_order_memory}
has already shown that this time is generated by the reservoir-induced
Hamiltonian transport covariance. Theorem~\ref{thm:inertial_ranges}
now determines which internal scales contribute to that covariance. Its
key step is to examine the limit of an internal resolved mode
$p\to0$ and to show that the nested Jacobian structure suppresses the
uniform-translation contribution before it can dominate the Eulerian
response.
%----------------------------------------------------------------------%

%===========================================================%
\subsection{Theorem: The Dual Inertial Cascades}
\label{sec:cascade_theorem}

%---------------------------------------------------------------%
\begin{figure*}[t]
\centering
\includegraphics[width=0.87\textwidth]{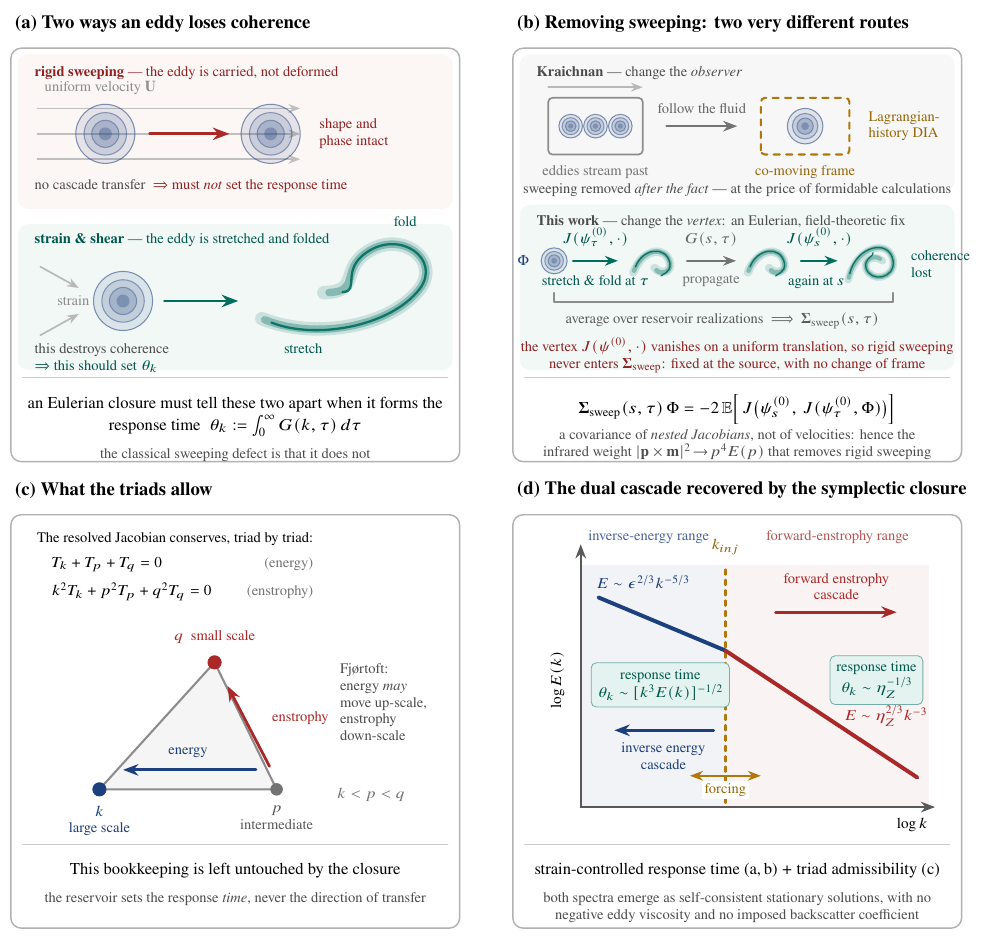}
\caption{\label{fig:sweeping_cascade}%
\textbf{The framework in action: sweeping suppression and the Kraichnan
dual cascade.}
(a) Two ways an eddy loses coherence at a fixed Eulerian point. A uniform
velocity $\mathbf U$ carries the eddy past the observer with its shape and
internal phase intact, producing no cascade transfer. Strain and shear
instead stretch the eddy into a filament and fold it back on itself,
destroying coherence and driving the cascade. A closure must distinguish
the two when it forms the response time
$\theta_k:=\int_0^\infty G(k,\tau)\,d\tau$; the classical sweeping defect
of Eulerian closures is precisely that it does not.
(b) Removing sweeping: two very different routes. Kraichnan's
resolution changes the \emph{observer}: the Lagrangian-history closures
reorganize the response around fluid trajectories so that the common
translation is factored out after the fact---a reformulation carried
through formidable calculations. The present theory instead changes the
transport \emph{vertex} itself: an Eulerian, field-theoretic fix applied {\it upfront}, before any renormalization is performed. The kernel is assembled from two Hamiltonian Lie transports
along the leading-order reservoir-induced drift
$\psi^{(0)}_\bullet=\frac12\gamma J(\bar\zeta,r^{(0)})$ of
Eq.~\eqref{eq:psi_sgs_zeroth}: the disturbance $\Phi$ is stretched and
folded at time $\tau$ by $J(\psi_\tau^{(0)},\cdot)$, propagated to time
$s$ by the dressed response $G(s,\tau)$, then stretched and folded again
by $J(\psi_s^{(0)},\cdot)$; averaging over reservoir realizations yields
$\mathbf\Sigma_{\rm sweep}(s,\tau)$.  Because the vertex $J(\psi^{(0)},\cdot)$ vanishes identically on
a uniform translation, rigid sweeping never enters
$\mathbf\Sigma_{\rm sweep}$, and no change of frame is required. The resulting kernel is a covariance
of nested Jacobians rather than of velocities, which is what produces the
infrared weight $p^4E(p)\,dp$ of Theorem~\ref{thm:inertial_ranges}.
(c) What the triads allow. The resolved Jacobian still conserves energy
and enstrophy triad by triad, so the Fj\o{}rtoft constraints continue to
permit upscale energy transfer together with downscale enstrophy
transfer. The closure leaves this bookkeeping untouched: it supplies the
response time, never the direction of transfer.
(d) The dual cascade recovered by the symplectic closure. With sweeping
removed, the surviving
decorrelation is strain-controlled,
$\theta_k\sim[k^3E(k)]^{-1/2}$ in the inverse range and
$\theta_k\sim\eta_Z^{-1/3}$ in the forward range. Combined with triad
admissibility, the classical spectra
$E(k)\sim\epsilon^{2/3}k^{-5/3}$ and $E(k)\sim\eta_Z^{2/3}k^{-3}$ emerge
as self-consistent stationary solutions of the dressed Eulerian transfer
theory, with no negative eddy viscosity and no imposed backscatter
coefficient.}
\end{figure*}
%---------------------------------------------------------------%

We now place the dressed transfer relation
Eq.~\eqref{Eq_Tk_response} in the standard statistically stationary setting
of forced two-dimensional turbulence. Let the resolved flow be sustained
by a narrow-band stochastic forcing concentrated near an injection
wavenumber $k_{inj}$. Away from the forcing band and from the large- and
small-scale dissipative cutoffs, the dynamics are effectively
conservative. Spectral space is therefore divided into an inverse-energy
range, $k\ll k_{inj}$, and a forward-enstrophy range,
$k\gg k_{inj}$, as in Kraichnan's dual-cascade phenomenology
\cite{Kraichnan1967,kraichnan1971,smithr1994}.

Let $\epsilon$ denote the stationary kinetic-energy injection rate.
Because the forcing is localized near $k_{inj}$, its associated
enstrophy injection rate satisfies
\beas
\eta_Z
\sim
k_{inj}^2\epsilon.
\eeas
The kinetic-energy and enstrophy fluxes are defined from the conservative
transfer spectrum by
%-----------------------------------%
\begin{subequations}
\begin{align}
\Pi_E(k)
&=
\int_k^\infty T(p),dp,
\label{eq:Pi_E_def}\\
\Pi_Z(k)
&=
\int_k^\infty p^2T(p),dp.
\label{eq:Pi_Z_def}
\end{align}
\end{subequations}
%-----------------------------------%
Statistical stationarity corresponds to an approximately constant
upscale energy flux in the inverse range and an approximately constant
downscale enstrophy flux in the forward range.

Section~\ref{sec:scaling_preliminaries} reduced the transfer estimate to
$T(k) \sim
\theta_k k^3E(k)^2$,
where the resolved Jacobian covariance determines the transfer
amplitude, while $\theta_k$, defined in
Eq.~\eqref{eq:theta_k_def}, measures the duration of its temporal
coherence. The remaining problem is therefore to determine which motions
control this response time.

The operator origin of $\theta_k$ has already been established in
Section~\ref{Sec_higher_order_memory}.
Theorem~\ref{thm:one_loop_emergence} derives the connected reservoir
transport covariance
$\mathbf\Sigma_{\rm sweep}(s,\tau)$ in
Eq.~\eqref{eq:Sigma_sweep_A0_covariance}, and
Corollary~\ref{cor:reservoir_tangent_self_energy} shows how the same two
OU-induced Hamiltonian advection operators enter the response-weighted
self-energy in Eq.~\eqref{eq:eta_response_dressed_core}. Thus the memory
time is generated internally by the stochastic symplectic reservoir,
rather than supplied through a prescribed eddy-damping law.

What remains is a spectral question. A conventional Eulerian response
can be dominated by energetic large-scale motions that merely translate
smaller eddies, producing rapid phase decorrelation without the relative
deformation responsible for cascade transfer
\cite{kraichnan1964kolmogorov,kraichnan1978sweeping}. The theorem below
examines the infrared content of the reservoir-generated covariance and
shows that its nested Jacobian structure suppresses precisely this
uniform-sweeping contribution. The surviving memory is therefore
controlled by strain-producing interactions. Under the standard
Fj\o{}rtoft--Kraichnan locality and stationary-flux assumptions, this
leads to the classical inverse-energy and forward-enstrophy spectra as
self-consistent stationary solutions of the dressed Eulerian transfer
theory.

%-----------------------------------------------------------------------------------------------------%
\bt[Geometric Sweeping Suppression and Inertial-Range Consistency]
\label{thm:inertial_ranges}
Let the macroscopic spectral transfer $T(k)$ be generated by the
inertial-limit form of the Symplectic Geometric Closure,
Eq.~\eqref{eq:zeta_inertial}. Let $E(k)$ denote the one-dimensional
kinetic-energy spectrum of the effective transport field
$\mathbf u_{\rm trs}=\bar{\mathbf u}-\mathbf u_{sgs}$, and assume the
no-anomalous-scaling result of Theorem~\ref{thm:no_anomalous_scaling}.
At the OU/one-loop level of Theorem~\ref{thm:one_loop_emergence}, the
geometric self-energy is generated by the covariance of two Hamiltonian
transport operators,
\be
\mathbf\Sigma_{\rm sweep}(s,\tau)\Phi
=
-2\,\mathbb{E}
\Big[
J\big(\psi_s^{(0)},J(\psi_\tau^{(0)},\Phi)\big)
\Big].
\ee
Then the diagonal spectral projection of this Lie-derivative covariance
has, in the linear-shell infrared sense, an internal resolved-mode
weighting proportional to
\be
p^4E(p)\,dp.
\ee
Consequently, the uniform-translation contribution responsible for the
classical Eulerian sweeping defect is suppressed at the level of the
self-energy integrand. Under the standard Fj\o{}rtoft--Kraichnan
locality and stationary-flux assumptions, the resulting geometrically
renormalized Volterra transfer theory admits the Kraichnan inverse-energy
cascade
\be
E(k)\sim \epsilon^{2/3}k^{-5/3}
\ee
and the forward-enstrophy cascade
\be
E(k)\sim \eta_Z^{2/3}k^{-3}
\ee
as self-consistent stationary inertial-range solutions.
\et
%-----------------------------------------------------------------------------------------------------%

%-----------------------------------------------------------------------------------------------------%
\noindent\textit{Idea of the proof.}
The proof separates the spatial and temporal roles of the one-loop
kernel. The nested Hamiltonian transport vertices determine which
motions can contribute to the self-energy: their cross-gradient structure
produces an infrared zero for a uniform resolved translation and, after
shell projection, the weighting $p^4E(p)\,dp$. The dressed response line
does not generate this spatial selection rule; it propagates the
admissible interaction in time and determines its memory duration. Once
the deep-infrared sweeping contribution has been suppressed, the standard
Fj\o{}rtoft--Kraichnan locality hypothesis identifies the remaining
decorrelation time with the local strain time. The stationary energy- and
enstrophy-flux balances then yield the classical dual-cascade spectra.

\begin{proof}
The proof has five steps. First, we rewrite the geometric self-energy in
the spectral form appropriate for infrared power counting. Second, we
show that the nested Jacobian structure produces a vertex-level zero for
uniform translations and the corresponding $p^4E(p)$ shell suppression.
Third, under the stated locality hypothesis, we identify the
geometrically renormalized triad-relaxation time. Fourth and fifth, we
close the stationary energy- and enstrophy-flux balances, respectively.

\textbf{Step 1: Spectral form of the one-loop transport covariance.}
Theorem~\ref{thm:one_loop_emergence} showed that the connected covariance
part of the third-order reservoir expansion produces the deterministic
transport-covariance insertion
\be
\mathbf{\Sigma}_{\rm sweep}(s,\tau)\Phi
=
-2\,\mathbb{E}
\Big[
J\big(\psi_s^{(0)},J(\psi_\tau^{(0)},\Phi)\big)
\Big],
\label{eq:Sigma_sweep_IR_start}
\ee
where
\be
\psi_s^{(0)}
=
\frac12\gamma J(\bar\zeta(s),r^{(0)}(s)).
\ee
This operator is the connected two-vertex covariance underlying the
Eulerian tangent-space self-energy $\boldsymbol{\eta}^{(0)}$. The two
representations organize the intermediate propagation differently. In
the ordered Volterra--Picard formulation, it remains explicit in the
causal sequence
\bes
G_D(t,s) \,
\mathcal{M}(\bar\zeta(t)) \,
G_D(s,\tau) \,
\mathbf{\Sigma}_{\rm sweep}(s,\tau) \,
G_D(\tau,\sigma) \,
\bar\zeta(\sigma),
\ees
whereas in the tangent-space Dyson formulation the corresponding
response is inserted between the same two stochastic transport vertices,
\bes
\boldsymbol{\eta}^{(0)}(s,\tau)\Phi
=
\mathbb{E}
\Big[
\mathcal{A}_s^{(0)}
\Big(
G(s,\tau)\mathcal{A}_\tau^{(0)}\Phi
\Big)
\Big].
\ees
The response line therefore supplies temporal propagation and memory,
while the explicit powers of the internal infrared mode are generated by
the spatial structure of the nested Hamiltonian vertices. The relevant
object for this power count is consequently the Lie-derivative covariance
in Eq.~\eqref{eq:Sigma_sweep_IR_start}.

Let $\mathbf{k}$ be the external target mode. For a Jacobian,
\be
\mathcal{F}_{\mathbf{k}}\{J(f,g)\}
=
\int_{\mathbf{m}+\mathbf{q}=\mathbf{k}}
(\mathbf{m}\times\mathbf{q})
f_{\mathbf{m}}g_{\mathbf{q}} \,d\mathbf{m},
\ee
where
$\mathbf{m}\times\mathbf{q}=m_1q_2-m_2q_1$. Since
$\mathbf{q}=\mathbf{k}-\mathbf{m}$, the vertex factor may be written as
\be
\mathbf{m}\times(\mathbf{k}-\mathbf{m})
=
\mathbf{m}\times\mathbf{k}.
\ee
Thus the pair of outer Lie-transport vertices in
Eq.~\eqref{eq:Sigma_sweep_IR_start} contributes the geometric factor
$|\mathbf{m}\times\mathbf{k}|^2$
in the diagonal self-energy projection. This factor couples the drift
mode $\mathbf{m}$ to the external target mode $\mathbf{k}$.

It remains to identify the spectral content of the drift mode
$\psi_{sgs}^{(0)}(\mathbf{m})$. From
$\psi_{sgs}^{(0)}=
\frac12\gamma J(\bar\zeta,r^{(0)}),$
we obtain, up to bounded Fourier weights generated by the smooth
coefficient $\gamma$,
\be
\psi_{sgs}^{(0)}(\mathbf{m},t)
\sim
\int
(\mathbf{p}\times\mathbf{m})
\bar\zeta(\mathbf{p},t)
r^{(0)}(\mathbf{m}-\mathbf{p},t) \,d\mathbf{p}.
\label{eq:psi_sgs_reduced_convolution_IR}
\ee
Indeed, the reservoir mode paired with $\bar\zeta(\mathbf{p},t)$ is
$\mathbf{m}-\mathbf{p}$, and the inner Jacobian vertex satisfies
\be
\mathbf{p}\times(\mathbf{m}-\mathbf{p})
=
\mathbf{p}\times\mathbf{m}.
\ee
The mode $\mathbf{p}$ is therefore the resolved macroscopic mode entering
the construction of the stochastic drift potential. The infrared
sweeping limit is the limit $\mathbf{p}\to0$ inside the covariance of
$\psi_{sgs}^{(0)}$, rather than a limit on the external target mode
$\mathbf{k}$.

The covariance of two drift-potential modes is obtained by contracting
the reservoir slots. The conjugate drift mode is
\bes
\psi_{sgs}^{(0)}(\mathbf{m},\tau)^*
\sim
\int
(\mathbf{p}'\times\mathbf{m})
\bar\zeta(\mathbf{p}',\tau)^*
r^{(0)}(\mathbf{m}-\mathbf{p}',\tau)^* \,d\mathbf{p}'.
\ees
Conditioning on the resolved field $\bar\zeta$, the only random factors
in this contraction are the OU reservoir modes. By the stationary
homogeneous covariance structure of
Appendix~\ref{App_Macro_Diffusion_Ope},
Eqs.~\eqref{eq:OU_two_time_covariance_kernel}--\eqref{eq:OU_two_time_covariance_fourier}, we have 
\be
\mathbb{E}
\Big[
r^{(0)}(\mathbf{q},t)
r^{(0)}(\mathbf{q}',\tau)^*
\Big]
=
\delta(\mathbf{q}-\mathbf{q}')C_{\mathbf{q}}(t,\tau).
\label{eq:OU_fourier_covariance_IR}
\ee
Substituting
\be
\mathbf{q}=\mathbf{m}-\mathbf{p},
\qquad
\mathbf{q}'=\mathbf{m}-\mathbf{p}',
\ee
the Wick contraction gives
\be
\delta(\mathbf{q}-\mathbf{q}')
=
\delta\big((\mathbf{m}-\mathbf{p})-(\mathbf{m}-\mathbf{p}')\big)
=
\delta(\mathbf{p}'-\mathbf{p}).
\ee
Thus the reservoir contraction identifies the two resolved labels and
yields
%----------------------------------------------%
\begin{widetext}
\be
\mathbb{E}
\Big[
\psi_{sgs}^{(0)}(\mathbf{m},t)
\psi_{sgs}^{(0)}(\mathbf{m},\tau)^*
\Big]
\sim
\int
|\mathbf{p}\times\mathbf{m}|^2
\bar\zeta(\mathbf{p},t)\bar\zeta(\mathbf{p},\tau)^*
C_{\mathbf{m}-\mathbf{p}}(t,\tau) \,d\mathbf{p}.
\label{eq:psi_sgs_covariance_two_time}
\ee
\end{widetext}
%----------------------------------------------%
The scalar factor $C_{\mathbf{m}-\mathbf{p}}(t,\tau)$ in
Eq.~\eqref{eq:psi_sgs_covariance_two_time} is the simplest Fourier
representation of the reservoir contraction. In the full
Lie-derivative covariance, the $r^{(0)}$ slots enter through
differentiated point-split covariances.
Appendix~\ref{App_Macro_Diffusion_Ope} makes these stochastic objects
precise:
\be
\mathbb{E}[
\partial_{x_i}r^{(0)}(\mathbf{x},t)
\partial_{y_j}r^{(0)}(\mathbf{y},\tau)]
=
\partial_{x_i}\partial_{y_j}
C_{t,\tau}(\mathbf{x}-\mathbf{y}),
\ee
with Fourier representation
\be
\partial_{x_i}\partial_{y_j}
C_{t,\tau}(\mathbf{x}-\mathbf{y})
=
\sum_{\mathbf{q}\neq0}
q_iq_j
C_{\mathbf{q}}(t,\tau)
e^{i\mathbf{q}\cdot(\mathbf{x}-\mathbf{y})}.
\ee
These are Eqs.~\eqref{eq:point_split_justification} and
\eqref{eq:OU_point_split_fourier_multiplier}. To connect the
point-split formulation with the preceding Fourier calculation, define
the matrix-valued Fourier symbol of the two-time reservoir-gradient covariance by:
\be
\mathsf{C}_{ij}^{\nabla}(\mathbf{q};t,\tau)
:=
q_iq_j \, C_{\mathbf{q}}(t,\tau).
\label{eq:OU_gradient_covariance_symbol_IR}
\ee
In the covariance of the inner Jacobian, this tensor is contracted with
the two resolved-gradient factors and the two antisymmetric tensors
defining the Jacobian. Writing $\epsilon_{ij}$ for the
two-dimensional Levi--Civita tensor, one obtains
\bea
\epsilon_{ai}p_a \,
\mathsf{C}_{ij}^{\nabla}(\mathbf{q};t,\tau) \,
\epsilon_{bj}p_b
&=
\epsilon_{ai}p_aq_i \,
\epsilon_{bj}p_bq_j \,
C_{\mathbf{q}}(t,\tau)
\nonumber\\
&=
|\mathbf{p}\times\mathbf{q}|^2 \,
C_{\mathbf{q}}(t,\tau).
\label{eq:OU_gradient_covariance_Jacobian_contraction}
\eea
Since
\be
\mathbf{q}=\mathbf{m}-\mathbf{p},
\ee
the antisymmetry of the cross product gives
\be
|\mathbf{p}\times\mathbf{q}|^2
=
|\mathbf{p}\times(\mathbf{m}-\mathbf{p})|^2
=
|\mathbf{p}\times\mathbf{m}|^2.
\ee
Thus the tensorial multiplier $q_iq_j$ appearing in the point-split
gradient covariance is already incorporated, after contraction with the
Jacobian vertices, into the factor
$|\mathbf{p}\times\mathbf{m}|^2$. The drift-potential covariance entering
the infrared estimate may therefore be written as
%----------------------------------------------%
\begin{widetext}
\be
\mathbb{E}
\Big[
\psi_{sgs}^{(0)}(\mathbf{m},t)
\psi_{sgs}^{(0)}(\mathbf{m},\tau)^*
\Big]
\sim
\int
|\mathbf{p}\times\mathbf{m}|^2 \,
\bar\zeta(\mathbf{p},t)\bar\zeta(\mathbf{p},\tau)^* \,
C_{\mathbf{m}-\mathbf{p}}(t,\tau)\,d\mathbf{p}.
\label{eq:psi_sgs_covariance_IR}
\ee
\end{widetext}
%----------------------------------------------%
For a smooth nonconstant coefficient $\gamma$, the corresponding
Fourier convolutions introduce bounded mode-coupling weights into
Eq.~\eqref{eq:psi_sgs_covariance_IR}. These weights are included in the
scaling symbol $\sim$ and do not alter the explicit infrared factor
$|\mathbf{p}\times\mathbf{m}|^2$.

For the infrared shell estimate, the two-time product
$\bar\zeta(\mathbf{p},t)\bar\zeta(\mathbf{p},\tau)^*$ must not be
replaced pointwise by an equal-time quantity. Instead, define the
resolved two-time vorticity shell density by
\be
\mathcal{Z}_{\bar\zeta}(p;t,\tau)\,dp
:=
\int_{|\mathbf{p}|\in[p,p+dp]}
\bar\zeta(\mathbf{p},t)
\bar\zeta(\mathbf{p},\tau)^*
\,d\mathbf{p}.
\label{eq:two_time_vorticity_shell_definition}
\ee
By Cauchy--Schwarz,
%----------------------------------------------%
\begin{widetext}
\be
|\mathcal{Z}_{\bar\zeta}(p;t,\tau)|\,dp
\le
\left[
\int_{|\mathbf{p}|\in[p,p+dp]}
|\bar\zeta(\mathbf{p},t)|^2\,d\mathbf{p}
\right]^{1/2}
\left[
\int_{|\mathbf{p}|\in[p,p+dp]}
|\bar\zeta(\mathbf{p},\tau)|^2\,d\mathbf{p}
\right]^{1/2}.
\label{eq:two_time_vorticity_shell_CS}
\ee
\end{widetext}
%----------------------------------------------%
To evaluate the two equal-time factors, let
$E_{\bar u}(p)$ denote the one-dimensional kinetic-energy spectrum of
the bare resolved velocity $\bar{\mathbf{u}}$, normalized so that
\be
\frac12
\int_{|\mathbf{p}|\in[p,p+dp]}
|\bar{\mathbf{u}}(\mathbf{p},t)|^2\,d\mathbf{p}
=
E_{\bar u}(p)\,dp
+
o(dp).
\label{eq:bare_velocity_shell_spectrum}
\ee
In Fourier variables,
\be
\bar\zeta(\mathbf{p},t)
=
i\mathbf{p}\times\bar{\mathbf{u}}(\mathbf{p},t).
\ee
Because $\bar{\mathbf{u}}$ is incompressible,
$\mathbf{p}\cdot\bar{\mathbf{u}}(\mathbf{p},t)=0$, and therefore
\be
|\bar\zeta(\mathbf{p},t)|^2
=
p^2|\bar{\mathbf{u}}(\mathbf{p},t)|^2.
\label{eq:vorticity_velocity_modal_identity}
\ee
Consequently, over a thin linear shell,
\bea
\int_{|\mathbf{p}|\in[p,p+dp]}
|\bar\zeta(\mathbf{p},t)|^2\,d\mathbf{p}
&=
\int_{|\mathbf{p}|\in[p,p+dp]}
|\mathbf{p}|^2
|\bar{\mathbf{u}}(\mathbf{p},t)|^2\,d\mathbf{p}
\nonumber\\
&=
2p^2E_{\bar u}(p)\,dp
+
o\big(p^2dp\big).
\label{eq:vorticity_shell_from_velocity_spectrum}
\eea
The same relation holds at time $\tau$ by stationarity. Substitution
into Eq.~\eqref{eq:two_time_vorticity_shell_CS} gives
\be
|\mathcal{Z}_{\bar\zeta}(p;t,\tau)|
\lesssim
p^2E_{\bar u}(p).
\label{eq:two_time_zeta_shell_bare_bound}
\ee
Theorem~\ref{thm:no_anomalous_scaling} states that the bare resolved
velocity $\bar{\mathbf{u}}$ and the effective transport velocity
$\mathbf{u}_{\rm trs}$ share the same inertial-range scaling exponent.
Thus, if $E(p)$ denotes the one-dimensional kinetic-energy spectrum of
$\mathbf{u}_{\rm trs}$, then
\be
E_{\bar u}(p)
\sim
c_{\bar u} \, E(p)
\ee
at the level of inertial-range power counting, where $c_{\bar u}$ is a
finite scale-independent amplitude ratio. Hence
\be
|\mathcal{Z}_{\bar\zeta}(p;t,\tau)|
\lesssim
p^2E(p).
\label{eq:two_time_zeta_shell_bound}
\ee
Equivalently, the integrated contribution of the linear shell is bounded
by
\be
|\mathcal{Z}_{\bar\zeta}(p;t,\tau)|\,dp
\lesssim
p^2E(p)\,dp.
\label{eq:two_time_zeta_shell_integrated_bound}
\ee
At equal times,
$\mathcal{Z}_{\bar\zeta}(p;t,t)$ is, up to the conventional factor $2$,
the one-dimensional enstrophy spectrum of the bare resolved field.

The finite gradient-energy trace condition
\eqref{eq:OU_gradient_trace_condition}, equivalently
\eqref{eq:OU_gradient_trace_fourier} in the commuting Fourier-diagonal
case, is precisely the SPDE regularity required to define the
matrix-valued covariance symbol
$\mathsf{C}_{ij}^{\nabla}$. Using
Eq.~\eqref{eq:OU_two_time_covariance_fourier}, the stationary OU modal
covariance is
\be
C_{\mathbf{q}}(t,\tau)
=
e^{-d_{\mathbf{q}}|t-\tau|}
(Q_\infty)_{\mathbf{q}},
\qquad
d_{\mathbf{q}}
=
\kappa|\mathbf{q}|^2+\mu_r,
\ee
so that
\be
\mathsf{C}_{ij}^{\nabla}(\mathbf{q};t,\tau)
=
q_iq_j
e^{-d_{\mathbf{q}}|t-\tau|}
(Q_\infty)_{\mathbf{q}}
\ee
is controlled by the same gradient trace. Consequently, the contracted
reservoir covariance in
Eq.~\eqref{eq:psi_sgs_covariance_IR} is finite in the point-split sense.
For fixed drift mode $\mathbf{m}$, it introduces no compensating
singular power as the internal resolved mode
$\mathbf{p}\to0$, since
$\mathbf{q}=\mathbf{m}-\mathbf{p}\to\mathbf{m}$.

Substituting Eq.~\eqref{eq:psi_sgs_covariance_IR} into the diagonal
spectral projection of the nested Lie-derivative covariance gives the
one-loop self-energy estimate. Indeed, in
\be
\mathbf{\Sigma}_{\rm sweep}(s,\tau)\Phi
=
-2\,\mathbb{E}
\Big[
J\big(\psi_s^{(0)},J(\psi_\tau^{(0)},\Phi)\big)
\Big],
\ee
the inner Lie derivative $J(\psi_\tau^{(0)},\Phi)$ couples the drift
mode $\mathbf{m}$ to the target mode $\mathbf{k}$, producing one factor
$\mathbf{m}\times\mathbf{k}$. The outer Lie derivative
$J(\psi_s^{(0)},\cdot)$ produces the second factor when the expression
is projected back onto the same target shell and the two drift-potential
modes are contracted. The diagonal projection of the two
Lie-derivative vertices therefore yields
$|\mathbf{m}\times\mathbf{k}|^2$.

After the homogeneous, diagonal shell projection used in
Section~\ref{sec:scaling_preliminaries}, the intermediate response
reduces to the scalar line $G(k,\tau)$ associated with the external
target shell. The resulting scaling form is
%--------------------------------------------------------%
\begin{widetext}
\be
\eta(k,\tau)
\sim
G(k,\tau)
\int
|\mathbf{m}\times\mathbf{k}|^2
\left[
\int
|\mathbf{p}\times\mathbf{m}|^2 \,
\bar\zeta(\mathbf{p},t)\bar\zeta(\mathbf{p},t-\tau)^* \,
C_{\mathbf{m}-\mathbf{p}}(\tau)\,d\mathbf{p}
\right]
d\mathbf{m}.
\label{eq:eta_spectral_full}
\ee
\end{widetext}
%--------------------------------------------------------%
Here stationarity has been used to express the covariance in terms of
the time separation $\tau$, so that
$C_{\mathbf{m}-\mathbf{p}}(\tau)$ abbreviates
$C_{\mathbf{m}-\mathbf{p}}(t,t-\tau)$. The bounded Fourier
mode-coupling weights associated with a smooth nonconstant $\gamma$ are
again understood within the scaling symbol $\sim$. The precise
tensorial coefficients before contraction are those of the point-split
operator \eqref{eq:Sigma_sweep_explicit_index}.

For the infrared scaling, every explicit dependence on the internal
resolved mode $\mathbf{p}$ is carried by
\be
|\mathbf{p}\times\mathbf{m}|^2 \,
\bar\zeta(\mathbf{p},t)\bar\zeta(\mathbf{p},t-\tau)^*,
\ee
whose shell size is controlled by
Eq.~\eqref{eq:two_time_zeta_shell_bound}. The factor
$|\mathbf{m}\times\mathbf{k}|^2$ belongs to the two outer
Lie-transport vertices acting on the drift mode $\mathbf{m}$ and the
external target mode $\mathbf{k}$. The response $G(k,\tau)$ supplies the
temporal line, while the OU factor
$C_{\mathbf{m}-\mathbf{p}}(\tau)$ supplies the finite reservoir
covariance. At this homogeneous diagonal one-loop level, none of these
factors introduces a negative power capable of cancelling the explicit
infrared zero generated as $\mathbf{p}\to0$.

%-------------------------------------------------------------------------------------%
\textbf{Step 2: Infrared zero of the symplectic transport vertex.}
We now evaluate the contribution of an infrared resolved mode
$\mathbf{p}\to0$ in Eq.~\eqref{eq:eta_spectral_full}. This is the regime
associated with the classical Eulerian sweeping defect: a large-scale
motion translates smaller structures through a fixed observation frame,
causing rapid apparent phase decorrelation without producing a
corresponding local deformation of those structures.

Because the time dependence has been retained in
Eq.~\eqref{eq:eta_spectral_full}, the relevant internal resolved factor
is not the pointwise equal-time quantity
$|\bar\zeta(\mathbf{p})|^2$. It is the two-time contribution
\be
|\mathbf{p}\times\mathbf{m}|^2
\bar\zeta(\mathbf{p},t)\bar\zeta(\mathbf{p},t-\tau)^*.
\ee
We therefore estimate the contribution after integrating over a linear
shell $|\mathbf{p}|\in[p,p+dp]$. By the definition of the two-time vorticity shell density in Step~1,
\be
\mathcal{Z}_{\bar\zeta}(p;t,t-\tau)\,dp
:=
\int_{|\mathbf{p}|\in[p,p+dp]}
\bar\zeta(\mathbf{p},t)
\bar\zeta(\mathbf{p},t-\tau)^*
\,d\mathbf{p},
\ee
and by Eq.~\eqref{eq:two_time_zeta_shell_integrated_bound},
\be
|\mathcal{Z}_{\bar\zeta}(p;t,t-\tau)|\,dp
\lesssim
p^2E(p)\,dp.
\ee
For fixed drift mode $\mathbf{m}$, the inner symplectic vertex satisfies
the exact identity
\be
|\mathbf{p}\times\mathbf{m}|^2
=
p^2m^2\sin^2\theta.
\ee
Hence the absolute size of the infrared shell contribution to the
drift-potential covariance is bounded by
\bea
|\mathbf{p}\times\mathbf{m}|^2
|\mathcal{Z}_{\bar\zeta}(p;t,t-\tau)|\,dp
&\lesssim
p^2m^2\sin^2\theta \,
p^2E(p)\,dp\\
&\sim
p^4E(p)\,dp,
\label{eq:p4Ep_suppression}
\eea
up to factors depending on $\mathbf{m}$, the angular variable, and the
finite OU covariance
$C_{\mathbf{m}-\mathbf{p}}(\tau)$. Equivalently, the corresponding
logarithmic-shell contribution is proportional to $p^5E(p)$.

Equation~\eqref{eq:p4Ep_suppression} is the geometric
sweeping-suppression factor. It shows generally that the infrared
contribution vanishes whenever
\bes
p^4E(p)\longrightarrow0
\qquad
\textrm{as }p\to0.
\ees
In particular, the inverse-energy-range candidate
$E(p)\sim p^{-5/3}$ satisfies
\be
p^4E(p)
\sim
p^{7/3}
\longrightarrow
0
\qquad
\textrm{as }p\to0.
\ee
This use of the inverse-cascade spectrum is a self-consistency check:
the $p^4$ selection rule is obtained before the stationary spectrum is
derived in Step~4. It then verifies that the resulting spectrum does not
reintroduce the deep-infrared sweeping divergence.

The mechanism is geometric. An exactly uniform translation has zero
vorticity and produces no cross-gradient twisting. In the present
closure, the stochastic memory kernel is not built from a covariance of
velocities alone; it is generated by a covariance of Hamiltonian
transport operators. The construction of the drift potential already
contains the Jacobian $J(\bar\zeta,r^{(0)})$,
and the self-energy contains the additional transport vertex
$J(\psi_{sgs}^{(0)},\cdot).$

The inner cross-gradient factor produces an exact zero at the uniform
mode and an asymptotic $p^2$ vertex suppression for nearly uniform
modes. Combined with the vorticity shell weighting
$p^2E(p)\,dp$, this yields the fourth-order infrared factor
$p^4E(p)\,dp$. Thus the deep-infrared Eulerian sweeping contribution is
suppressed at the vertex level, rather than removed by an externally
imposed transformation to Lagrangian coordinates.

The remaining factors in Eq.~\eqref{eq:eta_spectral_full} do not alter
this explicit infrared count. The differentiated reservoir covariance
has already been contracted with the inner Jacobian vertices, yielding
the finite OU factor $C_{\mathbf m-\mathbf p}(\tau)$ multiplying
$|\mathbf p\times\mathbf m|^2$. Its point-split regularity is guaranteed
by the finite gradient-energy trace condition, and for fixed
$\mathbf m$ it introduces no singular power as
$\mathbf p\to0$, since
$\mathbf m-\mathbf p\to\mathbf m$. The factor
$|\mathbf m\times\mathbf k|^2$ acts on the drift and target shells and
contains no inverse power of the internal resolved variable
$\mathbf p$. 

Finally, the diagonal line-renormalized response $G(k,\tau)$ depends on the external target
wavenumber $k$ and and supplies the
temporal memory by determining how the interaction is propagated in time.
It does not depend explicitly on the internal infrared wavenumber $p$.
It therefore does not change the factor $p^4E(p)$ obtained from the
vorticity shell and the Jacobian vertex.

%-------------------------------------------------------------------------------------%
\textbf{Step 3: Consequence for the inertial-range memory time.}
Recall the spectral transfer formula derived in
Section~\ref{sec:scaling_preliminaries}, namely 
\be
T(k) \sim \theta_k k^3E(k)^2,
\label{eq:Tk_scaling_inertial_proof}
\ee
where $\theta_k=\int_0^\infty G(k,\tau)\,d\tau$ is the dressed triad-relaxation time.

Theorem~\ref{thm:one_loop_emergence} identifies the origin of this
memory: it is generated by the Eulerian tangent-space self-energy whose
connected spatial covariance is represented by
$\mathbf{\Sigma}_{\rm sweep}$.
In bare Eulerian DIA, the memory integral is contaminated by the kinetic
energy of deep-infrared modes. Large eddies sweep smaller structures
across a fixed observation point and generate rapid Eulerian
decorrelation unrelated to the local cascade dynamics.
Equation~\eqref{eq:p4Ep_suppression} suppresses this non-distorting
contribution in the geometric self-energy. This suppression does not by
itself prove full spectral locality; rather, it removes the classical
sweeping obstruction that invalidates the bare Eulerian relaxation time.

Under the usual Fj\o{}rtoft--Kraichnan locality hypothesis, the dominant
inverse-cascade contribution then comes from strain-producing triads of
comparable scale,
\bes
p\sim m\sim k.
\ees
Let
\be
u_k:=|\mathbf{u}_{\rm trs}|_k
\sim
\sqrt{kE(k)}
\ee
be the characteristic velocity of the active transport field. The local
strain rate is
\be
\tau_k^{-1}
\sim
ku_k
\sim
\sqrt{k^3E(k)}.
\ee
The corresponding geometrically renormalized response-time estimate is
therefore
\be
\theta_k
\sim
\tau_k
\sim
\big(k^3E(k)\big)^{-1/2}.
\label{eq:theta_k_bound}
\ee
This is neither a passive viscous time nor the rigid-sweeping time of
bare Eulerian DIA. It is the strain-controlled memory time that remains
once the non-distorting deep-infrared translation has been suppressed by
the symplectic vertex.

%-------------------------------------------------------------------------------------%
\textbf{Step 4: Stationary inverse-energy cascade.}
The energy flux through wavenumber $k$ is obtained by integrating the
transfer spectrum over a local logarithmic band. Using
Eq.~\eqref{eq:Tk_scaling_inertial_proof}, this gives
\be
\Pi_E(k)
\sim
kT(k)
\sim
\theta_k k^4E(k)^2.
\label{Eq_Pi_E_scaling}
\ee
In the inverse cascade, stationarity requires a constant upscale energy
flux,
\be
\Pi_E(k)
=
-\epsilon.
\ee
Using Eq.~\eqref{eq:theta_k_bound} in
Eq.~\eqref{Eq_Pi_E_scaling} yields
\be
\epsilon
\sim
\big(k^3E(k)\big)^{-1/2}
k^4E(k)^2
=
k^{5/2}E(k)^{3/2}.
\ee
Solving for $E(k)$ gives
\be
E(k)
\sim
\epsilon^{2/3}k^{-5/3}.
\ee
Thus the inverse-energy-cascade spectrum is a self-consistent stationary
solution of the geometrically renormalized Volterra transfer theory.

%-------------------------------------------------------------------------------------%
\textbf{Step 5: Stationary forward-enstrophy cascade.}
In the forward cascade, stationarity requires a constant downscale
enstrophy flux,
\be
\Pi_Z(k)
=
\eta_Z.
\ee
Since the enstrophy density carries an additional factor $k^2$ relative
to the energy density, the enstrophy flux associated with
Eq.~\eqref{eq:Tk_scaling_inertial_proof} scales as
\be
\Pi_Z(k)
\sim
k^2\Pi_E(k)
\sim
\theta_k k^6E(k)^2.
\label{eq:Pi_Z_scaling}
\ee
For the enstrophy cascade, the characteristic nonlinear time is the
vorticity-strain time. In the Kraichnan enstrophy range, this time is
asymptotically independent of $k$, up to the familiar logarithmic
corrections. In the present setting, this independence is consistent
with the same geometric mechanism used above. The deep-infrared
translational component is suppressed by
Eq.~\eqref{eq:p4Ep_suppression}; the remaining phase scrambling is
controlled by finite strain, either from local enstrophy-range triads or
from the energy-containing strain field near the injection scale.
Dimensionally, this gives
\be
\theta_k
\sim
\eta_Z^{-1/3},
\label{eq:theta_enstrophy_constant}
\ee
with possible logarithmic refinements ignored at the level of the
present dimensional theorem.

Substituting Eq.~\eqref{eq:theta_enstrophy_constant} into
Eq.~\eqref{eq:Pi_Z_scaling} gives
\be
\eta_Z
\sim
\eta_Z^{-1/3}k^6E(k)^2.
\ee
Equivalently,
\be
E(k)^2
\sim
\eta_Z^{4/3}k^{-6},
\ee
and therefore
\be
E(k)
\sim
\eta_Z^{2/3}k^{-3}.
\ee
Thus the forward-enstrophy-cascade spectrum is also a self-consistent
stationary solution.

The proof also makes explicit the division of labor between geometry and
response. The Volterra representation supplies
\be
T(k)
\sim
\theta_k k^3E(k)^2.
\ee
The one-loop geometric calculation supplies the connected covariance
core
\be
-\mathbb{E}
\big[
J(\psi_s^{(0)},J(\psi_\tau^{(0)},\cdot))
\big]
=
\frac12\mathbf{\Sigma}_{\rm sweep}(s,\tau),
\ee
whose diagonal spectral projection carries the infrared weighting
$p^4E(p)$. The nested Hamiltonian vertices therefore determine which
large-scale motions are admitted into the memory kernel, while the
dressed response determines how long the surviving strain-producing
interactions remain coherent. Under the standard locality and
stationary-flux assumptions, this combination yields the classical
inverse-energy and forward-enstrophy spectra without appending a
separate Lagrangian-history correction: the relative-transport selection
required to suppress rigid sweeping is already encoded in the
Hamiltonian dressed-transport geometry.
\end{proof}
%-----------------------------------------------------------------------------------------------------%

%========================================%
\subsection{Physical Implications}
\label{subsec:physical_implications_inertial_ranges}

The main physical result of Theorem~\ref{thm:inertial_ranges} is not
simply the recovery of the familiar dual-cascade exponents. It is the
identification of which large-scale motions enter the decorrelation time
of the closure. In a fixed Eulerian frame, a small eddy can decorrelate
rapidly because a large-scale velocity carries it past the observer on
the sweeping time $(kU_{\rm rms})^{-1}$. This motion changes the observed
phase, but it does not deform the eddy or contribute directly to
inter-scale transfer. The classical sweeping problem is precisely that
an Eulerian response may confuse this translational decorrelation with
the strain-induced loss of coherence that controls the cascade
\cite{kraichnan1964kolmogorov}.

The Symplectic Geometric Closure separates these effects through the
form of its transport interaction. Its memory is generated not by a raw
velocity covariance, but by the covariance of the nested Hamiltonian
transport operators contained in
$\mathbf\Sigma_{\rm sweep}(s,\tau)$. For an internal resolved mode
$p\to0$, the corresponding shell contribution scales as
$p^4E(p) dp$. One factor $p^2E(p)dp$ is the vorticity content of the
large-scale velocity shell, while the additional factor $p^2$ comes from
the inner Jacobian vertex $|\mathbf p\times\mathbf m|^2$. An exactly
uniform translation has no vorticity and produces no cross-gradient
twisting; a nearly uniform mode is correspondingly suppressed as
$p\to0$.

This result removes the classical deep-infrared sweeping obstruction
without discarding dynamically active large scales. Large-scale strain,
shear, and cross-gradient deformation remain present because they
produce relative motion and can decorrelate interacting triads. What is
filtered is only the part of the large-scale velocity that translates
structures without deforming them. For the inverse-cascade spectrum,
$p^4E(p)\sim p^{7/3}\to0$, confirming that the resulting
$k^{-5/3}$ solution does not reintroduce the infrared divergence that
afflicts the bare Eulerian response.

The response line and the symplectic vertices therefore play distinct
roles. The vertices determine which motions contribute to the memory,
whereas the response determines how long those contributions persist.
Once rigid sweeping is suppressed, the standard
Fj\o{}rtoft--Kraichnan locality assumption identifies the remaining
inverse-range decorrelation time with the local strain time,
$\theta_k\sim[k^3E(k)]^{-1/2}$. The stationary energy-flux balance then
gives $E(k)\sim\epsilon^{2/3}k^{-5/3}$. In the forward range, the finite
vorticity-strain time $\theta_k\sim\eta_Z^{-1/3}$ similarly gives
$E(k)\sim\eta_Z^{2/3}k^{-3}$, up to the familiar logarithmic corrections
to the enstrophy cascade. These are the classical Kraichnan spectra and
the organizing laws observed in numerical studies of forced
two-dimensional turbulence
\cite{Kraichnan1967,kraichnan1971,mcwilliams1984emergence,smithr1994}.

This conclusion is unchanged by restricting the resolved dynamics to
their inertial form. As stated in
Remark~\ref{rem:deterministic_skeleton_independence}, the deterministic
Navier--Stokes--$\beta$ terms modify the background trajectory and the
baseline propagation, but the infrared weighting itself is fixed by the
OU reservoir covariance and the symplectic transport vertex. Likewise,
the forcing that maintains stationarity fixes the fluxes $\epsilon$ and
$\eta_Z$ but does not alter the $p^4E(p)$ suppression. The geometric
interaction therefore removes the specific infrared contribution that
obstructs an Eulerian locality argument, while the usual locality and
stationary-flux assumptions complete the inertial-range scaling theory.

Kraichnan's Lagrangian-history closures removed sweeping by reorganizing
the response around fluid trajectories, so that common translational
motion was factored out before evaluating the decorrelation of interacting
eddies \cite{kraichnan1965lagrangian,kraichnan1978sweeping}. The
Symplectic Geometric Closure achieves the corresponding separation at an
earlier structural level. The response remains Eulerian, but the
interaction that it propagates is already constrained to be a
Hamiltonian Lie transport. Because the subgrid drift is generated by
$\psi_{sgs}=\frac12\gamma J(\bar\zeta,r)$, the transport vertex depends
on vorticity gradients and cross-gradient deformation rather than on
absolute velocity alone. Uniform convection is therefore excluded from
the infrared memory before the response equation is dressed, while
strain-producing motions remain dynamically active.

In this sense, the Hamiltonian transport structure provides an Eulerian
geometric analogue of the correction sought in Kraichnan's Lagrangian
program. Kraichnan changed the representation of the response in order
to separate translation from deformation; the present closure changes
the interaction geometry before constructing that response. The
distinction is encoded pathwise in the prognostic stochastic field
dynamics: the resolved vorticity is coupled to the hidden reservoir
through the Lie-transport field generated by $\psi_{sgs}$, and only
afterward are the memory kernel and its infrared weighting obtained by
ensemble averaging and spectral projection. Sweeping suppression is
therefore neither appended as an eddy-damping prescription nor imposed
through a separate change of coordinates. It is inherited by the
Eulerian response from the Hamiltonian structure of the transport
vertex itself. This is what allows the closure to retain a finite-memory
Eulerian formulation while recovering the strain-controlled
decorrelation required by Kraichnan's dual-cascade phenomenology. For
broader context on Eulerian--Lagrangian decorrelation and
two-dimensional turbulence 
\cite{tennekes1975eulerian,boffetta2012two}.

%===========================================================%
\section{Energetic Directionality of the Geometric Transfer}
\label{Sec_phase_lag_upscale}

Section~\ref{Sec_inertial_ranges} established the temporal side of the
cascade mechanism. Theorem~\ref{thm:inertial_ranges} showed that the
geometric self-energy generated by the Symplectic Geometric Closure is not
a generic Eulerian eddy damping. Its Lie-derivative covariance produces an
infrared weight $p^4E(p)$, suppressing the uniform-translation component
responsible for Kraichnan's sweeping defect. Consequently, the memory time
entering the Volterra transfer integral is controlled by strain-sensitive
phase scrambling rather than by spurious sweeping decorrelation.

Spectral scaling alone, however, does not determine the direction of the
energy flux. The inverse cascade requires a second ingredient: the
nonlinear interactions must respect the conservation constraints that
permit kinetic energy to migrate upscale while enstrophy migrates
downscale. The present section supplies this complementary ingredient.
We show that the geometric transfer preserves the same energy--enstrophy
bookkeeping that underlies Fj\o{}rtoft's theorem \cite{fjortoft1953}. Thus the Symplectic
Geometric Closure provides both components needed for inverse-cascade
phenomenology: a finite strain-controlled memory time, established in
Section~\ref{Sec_inertial_ranges}, and a conservative triad architecture
that makes upscale kinetic-energy transfer energetically admissible.

The central point is that both properties come from the same structure.
The Lie-derivative covariance of Theorem~\ref{thm:inertial_ranges}
controls how long interacting triads remain coherent, while the
skew-symmetry of the Jacobian controls which directions of spectral
exchange are compatible with the inviscid invariants. The symplectic
geometry therefore acts twice: it regularizes the response in time, and it
preserves the conservation laws that constrain the direction of transfer.

%===========================================================%
\subsection{Energetic Exchange and Temporal Coherence}
\label{Sec_Dressed_Response}
We first identify the quantity whose sign measures the energetic exchange
between the resolved flow and the geometric subgrid transport. With the
transfer convention adopted here, define
\be
\mathcal{P}_{sgs}(t)
:=
\int_{\mathcal{D}}
\bar\psi(t)
\Pi_{sgs}(t)
\,d\mathbf{x},
\quad
\Pi_{sgs}
=
J(\psi_{sgs},\bar\zeta),
\label{eq:Psgs_definition}
\ee
where $\mathcal{P}_{sgs}$ denotes the SGS kinetic-energy transfer rate. 

Under this convention, $\mathcal P_{sgs}>0$ corresponds to a net
transfer of kinetic energy into the resolved flow, whereas
$\mathcal P_{sgs}<0$ corresponds to a net transfer from the resolved
flow into the hidden reservoir.

Using the cyclic integration-by-parts identity for the Jacobian on the
periodic domain,
\beas
\int_{\mathcal{D}}
A J(B,C) \,d\mathbf{x}
=
-\int_{\mathcal{D}}
B J(A,C) \,d\mathbf{x},
\eeas
we obtain the exact relation
\be
\mathbb{E}[\mathcal{P}_{sgs}(t)]
=
-\mathbb{E}
\left[
\int_{\mathcal{D}}
\psi_{sgs}(t)
J\big(\bar\psi(t),\bar\zeta(t)\big)
\,d\mathbf{x}
\right].
\label{eq:Psgs_IBP}
\ee
Introducing the resolved nonlinear tendency
$J_b=J(\bar\psi,\bar\zeta)$,
Eq.~\eqref{eq:Psgs_IBP} shows that the resolved--reservoir energy
exchange is controlled by the relative organization of $\psi_{sgs}$ and
$J_b$. Unlike an eddy-viscosity closure, the geometric interaction does
not impose a sign-definite drain. Depending on their relative phase, the
reservoir-induced transport may either remove kinetic energy from the
resolved flow or return it. Equation~\eqref{eq:Psgs_IBP} therefore
establishes energetic admissibility in both directions, but it does not
by itself determine which direction prevails statistically.

The temporal organization of this exchange has already been derived in
Sections~\ref{Sec_higher_order_memory} and
\ref{Sec_inertial_ranges}. Theorem~\ref{thm:one_loop_emergence} shows
that the third-order reservoir expansion generates the connected
two-time transport covariance
$\mathbf{\Sigma}_{\rm sweep}(s,\tau)$, while
Corollary~\ref{cor:reservoir_tangent_self_energy} identifies its
response-weighted form in the tangent-space self-energy. Theorem~\ref{thm:inertial_ranges}
then shows that the infrared part of this covariance suppresses uniform
sweeping, so that the response time is controlled by motions that deform
the interacting structures rather than merely translating them.

Accordingly, no additional response ansatz for $\psi_{sgs}$ is required
here. The relevant two-time transfer relation is already
Eq.~\eqref{eq:Tk_integral}. Defining the resolved-Jacobian covariance
\be
\mathcal{C}_J(k,\tau)
:=
\mathbb{E}
\Big[
\langle
J_k(t)J_k^*(t-\tau)
\rangle
\Big],
\label{eq:resolved_Jacobian_covariance}
\ee
that relation may be written compactly as
\be
T(k)
=
\mathrm{Re}
\int_0^\infty
k^{-2}G(k,\tau)
\mathcal{C}_J(k,\tau)
\,d\tau.
\label{eq:dressed_two_time_transfer}
\ee
The factor $k^{-2}$ is the elliptic conversion from vorticity response
to streamfunction work. The covariance $\mathcal{C}_J(k,\tau)$ records
the amplitude and relative phase of the nonlinear resolved tendency at
two different times, while the dressed response $G(k,\tau)$ determines
how strongly that past tendency continues to influence the transfer at
the present time. Equation~\eqref{eq:dressed_two_time_transfer} is the
geometric counterpart of the response--correlation transfer integrals
appearing in Kraichnan's DIA theory
\cite{kraichnan1958dia,kraichnan1964dia}.

The role of the dressed response should be stated carefully. It does not
select the direction of transfer independently of the nonlinear
dynamics. Rather, it determines how long a particular triadic exchange
remains coherent and therefore whether that exchange can accumulate
into a nonzero spectral flux. The direction in which energy may be
distributed among the members of a triad is fixed separately by the
simultaneous conservation of kinetic energy and enstrophy.

The inverse-cascade mechanism therefore separates into two complementary
requirements. The strain-controlled memory established in
Section~\ref{Sec_inertial_ranges} supplies the temporal persistence of
the nonlinear exchange. The energy--enstrophy conservation laws
determine which distributions of that exchange among interacting
wavenumbers are dynamically admissible. The following subsection turns
to this second requirement through the Fj\o{}rtoft--Kraichnan triad
constraints.
%===========================================================%
%===========================================================%
\subsection{Triad Mechanics and the Energetic Admissibility of Upscale Transfer}
\label{sec:triad_mechanics}
The preceding subsection established two properties of the geometric
transfer. First, the exact work relation
\eqref{eq:Psgs_IBP} does not impose a sign-definite exchange between the
resolved flow and the hidden reservoir. Second, the dressed transfer
relation \eqref{eq:dressed_two_time_transfer} determines how long the
resolved nonlinear interactions remain temporally coherent. Neither
property, by itself, selects the direction in which kinetic energy is
redistributed across scales. That direction is constrained by the
simultaneous conservation of kinetic energy and enstrophy within the
resolved nonlinear triads.

To expose this complementary constraint, consider the Fourier
decomposition of the resolved Jacobian
$J(\bar\psi,\bar\zeta)$ into interacting wavevector triads
$(\mathbf{k},\mathbf{p},\mathbf{q})$ satisfying
\be
\mathbf{k}+\mathbf{p}+\mathbf{q}=0.
\ee
For a given triad, let
$T_k= T(\mathbf{k},\mathbf{p},\mathbf{q})$ denote the rate at which
kinetic energy is transferred into mode $\mathbf{k}$ by its two partners.
The inviscid Jacobian conserves kinetic energy and enstrophy within each
triad, so that
\begin{subequations}
\begin{align}
T_k+T_p+T_q&=0,
\label{eq:triad_E}\\
k^2T_k+p^2T_p+q^2T_q&=0.
\label{eq:triad_Z}
\end{align}
\end{subequations}
These constraints were first emphasized by Fj\o{}rtoft
\cite{fjortoft1953} and later became central to Kraichnan's theory of the
dual cascade
\cite{Kraichnan1967,kraichnan1971}. Eliminating $T_q$ gives
\be
T_k
=
\frac{p^2-q^2}{q^2-k^2}\,T_p.
\label{eq:Fjortoft}
\ee
Consider an ordered triad with
\beas
k<p<q.
\eeas
If the intermediate mode loses kinetic energy, $T_p<0$, while the
highest-wavenumber mode receives enstrophy, the conservation constraints
require
\beas
\frac{p^2-q^2}{q^2-k^2}<0,
\qquad
T_k>0.
\eeas
The lowest-wavenumber member of the triad therefore receives kinetic
energy; this is the admissibility statement summarized in panel~(c) of
Fig.~\ref{fig:sweeping_cascade}. In a strongly nonlocal configuration, $k\ll p<q$, this recipient
is the largest physical scale participating in the interaction.

This result must be interpreted as an admissibility statement rather
than as a proof of a sustained inverse cascade. The conservation laws
allow kinetic energy to move toward lower wavenumbers while enstrophy is
preferentially transferred toward higher wavenumbers, but they do not
determine how frequently such triads occur, how their phases are
organized, or how long their transfers remain coherent
\cite{fjortoft1953,kraichnan1980,boffetta2000inverse,danilov2001,
vallis2017}. Those are dynamical questions.

The role of the Symplectic Geometric Closure is precisely to supply this
missing temporal ingredient without replacing the conservative triad
mechanics. The resolved transfer amplitude is still generated by the
Jacobian $J(\bar\psi,\bar\zeta)$ and therefore retains the
energy--enstrophy bookkeeping expressed by
Eqns.~\eqref{eq:triad_E}--\eqref{eq:triad_Z}. The hidden reservoir acts
instead through the dressed response: it changes the duration and phase
organization of the nonlinear exchange while leaving the underlying
Fj\o{}rtoft constraints intact. Thus the closure does not manufacture
upscale transfer through a prescribed negative viscosity. It provides a
finite-memory environment in which energetically admissible upscale
triad transfers can persist and accumulate into a macroscopic flux.

The inverse-cascade argument therefore contains two logically distinct
ingredients:
{\small
\bea
\text{inverse-cascade consistency}
&=
\underbrace{\text{energy--enstrophy admissibility}}
_{\text{Fj\o{}rtoft triad constraints}}
\nonumber\\
&\quad+
\underbrace{\text{temporal coherence}}
_{\text{dressed reservoir response}}.
\label{eq:inverse_cascade_factorization}
\eea
}
The first determines how energy and enstrophy may be distributed among
the members of a triad. The second determines whether those permitted
exchanges remain correlated long enough to produce a statistically
persistent flux.

Theorem~\ref{thm:inertial_ranges} supplies the decisive result for the
second ingredient. A conventional Eulerian response can be decorrelated
primarily by large-scale sweeping, on the time
$(kU_{\rm rms})^{-1}$, even though rigid translation contributes little
to inter-scale deformation. The nested Jacobian structure of the
reservoir-induced covariance instead weights an infrared shell by
$p^4E(p)\,dp$. Uniform translation is thereby suppressed, while
large-scale strain and shear remain dynamically active. The response
time entering Eq.~\eqref{eq:dressed_two_time_transfer} is consequently
controlled by deformation of the interacting structures rather than by
their passage through a fixed observation frame.

Sections~\ref{Sec_inertial_ranges} and
\ref{Sec_phase_lag_upscale} therefore establish complementary parts of
Kraichnan's inverse-cascade phenomenology. Section~\ref{Sec_inertial_ranges}
shows that the closure generates a finite, strain-controlled memory time
compatible with the inertial-range spectra. The present section shows
that the resolved Jacobian retains the conservation laws that make
upscale kinetic-energy transfer admissible. Together, these properties
provide both the temporal persistence and the energetic directionality
required for an inverse cascade, without imposing an \emph{ad hoc}
backscatter coefficient or negative eddy viscosity.

The common origin of these two properties is the Hamiltonian transport
architecture of the closure. Its nested Jacobian covariance suppresses
decorrelation by rigid sweeping, while its skew-symmetric Jacobian
structure preserves the conservative organization of the resolved
triads. The inverse-cascade phenomenology is therefore traced to the
geometry of the stochastic transport coupling rather than to a
prescribed sign of the SGS forcing.
%===========================================================%

%===========================================================%
\section{Random Galilean Invariance and the Symplectic Vertex}
\label{Sec_Galilean_Invariance}

The previous two sections established the two dynamical ingredients of
the geometric inverse-cascade mechanism. Section~\ref{Sec_inertial_ranges}
showed that the one-loop Lie-transport self-energy suppresses the
uniform-sweeping contribution through the infrared weight $p^4E(p)$.
Section~\ref{Sec_phase_lag_upscale} then showed that the same Jacobian
architecture preserves the energy--enstrophy conservation structure that
makes upscale kinetic-energy transfer admissible. We now isolate the
geometric reason why these two facts are linked.

The central point is that the Symplectic Geometric Closure does not
renormalize the usual Eulerian Navier--Stokes vertex. It renormalizes a
different object: a cross-gradient Hamiltonian transport vertex generated
by
\be
\psi_{sgs}
=
\frac12\gamma J(\bar\zeta,r),
\qquad
J(\psi_{sgs},\cdot)
=
\mathcal L_{X_{\psi_{sgs}}}(\cdot).
\ee
This vertex is built from Jacobians before any statistical closure is
performed. It therefore sees gradients, relative twisting, and
area-preserving deformation; it does not see absolute translational
velocity amplitudes. The $p^4E(p)$ infrared factor derived in
Theorem~\ref{thm:inertial_ranges} is the spectral manifestation of this
fact.

%----------------------------------------------------------------------------------------------------------------%
This observation places the present theory in direct contact with the
classical Random Galilean Invariance problem
\cite{kraichnan1964kolmogorov,kraichnan1965lagrangian,
kraichnan1971almost,McComb2005}. Kraichnan's original Eulerian DIA
famously suffered from the sweeping defect: in a fixed Eulerian frame,
inertial-range eddies are decorrelated by large-scale random sweeping on
the time scale $(kV_{\rm rms})^{-1}$, even though a spatially uniform
velocity merely translates turbulent structures and does not deform
them
\cite{kraichnan1964kolmogorov,tennekes1975eulerian,
chenKraichnan1989,sanadaShanmugasundaram1992,
kaneda1993lagrangian,eyink2011robert}. Herring and Kraichnan made
the consequence of this defect particularly explicit: Eulerian DIA
promotes the random-sweeping decorrelation time into an artificial
memory time for spectral transfer, whereas the test-field model and the
abridged Lagrangian-history DIA avoid this contamination
\cite{herring1972comparison}. More recent nonperturbative
renormalization-group analyses have confirmed that sweeping controls the
leading large-wavenumber Eulerian time dependence not only of two-point
correlations, but also of generic multi-point correlation and response
functions \cite{Tarpin2018}. Thus sweeping contamination is not merely a
peculiarity of the original DIA approximation; it is a genuine leading
contribution to fixed-frame Eulerian temporal decorrelation.

Kraichnan's Lagrangian-history formulation was designed to restore the correct insensitivity of cascade dynamics to this rigid sweeping by reorganizing the response around fluid histories rather than fixed Eulerian observations \cite{kraichnan1965lagrangian,kraichnan1971almost}. In the present theory, the same physical separation is achieved by a different mechanism. The Eulerian response is retained, but the interaction vertex being renormalized is modified at the outset: because the symplectic vertex depends on cross-gradients and relative deformation, it is already blind to spatially uniform translation. The theory therefore does not remove a sweeping contribution after it has entered the response; it prevents rigid sweeping from entering the self-energy through the interaction vertex in the first place.
%----------------------------------------------------------------------------------------------------------------%

%=========================================================%
\subsection{RGI Compatibility of the Symplectic Vertex}

In the coupled field closure, the effective interaction driving the
resolved vorticity is not the bare convective vertex alone. It is the
geometric transport
\be
(\bar\zeta,r)
\longmapsto
J(\psi_{sgs},\bar\zeta)
=
\frac12J\big(\gamma J(\bar\zeta,r),\bar\zeta\big).
\ee
The following proposition records the corresponding Random Galilean
compatibility at the level of the bare symplectic vertex.

\bprop[Random Galilean Compatibility of the Symplectic Vertex]
\label{prop:RGI_compatibility}
Let the resolved vorticity field and the hidden reservoir field be
observed in a randomly sweeping frame
\be
\mathbf x'=\mathbf x-\mathbf Vt,
\qquad
t'=t,
\ee
where $\mathbf V$ is a spatially uniform, time-independent random
velocity. Then the symplectic interaction vertex
\be
J\big(\gamma J(\bar\zeta,r),\cdot\big)
\ee
is equivariant under this transformation. In particular, the vertex
contains no explicit coupling to the uniform sweeping velocity
$\mathbf V$.
\eprop

\begin{proof}
The resolved vorticity and the reservoir are transported as scalar
fields under the coordinate shift:
\be
\bar\zeta'(\mathbf x',t')
=
\bar\zeta(\mathbf x,t),
\qquad
r'(\mathbf x',t')
=
r(\mathbf x,t).
\ee
Because $\mathbf V$ is spatially uniform, the spatial gradients are
unchanged:
\be
\nabla_{\mathbf x'}
=
\nabla_{\mathbf x}.
\ee
Therefore the inner cross-gradient interaction satisfies
\be
J_{\mathbf x'}(\bar\zeta',r')
=
\nabla_{\mathbf x'}^\perp\bar\zeta'\cdot\nabla_{\mathbf x'}r'
=
\nabla_{\mathbf x}^\perp\bar\zeta\cdot\nabla_{\mathbf x}r
=
J_{\mathbf x}(\bar\zeta,r).
\ee
Consequently,
\be
\psi_{sgs}'(\mathbf x',t')
=
\frac12\gamma J_{\mathbf x'}(\bar\zeta',r')
=
\frac12\gamma J_{\mathbf x}(\bar\zeta,r)
=
\psi_{sgs}(\mathbf x,t),
\ee
up to the passive coordinate relabelling. Applying the outer Jacobian
gives
\be
J_{\mathbf x'}(\psi_{sgs}',\bar\zeta')
=
J_{\mathbf x}(\psi_{sgs},\bar\zeta).
\ee
Thus the symplectic vertex is unchanged by a uniform random sweeping
velocity. The vertex depends only on spatial cross-gradients of the
fields, and a uniform translation supplies no such gradient.
\end{proof}

Proposition~\ref{prop:RGI_compatibility} should be read as a vertex-level
statement. It does not by itself prove all-order Random Galilean
Invariance of an arbitrary statistical truncation. What it does show is
that the fundamental interaction being renormalized is already
compatible with the symmetry that bare Eulerian DIA fails to respect. The
one-loop result of Theorem~\ref{thm:inertial_ranges} is the corresponding
spectral consequence: the uniform sweeping contribution is suppressed at
the level of the self-energy integrand.

%=========================================================%
\subsection{Diagrammatic Interpretation of the Symplectic Vertex}
\label{sec:diagrammatic_interpretation}

It is useful to translate the preceding geometric conclusions into the
language of diagrammatic turbulence theory, following Wyld's
field-theoretic formulation of the Navier--Stokes equations
\cite{wyld1961formulation}. The purpose of this comparison is
pedagogical rather than formal. We do not claim that the present SPDE
field closure is literally identical to Wyld's perturbation expansion.
Rather, the diagrammatic language provides a compact dictionary:
propagating lines represent responses, wavy lines represent two-time
covariances, and vertices represent the nonlinear interaction being
renormalized. In this dictionary, the decisive question is not only
whether one renormalizes the propagator, but which vertex is being
renormalized.

%---------------------------------------------------------------------%
\paragraph*{The bare Navier--Stokes vertex and DIA.}
In the Fourier representation of the Navier--Stokes equations, the
quadratic advection term couples interacting wavevectors
$\mathbf k=\mathbf p+\mathbf q$. Diagrammatically, the velocity field is
represented by propagating lines, and the interaction tensor
$M_{ijm}(\mathbf k,\mathbf p,\mathbf q)$ is the bare Navier--Stokes
vertex, denoted $\Gamma_{NS}$. It is represented by a black node where the
three interacting velocity lines meet:
\begin{center}
\begin{tikzpicture}[scale=0.9, thick]
    \draw (-1.5,0) node[left]{$u(k)$} -- (0,0);
    \draw (0,0) -- (1.2, 1) node[right]{$u(p)$};
    \draw (0,0) -- (1.2,-1) node[right]{$u(q)$};
    \filldraw[black] (0,0) circle (3pt) node[above=4pt]{$\Gamma_{NS}$};
\end{tikzpicture}
\end{center}
The black dot is therefore not decorative: it denotes the bare convective
interaction through which a velocity mode at one scale scatters two other
velocity modes.

A closure such as Kraichnan's DIA reorganizes the infinite perturbation
series by dressing the propagating lines while retaining this bare
interaction vertex \cite{kraichnan1959structure,wyld1961formulation,Leslie1973,McComb1990,Krommes2002}. In diagrammatic language, this gives two
conceptually distinct operations.

First, one may dress a line. These are \emph{line corrections}, or
self-energy insertions, which renormalize propagators or response functions \cite{PeskinSchroeder1995,ZinnJustin2002,Tauber2014}. Physically, they describe how a perturbation
propagating through the turbulent bath loses memory of its initial state.
In the simplest one-loop diagram, the straight internal arc is a response
propagator,
\be
G(k;t,\tau),
\ee
whereas the wavy internal arc is a two-time covariance. In the classical
Eulerian DIA setting this covariance is the Eulerian velocity covariance,
for example
\be
U_{ij}(k;t,\tau)
=
\mathbb E\big[u_i(\mathbf k,t)u_j(-\mathbf k,\tau)\big].
\ee
Thus the one-loop line correction has the symbolic structure
\be
\boldsymbol\Sigma
\sim
\int G\,U,
\ee
namely one response line and one covariance line attached to two
interaction vertices. This is the standard one-loop self-energy topology
in Wyld's formalism \cite{wyld1961formulation}. In the present notation,
the analogous object is the geometric self-energy
$\boldsymbol\eta(t,\tau)$ derived in
Theorem~\ref{thm:one_loop_emergence}.

Second, one may dress the interaction point itself. These are
\emph{vertex corrections}, which renormalize the effective interaction rather than only the propagating lines \cite{PeskinSchroeder1995,ZinnJustin2002,Tauber2014}. They modify the effective coupling between
interacting modes. Wyld's leading vertex correction has the familiar
triangular topology, with internal propagator and covariance lines
joining three interaction points \cite{wyld1961formulation}. The two
basic diagrammatic operations are therefore:
\begin{center}
\begin{tikzpicture}[scale=0.7, thick]
    \node at (-3, 2.1) {\textbf{(a) Line correction ($\boldsymbol\eta$)}};
    \draw (-4.5,0) node[left]{$G$} -- (-3.5,0);
    \draw (-1.5,0) -- (-0.5,0) node[right]{$G$};
    \draw[decorate, decoration={snake, amplitude=1.2pt, segment length=4pt}]
        (-3.5,0) arc (180:0:1)
        node[midway, above=5pt]{$U$ or $C$};
    \draw (-3.5,0) arc (180:360:1)
        node[midway, below=5pt]{$G$};
    \filldraw (-3.5,0) circle (3pt);
    \filldraw (-1.5,0) circle (3pt);

    \node at (3, 2.1) {\textbf{(b) Vertex correction ($\delta\Gamma$)}};
    \draw (1,0) -- (2,0);
    \draw (3,1.3) -- (4, 1.73);
    \draw (3,-1.3) -- (4, -1.73);
    \draw (2,0) -- (3,1.3) node[midway, above left]{$G$};
    \draw (2,0) -- (3,-1.3) node[midway, below left]{$G$};
    \draw[decorate, decoration={snake, amplitude=1.2pt, segment length=4pt}]
        (3,1.3) -- (3,-1.3)
        node[midway, right=5pt]{$U$ or $C$};
    \filldraw (2,0) circle (3pt);
    \filldraw (3,1.3) circle (3pt);
    \filldraw (3,-1.3) circle (3pt);
\end{tikzpicture}
\end{center}
In diagram (a), the external straight line is the response being dressed.
The lower internal straight arc is the propagator $G$. The upper wavy arc
is the covariance of the fluctuating bath that scatters the response. In
classical DIA this is the velocity covariance $U_{ij}$. In the present
geometric theory the analogous covariance is generated by the OU
reservoir, for instance
\be
C_{\mathbf q}(t,\tau)
=\mathbb{E}
\big[
r^{(0)}(\mathbf q,t)r^{(0)}(\mathbf q,\tau)^*
\big],
\ee
or, after the differentiated point-split contractions entering the
Lie-transport self-energy, the tensorial covariance weight
$\mathcal C_S(\mathbf q;t,\tau)$. Thus the wavy arc in the SGC diagram is
not a phenomenological eddy covariance; it is the covariance of the
explicit hidden reservoir field, pushed through the symplectic
cross-gradient vertex.

With this dictionary, the geometric self-energy $\boldsymbol\eta$ (Eq.~\eqref{eq:eta_reservoir_induced_operator}) has the same topological
role as a one-loop line correction, but with a different operator at the
endpoints:
\be
\boldsymbol\eta
\sim \mathbb{E}
\big[
\mathcal L_{X_{\psi_{sgs}}}
\,G\,
\mathcal L_{X_{\psi_{sgs}}}
\big],
\ee
or equivalently
\be
\boldsymbol\eta
\sim \mathbb{E}
\big[
J(\psi_{sgs},GJ(\psi_{sgs},\cdot))
\big].
\ee
The formal role is therefore the same as a self-energy insertion: a
response line is dressed by a two-time covariance. The difference lies in
the vertex carried by the endpoints of the loop.

This distinction is exactly where the Random Galilean Invariance problem
enters. In bare Eulerian DIA, one performs line renormalization while
keeping the bare Eulerian vertex,
\be
\Gamma=\Gamma_{NS}.
\ee
That vertex couples inertial-range modes directly to large-scale velocity
amplitudes. As a result, a spatially uniform sweeping velocity appears as
a decorrelating field in the Eulerian response, even though it only
translates small structures through the observation frame. In exact
RGI-respecting formulations, the missing cancellations may be viewed
diagrammatically as information carried by the class of vertex
renormalizations omitted by the bare Eulerian DIA truncation, or
operationally as the information restored by Lagrangian-history
formulations. The sweeping pathology is therefore not merely a bad
estimate of a time scale; it reflects a mismatch between line
renormalization and the unconstrained Eulerian vertex being renormalized.

%-------------------------------------------------------------------------------------%
\paragraph*{The symplectic vertex and intrinsic sweeping suppression.}
The Symplectic Geometric Closure changes this starting point. The
line-renormalized object is not built around the bare Navier--Stokes
velocity vertex. The elementary cross-scale interaction is the nested
Jacobian
\be
g_1
=
J(\psi_{sgs},\bar\zeta)
=
\frac12J\big(\gamma J(\bar\zeta,r),\bar\zeta\big),
\ee
with $\psi_{sgs} = \frac12\gamma J(\bar\zeta,r).$ 
We denote this differential interaction by the \emph{symplectic vertex}
$\Gamma_{symp}$:
\be\label{Gamma_symp_def}
\Gamma_{symp}
\sim
J\circ\gamma J.
\ee
Unlike the standard Navier--Stokes interaction tensor, which couples
velocity modes through the convective nonlinearity, $\Gamma_{symp}$
couples the resolved vorticity field $\bar\zeta$ to the hidden reservoir
field $r$ through sequential orthogonal gradients. It is therefore a
field-level cross-gradient vertex, not a moment-level eddy-damping
coefficient.

Diagrammatically, we represent this geometric coupling by a white square.
This explicitly distinguishes it from the black dot used for the bare
Navier--Stokes vertex in Feynman--Kraichnan diagrams. The line convention
is also different. Solid lines denote the observable resolved vorticity
field $\bar\zeta$ or its dressed response, while the dashed line denotes
the hidden kinematic reservoir $r$. The dashed line is meant to remind
the reader that $r$ is not a resolved fluid variable observed in the
macroscopic equation; it is an unobserved subgrid field whose covariance
is integrated out in the self-energy calculation:
\begin{center}
\begin{tikzpicture}[scale=0.9, thick]
    \draw (-2,0) node[left]{$\bar{\zeta}$} -- (0,0);
    \draw[dashed] (0,0) -- (1.2, 1) node[right]{$r$};
    \draw (0,0) -- (1.2,-1) node[right]{$\bar{\zeta}$};
    \filldraw[fill=white, draw=black, thick] (-0.2,-0.2) rectangle (0.2,0.2);
    \node at (0,0.5) {\small $\Gamma_{symp}$};
\end{tikzpicture}
\end{center}
Thus the white square should be read as a compact picture of the operator
\be
(\bar\zeta,r,\bar\zeta)
\longmapsto
J\big(\gamma J(\bar\zeta,r),\bar\zeta\big),
\ee
not as an ordinary algebraic three-wave vertex. The dashed reservoir leg
is Wick-contracted against another dashed reservoir leg when the
self-energy is formed. That contraction produces the OU covariance
$C_{\mathbf q}(t,\tau)$, and after spatial derivatives are taken it
produces the tensorial point-split weight $\mathcal C_S$.

This change of vertex is the crucial structural difference. The effective
subgrid velocity is $\mathbf u_{sgs}
=
\nabla^\perp\psi_{sgs}$,
and the drift potential is generated by $\psi_{sgs}=\frac12\gamma J(\bar\zeta,r).$

Thus the advecting subgrid field depends on cross-gradients of
$\bar\zeta$ and $r$, not on their absolute amplitudes. A uniform
translation has no spatial gradient and therefore cannot activate the
inner Jacobian. The vertex sees twisting and phase misalignment; it does
not see rigid sweeping.

The Fourier expression makes the same point in the most elementary way.
If the drift mode is $\mathbf m$ and the resolved infrared mode entering
the inner Jacobian is $\mathbf p$, then
\be
J(\bar\zeta,r)
\quad\longrightarrow\quad
(\mathbf p\times\mathbf m)\,
\bar\zeta(\mathbf p)\,
r(\mathbf m-\mathbf p).
\ee
Hence the inner vertex vanishes as $\mathbf p\to0$:
\be
\mathbf p\times\mathbf m
\longrightarrow0.
\ee
After the two-time reservoir contraction and the diagonal projection of
the two Lie-derivative vertices are performed, this vertex-level zero is
precisely what becomes the infrared weight
$p^4E(p)\,dp $
in Theorem~\ref{thm:inertial_ranges}. The two powers of $p$ from the
Jacobian vertex combine with the two powers carried by the vorticity
shell. Thus the $\mathcal O(p^4)$ self-energy suppression is not an
accident of the Dyson integral. It is the spectral footprint of the
symplectic vertex.

The diagrammatic interpretation of the theorem can now be stated
succinctly. The geometric self-energy is a line correction, but it is a
line correction generated by $\Gamma_{symp}$ rather than by
$\Gamma_{NS}$. In ordinary Eulerian DIA, the line is dressed while the
bare velocity vertex continues to transmit sweeping contamination. In the
Symplectic Geometric Closure, the line is dressed by a vertex that is
already blind to uniform translation. Consequently, the line
renormalization inherits a sweeping-suppressed geometry.

This is the sense in which the present closure accomplishes at the level
of its bare field vertex what Lagrangian-history closures were designed
to enforce at the level of the response. The Eulerian response is still
renormalized by a one-loop memory kernel, but the memory kernel is built
from Hamiltonian Lie transport rather than from raw velocity advection.
Uniform sweeping therefore does not enter as a physical decorrelation
mechanism; only cross-gradient deformation and strain-sensitive phase
scrambling remain active.

%=========================================================%
\subsection{Geometric Selection Rules}

The symplectic vertex is not only compatible with random Galilean
transformations. It also imposes local selection rules that clarify which
cross-scale configurations can activate the hidden reservoir.

\bprop[Geometric Selection Rules]
\label{prop:geometric_selection}
The cross-gradient interaction
\be
J(\bar\zeta,r)
=
\nabla^\perp\bar\zeta\cdot\nabla r
\ee
annihilates:
\begin{enumerate}
    \item uniform scalar offsets of either field;
    \item parallel-gradient configurations.
\end{enumerate}
Consequently, the reservoir is activated only by local cross-gradient
misalignment between the resolved vorticity and the hidden field.
\eprop

\begin{proof}
For constants $C_1$ and $C_2$,
\be
J(\bar\zeta+C_1,r+C_2)
=
J(\bar\zeta,r),
\ee
because $\nabla C_1=\nabla C_2=0$. Thus the interaction does not depend
on absolute scalar amplitudes.

If the gradients $\nabla\bar\zeta$ and $\nabla r$, are parallel,
then their two-dimensional cross product vanishes, and hence $J(\bar\zeta,r)=0.$
The vertex therefore activates only when the two fields possess local
cross-gradient phase misalignment. Geometrically, the coupling measures
twisting, not amplitude.
\end{proof}

These selection rules are the local version of the infrared result.
Uniform translation is a zero-gradient deformation and therefore lies in
the null direction of the vertex. Parallel gradients are also inactive:
they can change amplitudes, but they do not generate the transverse
twisting needed to form a Hamiltonian subgrid drift. The hidden reservoir
therefore couples most strongly to configurations in which the resolved
vorticity and subgrid field are locally out of phase in gradient space.

%==========================================================%
\subsection{Vertex Selection and the Causal Chain of Sweeping Suppression}
\label{sec:causal_chain}
\label{subsec:vertex_vs_closure}
The preceding results shift the emphasis of the Eulerian
renormalization problem. The question is not only how the response is
dressed, but also which interaction generates that dressing. The
resolved Navier--Stokes advection remains present in the baseline
operator $\mathcal{L}_0$. The distinctive feature of the Symplectic
Geometric Closure is that the reservoir-induced self-energy is generated
by the Hamiltonian transport operator
\beas
\mathcal{A}_{\rm res}\Phi
&=
J(\psi_{sgs},\Phi)
=
-\mathcal{L}_{X_{\psi_{sgs}}}\Phi,
\qquad
\psi_{sgs}
=
\frac12\gamma J(\bar\zeta,r),
\eeas
and hence by the nested Jacobian vertex
$\Gamma_{\rm symp}$ defined in
Eq.~\eqref{Gamma_symp_def}.

This distinction matters because a response resummation inherits the
spectral content of the interaction inserted into it. In the original
Eulerian DIA, the exact Navier--Stokes equations are Galilean invariant,
but the truncated Eulerian response does not retain all cancellations
needed to prevent a spatially uniform random velocity from contributing
to decorrelation
\cite{kraichnan1959structure,kraichnan1964kolmogorov,
kraichnan1965lagrangian,kraichnan1971almost,McComb1990}.
Kraichnan's Lagrangian-history constructions addressed this defect by
reorganizing the response around fluid trajectories. The present closure
acts at a different stage: it constrains the reservoir-induced
interaction before its contribution is assembled into the Eulerian
self-energy.

The resulting causal chain can be summarized directly in spectral
variables:
%---------------------------------------------%
\begin{widetext}
\bea
\Gamma_{\rm symp}=J\circ\gamma J
\quad
&\Longrightarrow\quad
|\mathbf{p}\times\mathbf{m}|^2
\nonumber\\
&\Longrightarrow\quad
\underbrace{p^2}_{\text{inner Jacobian vertex}}
\times
\underbrace{p^2E(p)\,dp}_{\text{vorticity content of the infrared shell}}
\nonumber\\
&\Longrightarrow\quad
p^4E(p)\,dp
\quad
\Longrightarrow\quad
\text{suppression of uniform sweeping in the self-energy}.
\label{eq:causal_chain_sweeping}
\eea
\end{widetext}
%---------------------------------------------%
The first factor of $p^2$ comes from the cross-gradient vertex
$|\mathbf{p}\times\mathbf{m}|^2$. It vanishes exactly at
$\mathbf{p}=0$, because a uniform translation produces no spatial
twisting. The second factor comes from the fact that the resolved
large-scale mode enters through its vorticity shell, whose weight is
$p^2E(p)\,dp$, rather than through an absolute velocity amplitude.
Their product gives the infrared contribution derived in
Eq.~\eqref{eq:p4Ep_suppression}.

The remaining parts of the self-energy do not reverse this count. The
point-split OU covariance supplies a finite weight
$C_{\mathbf{m}-\mathbf{p}}(\tau)$, while the outer Jacobian acts on the
drift and target modes and introduces no inverse power of the internal
wavenumber $p$. The response $G(k,\tau)$ determines how the disturbance
is propagated in time, but it is indexed by the external target shell
and does not alter the explicit $p\to0$ weighting. Thus the exact zero
of the uniform mode and the asymptotic suppression of nearly uniform
modes originate in the interaction vertex and remain present when that
interaction is incorporated into the dressed response.

The factor $p^4E(p)\,dp$ should therefore not be viewed as an accidental
power count produced only after resummation. It is the spectral
expression of a prior geometric selection rule. The reservoir-induced
transport responds to vorticity gradients and relative deformation,
whereas a rigid translation contains neither. Proposition~\ref{prop:RGI_compatibility}
establishes the Random-Galilean compatibility of the bare symplectic
vertex, and Theorem~\ref{thm:inertial_ranges} shows the corresponding
infrared consequence for the reservoir-generated self-energy.

This also clarifies the comparison with Kraichnan's program. The
Symplectic Geometric Closure does not introduce a new phenomenological
eddy-damping rule, nor does it append a separate Lagrangian-history
correction. The response remains Eulerian and Volterra-like. What changes
is the closure-induced interaction entering that response: the
self-energy is built from Hamiltonian cross-gradient transport rather
than from a covariance that is directly sensitive to absolute
large-scale velocity. The innovation lies therefore not in a new rule
for resumming the response, but in the geometry of the vertex being
resummed.

In this sense, the Hamiltonian transport structure performs upstream,
at the Eulerian field level, the separation that Kraichnan sought through
Lagrangian histories. Uniform convection is removed from the relevant
infrared memory because it produces no vorticity-bearing deformation,
while strain, shear, and cross-gradient twisting remain available to
decorrelate interacting modes. The resulting response is not claimed to
be fully Lagrangian, nor is full Random Galilean invariance inferred from
the vertex alone. The precise result is that the reservoir-induced
self-energy inherits a vertex-level distinction between translation and
deformation, and this distinction is sufficient to suppress the
deep-infrared sweeping contribution identified in
Theorem~\ref{thm:inertial_ranges}.

The conceptual conclusion of this section is therefore concise: the
vertex determines which motions enter the memory, and the response
determines how long their influence persists. By choosing the admissible
resolved--reservoir interaction to be a nested Hamiltonian Jacobian, the
Symplectic Geometric Closure removes rigid translation from the
infrared memory before statistical renormalization can promote it into
the dominant decorrelation mechanism.
%==========================================================%

%-------------------------------------------------------------------------------------%
\subsection{The Unique Symplectic Vertex and the Response Ansatz}
\label{sec:topological_uniqueness}
Appendix~\ref{app:anomalous_scaling} uses the causal response
representation
\be
\psi_{sgs}(t)
=
\mathcal{K}
\int_0^t
G(t,s)
J\big(\bar\psi(s),\bar\zeta(s)\big)\,ds,
\label{eq:app_ansatz_recalled}
\ee
where $\mathcal{K}$ maps the accumulated vorticity tendency into a
streamfunction-like drift potential. Equation~\eqref{eq:app_ansatz_recalled}
is not asserted to be an exact identity for the full coupled SPDE. It is
the response-level closure used for inertial-range exponent counting.
The question is why the temporal kernel appearing in this representation
may consistently be identified with the dressed macroscopic response
$G(t,s)$.

The structural reason is that the Symplectic Geometric Closure possesses
a single resolved--reservoir interaction vertex. Every exchange between
$\bar\zeta$ and $r$ is generated by
\beas
\psi_{sgs}
&=
\frac12\gamma J(\bar\zeta,r),
\\
\mathcal{A}_{\rm res}\Phi
&=
J(\psi_{sgs},\Phi)
=
-\mathcal{L}_{X_{\psi_{sgs}}}\Phi,
\eeas
or, equivalently, by the nested Jacobian vertex
$\Gamma_{\rm symp}
\sim
J\circ\gamma J,$
defined in Eq.~\eqref{Gamma_symp_def}. There is no additional
eddy-viscosity vertex, independent backscatter vertex, or separately
prescribed SGS-memory operator. Causal feedback from the reservoir is
therefore assembled from repeated applications of this same interaction:
\beas
\Gamma_{\rm symp}
\longrightarrow
\Gamma_{\rm symp}
\longrightarrow
\Gamma_{\rm symp}
\longrightarrow
\cdots .
\eeas
This differs from a closure in which an eddy viscosity, stochastic
backscatter term, scale-similarity stress, or other SGS operator is
appended to the resolved equation
\cite{smagorinsky1963general,leith1990stochastic,MasonThomson1992,
Schumann1995,ClarkFerzigerReynolds1979,BardinaFerzigerReynolds1980}.
Such additions introduce an independent dynamical channel and therefore
an additional propagator correction, interaction vertex, or noise
covariance in a Wyld--Martin--Siggia--Rose representation
\cite{wyld1961formulation,MartinSiggiaRose1973,
ForsterNelsonStephen1977,DeDominicisMartin1979,
BereraSalewskiMcComb2013}. No such second channel is present here.

To state the response closure precisely, first introduce a general
temporal kernel $K_{\rm SGS}(t,s)$:
\be
\psi_{sgs}(t)
=
\mathcal{K}
\int_0^t
K_{\rm SGS}(t,s)
J\big(\bar\psi(s),\bar\zeta(s)\big)\,ds.
\label{eq:KSGS_general_ansatz}
\ee
In the full stochastic system, $K_{\rm SGS}$ need not be literally
identical to the macroscopic tangent response. Different observables and
projections can generate different operator-valued kernels. The relevant
structural statement is narrower: both $K_{\rm SGS}$ and $G$ are built
from causal histories generated by the same vertex
$\Gamma_{\rm symp}$.

Indeed, Appendices~\ref{app:Dyson_derivation} and
\ref{app:variational_propagator} define
$G(t,s)=
\mathbb{E}[\mathcal{U}(t,s)],$
where $\mathcal{U}(t,s)$ is the pathwise fundamental operator of the
direct-advection tangent dynamics
\eqref{eq:direct_advection_tangent_skeleton}; see
Eq.~\eqref{eq:G_U_expectation}. Its Dyson--Volterra equation
\eqref{eq:Dyson_symplectic_propagator} is dressed by the self-energy
\eqref{eq:eta_reservoir_induced_operator}, whose two endpoints are
precisely the reservoir-induced operators
$\mathcal{A}_{\rm res}$. Thus the macroscopic response is obtained by
resumming the same symplectic interaction histories that generate the
reservoir memory.

Symbolically, this common origin may be expressed as
\bea
G_{\rm resolved}
&\sim
\mathcal{R}[\Gamma_{\rm symp}],
\nonumber\\
K_{\rm SGS}
&\sim
\mathcal{R}[\Gamma_{\rm symp}],
\label{eq:topological_same_R}
\eea
where $\mathcal{R}[\Gamma_{\rm symp}]$ denotes causal resummation over
histories formed from repeated insertions of the symplectic vertex. Here, the
symbol $\sim$ means generation by the same interaction architecture, not
exact equality of the resulting kernels.

At the line-renormalized response level adopted in the inertial-range
argument, the minimal consistent identification is therefore
\be
K_{\rm SGS}(t,s)
\approx
G(t,s).
\label{eq:KSGS_identification_G}
\ee
Substituting Eq.~\eqref{eq:KSGS_identification_G} into
Eq.~\eqref{eq:KSGS_general_ansatz} gives
Eq.~\eqref{eq:app_ansatz_recalled}. The approximation does not append an
independent eddy-memory law. It uses for the unresolved drift history the
same dressed propagation generated by the unique resolved--reservoir
interaction vertex.

%--------------------------------------------------------%
This identification also keeps the scaling argument internally
consistent. In Appendix~\ref{app:anomalous_scaling},
Eq.~\eqref{eq:app_ansatz} uses the dressed macroscopic response $G$ as the
temporal kernel of the reservoir-induced drift only for the purpose of
inertial-range exponent counting. The resulting
Theorem~\ref{thm:no_anomalous_scaling} shows that the bare macroscopic
velocity, the subgrid drift velocity, and the effective transport
velocity share the same scaling exponent. That conclusion is then used
in Section~\ref{sec:scaling_preliminaries} to estimate the resolved
Jacobian covariance entering the transfer relation
\eqref{eq:dressed_two_time_transfer}. Thus the response ansatz does not
enter that transfer formula by direct substitution; it supplies the
scaling consistency needed to express its variance factor in terms of
the spectrum $E(k)$. The strain-controlled response time and the
suppression of uniform sweeping are established separately through the
reservoir-generated self-energy analyzed in
Theorem~\ref{thm:inertial_ranges}.
%--------------------------------------------------------%

%======================================================================%
\section{Prospective Outlooks}
\label{Sec_prospective_outlook}

The Symplectic Geometric Closure developed in this work should not be
viewed only as a particular subgrid-scale parameterization. In
retrospect, the central structures constructed above---the emergent
subgrid drift, the nested-Jacobian transport, the
$\mathcal O(p^4E(p))$ infrared suppression, and the DIA-type geometric
self-energy---are all consequences of a single interaction functional:
\be
\G[\bar{\zeta},r]
=
\frac14
\int_{\mathcal D}
\gamma(\mathbf x)
\big[J(\bar\zeta,r)\big]^2\,d\mathbf x.
\label{eq:G_quartic_outlook}
\ee
This functional is not merely an algebraic device for generating a
forcing term. Because it satisfies the symplectic constraint
$\{\G,V\}=0,$
it simultaneously fixes the admissible exchange geometry between
$\bar\zeta$ and $r$, preserves the augmented enstrophy structure, and
determines the form of the stochastic transport vertex whose covariance
later appears as a self-energy. The closure is therefore organized by a
geometric principle before any statistical approximation is made.

This observation suggests a broader interpretation of the framework. The
quadratic Jacobian functional \eqref{eq:G_quartic_outlook} defines a
closed stochastic transport theory for the resolved vorticity and the
hidden reservoir. The DIA-type self-energy is then one statistical
representation of this transport theory, obtained after expanding,
averaging, and resumming the reservoir dynamics. Thus the hierarchy of
objects developed in the paper can be read as
\begin{widetext}
\be
\text{Symplectic generator } \G
\longrightarrow
(g_1,g_2)
\longrightarrow
(\bar\zeta,r)\ \textrm{transport}
\longrightarrow
\Gamma_{symp}
\longrightarrow
\boldsymbol\eta
\longrightarrow
\theta_k
\longrightarrow
E(k).
\ee
\end{widetext}
The important point is the direction of construction. We do not begin
with a prescribed statistical closure and search for a stochastic model
that realizes it. We begin with a symplectic generating functional, and
the transport law, the memory kernel, the self-energy, and the
inertial-range response follow from that object.

In this sense, the present construction is conceptually close to
Kraichnan's amplitude-equation representation of DIA
\cite{kraichnan1970convergents} and to the related model equations
introduced by Leith \cite{leith1971atmospheric}. Those works sought
dynamical realizations of renormalized statistical theories. The present
framework reverses the order of explanation: the stochastic transport
system is primary, and the DIA-like renormalized description emerges only
after analyzing its statistical consequences. Consequently, the
functional $\G$ does not define only a second-order closure. It defines an
amplitude-level field theory, namely Eqns.~\eqref{eq:zeta_transport}--
\eqref{eq:r_transport}, whose statistical reductions reproduce
DIA-like behavior while still allowing the underlying stochastic
transport equations to be simulated directly.

This distinction opens two closely related directions. The first is to
use the existing quadratic theory as a dynamical model beyond its
DIA-type reduction. The second is to enlarge the class of admissible
generating functionals and thereby construct a hierarchy of increasingly
rich symplectic effective theories.

%--------------------------------------------------------------------------------------------------------%
\subsection{Higher-Order Symplectic Functionals and Coherent Structures}

The present framework separates the underlying stochastic transport
theory from the particular statistical approximation used to analyze it.
The equations generated by the quadratic functional
\eqref{eq:G_quartic_outlook},
namely Eqns.~\eqref{eq:zeta_transport}--\eqref{eq:r_transport}, define a
closed nonlinear stochastic dynamics for the pair $(\bar\zeta,r)$. Their
solutions may be used, in principle, to investigate phenomena that are
not naturally captured by a two-point DIA closure: coherent vortex
formation, vortex mergers, filament roll-up, jet organization,
probability density functions, higher-order structure functions,
skewness, kurtosis, rare events, and Lagrangian transport. The
self-energy derived in this paper should therefore be understood as one
renormalized statistical projection of the symplectic transport theory,
not as the full content of the theory itself.

This distinction matters because classical DIA is fundamentally a
two-point theory \cite{orszag1970analytical,orszag1977lectures}. It is designed to describe response functions,
two-time covariances, memory effects, and stationary energy spectra. It
is much less naturally adapted to the strongly non-Gaussian structures
that dominate two-dimensional turbulence at the level of individual
realizations: long-lived vortices, intense shear layers, intermittent
filaments, vortex mergers, and coherent jets. These objects are organized
by higher-order correlations and phase alignments that are largely
compressed away by a two-point closure.

The symplectic framework provides a natural way to formulate this
problem. In the theory developed here, the elementary interaction vertex
$\Gamma_{symp}$ (Eq.~\eqref{Gamma_symp_def}) is generated by the quadratic Jacobian invariant
given by Eq.~\eqref{eq:G_quartic_outlook}. However, there is no fundamental reason to
restrict the generating functional to this lowest nontrivial invariant.
One can instead consider a hierarchy of admissible scalar functionals of
the cross-gradient Jacobian,
\begin{widetext}
\be
\G[\bar\zeta,r]
=
\int_{\mathcal D}
\Big[
a_2\big[J(\bar\zeta,r)\big]^2
+
a_4\big[J(\bar\zeta,r)\big]^4
+
a_6\big[J(\bar\zeta,r)\big]^6
+\cdots
\Big]\,d\mathbf x.
\label{eq:G_hierarchy}
\ee
\end{widetext}
Each term generates a distinct symplectic transport vertex while
preserving the basic exchange architecture between the resolved field and
the hidden reservoir. The quadratic term $\G_2$ produces the cubic
symplectic vertex analyzed throughout the paper. Higher-order terms
produce higher-leg vertices and hence provide a systematic route for
embedding non-Gaussian interactions directly into the underlying
geometric dynamics.

To see the structure explicitly, consider the quartic Jacobian functional
\be
\G_4[\bar\zeta,r]
=
\frac18
\int_{\mathcal D}
\alpha(\mathbf x)
\big[J(\bar\zeta,r)\big]^4\,d\mathbf x.
\ee
Using the Jacobian integration-by-parts identity on a doubly periodic
domain,
\be
\int_{\mathcal D} A\,J(B,C)\,d\mathbf x
=
\int_{\mathcal D} C\,J(A,B)\,d\mathbf x,
\ee
its variational derivative with respect to the hidden field gives
\be
g_1^{(4)}
=
\frac{\delta \G_4}{\delta r}
=
\frac12
J\Big(
\alpha\big[J(\bar\zeta,r)\big]^3,
\bar\zeta
\Big).
\ee
Equivalently, the associated higher-order drift potential is
\be
\psi_{sgs}^{(4)}
=
\frac12
\alpha\big[J(\bar\zeta,r)\big]^3.
\ee
Thus the quartic functional generates a drift that is cubic in the local
cross-gradient misalignment. 

This algebraic form suggests three concrete physical consequences. First, the interaction is strongly localized. Since
$\psi_{sgs}^{(4)}$ scales like the cube of $J(\bar\zeta,r)$, it is weak
where the resolved and reservoir gradients are nearly aligned or weak,
and it is amplified where their phase misalignment is extreme. Such
regions are precisely where one expects intense local strain, vortex
edges, shear layers, and filamentary roll-up. Higher-order symplectic
functionals therefore offer a natural mechanism for intermittent,
spatially localized transport without inserting localization by hand.

Second, higher-order functionals generate genuinely non-Gaussian
vertices. Schematically, $g_1^{(4)}$ is a seven-field interaction: three
copies of $\bar\zeta$ and three copies of $r$ appear inside
$[J(\bar\zeta,r)]^3$, and the outer Jacobian couples this composite
object to another $\bar\zeta$. Diagrammatically, such vertices couple
higher-order moments directly. They therefore provide a controlled way
for skewness, kurtosis, coherent phase alignments, and rare-event
statistics to enter the renormalized dynamics.

Third, higher-order functionals should strengthen the infrared geometric
filtering. The quadratic theory already suppresses sweeping through the
cross-gradient factor responsible for the
$\mathcal O(p^4E(p))$ infrared weight. In the quartic functional, the
inner factor $[J(\bar\zeta,r)]^3$ carries additional derivatives before
the outer Jacobian acts. This derivative structure suggests an even
stronger annihilation of uniform translations. A full renormalization
analysis of $\G_4$ remains to be carried out, but the vertex geometry
indicates that higher-order symplectic theories should retain, and
possibly amplify, the sweeping-blind character of the quadratic closure.

The crucial point is that these extensions do not abandon the
conservation structure. A broad class of scalar functionals of
$J(\bar\zeta,r)$ inherits the same symplectic exchange mechanism
responsible for $\{\G,V\}=0.$
Thus higher-order interactions can enrich the statistical and physical
content of the closure while remaining energetically bounded by the
augmented enstrophy geometry. The hierarchy therefore provides a path
from the scale-local cascade dynamics generated by $J^2$ toward
increasingly localized physical-space transport generated by invariants
such as $J^4$, without losing the fundamental conservation constraint.

The broader implication is that the turbulence closure problem can be
reformulated as the construction of admissible symplectic interaction
functionals. In the present work, choosing the lowest-order nontrivial
invariant $\G_2$ was enough to generate the macroscopic SGS transport, the
emergent symplectic drift, the non-Markovian memory, the Dyson
self-energy, the infrared suppression of sweeping, and the dual-cascade
inertial ranges. Once this generating object is known, more sophisticated
closures can be obtained not by adding phenomenological eddy viscosities
or empirical backscatter terms, but by enlarging the admissible
geometric functional.

In this sense, the present closure is best viewed as the first member of
a hierarchy of symplectic effective transport theories:
\begin{widetext}
\bea
\G_2
&\longrightarrow&
\textrm{cubic symplectic interaction vertex}
\longrightarrow
\textrm{DIA-type finite-memory theory},
\nonumber\\
\G_4
&\longrightarrow&
\textrm{higher-order symplectic vertices}
\longrightarrow
\textrm{non-Gaussian coherent-structure dynamics},
\nonumber\\
\G_6
&\longrightarrow&
\textrm{richer effective actions}
\longrightarrow
\textrm{higher-order transport statistics},
\nonumber
\eea
\end{widetext}
and so forth. The complexity of the resulting statistical theory is then
dictated by the structure of the generating functional itself.

Classical closure theory begins with a given nonlinear dynamics and asks:
\emph{Which statistical diagrams should be retained?} The present
framework begins one level deeper and asks:
\emph{Which geometric interaction functional should generate the
dynamics whose diagrams are to be renormalized?} Once the generating
functional is specified, the admissible transport laws, the hidden-sector
couplings, the response functions, and the hierarchy of statistical
interactions follow systematically. From this viewpoint, the functional
$\G$ plays the role of an effective action for symplectic turbulent
transport: it determines both the geometry of the underlying field
dynamics and the statistical diagrams that emerge under renormalization.

%----------------------------------------------------------------------%
This outlook also clarifies how the one-loop result proved in
Theorem~\ref{thm:one_loop_emergence} should be situated within the
larger structure of the theory. That theorem establishes the geometric
origin of the memory kernel at the first nontrivial renormalized level:
the third-order Volterra--Picard correction generated by the quadratic
functional $\G_2$. At this order, the reservoir expansion produces
quartic Gaussian moments of $r^{(0)}$, and the connected Wick--Isserlis
cross-pairings generate the one-loop transport-covariance insertion
$\mathbf\Sigma_{\rm sweep}$. This is the lowest perturbative order at
which a genuine two-time self-energy can appear.

The same mechanism extends structurally to higher orders. At order
$r^{(2n+1)}$, the Volterra--Picard hierarchy (Eq.~\eqref{eq:r_n_recursive}) produces Gaussian moments
of order $2n+2$. The connected temporal pairings among these higher-order
reservoir tensors generate increasingly nested contraction topologies.
Their temporal organization mirrors the hierarchy of higher-loop
corrections in a Dyson expansion. Thus the symplectic closure contains a
diagrammatic interpretation that separates two independent organizing
principles.

The first is the \emph{statistical hierarchy}. This hierarchy is
generated by perturbatively expanding the transport dynamics associated
with a fixed generating functional,
\bes
r
=
r^{(0)}
+
r^{(1)}
+
r^{(2)}
+
\cdots .
\ees
For the quadratic theory $\G_2$, this expansion keeps the same bare
symplectic vertex $\Gamma_{symp}$ but produces progressively higher-loop
renormalizations of the propagator. The first-order response contains no
connected memory insertion of the type needed to dress the line. The
third-order term is the first order that generates a connected
Wick--Isserlis contraction and therefore yields the one-loop geometric
self-energy. Higher odd orders generate higher-loop Volterra--Dyson
corrections of the same fixed-vertex theory.

The second is the \emph{geometric hierarchy}. This hierarchy is generated
not by expanding the reservoir state for a fixed functional, but by
changing the generating functional itself,
\be
\G_2
\longrightarrow
\G_4
\longrightarrow
\G_6
\longrightarrow
\cdots .
\ee
Changing $\G$ changes the bare interaction vertex before any statistical
renormalization is performed. Thus $\G_4$ and $\G_6$ do not merely add
higher-loop corrections to the $\G_2$ theory; they define new symplectic
transport theories with new geometric vertices, new selection rules, and
new possible non-Gaussian coherent-structure dynamics.

This distinction is crucial. The Volterra loop expansion operates within
a fixed interaction geometry. It asks what statistical self-energies are
generated by repeated applications of a given symplectic vertex. The
geometric hierarchy operates one level deeper. It asks which admissible
symplectic functional should define the vertex in the first place. In the
quadratic theory studied here, resolving the reservoir expansion through
$r^{(3)}$ is therefore both mathematically natural and conceptually
minimal: it is the first order capable of producing the connected
Gaussian contraction that dresses the response. This is precisely the
finite-memory, sweeping-blind renormalization needed for the inertial
cascade theory of Section~\ref{Sec_inertial_ranges}.

From this perspective, the present work identifies the first nontrivial
member of a broader program. The $\G_2$ theory supplies the minimal
symplectic mechanism for stochastic transport, memory generation,
one-loop self-energy, and dual-cascade scaling. Higher Volterra orders
refine the statistical renormalization of this same theory. Higher
generating functionals, by contrast, enlarge the admissible geometry
itself and open the possibility of systematically incorporating
non-Gaussian coherent structures, rare events, and localized intermittent
transport into the same symplectic framework. The fundamental question is
therefore not only \emph{which diagrams should be retained?} but
\emph{which geometric functional should generate the diagrams?}

%----------------------------------------------------------------------%
\subsection{Reframing Data-Driven Closures: Learning the Symplectic Generator}
\label{subsec:learning_generator}

The distinction between the statistical hierarchy and the geometric
hierarchy also suggests a different formulation of data-driven turbulence
closure. In conventional machine-learning approaches to Large Eddy
Simulation, the central objective is often posed as a direct regression
problem: given filtered DNS data, learn the exact subgrid forcing
$\Pi_{\rm true}$ as a function of the resolved state; e.g.~\cite{duraisamy2019turbulence,beck2019deep,maulik2019subgrid,srinivasan2024turbulence}.
Schematically, one seeks a neural closure by minimizing
\bes
\left\|
\Pi_{\rm true}
-
{\rm NN}(\bar\zeta)
\right\|^2 .
\ees
This formulation asks:
\emph{Which subgrid forcing should be learned?}

The symplectic perspective suggests a different question. If the
transport laws, memory kernels, self-energies, and statistical diagrams
are all downstream consequences of a generating functional, then the
data-driven problem should not be to learn the closure force directly.
It should be to learn the admissible geometric generator from which the
closure emerges. The relevant question becomes:
\emph{Which symplectic interaction functional should be learned?}

This reframing places the present proposal within the broader movement
toward structure-preserving machine learning for physical dynamics. In
Hamiltonian Neural Networks, the neural model learns a scalar
Hamiltonian, and the vector field is obtained by differentiating that
Hamiltonian through the canonical symplectic structure
\cite{greydanus2019hamiltonian}. Related developments include
Hamiltonian-system identification from data \cite{bertalan2019learning},
Lagrangian Neural Networks \cite{cranmer2020lagrangian}, Symplectic
ODE-Nets \cite{zhong2020symplectic}, intrinsic symplectic networks
\cite{jin2020sympnets}, and port-Hamiltonian neural networks for systems
with forcing or dissipation \cite{desai2021port}. The common lesson of
these approaches is that one should not ask a neural network to learn an
arbitrary vector field when the physics is generated by an underlying
geometric object.

The present setting follows the same principle, but in an infinite-
dimensional turbulence-closure context. Instead of parameterizing the
subgrid forcing itself, we parameterize a scalar generating functional
\be
\G_\theta[\bar\zeta,r],
\ee
where $\theta$ denotes trainable parameters. The induced resolved and
reservoir interaction fields are then obtained variationally:
\be
g_{1,\theta}(\bar\zeta,r)
=
\frac{\delta \G_\theta}{\delta r},
\qquad
g_{2,\theta}(\bar\zeta,r)
=
-\frac{\delta \G_\theta}{\delta\bar\zeta}.
\ee
Thus the closure force is no longer the primitive learned object. It is a
derived Hamiltonian exchange field generated by $\G_\theta$.

To preserve the stability mechanism of the analytic theory, the learned
functional must be restricted to an admissible symplectic class, or
projected onto one, so that $\{\G_\theta,V\}=0.$

This condition is the geometric guardrail. It ensures that the learned
interaction exchanges activity between the resolved field and the hidden
reservoir without directly violating the augmented enstrophy constraint.
Equivalently, the neural network may be highly expressive at the level of
the scalar functional $\G_\theta$, while the induced vector field remains
confined to the class of admissible symplectic exchanges.

This decouples expressivity from stability. In many neural closures,
increasing expressivity also increases the risk of unstable
\emph{a posteriori} behavior. Here, expressivity is assigned to the
generator, while stability is enforced by the variational architecture.
Even if $\G_\theta$ is imperfect on an unseen flow state, the induced
fields $g_{1,\theta}$ and $g_{2,\theta}$ retain the exchange structure
imposed by the bracket constraint. Errors in the learned geometry may
alter the conservative transport pathways, the local phase lag, or the
effective memory, but they are structurally prevented from acting as an
unconstrained source of total augmented enstrophy.

This also clarifies how modern neural architectures might enter. A
standard multilayer perceptron could be used to parameterize
$\G_\theta$, but so could architectures designed for symbolic or
interpretable scientific representation. Kolmogorov--Arnold Networks
(KANs), for example, replace fixed nodal activations by learnable
univariate functions on edges and have been proposed as interpretable
alternatives to multilayer perceptrons for scientific discovery
\cite{liu2025kan}. In the present context, KANs would
not by themselves guarantee stability or symplecticity. Their potential
role would be representational: they could provide a compact and
interpretable parametrization of the scalar generator
$\G_\theta[\bar\zeta,r]$. The admissibility condition
$\{\G_\theta,V\}=0$ and the variational construction of
$(g_{1,\theta},g_{2,\theta})$ would still be the mechanisms enforcing the
geometric constraint.

The two-hierarchy viewpoint is therefore essential. For a fixed learned
generator $\G_\theta$, the Volterra--Picard expansion of the corresponding
stochastic transport equations generates the statistical hierarchy:
response functions, memory kernels, self-energy insertions, and
higher-loop corrections. Changing the architecture or admissible class of
$\G_\theta$, by contrast, changes the geometric hierarchy: it changes the
bare interaction vertex before any statistical renormalization is
performed. Thus a learned generator does not merely approximate a missing
term in the resolved equation. It selects the effective interaction
geometry whose diagrams, memory operators, and transport statistics are
then generated dynamically.

Training can therefore be formulated at several levels without changing
the geometric principle. One may fit instantaneous filtered stresses,
short-time trajectory increments, fluxes, spectra, two-time correlations,
or response statistics. But these targets should constrain
$\G_\theta$, not replace it. The learned model should be judged not only
by pointwise agreement with $\Pi_{\rm true}$, but also by whether the
derived dynamics reproduce stable transport, correct flux direction,
realistic memory, coherent structures, and statistically consistent
long-time behavior.

Ultimately, this perspective provides a route toward structurally stable
machine-learning closures. It elevates data-driven modeling from the
regression of a kinematic forcing to the identification of the underlying
geometry of unresolved transport. In the language of the present work,
one should not learn the memory kernel, the self-energy, or the subgrid
forcing as independent objects. One should learn the generator from which
they emerge. The central maxim is therefore:
\beas
&\emph{Do not learn the closure; learn the symplectic generator}\\
& \emph{that
generates the closure.}
\eeas
The functional $\G_\theta$ becomes the fundamental learned object, while
the transport laws, hidden-reservoir dynamics, non-Markovian memory, and
statistical closures remain derived consequences of the same geometric
principle.

%%=====================================%
\section{Discussion}

The principal result of this work is the construction of a prognostic symplectic field closure in which unresolved transport is represented by an explicit hidden dynamical reservoir. The closure is not a moment model for prescribed covariances, nor an eddy-viscosity parametrization of a pre-averaged stress. It is an SPDE-level enlargement of the resolved vorticity dynamics: the resolved vorticity field $\bar\zeta$ is coupled to a hidden reservoir field $r$, and the subgrid drift potential
\bes
\psi_{sgs}=\frac12\gamma J(\bar\zeta,r)
\ees
is generated dynamically by their cross-gradient phase relation. The statistical objects that appear later---memory kernels, response functions, covariance insertions, and self-energies---are therefore not postulated at the outset. They arise by eliminating, averaging, or renormalizing an explicit stochastic field dynamics. This architecture is summarized in Fig.~\ref{fig:reservoir_loop}: the conservative resolved--reservoir exchange, the emergent drift potential it generates, and the memory left behind once the hidden sector is eliminated. That figure also settles a question the construction naturally raises, namely whether such a geometrically constrained hidden sector can effectively realize a faithful closure on a demanding flow. Its fields come from a closure run initialized from coarse-grained $\mathrm{Re}=2.5\times10^4$ turbulence \cite{srinivasan2024turbulence}, at a filter cutoff lying well inside the inertial range, where no scale separation is available.  That closure run is on the $\beta$-plane, and the closure is seen to sustain the zonal jets of $\beta$-plane turbulence together with the vortices and filaments embedded between them. The run is stable, and it reproduces the exact structural identities of the theory---$\langle\mathbf u_{sgs}\rangle=0$, $\nabla\!\cdot\!\mathbf u_{sgs}=0$, and $\int\bar\zeta\,\Pi\,d\mathbf x=0$---to machine precision, while leaving the kinetic-energy transfer free in sign; see Section~\ref{Sec_numerics}.
Carried over ten eddy-turnover times with the calibrated parameters of
Table~\ref{tab:dns_params}, the same run holds the enstrophy spectrum of the
filtered reference across the resolved band, keeps the grid-scale enstrophy
fraction at its reference level, and sustains the jets at their
climatological strength---numerical evidence, to be developed in
\cite{ChekrounMcWilliams2026BL}, that the geometric architecture is not
merely integrable but quantitatively skillful.

This starting point distinguishes the present construction from the usual organization of turbulence closure theory. Classical statistical closures, generalized Langevin formulations, and perturbative or functional renormalization-group approaches obtain an effective description by eliminating, projecting out, or progressively integrating over unresolved fluctuations \cite{Zhou2021,Canet2022,Verma2025}. In DIA and related closures, their cumulative influence is encoded in renormalized response and correlation functions; in Mori--Zwanzig representations, it appears as memory and fluctuating forcing; and in RG or functional RG formulations, it is absorbed into scale-dependent propagators, vertices, couplings, or effective actions. The Symplectic Geometric Closure proceeds in the opposite order. Rather than taking elimination as its starting operation, it first supplies a minimal prognostic realization of the unresolved sector and constrains its coupling to the resolved flow through Hamiltonian transport geometry. Only subsequently is that sector eliminated, producing the effective statistical objects of the reduced theory.

The appearance of memory kernels is nevertheless fully consistent with the broad Mori--Zwanzig and generalized Langevin perspective, in which eliminating unresolved degrees of freedom produces non-Markovian drift terms together with fluctuating forcing \cite{darve2009computing}. In fluid dynamics, this viewpoint underlies projection-based closure models ranging from the optimal prediction program of Chorin and collaborators \cite{chorin2000optimal,Chorin_al02} to finite-memory and non-Markovian turbulence closures \cite{stinis2007higher,parish2017_LES,parish2017nonmarkovian}. Those works demonstrated that unresolved degrees of freedom leave a persistent temporal footprint and that suitable approximations of this memory can substantially improve under-resolved simulations.

The present framework approaches the same reduction problem from the opposite direction. Instead of starting from a projection operator and then seeking a tractable approximation of the formally generated memory kernel, we introduce the hidden reservoir as a prognostic part of the dynamics. The enlarged resolved--reservoir system is Markovian in the extended phase space, whereas the reduced resolved dynamics become non-Markovian after the reservoir is eliminated. The memory kernel is therefore the dynamical footprint left by an explicit hidden transport sector, rather than an object reconstructed only after projection. Moreover, because that sector is coupled through Jacobian and Lie-transport operators, the resulting memory inherits geometric selection rules rather than appearing as an unconstrained convolution kernel. The conceptual departure is thus not merely a new approximation to an existing memory term, but a dynamical realization of the hidden sector from which that term arises.

The relation to Kraichnan's Direct Interaction Approximation should be
understood in the same spirit. The present theory does not replace DIA
as a statistical closure procedure. Rather, it supplies an explicit
stochastic transport system from which a DIA-type response architecture
can be derived. Theorem~\ref{thm:one_loop_emergence} shows that the
connected third-order reservoir contractions generate the two-time
transport-covariance insertion
$\mathbf{\Sigma}_{\rm sweep}(s,\tau)$, built from two OU-induced
Hamiltonian advection operators. Corollary~\ref{cor:reservoir_tangent_self_energy}
then identifies the precise connection with the tangent-space
self-energy: the same two operators appear in
\bes
\boldsymbol{\eta}^{(0)}(s,\tau)\Phi
=
\mathbb{E}
\Big[
\mathcal{A}_s^{(0)}
\Big(
G(s,\tau)
\mathcal{A}_\tau^{(0)}\Phi
\Big)
\Big],
\ees
with the intermediate response $G(s,\tau)$ inserted between them. Thus
the Volterra--Picard calculation supplies the explicit reservoir
covariance, while the tangent-space Dyson construction supplies its
response-weighted form. The response architecture familiar from DIA is
therefore not postulated directly as a closure for two-time moments; it
arises from the statistical analysis of a prognostic stochastic
transport system. Because the underlying operators are generated by the
Hamiltonian drift $\psi_{sgs}$, their covariance retains the nested
Jacobian structure that later produces the infrared sweeping suppression
established in Sections \ref{Sec_inertial_ranges} and \ref{Sec_Galilean_Invariance}. The conservative exchange and
stability properties, by contrast, originate from the symplectic
structure of the full resolved--reservoir dynamics established in
Section \ref{Sec_Field_Theory}.
%----------------------------------------------------------------------------------------------------------------%

From this perspective, Mori--Zwanzig memory and DIA self-energy are not competing descriptions. They are complementary reductions of the same hidden transport process. Eliminating the reservoir in the time domain produces a non-Markovian resolved equation with a Volterra memory kernel. Organizing the same dynamics in tangent space produces a Dyson--Volterra response equation with a self-energy kernel. The hidden reservoir therefore provides a common dynamical origin for two languages that are often treated separately: generalized Langevin memory and renormalized turbulent response.

%----------------------------------------------------------------------------------------------------------------%

A second consequence is that the renormalization is continuous and
dynamical. Classical effective-transport closures often precompute the
effect of unresolved scales, reduce it to coefficients such as an eddy
viscosity, and then freeze those coefficients into the resolved model.
Here the reservoir evolves concurrently with the resolved vorticity. The
effective subgrid transport is never prescribed \emph{a priori}; it is
continually regenerated from the evolving phase lag between
$\bar\zeta$ and $r$. In this sense, the closure performs a dynamical
renormalization in the field variables themselves: the hidden reservoir
does not merely represent missing variance, but carries a history-bearing
transport degree of freedom.

%-------------------------------------------------------------------------------------%

This view clarifies the role of stochasticity. The stochastic reservoir
is not an external random forcing added to mimic backscatter. It is the
hidden sector whose OU covariance supplies the statistical contractions
from which the macroscopic memory and self-energy are built. The
Volterra--Picard expansion shows how disconnected contractions generate
local deterministic diffusion insertions, while connected
cross-contractions generate the two-time Lie-transport covariance that
acts as the geometric self-energy. Thus stochastic backscatter,
realizability-preserving fluctuations, deterministic diffusion traces,
and non-Markovian memory appear as different projections of one SPDE
field closure.

The analysis of Section~\ref{Sec_Galilean_Invariance} adds the key
geometric ingredient. The nested-Jacobian interaction defining the
Symplectic Geometric Closure is compatible with random Galilean
transformations at the level of the bare interaction geometry. This is
not a cosmetic invariance statement. It explains why the one-loop
self-energy derived in Section~\ref{Sec_inertial_ranges} suppresses the
classical Eulerian sweeping contribution. The symplectic vertex
\be
\Gamma_{symp}
\sim
J\circ\gamma J,
\ee
depends on cross-gradients, not on absolute velocity amplitudes.
Proposition~\ref{prop:RGI_compatibility} shows that uniform sweeping
motions do not enter the vertex explicitly, while
Proposition~\ref{prop:geometric_selection} shows that the same vertex
annihilates uniform amplitudes and parallel-gradient configurations. The
infrared factor
$p^4E(p)\,dp,$ 
obtained in Theorem~\ref{thm:inertial_ranges} is therefore the spectral
signature of an interaction geometry that is already sweeping-blind
before statistical resummation begins.

This provides a different resolution of the classical Eulerian DIA
pathology. Kraichnan's original Eulerian DIA generated spurious
decorrelation because the Eulerian response remained sensitive to
large-scale translational sweeping. This motivated Lagrangian-history
formulations designed to restore Random Galilean Invariance by following
fluid histories rather than fixed Eulerian observations. The present
approach keeps an Eulerian Volterra response, but changes the vertex that
is being renormalized. The line renormalization is performed around the
symplectic cross-gradient vertex, not around the bare Navier--Stokes
velocity vertex. As a result, uniform translation lies in the null sector
of the interaction, while strain and cross-gradient twisting remain
active. In short, the closure distinguishes translation from deformation
at the level of the bare field vertex.  This selection rule is
visible in the realized fields of Fig.~\ref{fig:reservoir_loop}. The
hidden reservoir is driven by homogeneous, isotropic noise and relaxed by
a translation-invariant operator, so the only route by which the resolved
geometry reaches it is the nested Jacobian vertex itself; through that
single channel the emergent drift comes to concentrate on the resolved
vorticity filaments, with
$\mathrm{corr}(|\psi_{sgs}|,|\nabla\bar\zeta|)=+0.63$, and to organize
into counter-rotating dipoles that fold and stretch them
(Section~\ref{Sec_numerics}). The vertex places the closure's action
where the resolved and unresolved geometries fail to commute, not where
amplitudes happen to be large. The geometric origin of this
distinction is displayed in Fig.~\ref{fig:symplectic_backbone}, where the
emergent drift is shown to be confined to the zero-mean Hamiltonian
sector and therefore never to excite the uniform-translation modes;
Fig.~\ref{fig:sweeping_cascade} then follows the same distinction through
to its physical consequence, contrasting an eddy that is rigidly swept
with one that is stretched and folded, and tracing the surviving
strain-controlled memory into the dual-cascade spectra.

The inertial-range consequences are then transparent. Section
\ref{sec:scaling_preliminaries} expressed the transfer as a
response-weighted two-time covariance,
$T(k)
\sim
\theta_k k^3E(k)^2$ (Eq.~\eqref{Eq_Tk_response}), 
where the variance factor comes from the bare resolved advective
tendency and the memory time, $\theta_k$ given by Eq.~\eqref{eq:theta_k_def},
is generated by the geometric self-energy. Theorem
\ref{thm:inertial_ranges} then showed that the infrared sweeping
contribution is suppressed by the Lie-derivative covariance, leaving a
strain-controlled memory time in the inverse-energy range and a finite
vorticity-strain time in the forward-enstrophy range. Under the usual
Fj\o{}rtoft--Kraichnan locality and stationary-flux assumptions, the
resulting finite-memory transfer theory admits
$E(k)\sim\epsilon^{2/3}k^{-5/3}$,
for the inverse-energy cascade and
$E(k)\sim\eta_Z^{2/3}k^{-3}$,
for the forward-enstrophy cascade.

The energetic direction of these transfers is controlled by the
conservation laws. Section~\ref{Sec_phase_lag_upscale} showed that the
geometric interaction preserves the energy--enstrophy invariant
structure underlying Fj\o{}rtoft's theorem. Consequently, the closure
does not obtain upscale transfer by inserting a negative eddy viscosity.
Instead, it combines two ingredients already present in two-dimensional
Euler dynamics: conservative triad admissibility and finite
strain-sensitive memory. The former determines which transfers are
allowed; the latter determines whether those transfers remain coherent
long enough to produce a macroscopic flux. This recovers the two
foundations of Kraichnan's dual-cascade phenomenology, and is consistent
with the classical theoretical and numerical picture of forced
two-dimensional turbulence
\cite{Kraichnan1967,kraichnan1971,mcwilliams1984emergence,smithr1994}.

The closure also provides a dynamical counterpart to Kraichnan's use of
stochastic auxiliary constructions. Random Coupling Models and related
realizable closures introduced stochasticity as a mathematical device for
generating consistent turbulence statistics. Here the stochastic
auxiliary field is not introduced only at the statistical level. It is a
prognostic reservoir whose evolution continuously generates the transport
corrections experienced by the resolved flow. The resulting stochastic
transport is therefore not imposed kinematically; it emerges from the
coupled evolution of the resolved field and its hidden reservoir. In this
sense, several concepts that historically appeared as statistical
ingredients of renormalized turbulence theory---memory effects,
backscatter, realizability-preserving fluctuations, self-energy
corrections, and response damping---are converted into consequences of a
single geometric dynamical system evolving in an extended phase space.

The hidden reservoir also has a natural interpretation alongside
Generalized Lagrangian Mean (GLM) theory \cite{andrews1978generalized}. In GLM,
unresolved wave or displacement fields induce corrections to the mean
transport, such as Stokes drift. In the present theory, the reservoir
plays an analogous hidden-transport role, but the induced drift is not a
kinematic wave-mean correction. It is a non-Markovian Hamiltonian
transport velocity generated by the history of cross-gradient interaction
between $\bar\zeta$ and $r$; see
Section~\ref{Sec_dressed_streamfunction_GLM}. The result is a dressed
streamfunction that carries memory of the resolved dynamics rather than a
prescribed mean drift.

The construction is also related to generalized Langevin representations
of renormalized turbulence closures discussed by Krommes
\cite{krommes1996non}, to Kraichnan's amplitude-equation representation
of DIA \cite{kraichnan1970convergents}, and to the model equations
proposed by Leith \cite{leith1971atmospheric}. In those approaches,
stochastic effective equations are constructed to reproduce or model the
statistics implied by a renormalized closure. The Symplectic Geometric
Closure reverses the logic: it starts from a stochastic field dynamics in
an extended phase space and derives the response, memory, and self-energy
structure from that dynamics. The reservoir therefore realizes the
hidden sector whose elimination produces renormalized transport, rather
than representing a stochastic process fitted \emph{a posteriori} to an
already renormalized theory.

The scope of the present results should also be made explicit. The
inertial-range scaling theorem is an OU/one-loop and line-renormalized
result, derived under the isotropic inertial assumptions stated in
Section~\ref{sec:scaling_preliminaries} and the finite covariance trace
conditions of Appendix~\ref{App_Macro_Diffusion_Ope}. It is not an
all-orders proof of Random Galilean Invariance for every possible
statistical truncation. Rather, it identifies a precise mechanism by
which the dominant Eulerian sweeping contribution is suppressed in the
geometric self-energy. Higher-order vertex corrections, anisotropic
effects, finite-domain effects, beta-plane dynamics, and the numerical
realization of the coupled SPDE remain important directions for further
study. The point is that the leading obstruction that invalidates bare
Eulerian DIA is no longer present in the same form, because the vertex
being renormalized has changed.

Overall, the Symplectic Geometric Closure recovers, from a geometric
field-level construction, several central ingredients of Kraichnan's
turbulence phenomenology: finite-memory response, stochastic auxiliary
degrees of freedom, realizability-compatible covariance structure,
suppression of spurious sweeping decorrelation, and the
energy--enstrophy bookkeeping required for inverse transfer. The central
lesson is that the structure of the interaction being renormalized may be
as important as the renormalization procedure itself. In the present
framework, memory kernels, self-energy corrections, dual-cascade scaling,
and Random Galilean compatibility all arise from the same symplectic
interaction architecture. From this viewpoint, the closure may be
interpreted as transferring part of the work traditionally assigned to
statistical resummation into the geometry of the underlying transport
vertex itself.

\section*{Data availability statement}
All data that support the findings of this study are included within the article (and any supplementary files).

%---------------------------------------------------------------------------------------------------------------------------------------------------------%
\begin{acknowledgements}
MDC wishes to thank Kaushik Srinivasan for inspiring discussions on the closure problem and for sharing the DNS data used in \cite{srinivasan2024turbulence}.  
This work has been supported by the Office of Naval Research (ONR) Multidisciplinary University Research Initiative (MURI) grant N00014-20-1-2023, and by the National Science Foundation grant DMS-2407484. MDC acknowledges also the partial support provided by the 
Knell Family Institute for Artificial Intelligence, Weizmann Institute of Science, and by the Institute for Environmental Sustainability (IES) at the Weizmann Institute of Science. 
\end{acknowledgements}

%=====================================%
\appendix
%\apptoc
{\small

\appendix
\vspace{-2ex}
\section{Proof of Theorem \ref{Main_thm}}
\label{app:theorem_proof}

To avoid the technical difficulties associated with establishing a perfect stochastic cocycle in infinite dimensions \cite{arnold1995perfect,Arnold98}, we adopt the abstract attractor framework introduced by Crauel, Debussche, and Flandoli \cite{crauel1997random}. In that framework, once a shift-compatible stochastic flow is available, the existence of a compact pullback attracting set yields the corresponding random attractor.
\begin{proof}
We first establish pullback absorption. Let $z(t,\omega)$ be the stationary Ornstein--Uhlenbeck process solving
\be
dz+Dz\,dt=\Sigma\,dW_t.
\ee
By condition \eqref{cond3}, $z$ has $\mathbb P$-a.s.~continuous trajectories in the required regularity class. We perform the standard pathwise change of variables
\be
\tilde r(t)=r(t)-z(t,\omega),
\qquad
\tilde u(t)=(\bar\zeta(t),\tilde r(t)).
\ee
The stochastic system is thereby transformed into a deterministic non-autonomous system with random parameter:
\bea
\partial_t\bar\zeta
&=
f_1(\bar\zeta)
+
g_1(\bar\zeta,\tilde r+z),
\\
\partial_t\tilde r
&=
-D\tilde r
+
g_2(\bar\zeta,\tilde r+z).
\eea

We estimate the Lyapunov functional along the shifted dynamics. Using the uncoupled dissipation condition \eqref{cond1}, the cross-interaction bound \eqref{cond2}, and the polynomial translation bound \eqref{eq:V_translation_bound}, one obtains
\be
\frac{d}{dt}V(\tilde u(t))
\le
-\lambda V(\tilde u(t))
+
C_\epsilon\big(1+\|z(t,\omega)\|^m\big)+\Gamma,
\label{eq:abstract_shifted_dissipation}
\ee
where
\[
\lambda=\alpha-a>0.
\]
Applying Gronwall's lemma from an initial time $s\le0$ to the terminal time $0$ gives
\be
V(\tilde u(0))
\le
e^{\lambda s}V(\tilde u(s))
+
\int_s^0
e^{\lambda\tau}
\Big[
C_\epsilon\big(1+\|z(\tau,\omega)\|^m\big)+ \Gamma \Big]\,d\tau .
\label{eq:gronwall_H}
\ee
The stationary OU process is tempered; in particular, it has at most subexponential, hence polynomial, growth as $\tau\to-\infty$ \cite{crauel1997random,CLW15_vol1}. Therefore the integral in \eqref{eq:gronwall_H} converges to a finite random variable as $s\to-\infty$.

Let $\mathfrak B\subset\mathcal H$ be deterministic and bounded. Since the first term on the right-hand side of \eqref{eq:gronwall_H} tends to zero uniformly for initial data in $\mathfrak B$, there exists a random pullback time $s_{\mathfrak B}(\omega)<0$ and a finite random radius $R_{\mathcal H}(\omega)$ such that
\be
S(0,s;\omega)\mathfrak B
\subset
B_{\mathcal H}\big(0,R_{\mathcal H}(\omega)\big),
\qquad
s\le s_{\mathfrak B}(\omega).
\ee
Thus the solution operator possesses a pullback absorbing set in $\mathcal H$.

If the solution operator is additionally pullback asymptotically compact in $\mathcal H$, then the Crauel--Debussche--Flandoli attractor theorem applies \cite{crauel1997random}. The pullback absorbing set and asymptotic compactness together yield a unique measurable compact global random attractor $\mathcal A(\omega)$ in $\mathcal H$.
\end{proof}

\vspace{-2ex}
\section{Proof of Theorem \ref{thm:SGC_attractor}}
\label{app:SGC_attractor_proof}

\begin{proof}
The proof consists of verifying the hypotheses of the abstract random attractor
theorem, Theorem~\ref{Main_thm}. The Hilbert--Schmidt condition
\eqref{eq:OU_trace_condition} implies the OU regularity condition
\eqref{cond3}; this is shown in {\it Supplementary Note 1} by applying It\^o
isometry to the stationary convolution generated by $(D,\Sigma)$. The
pullback asymptotic compactness of the hyperviscous system is also verified in
 {\it Supplementary Note 1} by a Duhamel splitting argument adapting ideas used for viscoelastic fluids with memory \cite{CGH12}: the linear part
decays exponentially, while the nonlinear remainder is compactified by the
fourth-order analytic semigroup. It remains here to verify the Lyapunov
dissipation condition \eqref{cond1}, the polynomial translation bound
\eqref{eq:V_translation_bound}, and the conservative cross-interaction
condition \eqref{cond2}.

\textbf{Step 1: Derivation of the cross-interactions and verification of
\texorpdfstring{(H$_2$)}{(H2)}.}

We first compute the variational derivatives of the symplectic functional
$\G[\bar{\zeta},r]$ defined in Eq.~\eqref{eq:G_quartic}. On the periodic torus,
we use the standard Jacobian integration-by-parts identity
\be
\int_{\mathbb T^2} A\,J(B,C)\,d\mathbf x
=
\int_{\mathbb T^2} C\,J(A,B)\,d\mathbf x .
\label{eq:jacobian_identity}
\ee
Varying the symplectic generator $\G$ with respect to $r$ gives
\bea
\delta_r \G
&=
\frac12
\int_{\mathbb T^2}
\gamma J(\bar{\zeta},r)
J(\bar{\zeta},\delta r)\,d\mathbf x
\nonumber\\
&=
\frac12
\int_{\mathbb T^2}
\delta r\,
J\Big(\gamma J(\bar{\zeta},r),\bar{\zeta}\Big)
\,d\mathbf x .
\eea
Therefore
\be
g_1
=
\frac{\delta \G}{\delta r}
=
\frac12
J\Big(\gamma J(\bar{\zeta},r),\bar{\zeta}\Big).
\ee
Similarly, varying with respect to $\bar{\zeta}$ gives
\bea
\delta_{\bar{\zeta}}\G
&=
\frac12
\int_{\mathbb T^2}
\gamma J(\bar{\zeta},r)
J(\delta\bar{\zeta},r)\,d\mathbf x
\nonumber\\
&=
-\frac12
\int_{\mathbb T^2}
\delta\bar{\zeta}\,
J\Big(\gamma J(\bar{\zeta},r),r\Big)
\,d\mathbf x .
\eea
Thus
\be
g_2
=
-\frac{\delta \G}{\delta\bar{\zeta}}
=
\frac12
J\Big(\gamma J(\bar{\zeta},r),r\Big).
\ee

For the augmented enstrophy functional \eqref{Eq_V_enstrophy},
\be
\frac{\delta V}{\delta\bar{\zeta}}=\bar{\zeta},
\qquad
\frac{\delta V}{\delta r}=r.
\ee
Therefore the Poisson bracket between $G$ and $V$ is
\bea
\{\G,V\}
&=
\int_{\mathbb T^2}
\left(
\frac{\delta \G}{\delta r}\bar{\zeta}
-
\frac{\delta \G}{\delta\bar{\zeta}}r
\right)d\mathbf x
\nonumber\\
&=
\int_{\mathbb T^2}
\left(
g_1\bar{\zeta}
+
g_2 r
\right)d\mathbf x
\nonumber\\
&=
\frac12
\int_{\mathbb T^2}
\left[
\bar{\zeta}\,
J\big(\gamma J(\bar{\zeta},r),\bar{\zeta}\big)
+
r\,
J\big(\gamma J(\bar{\zeta},r),r\big)
\right]d\mathbf x .
\eea
Since $\int_{\mathbb T^2} A\,J(B,A)\,d\mathbf x=0$ on the periodic domain,
both terms vanish identically. Hence
\be
\{\G,V\}=0.
\ee
By Corollary~\ref{Main_corr}, the cross-interaction condition \eqref{cond2}
is therefore satisfied with
\be
a=0,
\qquad
b=0.
\ee
Thus the symplectic cross-interaction redistributes augmented enstrophy between
the resolved and hidden sectors but contributes exactly zero to the global
Lyapunov budget.

\textbf{Step 2: Verification of the uncoupled dissipation condition
\texorpdfstring{(H$_1$)}{(H1)}.} We evaluate the uncoupled macroscopic dynamics against $\bar{\zeta}$ in
$L^2(\mathbb T^2)$:
\bea
\left\langle f_{1,h}(\bar{\zeta}),\bar{\zeta}\right\rangle
&=
\left\langle -J(\bar{\psi},\bar{\zeta}),\bar{\zeta}\right\rangle
-
\beta
\left\langle \partial_x\bar{\psi},\bar{\zeta}\right\rangle
-
\mu\|\bar{\zeta}\|^2
\nonumber\\
&\quad
+
\nu\left\langle \Delta\bar{\zeta},\bar{\zeta}\right\rangle
-
\nu_h
\left\langle \Delta^2\bar{\zeta},\bar{\zeta}\right\rangle
+
\left\langle F_\zeta,\bar{\zeta}\right\rangle .
\eea
The Jacobian term conserves enstrophy:
\be
\left\langle J(\bar{\psi},\bar{\zeta}),\bar{\zeta}\right\rangle=0.
\ee
The beta-plane term also vanishes on the periodic domain, since
$\bar{\zeta}=\Delta\bar{\psi}$ and
\be
\left\langle \partial_x\bar{\psi},\bar{\zeta}\right\rangle
=
\left\langle \partial_x\bar{\psi},\Delta\bar{\psi}\right\rangle
=
-\frac12
\int_{\mathbb T^2}\partial_x|\nabla\bar{\psi}|^2\,d\mathbf x
=
0.
\ee
Integration by parts gives
\be
\nu\left\langle \Delta\bar{\zeta},\bar{\zeta}\right\rangle
=
-\nu\|\nabla\bar{\zeta}\|^2,
\qquad
-\nu_h
\left\langle \Delta^2\bar{\zeta},\bar{\zeta}\right\rangle
=
-\nu_h\|\Delta\bar{\zeta}\|^2.
\ee
Finally, Young's inequality yields
\be
\left\langle F_\zeta,\bar{\zeta}\right\rangle
\le
\frac{\mu}{2}\|\bar{\zeta}\|^2
+
\frac{1}{2\mu}\|F_\zeta\|^2.
\ee
Thus
\be
\left\langle f_{1,h}(\bar{\zeta}),\bar{\zeta}\right\rangle
\le
-\frac{\mu}{2}\|\bar{\zeta}\|^2
-\nu\|\nabla\bar{\zeta}\|^2
-\nu_h\|\Delta\bar{\zeta}\|^2
+
\frac{1}{2\mu}\|F_\zeta\|^2.
\ee

For the reservoir component, since $D$ is positive self-adjoint and
\[
D=\nu_rA^2+\kappa A+\mu_r I,
\]
we have
\be
\langle -Dr,r\rangle
=
-\nu_r\|Ar\|^2
-\kappa\|A^{1/2}r\|^2
-\mu_r\|r\|^2
\le
-\mu_r\|r\|^2.
\ee
Combining the two estimates gives
\bea
\left\langle f_{1,h}(\bar{\zeta}),\bar{\zeta}\right\rangle
+
\langle -Dr,r\rangle
&\le
-\frac{\mu}{2}\|\bar{\zeta}\|^2
-\mu_r\|r\|^2
+
\frac{1}{2\mu}\|F_\zeta\|^2 .
\eea
Since
\be
V(\bar{\zeta},r)
=
\frac12\left(\|\bar{\zeta}\|^2+\|r\|^2\right),
\ee
we may choose
\be
\alpha_*
=
\min\left\{\frac{\mu}{2},\mu_r\right\}>0,
\qquad
\Gamma_*
=
\frac{1}{2\mu}\|F_\zeta\|^2,
\ee
and obtain
\be
\left\langle f_{1,h}(\bar{\zeta}),\frac{\delta V}{\delta\bar{\zeta}}\right\rangle
+
\left\langle -Dr,\frac{\delta V}{\delta r}\right\rangle
+
\alpha_* V(\bar{\zeta},r)
\le
\Gamma_*.
\ee
This verifies the uncoupled dissipation condition \eqref{cond1}. Since
$a=0$ by Step 1, the stability requirement $a<\alpha_*$ is automatically
satisfied.

\textbf{Step 3: Verification of the polynomial translation bound.}

The augmented enstrophy is quadratic. Under the OU translation
$r=\tilde r+z$, we have
\bea
V(\bar{\zeta},\tilde r+z)
&=
\frac12\|\bar{\zeta}\|^2
+
\frac12\|\tilde r+z\|^2
\nonumber\\
&\le
\frac12\|\bar{\zeta}\|^2
+
\frac12(1+\epsilon)\|\tilde r\|^2
+
\frac12\left(1+\frac1\epsilon\right)\|z\|^2
\nonumber\\
&\le
(1+\epsilon)V(\bar{\zeta},\tilde r)
+
C_\epsilon\|z\|^2 .
\eea
Thus the translation bound, Eq.~\eqref{eq:V_translation_bound}, holds with
polynomial exponent $m=2$. Since the OU process satisfies the stronger
$H^2$ regularity condition implied by \eqref{eq:OU_trace_condition}, this
bound is compatible with the abstract framework.

We have therefore verified the Lyapunov dissipation condition \eqref{cond1},
the conservative cross-interaction condition \eqref{cond2}, the noise
regularity condition \eqref{cond3}, and the polynomial translation bound.
 {\it Supplementary Note 1} establishes pullback asymptotic compactness for
the hyperviscous closure under the same assumptions. Notably, the proof adapts the splitting philosophy used in  \cite{CGH12}  for viscoelastic fluids with memory. Applying
Theorem~\ref{Main_thm} yields a unique, measurable, compact global random
attractor in $\mathcal H$.
\end{proof}

%===========================================================%
\section{Derivation of the Dyson Equation and Geometric Self-Energy}
\label{app:Dyson_derivation}

In Section \ref{Sec_inertial_ranges}, we asserted that the stochastic dual-transport system analytically generates Kraichnan's triad memory kernel without invoking ad-hoc statistical closures. To formalize this, we explicitly construct the (renormalized) Dyson Volterra equation for the macroscopic propagator using projection operator techniques on the symplectic Lie algebra.

The derivation follows the diagrammatic structure of the Direct-Interaction Approximation. While the line-renormalization procedure is structurally equivalent to the DIA truncation, the fundamental novelty lies in the vertex operator: rather than the algebraic Navier-Stokes vertex, our stochastic transport operator is generated by the nested Jacobian geometry, which dictates the spectral structure of the resulting self-energy ${\bm \eta}$; see Section \ref{sec:diagrammatic_interpretation}.
As shown in Section~\ref{Sec_inertial_ranges}, this structure produces the infrared suppression mechanism responsible for eliminating the classical sweeping divergence.

%---------------------------------------------------------------------------------------%

%---------------------------------------------------------------------------------------%
\bt[Dyson Equation for the Symplectic Propagator]
\label{thm:Dyson_propagator}
Consider the stochastic dressed-transport skeleton
\bes
\partial_t\bar\zeta
+
J(\bar\psi,\bar\zeta)
=
J(\psi_{sgs},\bar\zeta).
\ees
Define the reservoir-induced stochastic advection operator by
\bea
\mathcal{A}_{\rm res}(t)\Phi
&:=
J\big(\psi_{sgs}(t),\Phi\big)
=
\mathcal{L}_{X_{\psi_{sgs}(t)}}\Phi,
\\
X_{\psi_{sgs}(t)}
&=
\nabla^\perp\psi_{sgs}(t).
\label{eq:A_res_appendix}
\eea
Assume that this stochastic Hamiltonian advection operator is centered:
\be
\mathbb{E}\big[\mathcal{A}_{\rm res}(t)\big]
=
0
\qquad
\text{for all }t.
\label{eq:A_res_centering_assumption}
\ee
Let $G(t,s)$ denote the ensemble-averaged fundamental solution of the
linearized stochastic dressed-transport equation. Under the DIA-type
line-renormalized smoothing, $G(t,s)$ satisfies the Volterra
integro-differential equation
\be
\partial_tG(t,s)
=
\mathcal{L}_0(t)G(t,s)
+
\int_s^t
\boldsymbol{\eta}(t,\tau)G(\tau,s)\,d\tau,
\qquad
t>s,
\label{eq:Dyson_symplectic_propagator}
\ee
with
\be
G(s,s)=\mathcal{I},
\qquad
\mathcal{L}_0(t)\Phi
=
-J\big(\bar\psi(t),\Phi\big).
\label{eq:bare_advection_operator}
\ee
The geometric self-energy is the operator-valued covariance of two
centered reservoir-induced Lie-transport vertices connected by the
intermediate response:
\bea
\boldsymbol{\eta}(t,\tau)\Phi
&=
\mathbb{E}
\Big[
\mathcal{A}_{\rm res}(t)
\Big(
G(t,\tau)\mathcal{A}_{\rm res}(\tau)\Phi
\Big)
\Big]\\
&=
\mathbb{E}
\Big[
\mathcal{L}_{X_{\psi_{sgs}(t)}}
\Big(
G(t,\tau)
\mathcal{L}_{X_{\psi_{sgs}(\tau)}}\Phi
\Big)
\Big]\\
&=
\mathbb{E}
\Big[
J\Big(
\psi_{sgs}(t),
G(t,\tau)
J\big(\psi_{sgs}(\tau),\Phi\big)
\Big)
\Big].
\label{eq:eta_reservoir_induced_operator}
\eea
\et
%---------------------------------------------------------------------------------------%

%---------------------------------------------------------------------%
\br[Centering of the Reservoir-Induced Advection]
\label{rem:centered_reservoir_advection}
Theorem~\ref{thm:Dyson_propagator} assumes that the
reservoir-induced advection operator is centered. For the full reservoir
dynamics, the corresponding decomposition is
\bea
&\mathcal{A}_{\rm res}(t)
=\overline{\mathcal{A}}_{\rm res}(t)+\widetilde{\mathcal{A}}_{\rm res}(t),\\
&\overline{\mathcal{A}}_{\rm res}(t):=\mathbb{E}\big[\mathcal{A}_{\rm res}(t)\big],\\
&\mathbb{E}\big[\widetilde{\mathcal{A}}_{\rm res}(t)\big]
=
0.
\label{eq:Ares_mean_fluctuation_decomposition}
\eea
The deterministic mean advection may then be absorbed into the bare
generator by defining
\be
\mathcal{L}_{\rm eff}(t)
:=
\mathcal{L}_0(t)
+
\overline{\mathcal{A}}_{\rm res}(t).
\label{eq:Leff_reservoir_mean}
\ee
The connected geometric self-energy is correspondingly constructed from
the centered fluctuation:
\be
\boldsymbol{\eta}_c(t,\tau)\Phi
=
\mathbb{E}
\Big[
\widetilde{\mathcal{A}}_{\rm res}(t)
\Big(
G(t,\tau)
\widetilde{\mathcal{A}}_{\rm res}(\tau)\Phi
\Big)
\Big].
\label{eq:eta_connected_centered}
\ee
Indeed, because $G(t,\tau)$ is the deterministic ensemble-averaged
response used in the DIA-type smoothing,
\bea
&\mathbb{E}
\Big[
\mathcal{A}_{\rm res}(t)
\Big(
G(t,\tau)
\mathcal{A}_{\rm res}(\tau)\Phi
\Big)
\Big]
\nonumber\\
&
=
\overline{\mathcal{A}}_{\rm res}(t)
\Big(
G(t,\tau)
\overline{\mathcal{A}}_{\rm res}(\tau)\Phi
\Big)
+
\mathbb{E}
\Big[
\widetilde{\mathcal{A}}_{\rm res}(t)
\Big(
G(t,\tau)
\widetilde{\mathcal{A}}_{\rm res}(\tau)\Phi
\Big)
\Big],
\label{eq:Ares_second_moment_decomposition}
\eea
the mixed terms vanishing by centering. The first term is a deterministic
mean--mean insertion and is not part of the connected self-energy; the
second is the genuine transport-covariance contribution.
At the leading OU baseline used in
Corollary~\ref{cor:reservoir_tangent_self_energy},
$\mathcal{A}_s^{(0)}$ depends linearly on the centered
Ornstein--Uhlenbeck field $r^{(0)}(s)$. Hence
\be
\mathbb{E}\big[\mathcal{A}_s^{(0)}\big]
=
0,
\qquad
\widetilde{\mathcal{A}}_{\rm res}(s)
\longrightarrow
\mathcal{A}_s^{(0)}.
\label{eq:Ares_OU_centering_reduction}
\ee
The self-energy appearing in that corollary is therefore already the
connected OU transport covariance, with no additional centering
required.
\er
%---------------------------------------------------------------------%
\begin{proof}
For each reservoir realization, we consider the tangent dynamics of the
stochastic dressed-transport skeleton with respect to the transported
resolved field, while treating $\psi_{sgs}$ as the prescribed random
coefficient generated along that realization. The resulting pathwise
linear operator is
\be
\mathcal{L}(t)
=
\mathcal{L}_0(t)
+
\mathcal{A}_{\rm res}(t),
\ee
where
\bea
\mathcal{L}_0(t)\Phi
&=
-J\big(\bar\psi(t),\Phi\big),
\\
\mathcal{A}_{\rm res}(t)\Phi
&=
J\big(\psi_{sgs}(t),\Phi\big)
=
\mathcal{L}_{X_{\psi_{sgs}(t)}}\Phi.
\eea
By assumption,
\be
\mathbb{E}[\mathcal{A}_{\rm res}(t)]
=
0.
\label{eq:A_res_centered_proof}
\ee
Let $\mathcal{U}(t,s)$ denote the exact path-dependent fundamental
solution of this linear equation. It satisfies
\be
\partial_t\mathcal{U}(t,s)
=
\Big[
\mathcal{L}_0(t)
+
\mathcal{A}_{\rm res}(t)
\Big]
\mathcal{U}(t,s),
\qquad
\mathcal{U}(s,s)
=
\mathcal{I}.
\label{eq:U_evolution}
\ee
The macroscopic response is defined as the ensemble-averaged propagator,
\be
G(t,s)
:=
\mathbb{E}[\mathcal{U}(t,s)].
\label{eq:G_expectation_definition}
\ee
We decompose the stochastic propagator into its mean and fluctuating
parts,
\be
\mathcal{U}(t,s)
=
G(t,s)
+
\mathcal{U}'(t,s),
\qquad
\mathbb{E}[\mathcal{U}'(t,s)]
=
0.
\label{eq:U_mean_fluctuation_decomposition}
\ee
Taking the expectation of Eq.~\eqref{eq:U_evolution} gives
\be
\partial_tG(t,s)
=
\mathcal{L}_0(t)G(t,s)
+
\mathbb{E}
\Big[
\mathcal{A}_{\rm res}(t)\mathcal{U}'(t,s)
\Big],
\label{eq:G_unclosed}
\ee
where
\bes
\mathbb{E}
\big[
\mathcal{A}_{\rm res}(t)G(t,s)
\big]
=
\mathbb{E}[\mathcal{A}_{\rm res}(t)]G(t,s)
=
0
\ees
has been used. The remaining correlation between the stochastic
advection and the fluctuating propagator is the unclosed feedback of the
reservoir-induced transport on the mean response.
Subtracting Eq.~\eqref{eq:G_unclosed} from the pathwise equation
\eqref{eq:U_evolution} yields
\bea
\partial_t\mathcal{U}'(t,s)
&=
\mathcal{L}_0(t)\mathcal{U}'(t,s)
+
\mathcal{A}_{\rm res}(t)G(t,s)
\nonumber\\
&\quad
+
\Big\{
\mathcal{A}_{\rm res}(t)\mathcal{U}'(t,s)
-
\mathbb{E}
\big[
\mathcal{A}_{\rm res}(t)\mathcal{U}'(t,s)
\big]
\Big\},
\label{eq:U_prime}
\eea
with
\be
\mathcal{U}'(s,s)=0.
\ee
The source term
$\mathcal{A}_{\rm res}(t)G(t,s)$ describes the scattering of the mean
response by a single realization of the reservoir-induced transport.
The centered term in braces contains the subsequent interaction of that
stochastic vertex with the already generated response fluctuation.
In the Bourret, or first-order smoothing, approximation
\cite{bourret1962stochastically}, the centered fluctuation--fluctuation
term in braces is neglected and the resulting equation is propagated
with the bare deterministic response generated by $\mathcal{L}_0$. The
DIA-type line-renormalized smoothing used here retains the same
first-order stochastic source but replaces that bare intermediate
propagation by the dressed mean response $G$. In this way, the repeated
feedback that would otherwise enter through higher-order corrections is
incorporated self-consistently into the response line. Duhamel's formula
then gives
\be
\mathcal{U}'(t,s)
\simeq
\int_s^t
G(t,\tau)
\mathcal{A}_{\rm res}(\tau)
G(\tau,s)
\,d\tau.
\label{eq:U_prime_inverted}
\ee
Thus the response fluctuation is represented as the history of the
stochastic advection vertex acting at time $\tau$, with its effect
propagated from $s$ to $\tau$ and subsequently from $\tau$ to $t$ by
the dressed response.
Substituting Eq.~\eqref{eq:U_prime_inverted} into
Eq.~\eqref{eq:G_unclosed}, and using the fact that the deterministic
response operators may be taken through the ensemble expectation, gives
\bea
\mathbb{E}
\Big[
\mathcal{A}_{\rm res}(t)\mathcal{U}'(t,s)
\Big]
&\simeq
\int_s^t
\mathbb{E}
\Big[
\mathcal{A}_{\rm res}(t)
G(t,\tau)
\mathcal{A}_{\rm res}(\tau)
\Big]
G(\tau,s)
\,d\tau.
\label{eq:Ares_Uprime_closure}
\eea
We therefore define the geometric self-energy by
\be
\boldsymbol{\eta}(t,\tau)\Phi
:=
\mathbb{E}
\Big[
\mathcal{A}_{\rm res}(t)
\Big(
G(t,\tau)
\mathcal{A}_{\rm res}(\tau)\Phi
\Big)
\Big].
\label{eq:eta_A_operator_form_app}
\ee
Because $\mathcal{A}_{\rm res}$ is centered, this second moment is the
operator-valued covariance of two reservoir-induced transport vertices
joined by the intermediate response.
Using the Hamiltonian Lie-transport representation of
$\mathcal{A}_{\rm res}$, Eq.~\eqref{eq:eta_A_operator_form_app} becomes
\bea
\boldsymbol{\eta}(t,\tau)\Phi
&=
\mathbb{E}
\Big[
\mathcal{L}_{X_{\psi_{sgs}(t)}}
\Big(
G(t,\tau)
\mathcal{L}_{X_{\psi_{sgs}(\tau)}}\Phi
\Big)
\Big]
\nonumber\\
&=
\mathbb{E}
\Big[
J\Big(
\psi_{sgs}(t),
G(t,\tau)
J\big(\psi_{sgs}(\tau),\Phi\big)
\Big)
\Big].
\label{eq:eta_nested_Jacobian_proof}
\eea
Insertion into Eq.~\eqref{eq:G_unclosed} yields
\be
\partial_tG(t,s)
=
\mathcal{L}_0(t)G(t,s)
+
\int_s^t
\boldsymbol{\eta}(t,\tau)G(\tau,s)
\,d\tau,
\ee
with $G(s,s)=\mathcal{I}$, which is the claimed Dyson--Volterra equation.
Hence, under DIA-type line-renormalized smoothing, the one-loop
self-energy is generated by the covariance of two centered
reservoir-induced Hamiltonian advection vertices, with the dressed
response inserted between them. Its nested Lie-derivative structure
shows that the resulting non-Markovian feedback inherits directly the
geometry of the stochastic symplectic transport.
\end{proof}

%===========================================================%
\section{The Variational Origin of the Linearized Propagator}
\label{app:variational_propagator}
A conceptual difficulty in translating nonlinear fluid dynamics into a
statistical response theory is the apparent paradox of introducing a
linear propagator for a nonlinear cascade. The resolved dressed-transport
equation,
\bes
\partial_t\bar\zeta
+
J(\bar\psi,\bar\zeta)
=
J(\psi_{sgs},\bar\zeta),
\ees
is nonlinear because
\bes
\bar\psi
=
\nabla^{-2}\bar\zeta,
\qquad
\psi_{sgs}
=
\frac12\gamma J(\bar\zeta,r),
\ees
so both the resolved and reservoir-induced advecting fields depend on the
evolving state.
Consequently, the operators
\be
\mathcal{L}_0(t)\Phi
=
-J\big(\bar\psi(t),\Phi\big),
\qquad
\mathcal{A}_{\rm res}(t)\Phi
=
J\big(\psi_{sgs}(t),\Phi\big)
\label{eq:variational_direct_operators}
\ee
cannot be interpreted as defining a linear solution operator for the
original nonlinear state equation. Their role is instead variational:
they act on infinitesimal perturbations along a prescribed realization of
the nonlinear background flow.
The macroscopic propagator $G(t,s)$ is therefore not the solution
operator of the primary nonlinear SPDE. It is the ensemble-averaged
fundamental solution of a tangent-space equation obtained by linearizing
the dynamics along the stochastic solution manifold. This distinction
can be made precise through the functional derivative of the resolved
state with respect to an auxiliary forcing.
Introduce an infinitesimal forcing perturbation $\delta f(\mathbf{x},t)$
on the right-hand side of the resolved equation. The ensemble-averaged
response function is formally defined by
\be
G(t,s)
:=
\mathbb{E}
\left[
\frac{\delta\bar\zeta(t)}
{\delta f(s)}
\right].
\label{eq:G_functional_def}
\ee
Taking the first variation of the nonlinear resolved equation gives
\be
\partial_t\delta\bar\zeta
+
J(\delta\bar\psi,\bar\zeta)
+
J(\bar\psi,\delta\bar\zeta)
=
J(\delta\psi_{sgs},\bar\zeta)
+
J(\psi_{sgs},\delta\bar\zeta)
+
\delta f(t),
\label{eq:tangent_SPDE}
\ee
with
\be
\delta\bar\psi
=
\nabla^{-2}\delta\bar\zeta.
\label{eq:delta_psi_resolved}
\ee
If the full coupled closure is varied, then
$\delta\psi_{sgs}$ also contains the variation induced by the reservoir
perturbation $\delta r$, and Eq.~\eqref{eq:tangent_SPDE} forms the
resolved component of the full block tangent system for
$(\delta\bar\zeta,\delta r)$. In either case, the equation is linear in
the perturbations, while its coefficients are determined by the
unperturbed stochastic trajectory.
For each realization in the underlying probability space, the
background fields
\bes
\bar\zeta(t),
\qquad
\bar\psi(t),
\qquad
r(t),
\qquad
\psi_{sgs}(t)
\ees
are therefore regarded as known time-dependent coefficients along that
trajectory. The tangent equation describes how an infinitesimal
perturbation propagates through this realized turbulent background; it
does not replace the nonlinear evolution that generates the background
itself.
The exact resolved component of the tangent dynamics may be organized as
\bea
\partial_t\delta\bar\zeta
&=
\Big[
\mathcal{L}_0(t)
+
\mathcal{A}_{\rm res}(t)
\Big]
\delta\bar\zeta
\nonumber\\
&\quad
-
J\big(
\nabla^{-2}\delta\bar\zeta,
\bar\zeta
\big)
+
J\big(
\delta\psi_{sgs},
\bar\zeta
\big)
+
\delta f(t).
\label{eq:exact_tangent_decomposition}
\eea
This decomposition separates two classes of tangent interactions. The
operators $\mathcal{L}_0$ and $\mathcal{A}_{\rm res}$ describe the direct
advection of the perturbation by the resolved and
reservoir-induced velocity fields. The remaining terms describe
variations of the transport vertices themselves: the perturbation changes
the resolved streamfunction and, through $\delta\psi_{sgs}$, the
reservoir-induced drift potential.
The Dyson construction of Appendix~\ref{app:Dyson_derivation} is based
on the direct-advection tangent skeleton obtained by neglecting these
explicit vertex-variation terms while retaining the stochastic
reservoir-induced scattering:
\be
\partial_t\delta\bar\zeta
=
\Big[
\mathcal{L}_0(t)
+
\mathcal{A}_{\rm res}(t)
\Big]
\delta\bar\zeta
+
\delta f(t).
\label{eq:direct_advection_tangent_skeleton}
\ee
This is the tangent-space counterpart of the DIA line-renormalized
truncation. The transport vertices are kept fixed along each background
realization, whereas the propagation lines generated by their repeated
action are subsequently dressed self-consistently through the
ensemble-averaged response. Thus the approximation does not linearize
away the stochastic transport; it neglects explicit vertex
renormalization while retaining the memory generated by repeated
reservoir-induced advection.
Under the standard well-posedness assumptions for this linear random
evolution equation, there exists for every realization a path-dependent
fundamental operator $\mathcal{U}(t,s)$ satisfying
\be
\partial_t\mathcal{U}(t,s)
=
\Big[
\mathcal{L}_0(t)
+
\mathcal{A}_{\rm res}(t)
\Big]
\mathcal{U}(t,s),
\qquad
\mathcal{U}(s,s)
=
\mathcal{I}.
\label{eq:variational_U_evolution}
\ee
The forced perturbation is then represented by Duhamel's formula,
\be
\delta\bar\zeta(t)
=
\int_s^t
\mathcal{U}(t,\tau)
\, \delta f(\tau) \, d\tau.
\label{eq:variational_Duhamel}
\ee
Taking the functional derivative with respect to an impulse applied at
time $s$ gives
\bea
G(t,s)
&=
\mathbb{E}
\left[
\frac{\delta}{\delta f(s)}
\int_s^t
\mathcal{U}(t,\tau)
\, \delta f(\tau) \, d\tau
\right]\\
&=
\mathbb{E}[\mathcal{U}(t,s)].
\label{eq:G_U_expectation}
\eea
Equation~\eqref{eq:G_U_expectation} is the variational origin of the
macroscopic response used throughout the Dyson theory. It is exact for
the direct-advection tangent skeleton
\eqref{eq:direct_advection_tangent_skeleton}; the subsequent closure
problem is to determine the ensemble average of the random fundamental
operator.
As shown in Appendix~\ref{app:Dyson_derivation}, the evolution of
$\mathbb{E}[\mathcal{U}(t,s)]$ couples the reservoir-induced advection
operator $\mathcal{A}_{\rm res}$ to the fluctuating part of the pathwise
propagator. Under DIA-type line-renormalized smoothing, this correlation
is reorganized into the geometric self-energy
\bea
\boldsymbol{\eta}(t,\tau)\Phi
&=
\mathbb{E}
\Big[
\mathcal{A}_{\rm res}(t)
\Big(
G(t,\tau)
\mathcal{A}_{\rm res}(\tau)\Phi
\Big)
\Big]
\nonumber\\
&=
\mathbb{E}
\Big[
\mathcal{L}_{X_{\psi_{sgs}(t)}}
\Big(
G(t,\tau)
\mathcal{L}_{X_{\psi_{sgs}(\tau)}}\Phi
\Big)
\Big].
\label{eq:variational_eta}
\eea
The background nonlinearity therefore enters the averaged tangent
response through the statistics of the stochastic coefficients evaluated
along the nonlinear trajectories. Their two-time covariance, joined by
the intermediate dressed response, generates the non-Markovian feedback
on the mean propagator.
The variational response consequently satisfies the Dyson--Volterra
equation
\bea
\partial_tG(t,s)
&=
\mathcal{L}_0(t)G(t,s)
+
\int_s^t
\boldsymbol{\eta}(t,\tau)
G(\tau,s)\,d\tau,
\qquad
t>s,
\label{eq:G_final_Volterra}
\\
G(s,s)
&=
\mathcal{I}.
\eea
This resolves the apparent paradox. The propagator $G(t,s)$ is not a
linear solution operator for the original nonlinear fluid equation. It
is the dressed, ensemble-averaged Green's function governing the
tangent-space response of the stochastic transport skeleton. The
nonlinear background dynamics remain encoded in its time-dependent
random transport vertices and, after averaging, in the geometric
self-energy and the associated memory kernel.
%===========================================================%

%===========================================================%
\section{The Ornstein-Uhlenbeck Reservoir and its Macroscopic Diffusion Operator}\label{App_Macro_Diffusion_Ope}

Before formalizing the closure, we must rigorously evaluate the mean deterministic effect of the stochastic reservoir on the resolved flow. To do so we adopt the formalism of stochastic analysis of SPDEs as framed in standard textbooks \cite{da2006introduction,DZ96,DPZ08}.  The uncoupled fast dynamics of the hidden subgrid kinematic reservoir $r^{(0)}(t)$ are governed by the linear SPDE
\begin{equation}
\partial_t r^{(0)}=-D r^{(0)}+\Sigma\dot W_t,
\label{eq:OU_reservoir_note3}
\end{equation}
where the state space
\be
\mathcal H=L^2_0(\mathbb T^2),
\ee
is the separable Hilbert space of mean-zero square-integrable functions on the doubly periodic normalized torus, with $|\mathbb T^2|=1$, and where $\Sigma\in\mathcal L(\mathcal H)$ denotes the noise amplitude, and $W_t$ is a cylindrical Wiener process.

In the physical Navier--Stokes realization considered here, we specialize the reservoir relaxation operator to
\be
D=\kappa(-\Delta)+\mu_r I,
\qquad
\kappa,\mu_r>0.
\label{Eq_D_no_hperviscous_app}
\ee
For the pullback $L^2$ stability theory of the non-hyperviscous closure, it is sufficient to assume the finite-energy condition
\be
\mathrm{Tr}\left(D^{-1}\Sigma\Sigma^*\right)<\infty.
\label{eq:OU_L2_trace_condition}
\ee
This condition ensures that the OU reservoir is $L^2$-valued and that the stochastic reservoir contributes finite variance at the level required by the augmented-enstrophy estimates.

The derivation of the macroscopic diffusion operator and of the geometric self-energy requires a stronger covariance regularity. Indeed, these calculations involve point-split gradient covariances of the form
\be
\mathbb{E}
\left[
\partial_i r^{(0)}(\mathbf x,t)
\partial_j r^{(0)}(\mathbf x,\tau)
\right],
\ee
which are meaningful only when the invariant covariance has finite gradient energy. We therefore impose, for the present Supplementary Note, the strengthened trace assumption
\be
\mathrm{Tr}\big((-\Delta)Q_\infty\big)<\infty,
\qquad
Q_\infty = \int_0^\infty e^{-sD}\Sigma\Sigma^*e^{-sD}\,ds.
\label{eq:OU_gradient_trace_condition}
\ee
Thus Eq.~\eqref{eq:OU_L2_trace_condition} controls the $L^2$ variance needed for stability, whereas Eq.~\eqref{eq:OU_gradient_trace_condition} controls the differentiated covariance needed for the point-split Wick contractions. Since $-\Delta$ has a strictly positive spectral gap on $L^2_0(\mathbb T^2)$, the gradient-energy condition implies the usual trace-class condition $\mathrm{Tr}(Q_\infty)<\infty$, and hence the linear SPDE \eqref{eq:OU_reservoir_note3} admits the invariant centered Gaussian measure
 \be 
 \mu=\mathcal N(0,Q_\infty). \label{eq:OU_invariant_gaussian_measure} 
\ee

The specialization \eqref{Eq_D_no_hperviscous_app} also makes the spectral meaning of the strengthened trace condition transparent. On the mean-zero torus $L^2_0(\mathbb T^2)$, the operator $D$ is diagonal in the discrete Fourier basis indexed by $k\in\mathbb Z^2\setminus\{0\}$, with eigenvalues \be d_k=\kappa |k|^2+\mu_r. \label{eq:OU_decay_rates_fourier} \ee We further assume that the noise covariance $\Sigma\Sigma^*$ is translation-invariant and thus diagonal in the same Fourier basis, ensuring it commutes with $D$. Under this simultaneous diagonalization, the invariant covariance reduces to 
\be 
Q_\infty = \frac12D^{-1}\Sigma\Sigma^*. 
\label{eq:OU_covariance_commuting} 
\ee 
The strengthened gradient trace condition therefore becomes
 \be 
 \mathrm{Tr}\big((-\Delta)D^{-1}\Sigma\Sigma^*\big)<\infty. 
 \label{eq:OU_gradient_trace_commuting} \ee
  In Fourier variables, this condition reads 
  \be
   \sum_{k\in\mathbb Z^2\setminus\{0\}} \frac{|k|^2}{\kappa |k|^2+\mu_r}\sigma_k^2 <\infty, \label{eq:OU_gradient_trace_fourier} 
   \ee 
  where $\sigma_k^2$ denotes the variance injected by $\Sigma\Sigma^*$ into the $k$-th Fourier mode. 
Since
\be
\frac{|k|^2}{\kappa |k|^2+\mu_r}
\longrightarrow
\frac1\kappa
\qquad
\text{as } |k|\to\infty,
\ee
the condition requires the stochastic forcing covariance to be ultraviolet-regularized. In particular, spatially white forcing, for which $\sigma_k^2$ does not decay with $|k|$, is excluded. This is the precise regularity assumption under which the point-split covariance identities and the deterministic diffusion trace $\mathcal D_\gamma$ are rigorously defined in the non-hyperviscous Navier--Stokes realization.

%--------------------------------------------------------------------------------------------------%
To analyze the macroscopic effect of this stochastic reservoir, we use
the invariant Gaussian measure $\mu$ and the associated Wiener--It\^o
chaos decomposition \cite{janson1997gaussian,nualart2006malliavin,peccati2011wiener}. The
Hilbert space of square-integrable observables over the Gaussian
reservoir admits the orthogonal decomposition
\be
L^2(\mathcal H,\mu)
=
\bigoplus_{n=0}^{\infty}\mathcal H_n,
\ee
where $\mathcal H_0$ consists of deterministic constants, $\mathcal H_1$
is the closure of continuous linear functionals of the Gaussian field,
and $\mathcal H_2$ contains centered quadratic functionals. The
Ornstein--Uhlenbeck transition semigroup preserves this chaos
stratification,
\be
P_t(\mathcal H_n)\subset \mathcal H_n,
\qquad n\ge0.
\ee
Consequently, expectation with respect to $\mu$ is the projection onto
the zeroth chaos. Linear centered observables vanish under this
projection, whereas quadratic observables are evaluated exactly through
the covariance operator $Q_\infty$. This is the functional-analytic counterpart of the Wick--Isserlis
contraction mechanism for Gaussian variables
\cite{isserlis1918formula,janson1997gaussian,nualart2006malliavin}: disconnected quadratic contractions collapse to
deterministic covariance traces, while connected contractions retain the
two-time covariance structure that generates memory.
%--------------------------------------------------------------------------------------------------%

The continuous geometric operator $\mathcal M(r^{(0)})$ acting on the macroscopic vorticity is quadratic in the Gaussian field $r^{(0)}$. Expanding the nested Jacobians into divergence form and taking expectation over $\mu$ isolates the deterministic trace. We define the positive macroscopic diffusion operator $\mathcal D_\gamma$ (evaluated explicitly in Proposition \ref{prop:macroscopic_diffusion_trace} below) by the sign convention
\begin{equation}
-\mathcal D_\gamma\bar\zeta
:=
\mathbb{E}[\mathcal M(r^{(0)})\bar\zeta]
=
-\frac12\nabla\cdot
\Big(
\gamma\,
\mathbb{E}[
\nabla^\perp r^{(0)}\otimes\nabla^\perp r^{(0)}
]
\nabla\bar\zeta
\Big).
\label{Eq_D_gamma_def}
\end{equation}
By the skew-symmetry of the symplectic interaction,
\be
\mathcal B_{\bar\zeta}(r,r)
=
-\mathcal M(r)\bar\zeta.
\ee
Therefore the expected bilinear reservoir self-interaction has the opposite sign:
\be
\mathbb{E}[
\mathcal B_{\bar\zeta}(r^{(0)},r^{(0)})
]
=
+\mathcal D_\gamma\bar\zeta.
\label{Eq_D_gamma_B_identity}
\end{equation}
The following proposition evaluates this operator explicitly from the spatial covariance trace.

%=================================================%
\bprop[Macroscopic Diffusion via Spatial Covariance Trace]
\label{prop:macroscopic_diffusion_trace}
Let the OU reservoir \eqref{eq:OU_reservoir_note3} be posed on 
$\mathcal H=L^2_0(\mathbb T^2)$ with the non-hyperviscous relaxation operator 
\eqref{Eq_D_no_hperviscous_app}. Assume that the stochastic forcing covariance 
$\Sigma\Sigma^*$ is translation-invariant and locally isotropic---meaning it is diagonalized by the discrete Fourier basis with modal variances $\sigma_k^2$ that depend exclusively on the wavenumber magnitude $|k|$---and that it satisfies the ultraviolet regularity condition \eqref{eq:OU_gradient_trace_commuting}. Then the invariant covariance 
\be
Q_\infty=\frac12D^{-1}\Sigma\Sigma^*
\ee
is spatially homogeneous and locally isotropic, and the mean deterministic effect of the stochastic geometric closure is exactly
\begin{equation}
\label{Eq_D_gamma_Doperator}
\mathbb{E}[\mathcal M(r^{(0)})\bar\zeta]
=
-\mathcal D_\gamma\bar\zeta,
\end{equation}
where the macroscopic diffusion operator is given by the algebraic trace formula
\begin{equation}
\mathcal D_\gamma\bar\zeta
=
\frac{1}{8}
\mathrm{Tr}\big(-\Delta D^{-1}\Sigma\Sigma^*\big)
\nabla\cdot(\gamma\nabla\bar\zeta).
\label{Eq_D_gamma_commuting_formula}
\end{equation}
Equivalently, the expected bilinear self-interaction yields the positive diffusion
\begin{equation}
\mathbb{E}[
\mathcal B_{\bar\zeta}(r^{(0)},r^{(0)})
]
=
+\mathcal D_\gamma\bar\zeta.
\label{Eq_B_mean_Dgamma}
\end{equation}
\eprop
%=================================================%
\begin{proof}[Proof of Proposition \ref{prop:macroscopic_diffusion_trace}]
To rigorously evaluate the expected macroscopic operator, we must bridge the infinite-dimensional covariance of the background noise to the local, pointwise evaluation of the spatial gradients.

\textbf{Step 1: Conditioning and Tensor Rearrangement.}
The continuous operator $\mathcal{M}(r^{(0)})$ acting on the macroscopic state $\bar{\zeta}$ expands via the divergence of nested Jacobians. Conditioning on the prescribed resolved state $\bar{\zeta}$, the OU field $r^{(0)}$ is a centered Gaussian with law $\mu$. Expanding the operator:
\bes
\mathcal{M}(r^{(0)})\bar{\zeta} = \frac{1}{2} J\big(\gamma J(r^{(0)}, \bar{\zeta}), r^{(0)}\big) = -\frac{1}{2} \nabla \cdot \Big( \gamma J(r^{(0)}, \bar{\zeta}) \nabla^\perp r^{(0)} \Big)
\ees
Substituting $J(r^{(0)}, \bar{\zeta}) = \nabla^\perp r^{(0)} \cdot \nabla \bar{\zeta}$, we use the standard tensor identity $(\mathbf{A} \cdot \mathbf{B})\mathbf{A} = (\mathbf{A} \otimes \mathbf{A})\mathbf{B}$ to pass the expectation directly onto the noise gradients:
\begin{equation}
\mathbb{E}[\mathcal{M}(r^{(0)})\bar{\zeta}] = -\frac{1}{2} \nabla \cdot \Big( \gamma \mathbb{E}[\nabla^\perp r^{(0)} \otimes \nabla^\perp r^{(0)}] \nabla \bar{\zeta} \Big)
\end{equation}

\textbf{Step 2: Homogeneity and Pointwise Gradient Variance.}
The covariance operator $Q_\infty$ is an integral operator whose action on a test function $\phi(\mathbf{x})$ is defined by $(Q_\infty \phi)(\mathbf{x}) = \int C(\mathbf{x},\mathbf{y}) \phi(\mathbf{y}) d\mathbf{y}$, where $C(\mathbf{x},\mathbf{y}) = \mathbb{E}[r^{(0)}(\mathbf{x}) r^{(0)}(\mathbf{y})]$ is the two-point spatial covariance kernel. By the stated assumptions, $\Sigma\Sigma^*$ is translation-invariant (as is the specialized operator $D$), enforcing spatial homogeneity of the field; thus the kernel depends only on the separation distance, $C(\mathbf{x},\mathbf{y}) = c(\mathbf{x} - \mathbf{y})$.

By the properties of Gaussian fields, differentiating the field is equivalent to differentiating the covariance kernel:
\begin{equation}
\mathbb{E}[\nabla_x r^{(0)}(\mathbf{x}) \otimes \nabla_y r^{(0)}(\mathbf{y})] = -\mathbf{H}_c(\mathbf{x} - \mathbf{y})
\end{equation}
where $\mathbf{H}_c$ is the $2 \times 2$ Hessian matrix of the function $c$. Evaluating this at the exact same spatial coordinate ($\mathbf{x} = \mathbf{y}$), the pointwise gradient covariance tensor becomes $\mathbb{E}[\nabla r^{(0)}(\mathbf{x}) \otimes \nabla r^{(0)}(\mathbf{x})] = -\mathbf{H}_c(0)$.

\textbf{Step 3: Isotropy and the Trace.}
Because the unresolved background noise is locally isotropic, the Hessian matrix evaluated at zero must be strictly proportional to the identity matrix $\mathbf{I}$, such that $-\mathbf{H}_c(0) = q \mathbf{I}$ for some scalar amplitude $q$. 

Taking the matrix trace of both sides in two dimensions yields $2q = \mathrm{Tr}(-\mathbf{H}_c(0))$. The trace of the negative Hessian evaluated at zero is exactly the negative Laplacian of the covariance kernel at the origin, $-\Delta c(0)$. On our normalized torus ($|\mathbb{T}^2|=1$), this pointwise quantity corresponds mathematically to the global operator trace of $(-\Delta) Q_\infty$. Therefore:
\begin{equation}
q = \frac{1}{2} \mathrm{Tr}\big( -\Delta Q_\infty \big)
\end{equation}
Because the orthogonal gradient operator $\nabla^\perp = (-\partial_y, \partial_x)$ merely rotates the coordinate axes, it preserves the diagonal variance under isotropy. Thus, the orthogonal gradient tensor evaluates to:
\begin{equation}
\mathbb{E}[\nabla^\perp r^{(0)} \otimes \nabla^\perp r^{(0)}] = \frac{1}{2} \mathrm{Tr}\big( -\Delta Q_\infty \big) \mathbf{I}
\end{equation}

\textbf{Step 4: The Macroscopic Diffusion Operator.}
Substituting this strictly scalar matrix back into the divergence operator, we extract the deterministic macroscopic diffusion operator:
\bea
\mathbb{E}[\mathcal{M}(r^{(0)})\bar{\zeta}] &= -\frac{1}{2} \nabla \cdot \left( \gamma \left[ \frac{1}{2} \mathrm{Tr}\big( -\Delta Q_\infty \big) \mathbf{I} \right] \nabla \bar{\zeta} \right)\nonumber\\
& = -\frac{1}{4} \mathrm{Tr}\big( -\Delta Q_\infty \big) \nabla \cdot \big( \gamma \nabla \bar{\zeta} \big) := -\mathcal{D}_\gamma \bar{\zeta}.
\eea
Since $D = \kappa(-\Delta) + \mu_r I$ and $\Sigma \Sigma^*$ is translation-invariant, both operators are diagonal in the Fourier basis; hence they commute. The invariant covariance integral simplifies algebraically to $Q_\infty = \frac{1}{2} D^{-1} \Sigma \Sigma^*$. Substituting this algebraic reduction directly into the trace operation yields the explicit operator:
\begin{equation}
\mathcal{D}_\gamma \bar{\zeta} = \frac{1}{8} \mathrm{Tr}\big( -\Delta D^{-1} \Sigma \Sigma^* \big) \nabla \cdot \big( \gamma \nabla \bar{\zeta} \big).
\end{equation}

Finally, by the skew-symmetry of the symplectic interaction, $\mathcal B_{\bar\zeta}(r,r)=-\mathcal M(r)\bar\zeta$. Therefore, the expected bilinear self-interaction satisfies:
\be
\mathbb{E}\big[\mathcal B_{\bar\zeta}(r^{(0)},r^{(0)})\big] = -\mathbb{E}\big[\mathcal M(r^{(0)})\bar\zeta\big] = +\mathcal D_\gamma\bar\zeta.
\ee
This completes the proof.
\end{proof}
%----------------------------------------------------------------------------------------------------------------------------------------------%

%----------------------------------------------------------------------------------------------------------------------------------------------%
The rigorous evaluation of the macroscopic diffusion operator
$\mathcal D_\gamma$ via the trace formula
\eqref{Eq_D_gamma_commuting_formula} serves a dual purpose: it
mathematically formalizes the macroscopic diffusion limit, and it
provides the analytic foundation required for the Wick contractions \cite{wick1950evaluation,peccati2011wiener}
deployed in the derivation of the geometric self-energy. In evaluating
the memory kernel from the third-order Volterra--Picard expansion, the
statistical expectation must be pulled through the spatial derivatives of
nested Jacobians. The finite gradient-energy trace established above,
\be
\mathrm{Tr}\big((-\Delta)Q_\infty\big)<\infty,
\ee
guarantees that the space-time point-split covariances appearing in the
main derivation are rigorously well-defined in the diagonal covariance
sense.

Because the relaxation operator
\be
D=\kappa(-\Delta)+\mu_r I
\ee
imposes wavenumber-dependent temporal decay rates
\be
d_k=\kappa |k|^2+\mu_r,
\ee
the space-time covariance does not generally factorize into a scalar
temporal covariance times a fixed spatial covariance. Instead, under the
stationary OU measure, it is represented by the two-time spatial
covariance kernel
\be
C_{t,\tau}(\mathbf x-\mathbf y)
:=
\mathbb{E}
\big[
r^{(0)}(\mathbf x,t)r^{(0)}(\mathbf y,\tau)
\big].
\label{eq:OU_two_time_covariance_kernel}
\ee
In the commuting Fourier-diagonal case,
\eqref{eq:OU_covariance_commuting}, the corresponding modal covariance is
\bea
& C_{\mathbf q}(t,\tau)
=
e^{-d_{\mathbf q}|t-\tau|}
(Q_\infty)_{\mathbf q}
=
e^{-d_{\mathbf q}|t-\tau|}
\frac{\sigma_{\mathbf q}^2}{2d_{\mathbf q}},\\
& d_{\mathbf q}=\kappa|\mathbf q|^2+\mu_r.
\label{eq:OU_two_time_covariance_fourier}
\eea

The point-split gradient covariance is therefore evaluated as
\be
\mathbb{E}
\Big[
\partial_{x_i}r^{(0)}(\mathbf x,t)
\partial_{y_j}r^{(0)}(\mathbf y,\tau)
\Big]
=
\partial_{x_i}\partial_{y_j}
C_{t,\tau}(\mathbf x-\mathbf y).
\label{eq:point_split_justification}
\ee
Equivalently, on the diagonal,
\be
\mathbb{E}
\Big[
\partial_{x_i}r^{(0)}(\mathbf x,t)
\partial_{x_j}r^{(0)}(\mathbf x,\tau)
\Big]
=
\left[
\partial_{x_i}\partial_{y_j}
C_{t,\tau}(\mathbf x-\mathbf y)
\right]_{\mathbf y=\mathbf x}.
\label{eq:point_split_diagonal_covariance}
\ee
Writing $z=\mathbf x-\mathbf y$, this is the same as
\be
\partial_{x_i}\partial_{y_j}
C_{t,\tau}(\mathbf x-\mathbf y)
=
-\partial_{z_i}\partial_{z_j}C_{t,\tau}(z).
\label{eq:point_split_negative_hessian}
\ee
In Fourier variables, each differentiated covariance slot is represented
by the tensorial multiplier
\be
\partial_{x_i}\partial_{y_j}
C_{t,\tau}(\mathbf x-\mathbf y)
=
\sum_{\mathbf q\neq0}
q_iq_j
C_{\mathbf q}(t,\tau)
e^{i\mathbf q\cdot(\mathbf x-\mathbf y)}.
\label{eq:OU_point_split_fourier_multiplier}
\ee
The strengthened trace condition
\eqref{eq:OU_gradient_trace_condition}, equivalently
\eqref{eq:OU_gradient_trace_fourier} in the commuting case, is precisely
the statement that these differentiated covariance weights are summable
in the sense required by the point-split Wick contractions \cite{wick1950evaluation,peccati2011wiener}.

At equal times, $t=\tau$, the kernel $C_{t,t}$ reduces to the invariant
spatial covariance kernel $c(\mathbf x-\mathbf y)$ associated with
$Q_\infty$.

Consequently, this functional analytic framework gives precise
mathematical meaning to the algebraic organization of the
Dyson--Volterra series. 
When the Wick--Isserlis theorem for Gaussian moments
\cite{isserlis1918formula,peccati2011wiener,janson1997gaussian}
is applied to the quartic moments of the subgrid reservoir, the
disconnected intra-time pairings factorize into expectations of the
bilinear operator at equal times.  By the identity established in Eq.~\eqref{Eq_B_mean_Dgamma}, these
pairings collapse into the positive deterministic spatial diffusion
\be
\mathbb{E}
\Big[
\mathcal B_{\bar\zeta(s)}
\big(r^{(0)}(s),r^{(0)}(s)\big)
\Big]
=
+\mathcal D_\gamma\bar\zeta(s).
\label{eq:disconnected_diffusion_identity}
\ee
The connected cross-contractions, by contrast, exchange differentiated
reservoir slots across distinct historical times through
\eqref{eq:point_split_justification}--\eqref{eq:OU_point_split_fourier_multiplier}.
These are the contractions that generate the causal, phase-scrambling
sweeping self-energy $\mathbf\Sigma_{\rm sweep}$. Thus the DIA-like
Volterra structures used below are grounded in the spectral regularity of
the underlying SPDE, rather than introduced as independent statistical
assumptions.
%----------------------------------------------------------------------------------------------------------------------------------------------%

%===========================================================%
\section{Dimensional Consistency and the Absence of Anomalous Scaling}
\label{app:anomalous_scaling}

In Section \ref{sec:scaling_preliminaries}, to extract the inertial-range energy flux, we invoked the dimensional hypothesis that the effective transport velocity $\mathbf{u}_{\rm trs}$ and the bare macroscopic velocity $\bar{\mathbf{u}}$ share the same scaling exponent. We now formally demonstrate that within the Symplectic Geometric Closure, this is not an independent phenomenological assumption. It is a dynamical consequence of the Volterra response topology.

\bt[Preservation of Inertial-Range Scaling Exponents]
\label{thm:no_anomalous_scaling}
Let the macroscopic cascade be governed by the effective transport velocity $\mathbf{u}_{\rm trs} = \bar{\mathbf{u}} - \mathbf{u}_{sgs}$, with the emergent subgrid drift $\psi_{sgs}$ structurally determined by the causal Volterra response topology established in Section \ref{sec:topological_uniqueness}. Under the hypothesis of a stationary, scale-invariant inertial range, the bare macroscopic velocity $\bar{\mathbf{u}}$, the subgrid drift velocity $\mathbf{u}_{sgs}$, and the effective transport velocity $\mathbf{u}_{\rm trs}$ must strictly share the same scaling exponent. Consequently, the symplectic subgrid reservoir cannot generate an anomalous scaling dimension.
\et

\begin{proof}
\textbf{Step 1: The Dimensional Algebraic Constraint.}
We begin with the response representation of the emergent drift potential, justified by the topological uniqueness of the symplectic interaction (Section \ref{sec:topological_uniqueness}). This structural constraint models the subgrid streamfunction as the causal accumulation of past nonlinear macroscopic forcings transported by the dressed propagator $G(t,s)$:
\be
\psi_{sgs}(t) = \mathcal{K} \int_0^t G(t,s) J\big(\bar{\psi}(s), \bar{\zeta}(s)\big) ds.
\label{eq:app_ansatz}
\ee
Here, the operator $\mathcal{K}$ represents the topological mapping from the macroscopic interaction history to the emergent subgrid drift. We evaluate the dimensional scaling of this relation for a characteristic wavenumber $k$ in the inertial range. 

%------------------------------------------------%
The bare macroscopic Jacobian contains one derivative acting on each advected factor. Equivalently, using $\bar{\mathbf u}=\nabla^\perp\bar{\psi}$ and $\bar{\zeta}=\nabla^2\bar{\psi}$, it scales as
\be
J_k\sim \|\nabla\bar{\psi}\|_k\,\|\nabla\bar{\zeta}\|_k
\sim \bar u_k\,(k\bar{\zeta}_k)
\sim \bar u_k\,(k^2\bar u_k)
=
k^2\bar u_k^2,
\ee
where $\bar u_k:=\|\bar{\mathbf u}\|_k$ is the characteristic amplitude of the bare resolved velocity. The temporal integration over the dressed propagator $G(t,s)$ dimensionally evaluates to the triad relaxation time, $\theta_k = \int_0^\infty G(k,\tau)d\tau$. 
%------------------------------------------------%

Crucially, while the time integral of the macroscopic Jacobian evaluates dimensionally to a vorticity-like tendency ($T^{-1}$), the dual-transport equation interpretation of the Symplectic Geometric Closure mandates that $\psi_{sgs}$ acts precisely as a divergence-free streamfunction ($L^2 T^{-1}$). To ensure dimensional consistency within the inertial range, the mapping operator $\mathcal{K}$ must therefore inherently possess a spatial scaling of $L^2$. In Fourier space, this dictates that $\mathcal{K}$ scales as $\mathcal{K}(k) \sim k^{-2}$. We emphasize that the argument does not require $\mathcal{K}$ to coincide exactly with the inverse Laplacian. The proof uses only the inertial-range dimensional behavior $\mathcal{K}(k) \sim k^{-2}$. Any operator belonging to the same scaling class yields the same exponent-counting argument below.

Applying this dimensional scaling, the response ansatz dictates that the subgrid streamfunction scales dimensionally as:
\be
\psi_{sgs,k} \sim \mathcal{K}(k) \, \theta_k \, J_k \sim k^{-2} \, \theta_k \, (k^2 \bar{u}_k^2) = \theta_k \bar{u}_k^2.
\ee
The corresponding subgrid drift velocity, $\mathbf{u}_{sgs} = \nabla^\perp \psi_{sgs}$, acquires an additional gradient factor $k$, yielding:
\be
u_{sgs,k} \sim k \, \theta_k \bar{u}_k^2.
\label{eq:app_u_sgs_1}
\ee

\textbf{Step 2: The Decorrelation Timescale.} As established by the geometric self-energy theory, the triad relaxation time $\theta_k$ is governed by the macroscopic strain of the actual advecting field. In the inertial limit, the governing equation $\partial_t \bar{\zeta} + J(\bar{\psi}-\psi_{sgs}, \bar{\zeta}) = 0$ identifies the advecting field strictly as the effective transport velocity $\mathbf{u}_{\rm trs}$. Thus, the decorrelation time is kinematically locked to:
\be
\theta_k \sim (k \, u_{{\rm trs},k})^{-1}.
\ee
Substituting this timescale back into Eq.~\eqref{eq:app_u_sgs_1} cancels the wavenumber $k$ exactly, yielding the fundamental algebraic constraint of the closure:
\be
u_{sgs,k} \sim \frac{\bar{u}_k^2}{u_{{\rm trs},k}}.
\label{eq:app_fundamental_constraint}
\ee

\textbf{Step 3: Exponent Analysis and Dimensional Homogeneity.} Assume the system admits a pure, power-law inertial range lacking multifractal or logarithmic corrections. The velocity fields are thus characterized by respective scaling exponents $a, b,$ and $c$:
\be
\bar{u}_k \sim k^{-a}, \qquad u_{{\rm trs},k} \sim k^{-b}, \qquad u_{sgs,k} \sim k^{-c}.
\ee
Inserting these power laws into the fundamental algebraic constraint Eq.~\eqref{eq:app_fundamental_constraint} implies:
\be
k^{-c} \sim \frac{k^{-2a}}{k^{-b}} \implies c = 2a - b.
\ee
Simultaneously, the definition of the effective transport velocity imposes the exact kinematic relation $
\mathbf{u}_{\rm trs} = \bar{\mathbf{u}} - \mathbf{u}_{sgs}.$
For a scale-invariant cascade to maintain a nontrivial dressed state over an extended inertial range, the constituent fields of this linear combination must be dimensionally homogeneous. If $a \neq c$, one component asymptotically dominates the effective transport velocity at disparate scales, implying that either the emergent subgrid drift or the bare macroscopic velocity becomes asymptotically irrelevant. The closure would then cease to represent a genuine dressed transport field in the inertial range.

Thus, physical and structural consistency demands:
\be
a = c = b.
\ee
Substituting $a = b = c$ into the algebraic constraint $c = 2a - b$ yields a trivial identity ($a = 2a - a$), proving that, under the hypothesis of a stationary scale-invariant cascade, this homogeneous scaling is the unique consistent solution. 

In renormalization-group language, the geometric reservoir behaves as a marginal dressing of the inertial-range cascade: while it renormalizes fluctuation amplitudes, decorrelation times, and modal phases, it does not alter the scaling dimension of the transport velocity. The bare velocity, the subgrid drift, and the effective transport velocity must therefore belong to the exact same inertial-range scaling class.
\end{proof}

}
\bibliography{NGM_paper_symplectic}
\end{document}